\documentclass[a4paper,11pt,oneside]{article}

\usepackage[left=2cm,right=2cm,top=2.5cm,bottom=3cm]{geometry}
\usepackage[centertags]{mathtools}
\usepackage{amssymb, mathrsfs, amsthm}
\usepackage{amsmath}
\usepackage{graphicx}
\usepackage{multicol}
\usepackage{multirow}
\usepackage[svgnames]{xcolor}
\usepackage{csquotes}
\usepackage{bbm}
\usepackage{bm}
\usepackage{mwe}
\usepackage{pdflscape}
\usepackage[noblocks,affil-sl]{authblk} 
\usepackage{caption}
\usepackage{subcaption}
\usepackage{hhline}
\usepackage{adjustbox}
\usepackage{dsfont}
\usepackage{booktabs}
\usepackage{colortbl}
\usepackage[authoryear]{natbib}
\usepackage{cases}
\usepackage{soul}
\usepackage{cancel}
\usepackage{makecell}
\usepackage{float}
\usepackage{enumerate}
\usepackage{xfrac} 
\mathtoolsset{showonlyrefs}
\usepackage{hyperref}
\usepackage{verbatim}
\allowdisplaybreaks 

\counterwithin{figure}{section}
\counterwithin{table}{section}
\numberwithin{equation}{section}

\theoremstyle{plain}
\newtheorem{thm}{Theorem}[section]

\newtheorem{prop}[thm]{Proposition}
\newtheorem{cor}[thm]{Corollary}
\newtheorem{defn}[thm]{Definition}

\newtheorem{remark}[thm]{Remark}

\newtheorem{example}{Example}[section]

\makeatletter
\def\@makefnmark{\hbox{\@textsuperscript{\normalfont(\@thefnmark)}}}
\makeatother 

\newcommand{\R}{\mathbb{R}}
\newcommand{\bmu}{\bm{\mu}}
\newcommand{\diag}{\text{diag}}
\newcommand{\bmhatTheta}{\mathbf{\hat\Theta}}
\newcommand{\bma}{\mathbf{a}}
\newcommand{\bmb}{\mathbf{b}}
\newcommand{\bme}{\mathbf{e}}
\newcommand{\bmS}{\mathbf{S}}
\newcommand{\bmI}{\mathbf{I}}
\newcommand{\bmZ}{\mathbf{Z}}
\newcommand{\bmW}{\mathbf{W}}
\newcommand{\bmpi}{\bm{\pi}}
\newcommand{\bmtheta}{\bm{\theta}}
\newcommand{\bmtSigma}{\mathbf{\tilde\Sigma}}
\newcommand{\bmSigma}{\mathbf{\Sigma}}
\newcommand\de{\mathrm{d}}

\DeclareMathOperator*{\argmax}{arg\,max} 

\begin{document}
\author[1]{Katia Colaneri  \thanks{katia.colaneri@uniroma2.it}}
\author[2]{Federico D'Amario \thanks{ federico.d’amario@bankofengland.co.uk}}
\author[3]{Daniele Mancinelli\thanks{Corresponding author: daniele.mancinelli@polimi.it}}
\affil[1]{Department of Economics and Finance, University of Rome Tor Vergata.}
\affil[2]{Bank of England.}
\affil[3]{Department of Mathematics, Politecnico di Milano.}
\title{Carbon-Penalised Portfolio Insurance Strategies in a Stochastic Factor Model with Partial Information}
\date{\today}
\maketitle
\vspace{-0.5cm}
\bigskip
\begin{abstract}
\noindent  
We investigate optimal proportional portfolio insurance (PPI) strategies aimed at reducing exposure to carbon intensive stocks. PPI strategies enable investors to mitigate downside risk while retaining the potential for upside gains. In this paper we determine the PPI strategies to maximise the expected utility of the terminal cushion, where the terminal cushion is penalised proportionally to the realised volatility of stocks issued by firms operating in carbon-intensive sectors. We model the risky assets’ dynamics using geometric Brownian motions whose drift rates are modulated by an unobservable common stochastic factor to capture market-specific or economy-wide state variables that are typically not directly observable. Using the classical stochastic filtering theory, we formulate a suitable optimisation problem and solve it for the CRRA utility function. We characterise optimal carbon-penalised PPI strategies and optimal value functions under full and partial information. 
We also carry a numerical analysis  showing that the proposed strategy reduces carbon-emissions intensity without compromising financial performance.
\end{abstract}

\section{Introduction}

As recently documented in several studies, including \cite{hartzmark2019investors}, \cite{lagerkvist2020preferences}, and \cite{anquetin2022scopes}, stakeholders around the world increasingly perceive climate change as a major global risk. As a result, institutional investors progressively incorporate environmental, social and governance (ESG) considerations into portfolio management. Among the different ESG dimensions, particular attention has been devoted to climate-related risks associated with firms operating in carbon-intensive sectors.

In this context, investors commonly evaluate the exposure to assets in their portfolios. 
Monitoring such exposure has become increasingly relevant not only from a sustainability perspective but also from a financial risk perspective, since carbon intensive firms with high carbon intensity may be more exposed to regulatory changes, technological transitions, and market repricing associated with the transition to a low-carbon economy.

Large institutional investors have already incorporated these considerations into their investment policies. For example, as reported by \cite{peng2024optimal}, the Government Pension Investment Fund has allocated \$163 trillion yen to passive ESG index products, while the California Public Employees' Retirement System follows a ``social change investment'' approach guided by ESG principles.\\ Although ESG encompasses several dimensions, this article focuses specifically on \emph{carbon risk}, which carries relevant regulatory, market and reputational implications for institutional investors. Measuring firms’ exposure to carbon risk is therefore an important step in portfolio construction. Two widely used indicators are the Brown–Green Score introduced by \cite{gorgen2020carbon} and carbon intensity\footnote{Carbon intensity is defined as the ratio between a firm’s greenhouse-gas emissions and its revenues.} proposed by \cite{hellmich2021carbon}; in this article we adopt the latter measure.

\bigskip

Against this backdrop, \emph{Proportional Portfolio Insurance} (PPI) strategies provide a natural framework for combining downside protection with dynamic participation in financial markets. PPI strategies were originally introduced following the 1973–1974 market collapse by \cite{rubinstein1976evolution} and \cite{brennan1976pricing} and have since become widely used by institutional investors such as pension funds, insurance companies and mutual funds (see, e.g., \cite{temocin2018constant} and \cite{di2024pension}). Their main objective is to guarantee a predetermined level of wealth at the end of the investment horizon while allowing participation in favourable market movements (see \cite{grossman1989portfolio} and \cite{basak2002comparative}). The strategy relies on a dynamic allocation between a risky reference portfolio and a reserve asset. The allocation rule is governed by the concept of the \emph{cushion}, defined as the difference between the current portfolio value and a protection floor representing the minimum wealth level to be preserved. The exposure to the risky portfolio is proportional to the cushion through a time-varying coefficient known as the \emph{multiplier}. When the cushion is positive, the strategy allows participation in equity market gains, whereas if the cushion reaches zero the portfolio is entirely reallocated to the reserve asset, thereby ensuring the protection of the guaranteed wealth level.

\bigskip
Despite the growing literature on sustainable portfolio allocation, relatively little attention has been devoted to the integration of carbon-related risks within portfolio insurance strategies. In particular, the interaction between downside protection mechanisms and exposures to carbon-sensitive assets has, to the best of our knowledge, not yet been systematically studied in a continuous-time stochastic control framework.

In this paper we investigate how carbon-related considerations can be incorporated into the construction of portfolio insurance strategies. We propose a modified PPI framework in which the investor maximises the expected utility of a \emph{carbon-penalised terminal cushion}. Specifically, the terminal cushion is adjusted by a penalty term proportional to the realised variance of stocks issued by firms operating in carbon-intensive sectors. The proportionality parameter represents the investor’s degree of \emph{carbon aversion} and determines the strength of the penalisation applied to exposures to carbon-intensive assets. A similar idea has recently been explored in \cite{colaneri2025design} within a general portfolio optimisation framework. Our modelling approach differs from a large part of the existing literature on low-carbon portfolio construction (see, for instance, \cite{andersson2016hedging}, \cite{bolton2022net}, and \cite{LeGuenedalRoncalli2023}). In many existing contributions, sustainability considerations are incorporated through exclusion rules or explicit constraints on the whole carbon intensity of the portfolio. By contrast, in our framework sustainability considerations enter directly into the investor’s objective function through a penalty term. This modelling choice allows the optimisation problem to account jointly for financial performance and carbon exposure. In particular, a carbon-intensive asset may still be included in the portfolio if its risk–return profile is sufficiently favourable to compensate for the additional penalty. Such flexibility is particularly relevant for portfolio insurance strategies, whose primary objective remains the preservation of the guaranteed wealth level.

\bigskip

{From an actuarial perspective, PPI strategies are widely used in the accumulation phase of defined contribution pension schemes and in capital-protected investment products, where benefits depend on investment performance while downside protection remains essential. In this setting, our framework delivers an optimal dynamic allocation rule, determined jointly by the multiplier and the composition of the risky portfolio, that balances participation in market returns with the protection of the guarantee. The introduction of a carbon penalty can also be interpreted from a prudential perspective. Supervisory authorities increasingly emphasise the importance of incorporating climate-related risks into insurers’ risk-management frameworks. In particular, \cite{EIOPA2021ORSAClimate} highlights that insurers should assess the potential impact of climate-related transition risks within their Own Risk and Solvency Assessment (ORSA). Within this perspective, the carbon penalty in our model can be interpreted not merely as an ESG preference parameter but as a forward-looking risk adjustment reflecting the potential volatility and repricing associated with exposures to carbon-intensive sectors. Since portfolio insurance strategies are designed to preserve a protection floor, limiting excessive exposure to assets that are more sensitive to transition risk contributes to the robustness of the guaranteed payoff and is therefore consistent with the Solvency II objective of policyholder protection.}

\bigskip

Our contributions can be summarised as follows. First, we introduce a novel framework for incorporating carbon-related risk considerations into portfolio insurance strategies. We extend the classical proportional portfolio insurance setting by introducing a carbon-penalised terminal cushion, in which exposures to carbon-intensive assets are penalised through a term proportional to their realised variance. {This modeling approach embeds carbon considerations directly into the investor’s objective function, so that the optimal allocation endogenously reflects both financial performance and exposure to carbon-intensive assets. From a modeling perspective, we assume that the rate of return of risky assets is affected by a latent process. This specification is motivated by the fact that returns of risky assets are hard to estimate and typically depend on many risk macroeconomic, political, and environmental factors that cannot be directly measured (see, among others, \cite{lakner1995utility, lakner1998optimal}, \cite{xia2001learning}, \cite{brendle2006portfolio} and \cite{bjork2010optimal}). The choice of the dynamics has been made to allow for analytical tractability of the optimisation problem and interpretability of optimal carbon-penalised PPI strategies.}
Second, we characterise the optimal design of the carbon-penalised PPI strategy in a continuous-time stochastic control framework. In particular, we determine jointly the optimal multiplier and the optimal composition of the risky reference portfolio that maximise the expected CRRA utility of the terminal carbon-penalised cushion. By modelling this latent factor as an Ornstein–Uhlenbeck process and applying stochastic filtering techniques, we derive the optimal strategy under incomplete information and obtain explicit expressions for the associated value function. We perform several comparisons of the optimal value and PPI strategies under full information (i.e. pretending that an insurer is able to perfectly observe the factor process) and under partial information. In particular, we quantify the value of information superiority by measuring the loss of utility arising from the inability to observe the latent factor directly.

\bigskip

The remainder of the paper is organised as follows. Section \ref{sec:review} discusses further references. Section \ref{sect:model_setting} introduces the model setting. Section \ref{sect:PPI} presents the carbon-penalised PPI strategy. Section \ref{sect:opt_problem_full_info} studies the optimisation problem under full information, while Section \ref{sect:opt_partial_info} considers the partial-information case. Section \ref{sect:num_experiments} reports the numerical analysis and Section \ref{sect:conclusions} concludes. Proofs of all results are collected in the Appendix.

\subsection{Literature review}\label{sec:review}
This article relates to different strands of the literature that study how sustainability considerations, typically measured through carbon emissions or ESG indicators, can be incorporated into portfolio optimisation alongside the traditional objectives of return maximisation and risk control. From a methodological perspective, three main approaches can be identified.

The first approach consists of excluding assets that do not meet predetermined sustainability criteria from the investment universe. A pioneering contribution is \cite{andersson2016hedging}, who propose excluding stocks with high carbon intensity and selecting the remaining assets so as to minimise tracking error relative to a benchmark portfolio. Their results show that the portfolio’s carbon footprint can be substantially reduced while maintaining negligible tracking error. \cite{bolton2022net} extend this idea by constructing portfolios consistent with the targets of the Paris Agreement, achieving a gradual reduction in carbon exposure while preserving a close tracking of major market indices.

The second approach keeps the investment universe unchanged but imposes sustainability constraints at the portfolio level. In this context, \cite{LeGuenedalRoncalli2023} study an optimisation problem in which tracking error relative to a benchmark is minimised subject to a constraint on carbon risk. Similarly, \cite{de2023esg} construct portfolios satisfying sustainability requirements based on both carbon intensity and ESG ratings. Their empirical analysis shows that portfolios with lower ESG scores may initially exhibit higher performance, while higher-ESG portfolios tend to perform better over longer horizons. \cite{BoltonEtAl2024} propose a framework for constructing equity portfolios aligned with net-zero targets by imposing a time-varying carbon budget within a tracking-error minimisation problem, obtaining substantial reductions in carbon exposure while preserving diversification.

The third approach, which is closer to our contribution, incorporates sustainability directly into investor preferences, without imposing it through explicit constraints. For instance, \cite{pastor2021sustainable} develop an equilibrium model in which investors derive utility from holding sustainable assets, leading to a lower cost of capital for environmentally friendly firms. Similarly, \cite{escobar2022multivariate} introduce a multivariate CRRA utility framework allowing investors to assign different levels of risk aversion to green and brown assets, showing that stronger aversion to brown assets significantly increases optimal allocations to sustainable investments. In contrast to the above approaches, this paper incorporates carbon-related risk considerations into the design of  downside protection portfolio insurance strategies by embedding a carbon penalty directly into the optimisation problem.

\section{The market setup}\label{sect:model_setting}
Let $\left(\Omega,\mathbb{G},\mathbb{P}\right)$ be a fixed probability space and $T$ a finite time horizon coinciding with the terminal time of an investment. We also introduce a $\mathbb{P}$-complete and right-continuous filtration $\mathbb{G}=\left\lbrace\mathcal{G}_t \right\rbrace_{t\in[0,T]}$ representing the global information flow, and we assume that all processes below are $\mathbb{G}$-adapted. We consider a financial market model consisting of $n$ stocks with $n$-dimensional price processes $\bmS=\left\lbrace\bmS_t\right\rbrace_{t\in[0,T]}$ 
where $\bmS_t=(S^1_t,\,\dots,\,S^n_t)^\top$ for all $t \in [0,T]$, and one risk-free asset $S^0$, that are traded continuously on $[0,T]$. The dynamics of the risk-free are given by 
\begin{equation}\label{eq:cash_account}
\de S^0_t=rS^0_t\de t,\quad S^0_0=1, 
\end{equation}
where $r>0$ denotes the constant risk-free interest rate. The price dynamics of the risky assets $\bmS$ are given by
\begin{equation}\label{eq:risk_assets_dyn}
\de\bmS_t=\diag\left(\bmS_t\right)\left(\bmu_t\de t+\bmSigma_{\bmS}\de\bmW^{\bmS}_t\right),
\end{equation}
where $\bmS_0=(S^1_0,\,\dots,\,S^n_0)^\top$ and $S^i_0\in \R_+$ for all $i=1, \dots, n$. In equation \eqref{eq:risk_assets_dyn}, $\bmSigma_{\bmS}=\diag\left(\sigma_1,\dots,\sigma_n\right)$, with $\sigma_i>0$ for every $i=1,\dots,n$, and $\bmW^\bmS=\{\bmW^{\bmS}\}_{t\in[0,T]}$ is a standard $\mathbb{G}$-Brownian motion in $\R^n$ with correlated components, namely $\de\langle W^{\bmS}_{i},\,W^{\bmS}_{j}\rangle_t=\rho_{i,j}\de t,$ for constant correlation coefficients $\rho_{i,j}\in[-1,1]$, such that $\rho_{i,j}=\rho_{j,i}$, for every $i,\,j=1,\dots,n$, and $\rho_{i,i}=1$, for every $i=1,\dots,n$. 
Moreover, $\bm{\mu}_t$ is stochastic and unobservable. This assumption is motivated by the fact that drifts of financial assets are rarely constant and subject to random fluctuations. In particular, we assume that the drift process $\bmu=\left\lbrace \bmu_t\right\rbrace_{t\in[0,T]}$ is of the form $\bmu_t=\bmu(Y_t)=\bm{a}Y_t+\bmb$ for every $t\in[0,T]$, with $\bm{a}\in \R^n$ and $\bmb \in \R^n$, where $Y=\{Y_t\}_{t\in[0,T]}$ is the common unobservable factor process. Indeed, $Y_t$ can represents macro-financial states that are hard to observe cleanly over time. Typical examples include the business cycle, monetary and financial conditions, credit and funding conditions, systemic liquidity, inflation pressures, and transition-to-net-zero emissions pressure. While these variables have observable proxies, none of them is directly observed in a noise-free way. Consequently, a partial information framework is necessary to model these state processes. In this paper, we model the common latent factor $Y$ as an Ornstein-Uhlenbeck (OU) process, namely, 
\begin{equation}\label{eq:latent_factor}
\de Y_t=\left(\lambda Y_t+\beta\right)\de t+\sigma_Y\de W_t^Y, \quad Y_0\sim N\left(\Gamma_0,P_0\right),
\end{equation}
with $\lambda,\,\beta\in\R,\,\sigma_Y>0$. Here, $W^Y=\left\lbrace W_t^Y\right\rbrace_{t\in[0,T]}$ is a standard one-dimensional $\mathbb{G}$-Brownian motion correlated with $\bmW^{\bmS}$ with $\de\langle W^{Y},\,W_i^{\bmS}\rangle_t=\rho_{i,Y}\de t$, where $\rho_{i,Y}\in[-1,1]$ for every $i=1,\dots,n$. The OU choice captures the cyclical, mean-reverting nature of the above macro-financial variables while preserving the linear–Gaussian structure that makes filtering under partial information analytically tractable (see Section \ref{sect:opt_partial_info}). 
{\begin{remark}
In this work, partial information refers to the unobservability of the expected return (drift), which represents a key source of informational incompleteness in portfolio choice. Indeed, while volatility may be stochastic, in continuous-time settings it can be identified from the quadratic variation of observed price paths and is therefore naturally adapted to the filtration generated by asset prices.\footnote{Alternative approaches to modelling price processes under partial information exist. For instance, \cite{frey2001nonlinear} propose a framework based on high-frequency data in which stock prices are modelled through marked point processes and nonlinear filtering techniques are used to estimate asset price volatility.} This modelling choice is standard in the portfolio-optimization literature and allows us to capture learning about expected returns while preserving analytical tractability. More generally, our specification should be interpreted as a parsimonious framework aimed at highlighting the analysis of carbon-penalised PPI strategies, while still incorporating more realistic features than models in which the rate of return is assumed to be constant and directly observable.
\end{remark}}
Stocks are assumed to be issued by firms with different levels of carbon emissions, measured by \textit{carbon intensity}. A firm's carbon intensity is defined as the ratio between the total greenhouse gas emissions in metric tonnes of CO$_2$ and total revenues (in USD millions). Based on carbon intensity, we cluster the stocks into two groups; in particular, we assume that the first $k$ assets are characterised by low carbon intensity (green stocks) and the remaining $n-k$ assets by high carbon intensity (brown stocks). {From a practical perspective, a common approach (see, e.g., \cite{ardia2023factor}) is to rank firms’ carbon intensity cross-sectionally and identify the two groups using percentiles. For instance, firms with carbon intensity above (respectively, below) the $p$-th (respectively, $(1-p)$-th) percentile are labeled as brown (respectivel, green)}.
\paragraph{A convenient representation for the latent factor--stock model.}
We denote by $\mathbf{R}$ the positive definite correlation matrix of $\left(\bmW^{\bmS},\,W^Y\right)^\top$, 
\begin{equation}
\mathbf{R}=
\begin{pmatrix}
1          & \rho_{1,2} & \dots  & \rho_{1,n}   & \rho_{1,Y}  \\
\rho_{1,2} & 1          & \dots  & \rho_{2,n}   & \rho_{2,Y}  \\
\vdots     & \vdots     & \ddots & \vdots       & \vdots      \\
\rho_{1,n} & \rho_{2,n} & \dots  & 1            & \rho_{n,Y}\\
\rho_{1,Y} & \rho_{2,Y} & \dots  & \rho_{n,Y} & 1
\end{pmatrix}.
\end{equation}
We express $\bmW^\bmS$ and $W^Y$ as a linear combination of uncorrelated standard $\mathbb{G}$-Brownian motions, namely $\bmZ=\left(\bmZ^{\bmS},\,Z^Y\right)^\top=\left(Z^S_1,\,\dots,\,Z_n^S,\,Z^Y\right)^\top$, as follow
\begin{equation}
\begin{pmatrix}\bmW^{\bmS}_t\\W^Y_t\end{pmatrix}=\mathbf{L}\begin{pmatrix}\bmZ^{\bmS}_t\\Z^Y_t\end{pmatrix},\quad t\in[0,T],
\end{equation}
where $\mathbf{L}=\left(l_{i,j}\right)_{i,j\in\left\lbrace1,\dots,n+1\right\rbrace}\in\R^{\left(n+1\right)\times\left(n+1\right)}$ is a lower triangular matrix obtained through the Cholesky decomposition of the correlation matrix $\mathbf{R}$, that is $\mathbf{R}=\mathbf{L}\mathbf{L}^\top$. Thus, the dynamics in \eqref{eq:latent_factor} and \eqref{eq:risk_assets_dyn} can be rewritten as
\begin{align}
\label{eq:risky_asset_dyn_uncorrelated}
\de Y_t&=\left(\lambda Y_t+\beta\right)\de t+\bmtSigma_Y\de\bmZ_t^{\bmS}+\tilde\sigma_Y\de Z_t^Y, \quad Y_0\sim N\left(\Gamma_0,P_0\right),\\
\de\bmS_t&=\diag\left(\bmS_t\right)\left[\left(\bma Y_t+\bmb\right)\de t+\bmtSigma_{\bmS}\de\bmZ^{\bmS}_t\right],\quad\bmS_0\in\R^{n}_{+},
\end{align}
respectively, where $\bmtSigma_Y=\sigma_{Y}\mathbf{L}_{Y}\in\R^{1\times n}$, $\tilde\sigma_Y=\sigma_{Y}l_{n+1,n+1}\in\R$, $\mathbf{\tilde\Sigma_S}=\bmSigma_{\bmS}\mathbf{L}_{\bmS}=\left(\tilde\sigma_{i,j}\right)_{i,j\in\left\lbrace1,\dots,n\right\rbrace}\in\R^{n\times n}$,  with $\mathbf{L}_Y=\left(l_{n+1,j}\right)_{j\in\left\lbrace1,\dots,n\right\rbrace}\in\R^{1\times n}$ and $\mathbf{L}_{\bmS}=\left(l_{i,j}\right)_{i,j\in\left\lbrace1,\dots,n\right\rbrace}\in\R^{n\times n}$. 

\section{The carbon-penalised proportional portfolio insurance strategy}\label{sect:PPI}
The portfolio insurer employs a proportional portfolio insurance (PPI) strategy. Such strategies are designed to capitalise on the returns of the risky assets traded on the market while securing a pre-specified amount $G$ at maturity $T$. To achieve this goal, the fund manager divides her position between the bank account $S^0$, and a risky reference portfolio with value $X=\left\lbrace X_t\right\rbrace_{t\in[0,T]}$. The fund manager defines a floor process $F=\left\lbrace F_t\right\rbrace_{t\in[0,T]}$ and a cushion process $C=\left\lbrace C_t\right\rbrace_{t\in[0,T]}$. The floor $F$ is given by the present value of the guarantee amount $G$ at maturity, that is $F_t=Ge^{-r(T-t)}$ for all $t \in [0,T]$, and represents the capital to be protected at every time.\footnote{Typically, the guaranteed amount $G$ is a pre-specified percentage of the initial endowment $V_0$, namely $G = V_0 \cdot PL$, where $\mathrm{PL} \in (0,1]$ is the so-called \textit{protection level}.} The cushion $C$ is the difference between the current PPI portfolio value $V=\left\lbrace V_t\right\rbrace_{t\in[0,T]}$ and the floor, that is $C_t=V_t-F_t$ for every $t\in[0,T]$. The exposure to the risky reference portfolio $X$ is linked to the cushion in the following way. At every time $t\in[0,T]$, if $V_t>F_t$ the exposure to $X$ is given by $m_tC_t$, where $m=\left\lbrace m_t\right\rbrace_{t\in[0,T]}$ is the proportionality factor known as multiplier. However, if there exists a time $\tau:=\inf\left\lbrace t>0:V_t\leq F_t\right\rbrace\wedge T$, the portfolio value is entirely invested into the bank account $S^0$, since $C_t=0$ for all $t\in\left[\tau\wedge T,T\right]$. To summarize, the exposure to the market index is given by $m_t\left(C_t\right)^+$ for every $t\in[0,T]$. Hence, the dynamics of the PPI portfolio is given by 
\begin{equation}\label{eq:PPI_strategy_0}
\de V_t=
\begin{cases}
\begin{aligned}
&rV_{t}\de t+\left(V_t-F_t\right)m_t\left( \frac{\de X_t}{X_t} - r dt\right),\quad t<\tau,\\[8pt]
&rV_{t}\de t,\quad t\geq\tau, 
\end{aligned}
\end{cases}
\end{equation}
with $V_0=v_0$ being the initial endowment, and the dynamics of the cushion $C=\{C_t\}_{t\in[\tau\wedge T, T]}$ are
\begin{align}\label{eq:cushion_process_0}
\dfrac{\de C_t}{C_t}=& r \de t +m_t \left(\frac{\de X_t}{X_t} - r \de t \right),\quad C_0=c_0=v_0-F_0.
\end{align}
Next, we introduce the dynamics of the \textit{risky reference portfolio}. Let $\bmpi=\left\lbrace\pi_{1,t},\dots,\pi_{n,t}\right\rbrace_{t\in[0,T]}$ be the vector-valued process in $\R^n$ containing the composition percentage of the $i$-th stock in the risky reference portfolio, for every $i=1,\dots,n$ and $t\in[0,T]$. Hence, the dynamics of $X^{\bmpi}=\left\lbrace X_t^{\bmpi}\right\rbrace_{t\in[0,T]}$ read as
\begin{equation}
\dfrac{\de X_t^{\bmpi}}{X_t^{\bmpi}}=\bmpi_t^\top\left(\bma Y_t+\bmb\right)\de t+\bmpi^\top_t\bmtSigma_{\bmS}\de\bmZ^{\bmS},\quad X_0= x_0.
\end{equation}
Assuming that $\sum_{i=1}^n\pi_{i,t}=1$ for every $t\in[0,T]$, for any given couple $\left(m,\bmpi\right)=\left\lbrace m_t,\bmpi_t\right\rbrace_{t\in[0,T]}$, equation \eqref{eq:PPI_strategy_0} becomes   
\begin{equation}\label{eq:PPI_strategy}
\de V_t^{m,\bmpi}=
\begin{cases}
\begin{aligned}
&rV_{t}^{m,\bmpi}\de t+\left(V_t^{m,\bmpi}-F_t\right)m_t\left[\bmpi_{t}^\top\left(\bma Y_t+\bmb-\mathbf{r}_n\right)\de t+\bmpi_{t}^\top\bmtSigma_{\bmS}\de\bmZ^{\bmS}_t\right],\quad t<\tau,\\[8pt]
&rV_{t}^{m,\bmpi}\de t,\quad t\geq\tau, 
\end{aligned}
\end{cases}
\end{equation}
with $V_0^{m,\bmpi}=v_0$ being the initial endowment, and consequently, \eqref{eq:cushion_process_0} is 
\begin{align}\label{eq:cushion_process}
\dfrac{\de C_t^{m,\bmpi}}{C_t^{m,\bmpi}}=&\left[r+m_t\bmpi_{t}^\top\left(\bma Y_t+\bmb-\mathbf{r}_n\right)\right]\de t+m_t\bmpi_{t}^\top\bmtSigma_{\bmS}\de\bmZ^{\bmS}_t,\quad C_0^{m,\bmpi}=c_0.
\end{align}
Here, we stress the dependence of the risky reference portfolio $X$ on its composition $\bmpi$, and the dependence of both the PPI portfolio value $V$ and the cushion $C$ on $\bmpi$ and the multiplier $m$. The fund manager's objective is to maximise the expected utility from the terminal cushion in a carbon-penalised setting. In particular, the fund manager wants to prevent a high exposure of the strategy to brown stocks by adding a penalty term to the terminal cushion. In the same spirit of \cite{rogers2013optimal}, we assume that such penalisation is proportional to the riskiness of brown stocks, which is measured according to their realised variance. The carbon-penalised cushion at maturity is given by 
\begin{equation}
\hat{C}^{m,\bmpi}_T=C^{m,\bmpi}_T\exp\left\lbrace-\dfrac{1}{2}\int_0^Tm_s^2\bmpi_{s}^\top\left(\bmSigma_{\bmS}\bmSigma_{\bmS}^\top\odot\bme\right)\bmpi_s\de s\right\rbrace,
\end{equation}
where $\odot$ denotes the Hadamard product, and $\bme=\begin{pmatrix}\bm{0}_{k} & \bm{1}_{n-k}\varepsilon\end{pmatrix}^\top\in\R^{n}$ with $\varepsilon\geq 0$ represents the fund manager’s carbon aversion with respect to brown stocks. It follows from It\^{o}'s formula that the dynamics of $\hat{C}^{m,\bmpi}=\{\hat C_t^{m,\bmpi}\}_{t\in[\tau\wedge T,T]}$ is given by
\begin{align}\label{eq:penalised_cushion_dyn}
\dfrac{\de\hat{C}_t^{m,\bmpi}}{\hat{C}_t^{m,\bmpi}}=\left[r+m_t\bmpi_t^\top\left(\bma Y_t+\bmb-\mathbf{r}_n\right)-\dfrac{1}{2}m_t^2\bmpi^\top_t\left(\bmSigma_{\bmS}\bmSigma_{\bmS}^\top\odot\bme\right)\bmpi_t\right]\de t+m_t\bmpi_t^\top\bmtSigma_\bmS\de\bmZ^{\bmS}_t,\quad\hat C_0^{\bmpi}=\hat c_0.
\end{align}
\begin{remark}
\begin{itemize}
\item[(i)] The penalisation embeds sustainability into the portfolio insurer’s preferences by increasing risk aversion specifically toward high–carbon-intensity stocks. Unlike \cite{rogers2013optimal}, our penalty excludes the variance–covariance matrix to avoid bias from negatively correlated brown stocks; instead, it relies solely on realised variance. Moreover, we do not impose a fixed sustainability target as in \cite{bolton2022net} and \cite{LeGuenedalRoncalli2023}. Instead this approach accounts jointly for financial performance and carbon exposure: high-carbon assets may still be held if their low volatility or high expected return compensates for their emissions. This is crucial for PI strategies, whose main goal is capital protection, as it prevents excessive penalisation of low-risk brown assets.
\item[(ii)] The carbon penalty admits two interpretations. It can be seen as (i) a proportional cost on carbon-intensive holdings, balancing risk premia against reputational or regulatory costs, or (ii) an endogenous increase in the insurer’s risk aversion toward brown stocks. As shown in Example \ref{Ex:example_CRRA_full_info}, the effective risk aversion to such assets equals the market risk-aversion parameter plus the penalty term, naturally reducing exposure to carbon-intensive stocks (see, e.g. \cite{colaneri2025design} for more details on this point).
\end{itemize}
\end{remark}
To reduce the number of controls of the optimisation problem, we introduce the process $\bmtheta=\{\bmtheta_t\}_{t\in[0,T]}$ such that $\bmtheta_t=m_t\bmpi_t$, for every $t\in[0,T]$. Hence, the dynamics of the carbon-penalised cushion can be rewritten as
\begin{align}\label{eq:carbon_penalised_cushion_FULL_INFO}
\dfrac{\de\hat{C}_t^{\bmtheta}}{\hat{C}_t^{\bmtheta}}=\left[r+\bmtheta_t^\top\left(\bma Y_t+\bmb-\mathbf{r}_n\right)-\dfrac{1}{2}\bmtheta^\top_t\left(\bmSigma_{\bmS}\bmSigma_{\bmS}^\top\odot\bme\right)\bmtheta_t\right]\de t+\bmtheta_t^\top\bmtSigma_\bmS\de\bmZ^{\bmS}_t,\quad\hat C_0^{\bmtheta}=\hat c_0.
\end{align}
In the next section, we address the optimisation problem of the portfolio insurer under two different information settings. We begin with the case where she has full information on all factor processes that drive the market, and we refer to this as the full information case. Second, we assume that she cannot observe the common stochastic factor $Y$ directly, but she can only infer its value from the observation of stock prices, and we call this case the partial information setting. 
\section{Optimisation problem under full information}\label{sect:opt_problem_full_info}
We introduce the set of admissible strategies.
\begin{defn}\label{defn:G_admissible_strategies_theta}
A $\mathbb{G}$-admissible carbon-penalised PPI strategy $\bmtheta=\left\lbrace\bmtheta_t\right\rbrace_{t\in[0,T]}$ is a self-financing, $\mathbb{G}$-predictable process such that 
\begin{itemize}
\item[(i)] $\mathbb{E}\left[\int_0^T|Y_s|\|\bmtheta_s\|_1+\|\bmtheta_s\|_2^2\de s\right]<\infty$,
\item[(ii)] $\displaystyle\sup_{t\in[0,T]}\mathbb{E}\left[(\hat C_t^{\bmtheta})^{d\left(1-\delta\right)(1+\alpha)}\right]<\infty$, for some $\alpha>0$ and $d>1$.
\end{itemize}
We denote the set of $\mathbb{G}$-admissible strategies by $\mathcal{A}^{\mathbb{G}}$.
\end{defn}

Note that we can equivalently rewrite the set of admissible strategies in terms of $\left(m,\,\bmpi\right)$ as follows. Precisely, a $\mathbb{G}$-admissible carbon-penalised PPI strategy $\left(m,\,\bmpi\right)=\left\lbrace m_t,\,\bmpi_t\right\rbrace_{t\in[0,T]}$ is a self-financing, $\mathbb{G}$-predictable process such that 
\begin{itemize}
\item[(i)] the following integrability condition holds
\begin{equation}\label{eq:integrability_conds}
\mathbb{E}\left[\int_0^T|Y_s||m_s|\|\bmpi_s\|_1+m_s^2\|\bmpi_s\|_2^2\de s\right]
<\infty,
\end{equation}
where $\|\cdot \|_1$ and $\|\cdot \|_2$ denote the $\ell_1$ and $\ell_2$ norms in $\R^n$,
\item[(ii)] $\displaystyle\sup_{t\in[0,T]}\mathbb{E}\left[(\hat C_t^{m,\,\bmpi})^{d\left(1-\delta\right)(1+\alpha)}\right]<\infty$, for some $\alpha>0$ and $d>1$.
\end{itemize}

A fully informed portfolio insurer seeks to solve the following optimisation: 
\begin{equation}\label{eq:optimisation_problem_full_info_CRRA}
\mbox{Maximise }\mathbb{E}^{t,c,y}\left[\dfrac{(\hat C_T^{\bmtheta})^{1-\delta}}{1-\delta}\right],\mbox{ over all }\bmtheta\in\mathcal{A}^{\mathbb{G}},
\end{equation}
where $\delta\in\left(0,1\right)\cup\left(1,+\infty\right)$ represents the fund manager's risk aversion parameter, and $\mathbb{E}^{t,c,y}$ denotes the conditional expectation given $\hat C_t=c$ and $Y_t=y$. The value function of the optimisation problem in equation \eqref{eq:optimisation_problem_full_info_CRRA}, is given by
\begin{equation}\label{eq:value_function_full_info_CRRA}
\hat{v}(t,c,y):=\sup_{\bmtheta\in\mathcal{A}^{\mathbb{G}}}\mathbb{E}^{t,c,y}\left[\dfrac{(\hat C_T^{\bmtheta})^{1-\delta}}{1-\delta}\right].
\end{equation}
The problem is solved by employing dynamic programming principle. We consider the following Hamilton-Jacobi-Bellman equation
\begin{equation}\label{eq:HJB_equation_FULL_INFO}
\begin{cases}
\displaystyle\sup_{\bmtheta\in\mathcal{A}}\hat{v}_{t}(t,c,y)+\mathcal{L}^{\bmtheta}\hat{v}(t,c,y)=0,&(t,c,y)\in[0,T)\times\R_+\times\R,\\[8pt]
\hat{v}(T,c,y)=\dfrac{c^{1-\delta}}{1-\delta},& (c,y)\in\R_+\times\R,
\end{cases}
\end{equation}
where for any constant control $\bmtheta\in\R^n$, the operator $\mathcal{L}^{\theta}$ denotes the infinitesimal generator of the process $(\hat{C}_t^{\bmtheta},Y_t)$ which is given by 
\begin{align}
\mathcal{L}^{\bmtheta}F(t,c,y)=&c\left[r+\bmtheta^\top\left(\bma y+\bmb-\bm{r}_{n}\right)-\dfrac{1}{2}\bmtheta^\top\left(\bmSigma_{\bmS}\bmSigma_{\bmS}^\top\odot\bme\right)\bmtheta\right]F_{c}(t,c,y)+\frac{c^2}{2}\bmtheta^\top\bmtSigma_{\bmS}\bmtSigma_{\bmS}^\top\bmtheta F_{c,c}(t,c,y)\\
&+\left(\lambda y+\beta\right)F_y(t,c,y)+\dfrac{\sigma_Y^2}{2}F_{y,y}(t,c,y)+c\bmtheta^\top\bmtSigma_{\bmS}\bmtSigma_Y^\top F_{c,y}(t,c,y),
\end{align}
for every function $F(\cdot)\in\mathcal{C}^{1,2,2}\left([0,T]\times\R_{+}\times\R\right)$. In the sequel, we prove that the value function, defined in equation \eqref{eq:value_function_full_info_CRRA}, solves the equation \eqref{eq:HJB_equation_FULL_INFO}. We begin our analysis of the optimisation problem under full information with a verification result. 
\begin{thm}[Verification Theorem]\label{thm:verification_thm_full_info_CRRA}
Let $f(t,c,y)\in\mathcal{C}^{1,2,2}([0,T]\times\R_{+}\times\R)$ be a classical solution to the HJB equation \eqref{eq:HJB_equation_FULL_INFO} and assume that the following conditions hold:
\begin{itemize}
\item[(i)] for any $\bmtheta\in\mathcal{A}^{\mathbb{G}}$ 
the family $\{f(t \wedge \tau, \hat{C}_{t \wedge \tau}, Y_{t \wedge \tau}), \text{ for all } \mathbb{G}\text{--stopping times } \tau \}$ is uniformly integrable;
\item[(ii)] there exists $\bmtheta^\star$ at which the supremum in equation \eqref{eq:HJB_equation_FULL_INFO} is attained.
\end{itemize}
Then $f(t,c,y)=\hat{v}(t,c,y)$ and if $\{\bmtheta^\star(t,Y_t)\}_{t\in[0,T]}\in\mathcal{A}^{\mathbb{G}}$ this is an optimal Markovian control.
\end{thm}
\begin{proof}
See Appendix \ref{app:A_1}.
\end{proof}
\begin{thm}\label{thm:existence_CRRA}
Let $\hat{f}(t),\hat{g}(t),\hat{h}(t)\in\mathcal{C}_b^{1}([0,T])$ be the unique solutions to the following system of ODEs
\begin{align}
0=&\hat{f}_t(t)+\left[\left(1-\delta\right)\bmtSigma_Y\bmtSigma_{\bmS}^\top\bmhatTheta^{-1}\bmtSigma_{\bmS}\bmtSigma_Y^\top+\sigma_Y^2\right]\hat{f}^2(t)+2\left[\left(1-\delta\right)\bmtSigma_Y\bmtSigma_{\bmS}^\top\bmhatTheta^{-1}\bma+\lambda\right]\hat{f}(t)\\
\label{eq:f_hat}&+\left(1-\delta\right)\bma^\top\bmhatTheta^{-1}\bma,\\[8pt]
0=&\hat{g}_t(t)+\left[\left(1-\delta\right)\bmtSigma_Y\bmtSigma_{\bmS}^\top\bmhatTheta^{-1}\bma+\lambda\right]\hat{g}(t)+\left[\left(1-\delta\right)\bmtSigma_Y\bmtSigma_{\bmS}^\top\bmhatTheta^{-1}\left(\bmb-\mathbf{r}_{n}\right)+\beta\right]\hat{f}(t)\\
\label{eq:g_hat}&+\left[\left(1-\delta\right)\bmtSigma_Y\bmtSigma_{\bmS}^\top\bmhatTheta^{-1}\bmtSigma_{\bmS}\bmtSigma_Y^\top+\sigma_Y^2\right]\hat{f}(t)\hat{g}(t)+\left(1-\delta\right)\bma^\top\bmhatTheta^{-1}\left(\bmb-\bm{r}_{n}\right),\\[8pt]
0=&\hat{h}_t(t)+\left(1-\delta\right)r+\left[\left(1-\delta\right)\bmtSigma_Y\bmtSigma_{\bmS}^\top\bmhatTheta^{-1}\left(\bmb-\mathbf{r}_n\right)+\beta\right]\hat{g}(t)+\dfrac{\sigma_Y^2}{2}\hat{f}(t)\\
&\label{eq:h_hat}+\dfrac{1}{2}\left[\left(1-\delta\right)\bmtSigma_Y\bmtSigma_{\bmS}^\top\bmhatTheta^{-1}\bmtSigma_{\bmS}\bmtSigma_Y^\top+\sigma_Y^2\right]\hat{g}^2(t)+\dfrac{1-\delta}{2}\left(\bmb-\bm{r}_{n}\right)^\top\bmhatTheta^{-1}\left(\bmb-\bm{r}_{n}\right),
\end{align}
with terminal conditions $\hat{f}(T)=\hat{g}(T)=\hat{h}(T)=0$, where $\bmhatTheta=\left(\bmSigma_{\bmS}\bmSigma_{\bmS}^\top\right)\odot\bme+\delta\bmtSigma_{\bmS}\bmtSigma_{\bmS}^\top$. Then, the optimal control is given by 
\begin{equation}\label{eq:opt_controls_CRRA_full}
\bmtheta^\star(t,y)=\bmhatTheta^{-1}\left(\bma y+\bmb-\bm{r}_{n}\right)+\bmhatTheta^{-1}\bmtSigma_{\bmS}\bmtSigma_Y^\top\left(\hat{f}(t)y+\hat{g}(t)\right),
\end{equation}
and the value function satisfies
\begin{equation}\label{eq:value_fun_CRRA_full_info}
\hat v(t,c,y)=\dfrac{c^{1-\delta}}{1-\delta}\exp\left\lbrace\frac{\hat{f}(t)}{2}y^2+\hat{g}(t)y+\hat{h}(t)\right\rbrace.
\end{equation}
\end{thm}
\begin{proof}
See Appendix \ref{proof_existence_CRRA}.
\end{proof}
We now characterise the range of risk aversion parameters that guarantee $\hat{f}(t)\in\mathcal{C}^{1}_b([0,T])$. We define the function $\Delta(x):(0,+\infty)\to\R$ as follows
\begin{equation}\label{eq:discriminant}
\Delta(x)=4\left\lbrace\left[\left(1-x\right)\bmtSigma_Y\bmtSigma_{\bmS}^\top\mathbf{\hat{\Theta}}^{-1}\bma+\lambda\right]^2-\left[\left(1-x\right)^2\bmtSigma_Y\bmtSigma_{\bmS}^\top\mathbf{\hat{\Theta}}^{-1}\bmtSigma_{\bmS}\bmtSigma_Y^\top+\left(1-x\right)\sigma_Y^2\right]\bma^\top\mathbf{\hat{\Theta}}^{-1}\bma\right\rbrace,
\end{equation}
which represents the \textit{discriminant} of the Riccati ODE $\hat{f}$ in \eqref{eq:f_hat}, and define the set $\mathcal{P}=\{\delta\in(0,1)\cup(1,+\infty):\Delta(\delta)>0\}$. The set $\mathcal{P}$ represents set of risk aversion parameters for which $\hat{f}(t)\in\mathcal{C}^1_b([0,T])$.
\begin{prop}\label{prop:P_non_empty}
The set $\mathcal{P}$ is not empty.
\end{prop}
\begin{proof}
This result is a consequence of the fact that $\Delta(x)$ is a continuous function and that $\Delta(1)=\lambda^2>0$; hence, there exists a neighborhood of $\delta=1$ contained in $\mathcal{P}$ such that $\Delta(\delta)>0$. 
\end{proof}
By virtue of Proposition \ref{prop:P_non_empty}, there exist values of $\delta$ contained in $\mathcal{P}$ such that $\hat{f}(t)\in\mathcal{C}^{1}_b([0,T])$. As a consequence, the solutions of the linear ODEs in equations \eqref{eq:g_hat} and \eqref{eq:h_hat} also exist and share the same regularity.
\begin{remark}
Proposition \ref{prop:P_non_empty} ensures that the system of ODEs in equations \eqref{eq:f_hat}, \eqref{eq:g_hat}, and \eqref{eq:h_hat} admits a solution that does not explode in finite time, for some values of the risk aversion parameter $\delta$. In particular, it guarantees the existence of a solution for risk aversion parameters that are close to logarithmic utility. In a multidimensional setting, such as the one considered in this paper, deriving conditions for the existence of a solution over a broader range of $\delta$ is not straightforward. As a result, identifying the largest possible set $\mathcal{P}$, which depends on several model parameters (e.g., the variance-covariance matrices), remains a challenging task. Nevertheless, $\mathcal{P}$ can be explicitly identified in a simplified setting with two uncorrelated assets, independent of the common stochastic factor $Y$ (see Appendix \ref{sect:example}).
\end{remark}
The optimal candidate strategy $\bmtheta^\star=\{\bmtheta^\star(t,Y_t)\}_{t\in[0,T]}$, where $\bmtheta^\star(t,y)$ is defined by equation \eqref{eq:opt_controls_CRRA_full}, is Markovian, as it depends exclusively on time and the exogenous factor $Y$. We now provide conditions on the model parameters ensuring that condition (i) of Theorem \ref{thm:verification_thm_full_info_CRRA} is satisfied and that $\bmtheta^\star$ is an admissible control, according to Definition \ref{defn:G_admissible_strategies_theta}. These results are stated and proved in the following propositions.
\begin{prop}\label{prop:sufficient_cond_ver_FULL_INFO}
Assume that one of the two following conditions holds
\begin{itemize}
\item[(i)] $\delta\in\mathcal{P}\cap(1,+\infty)$,
\item[(ii)] $\delta\in\mathcal{P}\cap(0,1)$ and
\begin{equation}\label{eq:cond_1}
1-q(1+\alpha)\hat{f}(0)\max\left\lbrace P_0,\text{Var}[Y_T]\right\rbrace>0,
\end{equation}
for some $q>1$.
\end{itemize}
Then, for any admissible strategy $\bmtheta\in\mathcal{A}^{\mathbb{G}}$, $\{v(\tau, \hat{C}_{\tau}, Y_{\tau}),\mbox{ for all }\mathbb{G}\mbox{-stopping times }\tau\le T\}$ forms a uniformly integrable family.
\end{prop}
\begin{proof}
The proof is provided in Appendix \ref{app:A_4}.
\end{proof}
In the next Proposition, we provide sufficient conditions for admissibility of the optimal strategy.
\begin{prop}\label{prop:sufficient_cond_admissibility_FULL_INFO}
Assume that one of the two following conditions holds
\begin{itemize}
\item[(i)] $\delta\in\mathcal{P}\cap(0,1)$ and 
\begin{equation}\label{eq:cond_1_ammissibility_delta_in_0_1}
1-8d(1-\delta)(1+\alpha)nT\left[\left(1\vee d(1-\delta)(1+\alpha)w\right)c_1^2+a_M^2\right]\max\left\lbrace P_0,\text{Var}[Y_T]\right\rbrace>0,
\end{equation}
\item[(ii)] $\delta\in\mathcal{P}\cap(1,+\infty)$ and
\begin{equation}\label{eq:cond_1_ammissibility_delta_in_1_infty}
1-8d(1-\delta)(1+\alpha)nT\left[\left(-(1+w)\wedge d(1-\delta)(1+\alpha)\tilde{w}\right)c_1^2-a_M^2\right]\max\left\lbrace P_0,\text{Var}[Y_T]\right\rbrace>0,
\end{equation}
for some $d>1$, where
\begin{align}
\label{eq:a}
a_M&=\max_{i=1,\dots,n}|\left(\mathbf{a}\right)_i|,\\
\label{eq:exp_for_w}
w&=\max_{i,j=1,\dots,n}\left|\left(\mathbf{\tilde\Sigma}_{\mathbf{S}}\mathbf{\tilde\Sigma}_{\mathbf{S}}^\top\right)_{i,j}\right|,\\
\label{eq:exp_for_tilde_w}
\tilde{w}&=\max_{i,j=1,\dots,n}|(\mathbf{\hat\Theta})_{i,j}|,\\
\label{eq:exp_for_c_1}
c_1&=\max_{i=1,\dots,n}\bigg|\left(\mathbf{\hat\Theta}^{-1}\left(\mathbf{a}+\bm{\tilde\Sigma}_{\mathbf{S}}\bm{\tilde\Sigma}_Y^\top\sup_{t\in[0,T]}\hat{f}(t)\right)\right)_i\bigg|.
\end{align}
\end{itemize}
Then, the process $\bmtheta^\star$ given by equation \eqref{eq:opt_controls_CRRA_full} is an admissible strategy.
\end{prop}
\begin{proof}
The proof is provided in Appendix \ref{app:A_5}.
\end{proof}
Under the assumption of Proposition \ref{prop:sufficient_cond_ver_FULL_INFO}, the value function $\hat v$ is the unique solution of the optimisation problem \ref{eq:optimisation_problem_full_info_CRRA} and $\bmtheta^\star\in\mathcal{A}$. Given $\bmtheta^\star$, we can characterise the optimal multiplier $m^\star$ and the optimal stock composition percentages $\bmpi^\star$ of the risky reference portfolio as in the following Proposition.   
\begin{prop}\label{prop:originalcontrols} The optimal multiplier is given by $m^\star_t=\bm{\theta}^{\star,\top}\mathbf{1}_n$ and the optimal composition percentage of the $i$-th stock in the risky reference portfolio $X$ is given by $\pi^\star_{i,t}=\frac{\theta^\star_{i,t}}{\bm{\theta}^{\star,\top}\mathbf{1}_n}$, for every $i=1,\dots,n$, and $t\in[0,T]$.
\end{prop}
\begin{proof}
The proof is provided in Appendix \ref{app:A_6}.    
\end{proof}
\begin{remark}[Structure of the optimal strategy and role of the carbon penalty]
{
The optimal allocation in \eqref{eq:opt_controls_CRRA_full} admits the classical decomposition into a myopic demand and an intertemporal hedging component. The first term, $\hat{\bm{\Theta}}^{-1}(ay+b-r\mathbf{1}_n)$, represents the instantaneous mean-variance trade-off, while the second term captures the investor’s desire to hedge against future changes in investment opportunities driven by the factor process $Y$.
A key feature of the model is that both components are multiplied by the matrix $\hat{\bm{\Theta}}^{-1}$, which embeds not only the market risk aversion parameter $\delta$ but also the carbon penalisation parameter $\varepsilon$. In particular, the carbon penalty affects the effective risk loading in an asset-specific way. As a consequence, the model does not simply rescale the overall level of risk, but modifies the relative attractiveness of assets by altering their risk-adjusted returns. This mechanism can be interpreted as an endogenous increase in risk aversion toward carbon-intensive assets, leading to a systematic reweighting of the portfolio. To better explain this mechanism, we provide below a stylized example with two assets only.}
\end{remark}
\begin{example}\label{Ex:example_CRRA_full_info}
To analyze the optimal PPI strategy, we consider the case in which only two stocks, $S_1$ and $S_2$, are traded on the market, representing a green and a brown stock, respectively. For simplicity, we assume that $S_1$ and $S_2$ are driven by independent Brownian motions. Applying Proposition \ref{prop:originalcontrols}, the optimal multiplier $m^\star$ reads as follows 
\begin{equation}
m^\star(t,y;\delta,\varepsilon)=\theta_{1}^{\star}(t,y;\delta)+\theta_{2}^{\star}(t,y;\delta,\varepsilon),
\end{equation}
where $\theta_1^{\star}(t,y;\delta)=\xi_{1}^M(t,y;\delta)+\xi_{1}^I(t,y;\delta)$ and $\theta_2^{\star}(t,y;\delta,\varepsilon)=\xi_{2}^M(t,y;\delta,\varepsilon)+\xi_{2}^I(t,y;\delta,\varepsilon)$, with 
\begin{align}
\label{eq:xi_1}\xi_1^M(t,y;\delta)&=\dfrac{1}{\delta}\dfrac{a_1y+b_1-r}{\sigma_1^2},\quad \xi_1^I(t,y;\delta)=\dfrac{1}{\delta}\dfrac{\sigma_Y\rho_{1,Y}}{\sigma_1}\left(\hat{f}(t)y+\hat{g}(t)\right),\\
\label{eq:xi_2}\xi_2^M(t,y;\delta,\varepsilon)&=\dfrac{1}{\varepsilon+\delta}\dfrac{a_2y+b_2-r}{\sigma_2^2},\quad \xi_2^I(t,y;\delta,\varepsilon)=\dfrac{1}{\varepsilon+\delta}\dfrac{\sigma_Y\rho_{2,Y}}{\sigma_2}\left(\hat{f}(t)y+\hat{g}(t)\right),
\end{align}
for every $\left(t,y\right)\in[0,T]\times\R$. The optimal multiplier is the sum of the myopic and intertemporal hedging demand relative to each of the two stocks included in the risky reference portfolio. Both the myopic and the intertemporal components relative to the brown stock depend on the carbon aversion factor $\varepsilon$. Hence, by introducing a penalty term proportional to the realised volatilities of brown stocks in the objective function, we have effectively increased the fund manager's risk aversion toward this category of assets. The optimal composition percentages of the stocks in the risky reference portfolio $\left(\pi^\star_1,\pi^\star_2\right)$ are given by

\begingroup
\small{
\begin{align}
\pi_1^\star(t,y;\delta,\varepsilon)&=\dfrac{\left(\varepsilon+\delta\right)\left[a_1y+b_1-r+\sigma_1\sigma_Y\rho_{1,Y}\left(\hat{f}(t)y+\hat{g}(t)\right)\right]\sigma_2^2}{\left(\varepsilon+\delta\right)\left(a_1y+b_1-r\right)\sigma_2^2+\delta\left(a_2y+b_2-r\right)\sigma_1^2+\left[\left(\varepsilon+\delta\right)\sigma_1\sigma_2^2\rho_{1,Y}+\delta\sigma_1^2\sigma_2\rho_{2,Y}\right]\sigma_Y(\hat{f}(t)y+\hat{g}(t))},\\[3pt]
\pi_2^\star(t,y;\delta,\varepsilon)&=\dfrac{\delta\left[a_2y+b_2-r+\sigma_2\sigma_Y\rho_{2,Y}(\hat{f}(t)y+\hat{g}(t))\right]\sigma_1^2}{\left(\varepsilon+\delta\right)\left(a_1y+b_1-r\right)\sigma_2^2+\delta\left(a_2y+b_2-r\right)\sigma_1^2+\left[\left(\varepsilon+\delta\right)\sigma_1\sigma_2^2\rho_{1,Y}+\delta\sigma_1^2\sigma_2\rho_{2,Y}\right]\sigma_Y(\hat{f}(t)y+\hat{g}(t))},
\end{align}
}
\endgroup

for every $\left(t,y\right)\in[0,T]\times\R$. We observe that $\pi_1^\star$ (respectively,  $\pi_2^\star$) is increasing (respectively,  decreasing) with respect to the carbon aversion parameter $\varepsilon$. As expected, the higher $\varepsilon$, the lower (respectively,  higher) the presence of brown (respectively,  green) stock in $X$. Hence, any increase of $\varepsilon$ results in a reduction of the overall carbon intensity of the risky reference portfolio and, consequently, of the PPI strategy. In the limiting case where $\varepsilon\to\infty$, $\pi^\star_1=1$ and $\pi^\star_2=0$, meaning that the risky reference portfolio fully coincides with the green stock. Moreover, the optimal multiplier becomes
\begin{equation}
m^\star(t,y;\delta,\varepsilon=+\infty)=\dfrac{1}{\delta}\left[\dfrac{a_1y+b_1-r}{\sigma_1^2}+\dfrac{\sigma_Y\rho_{1,Y}(\hat{f}(t;\varepsilon=\infty)y+\hat{g}(t;\varepsilon=\infty))}{\sigma_1}\right],
\end{equation}
recovering the optimal PPI strategy with one single investment asset, see, e.g., \cite{zieling2014performance}.
\end{example}
\paragraph{Logarithmic case.} We assume that the fund manager is endowed with a logarithmic utility function. In such a case, the optimisation problem \eqref{eq:optimisation_problem_full_info_CRRA} can be reformulated as follows
\begin{equation}\label{eq:LOG_optimisation_problem_full_info}
\mbox{Maximise }\mathbb{E}^{t,c,y}\left[\log(\hat{C}_T^{\bmtheta})\right],
\end{equation}
over all $\bmtheta\in\mathcal{A}^{\mathbb{G}}$, and the corresponding value function is given by
\begin{equation}\label{eq:_LOGvalue_function_full_info}
v(t,c,y):=\sup_{\bmtheta\in\mathcal{A}^{\mathbb{G}}}\mathbb{E}^{t,c,y}\left[\log(\hat{C}_T^{\bmtheta})\right].
\end{equation}
For the logarithmic case, the optimal strategy can be derived by applying pointwise maximisation, which also yields an explicit characterisation for the value function. This result is presented in the following corollary.
\begin{cor}\label{cor:solution_full_info_log}
Consider a fund manager endowed with a logarithmic utility function and a carbon aversion $\varepsilon\geq 0$, then the optimal controls $\bmtheta^\star\in\mathcal{A}^{\mathbb{G}}$ is given by
\begin{align}
\bmtheta^\star(t,y)=\mathbf{\Theta}^{-1}\left(\bma y+\bmb-\mathbf{r}_{n}\right),
\end{align}
where $\mathbf{\Theta}=\left(\bmSigma_{\bmS}\bmSigma_{\bmS}^\top\right)\odot\bme+\bmtSigma_{\bmS}\bmtSigma_{\bmS}^\top$. The value function reads as
\begin{equation}\label{eq:value_fun_full_info_log}
v(t,c,y)=\log(c)+r(T-t)+f(t)y^2+g(t)y+h(t), 
\end{equation}
where
\begin{align}
&f(t)=\frac{\mathbf{a}^\top\mathbf{\Theta}^{-1}\mathbf{a}}{2\lambda}\left(e^{2\lambda\left(T-t\right)}-1\right),\\[6pt]
&g(t)=\frac{\mathbf{a}^\top\mathbf{\Theta}^{-1}(\mathbf{b}-\mathbf{r}_n)}{\lambda}\left(e^{\lambda\left(T-t\right)}-1\right)+\beta\frac{\mathbf{a}^\top\mathbf{\Theta}^{-1}\mathbf{a}}{2\lambda^2}\left(e^{\lambda\left(T-t\right)}-1\right)^2,\\[6pt] 
&h(t)=\left[r+\dfrac{1}{2}\left(\mathbf{b}-\mathbf{r}_n\right)^\top \mathbf{\Theta}^{-1}\left(\mathbf{b}-\mathbf{r}_n\right)\right]\left(T-t\right)+ \beta\dfrac{\mathbf{a}^\top\mathbf{\Theta}^{-1}\left(\mathbf{b}-\mathbf{r}_n\right)}{\lambda}\left[\frac{e^{\lambda\left(T-t\right)}-1}{\lambda}-\left(T-t\right)\right]\\
&+\beta^2\dfrac{\mathbf{a}^\top\mathbf{\Theta}^{-1}\mathbf{a}}{2\lambda^2}\left[\dfrac{e^{2\lambda\left(T-t\right)}-1}{2\lambda}-\frac{2}{\lambda}\left(e^{\lambda\left(T-t\right)}-1\right)+T-t\right]+ \frac{\sigma_Y^2}{2}\dfrac{\mathbf{a}^\top\mathbf{\Theta}^{-1}\mathbf{a}}{2\lambda}\left[\frac{e^{2\lambda\left(T-t\right)}-1}{2\lambda}-\left(T-t\right)\right],
\end{align}
for every $t\in[0,T]$.
\end{cor}
\begin{proof}
The proof is provided in Appendix \ref{app:A_7}.    
\end{proof}
In the case of the logarithmic utility function, the optimal strategy $\left(m^\star,\pi_1^\star,\pi_2^\star\right)$ discussed in Example \ref{Ex:example_CRRA_full_info}, becomes
\begin{align}
m^\star(t,y;1,\varepsilon)&=\xi^M_1(t,y;1)+\xi^M_2(t,y;1,\varepsilon),\\
\pi^\star_1(t,y;1,\varepsilon)&=\dfrac{\xi_1^M(t,y;1)}{m^\star(t,y;1,\varepsilon)},\quad\pi^\star_2(t,y;1,\varepsilon)=\dfrac{\xi_2^M(t,y;1,\varepsilon)}{m^\star(t,y;1,\varepsilon)},
\end{align}
for every $\left(t,y\right)\in[0,T]\times\R$. As expected by the nature of the utility function, the optimal multiplier presents only the myopic component. The factor $\xi_2$ depends on carbon penalisation in the same form as for the power utility case. Similar considerations on $\left(\pi_1^\star,\pi_2^\star\right)$, as for the power utility case, hold for logarithmic utility.
\section{Optimisation problem under partial information}\label{sect:opt_partial_info}
In this section, we address the utility maximisation problem faced by a portfolio insurer who cannot directly observe the common stochastic factor $Y$. The portfolio insurer's available information is limited to observing the price processes of green and brown stocks. Mathematically, the information flow accessible to the fund manager is given by the natural filtration generated by $\bmS$, referred to as $\mathbb{F}=\left\lbrace\mathcal{F}_t\right\rbrace_{t\in[0,T]}$, where $\mathcal{F}_t=\sigma\left\lbrace\bmS_u,\,0\leq u\leq t\right\rbrace\vee\mathcal{N}$ such that $\mathcal{F}_t\subset\mathcal{G}_t$. Here, $\mathcal{N}$ represents the collection of $\mathbb{P}$-null sets, and $\mathcal{F}_0$ is the trivial $\sigma$-algebra. The portfolio insurer, operating under partial information, seeks to maximise the expected CRRA utility of the terminal carbon-penalised cushion over the set of $\mathbb{F}$-admissible strageies  $\mathcal{A}^{\mathbb{F}}$ defined below (see Definition \ref{defn:F_admissible_strategies_theta}). To address the optimisation problem with partial observations, we introduce the conditional distribution of the unobservable factor process $Y$, using stochastic filtering theory.  

Let $\Gamma$ and $P$ be the conditional expectation and the conditional variance of the common stochastic factor $Y$ given the available information, that is, $\Gamma_t:=\mathbb{E}\left[Y_t|\mathcal{F}_t\right]$ and $P_t:=\mathbb{E}\left[\left(Y_t-\Gamma_t\right)^2|\mathcal{F}_t\right]$ for every $t\in[0,T]$, respectively. Since the conditional distribution of $Y$ is Gaussian, it is fully characterised by its conditional mean and variance dynamics. Moreover, since $\mathcal{F}_0$ is the trivial $\sigma$-algebra, the initial values $\Gamma$ and $P$ correspond to the parameters of the initial distribution of $Y$, that is, $Y_0\sim N(\Gamma_0, P_0)$.
To characterise the dynamics of $\Gamma$ and $P$, we introduce the innovation process $\bmI^{\bmS}=\left\lbrace \bmI^{\bmS}_t\right\rbrace_{t\in[0,T]}$,
\begin{align}
\bmI^{\bmS}_t=\bmtSigma_{\bmS}^{-1}\bma\int_0^t\left(Y_s-\Gamma_s\right)\de s+\bmZ_t^{\bmS},
\end{align}
for every $t\in[0,T]$. As proven in \cite[Section $10.3$]{liptser2013statistics},  $\bmI^{\bmS}$ is an $\left(\mathbb{F},\mathbb{P}\right)$-Brownian motion in $\R^{n}$, and the processes $\Gamma$ and $P$ are the unique solutions
to the system
\begin{align}
\label{eq:cond_exp}\de\Gamma_t&=\left(\lambda\Gamma_t+\beta\right)\de t+\mathbf{\bar{P}}_t\left(\bmtSigma_{\bmS}^\top\right)^{-1}\de\bmI_t^{\bmS},\quad\Gamma_0\in\R,\\
\label{eq:variance}\dfrac{\de P_t}{\de t}&=2\lambda P_t+\sigma_Y^2-\mathbf{\bar{P}}_t\left(\bmtSigma_{\bmS}\bmtSigma_{\bmS}^\top\right)^{-1}\mathbf{\bar{P}}_t^\top,\quad P_0\in\R_{+},
\end{align}
where $\mathbf{\bar{P}}_t=\bmtSigma_Y\bmtSigma_{\bmS}^\top+P_t\bma^\top$ for every $t\in[0,T]$, and $P_t$ and $\mathbf{\bar P}_t$ are deterministic functions. To highlight this property, from now on we will write $P(t)$ and $\mathbf{\bar P}(t)$ instead of $P_t$ and $\mathbf{\bar P}_t$, respectively. The semimartingale representations of $\mathbf{S}$ with respect to the information filtration $\mathbb{F}$ are given by
\begin{align}
\de\bmS_t&=\diag\left(\bmS_t\right)\left[\left(\bma\Gamma_t+\bmb\right)\de t+\bmtSigma_{\bmS}\de\bmI_t^{\bmS}\right],\quad \bmS_0\in\R_+^{n},
\end{align}
leading to the following representation for the carbon-penalised cushion process
\begin{align}\label{eq:filtered_penalised_cushion_dyn}
\dfrac{\de\hat{C}_t^{\bmtheta}}{\hat{C}_t^{\bmtheta}}=\left[r+\bmtheta_t^\top\left(\bma\Gamma_t+\bmb-\mathbf{r}_n\right)-\dfrac{1}{2}\bmtheta^\top_t\left(\bmSigma_{\bmS}\bmSigma_{\bmS}^\top\odot\bme\right)\bmtheta_t\right]\de t+\bmtheta_t^\top\bmtSigma_\bmS\de\bmI^{\bmS}_t,\quad\hat C_0^{\bmtheta}=\hat c_0.
\end{align}
Since the portfolio insurer's decisions depend on the information available at time $t$, we define the set of admissible strategies $\bmtheta$ as follows.
\begin{defn}\label{defn:F_admissible_strategies_theta}
A $\mathbb{F}$-admissible carbon-penalised PPI strategy $\bmtheta=\left\lbrace\bmtheta\right\rbrace_{t\in[0,T]}$ is a self-financing, $\mathbb{F}$-predictable process such that 
\begin{itemize}
\item[(i)] $\mathbb{E}\left[\int_0^T|\Gamma_s|\|\bmtheta_s\|_1+\|\bmtheta_s\|_2^2\de s\right]<\infty,$
\item[(ii)] $\displaystyle\sup_{t\in[0,T]}\mathbb{E}\left[(\hat C_t^{\bmtheta})^{d\left(1-\delta\right)(1+\alpha)}\right]<\infty$, for some $\alpha>0$ and $d>1$.
\end{itemize}
We denote the set of $\mathbb{F}$-admissible strategies by $\mathcal{A}^{\mathbb{F}}$.\footnote{As in the full-information case, the set of admissible strategies can also be characterised in terms of $m$ and $\bmpi$, but we omit reporting it here for brevity.}
\end{defn}
Thanks to uniqueness of the solution of the filtering equation, we can consider $\hat{C}$ and $\Gamma$ as state processes and formulate the separated problem as follows 
\begin{equation}\label{eq:CRRA_opt_problem_partial_info}
\mbox{Maximise }\mathbb{E}^{t,c,\gamma}\left[\dfrac{(\hat C_T^{\bmtheta})^{1-\delta}}{1-\delta}\right],\mbox{ over all }\bmtheta\in\mathcal{A}^{\mathbb{F}},
\end{equation}
where $\mathbb{E}^{t,c,\gamma}$ denotes the conditional expectation given $\hat C_t=c$ and $\Gamma_t =\gamma$, where $\left(c,\gamma\right)\in\R_+\times\R$. We define the value function by
\begin{equation}
\hat{V}(t,c,\gamma):=\sup_{\bmtheta\in\mathcal{A}^{\mathbb{F}}}\mathbb{E}^{t,c,\gamma}\left[\dfrac{(\hat C_T^{\bmtheta})^{1-\delta}}{1-\delta}\right].
\end{equation}
Also in this case, we resort to dynamic programming principle. The HJB equation is given by  \begin{equation}\label{eq:HJB_partial_info_CRRA}
\begin{cases}
\displaystyle\sup_{\bmtheta\in\mathcal{A}^{\mathbb{F}}}\hat{V}_t(t,c,\gamma)+\mathcal{L}^{\bmtheta}\hat{V}(t,c,\gamma)=0,&(t,c,\gamma)\in[0,T)\times\R_+\times\R,\\[8pt]
\hat{V}(T,c,\gamma)=\dfrac{c^{1-\delta}}{1-\delta},&(c,\gamma)\in\R_+\times\R,
\end{cases}
\end{equation}
where for any constant control $\bmtheta\in\R^n$, the operator $\mathcal{L}^{\bmtheta}$ is given by
\begin{align}
\mathcal{L}^{\bmtheta}F(t,c,\gamma)=&c\left[r+\bmtheta_{t}^\top\left(\bma\gamma+\bmb-\mathbf{r}_n\right)-\dfrac{1}{2}\bmtheta^\top\left(\bmSigma_{\bmS}\bmSigma_{\bmS}^\top\odot\bme\right)\bmtheta\right]F_{c}(t,c,\gamma)\\
&+\frac{c^2}{2}\bmtheta^\top\bmtSigma_{\bmS}\bmtSigma_{\bmS}^\top\bmtheta^\top F_{c,c}(t,c,\gamma)+\left(\lambda\gamma+\beta\right)F_\gamma(t,c,\gamma)\\
&+\frac{1}{2}\mathbf{\bar{P}}(t)\left(\bmtSigma_{\bmS}\bmtSigma_{\bmS}^\top\right)^{-1}\mathbf{\bar{P}}(t)^\top F_{\gamma,\gamma}(t,c,\gamma)+c\bmtheta^\top\mathbf{\bar{P}}(t)^\top F_{c,\gamma}(t,c,\gamma),
\end{align}
for every function $F\left(\cdot\right)\in \mathcal{C}^{1,2,2}\left([0,T]\times\R_{+}\times\R\right)$. First, we establish the following verification result.
\begin{thm}[Verification Theorem]\label{thm:verification_thm_PARTIAL_INFO}
Let $f(t,c,\gamma)\in\mathcal{C}^{1,2,2}([0,T]\times\R_{+}\times\R)$ be a classical solution to the HJB equation \eqref{eq:HJB_equation_FULL_INFO} and assume that the following conditions hold:
\begin{itemize}
\item[(i)] for any $\bmtheta\in\mathcal{A}^{\mathbb{F}}$ 
the family $\{f(t \wedge \tau, \hat{C}_{t \wedge \tau}, \Gamma_{t \wedge \tau}), \text{ for all } \mathbb{F}-\text{stopping times } \tau \}$ is uniformly integrable;
\item[(ii)] there exists $\bar\bmtheta^\star$ at which the supremum in equation \eqref{eq:HJB_partial_info_CRRA} is attained.
\end{itemize}
Then $f(t,c,\gamma)=\hat{V}(t,c,\gamma)$ and if $\{\bar\bmtheta^\star(t,\Gamma_t)\}_{t\in[0,T]}\in\mathcal{A}^{\mathbb{F}}$ this is an optimal Markovian control.
\end{thm}
\begin{proof}
The proof replicates the line of that of Theorem \ref{thm:verification_thm_full_info_CRRA}.
\end{proof}
In view of the Verification Theorem, we characterise the value function as the unique classical solution
of the HJB equation. Also in this case, we resort to a guess-and-verify approach. The following result presents a candidate for the value function $\hat{V}$ and the optimal control $\bar{\bmtheta}^\star$ under partial information. We let $\bmhatTheta$ be the same of Theorem \ref{thm:existence_CRRA} and we introduce the following system od ODEs: 
\begingroup
\small
\begin{align}
0=&\bar{f}_t(t)+\left[\left(1-\delta\right)\mathbf{\bar{P}}(t)\mathbf{\hat{\Theta}}^{-1}\left(\mathbf{\bar{P}}(t)\right)^\top+\mathbf{\bar{P}}(t)\left(\mathbf{\tilde\Sigma}_{\mathbf{S}}\mathbf{\tilde\Sigma}_{\mathbf{S}}^\top\right)^{-1}\left(\mathbf{\bar{P}}(t)\right)^\top\right]\bar{f}^2(t)\\
\label{eq:bar_f}&+2\left[\left(1-\delta\right)\mathbf{\bar{P}}(t)\mathbf{\hat{\Theta}}^{-1}\mathbf{a}+\lambda\right]\bar{f}(t)+\left(1-\delta\right)\mathbf{a}^\top\mathbf{\hat{\Theta}}^{-1}\mathbf{a},\\[8pt]
0=&\bar{g}_t(t)+\left[\left(1-\delta\right)\mathbf{\bar{P}}(t)\mathbf{\hat{\Theta}}^{-1}\mathbf{a}+\lambda\right]\bar{g}(t)+\left[\left(1-\delta\right)\mathbf{\bar{P}}(t)\mathbf{\hat{\Theta}}^{-1}\left(\mathbf{b}-\mathbf{r}_n\right)+\beta\right]\bar{f}(t)\\
&+\left[\left(1-\delta\right)\mathbf{\bar{P}}(t)\mathbf{\hat{\Theta}}^{-1}\left(\mathbf{\bar{P}}(t)\right)^\top+\mathbf{\bar{P}}(t)\left(\mathbf{\tilde\Sigma}_{\mathbf{S}}\mathbf{\tilde\Sigma}_{\mathbf{S}}^\top\right)^{-1}\left(\mathbf{\bar{P}}(t)\right)^\top\right]\bar{f}(t)\bar{g}(t)\\
\label{eq:bar_g}&+\left(1-\delta\right)\mathbf{a}^\top\mathbf{\hat{\Theta}}^{-1}\left(\mathbf{b}-\mathbf{r}_n\right),\\[8pt]
0=&\bar{h}_t(t)+(1-\delta)r+\left[\left(1-\delta\right)\mathbf{\bar{P}}(t)\mathbf{\hat{\Theta}}^{-1}\left(\mathbf{b}-\mathbf{r}_n\right)+\beta\right]\bar{g}(t)+\frac{1}{2}\mathbf{\bar{P}}(t)\left(\mathbf{\tilde\Sigma}_{\mathbf{S}}\mathbf{\tilde\Sigma}_{\mathbf{S}}^\top\right)^{-1}\left(\mathbf{\bar{P}}(t)\right)^\top\bar{f}(t)\\
&+\dfrac{1}{2}\left[\left(1-\delta\right)\mathbf{\bar{P}}(t)\mathbf{\hat{\Theta}}^{-1}\left(\mathbf{\bar{P}}(t)\right)^\top+\mathbf{\bar{P}}(t)\left(\mathbf{\tilde\Sigma}_{\mathbf{S}}\mathbf{\tilde\Sigma}_{\mathbf{S}}^\top\right)^{-1}\left(\mathbf{\bar{P}}(t)\right)^\top\right]\bar{g}^2(t)\\
\label{eq:bar_h}&+\dfrac{1-\delta}{2}\left(\mathbf{b}-\mathbf{r}_n\right)^\top\mathbf{\hat{\Theta}}^{-1}\left(\mathbf{b}-\mathbf{r}_n\right).
\end{align}
\endgroup

\begin{thm}\label{thm:CRRA_case_PARTIAL_INFO} 
Let $\bar{f}(\cdot),\,\bar{g}(\cdot),\,\bar{h}(\cdot)\in\mathcal{C}^1_b([0,T])$ be the unique solutions of the following system of ODEs \eqref{eq:bar_f},\eqref{eq:bar_g},\eqref{eq:bar_h},  
with terminal conditions $\bar{f}(T)=\bar{g}(T)=\bar{h}(T)=0$. Then, the optimal control $\bar{\bmtheta}^\star$ is given by $\bar{\bmtheta}^\star_t=\bar{\bmtheta}^\star(t, \Gamma_t)$ where
\begin{equation}\label{eq:candidate_control_CRRA_PARTIAL}
\bar{\bmtheta}^\star(t,\gamma)=\bmhatTheta^{-1}\left(\bma\gamma+\bmb-\mathbf{r}_n\right)+\bmhatTheta^{-1}\mathbf{\bar{P}}(t)^\top\left(\bar{f}(t)\gamma+\bar{g}(t)\right),
\end{equation}
and the value function satisfies
\begin{equation}\label{eq:VAL_F_P_INFO}
\hat{V}(t,c,\gamma)=\dfrac{c^{1-\delta}}{1-\delta}\exp\left\lbrace\dfrac{\bar{f}(t)}{2}\gamma^2+\bar{g}(t)\gamma+\bar{h}(t)\right\rbrace.
\end{equation}
Moreover, let $(\hat{f}(t),\,\hat{g}(t),\,\hat{h}(t))$ be the unique solutions on $[0,T]$ of the systems of ODEs given by equations \eqref{eq:f_hat}, \eqref{eq:g_hat},\eqref{eq:h_hat} with $\hat{f}(T)=\hat{g}(T)=\hat{h}(T)=0$. Then, for all $t\in[0,T]$, $1-P(t)\hat{f}(t)>0$ and 
\begin{align}
\label{eq:rel_hat_f_bar_f}\bar{f}(t)=&\dfrac{\hat{f}(t)}{1-P(t)\hat{f}(t)},\\[6pt]
\label{eq:rel_hat_g_bar_g}\bar{g}(t)=&\dfrac{\hat{g}(t)}{1-P(t)\hat{f}(t)},\\[6pt]
\bar{h}(t)=&\hat{h}(t)-\dfrac{1}{2}\log\left(1-P(t)\hat{f}(t)\right)+\dfrac{1}{2}\dfrac{\hat g^2(t)P(t)}{1-P(t)\hat{f}(t)}\\
\label{eq:rel_hat_h_bar_h}&-\dfrac{1-\delta}{2}\int_t^T\dfrac{P(s)}{1-P(s)\hat{f}(s)}\left[\bm{\tilde\Sigma}_Y\bm{\tilde\Sigma}_{\mathbf{S}}^\top\hat{f}(s)+\mathbf{a}^\top\right]\mathbf{\hat\Theta}^{-1}\left[\bm{\tilde\Sigma}_Y\bm{\tilde\Sigma}_{\mathbf{S}}^\top\hat{f}(s)+\mathbf{a}^\top\right]^\top\de s,
\end{align}
implying that $\bar{f}(t),\,\bar{g}(t),\,\bar{h}(t)\in\mathcal{C}^1_{b}([0,T]).$ 
\end{thm}
\begin{proof}
The proof is provided in Appendix \ref{app:B_1}.
\end{proof}
Note that, in view of the relationship between $\hat{f}, \hat{g}, \hat{h}$ and $\bar{f}, \bar{g}, \bar{h}$ and the properties of the solution of the system \eqref{eq:f_hat}, \eqref{eq:g_hat}, \eqref{eq:h_hat}, we immediately get that the system  \eqref{eq:bar_h}, \eqref{eq:bar_g} and \eqref{eq:bar_h} admits a unique solution in $\mathcal{C}_b^{1}([0,T])$.

\begin{remark}
{The optimal strategy in equation \eqref{eq:candidate_control_CRRA_PARTIAL} id Markovian and preserves the same qualitative structure as in the full-information case, with the factor process $Y$ replaced by its filtered estimate $\Gamma$.  
However, partial information introduces an additional channel through which uncertainty affects portfolio choice. In particular, the intertemporal hedging component now depends on the process $\bar{P}(t)$, which captures the conditional variance of the estimation error. This generates an additional demand that can be interpreted as a hedge against parameter uncertainty (or learning risk), reflecting the investor’s need to account for imperfect information about expected returns. All components of the optimal allocation are multiplied by $\hat{\bm{\Theta}}^{-1}$. Hence, the carbon penalty  distorts the standard risk-return trade-off, but also affects the way the investor hedges estimation risk.} \end{remark}
We now provide sufficient conditions on model parameters that guarantee that condition (ii) of Theorem \ref{thm:verification_thm_PARTIAL_INFO} is satisfied and that $\bar{\bm{\theta}}^\star$ given by equation \eqref{eq:candidate_control_CRRA_PARTIAL} is an admissible control, according to Definition \ref{defn:F_admissible_strategies_theta}. The following Proposition is a preliminary results.
\begin{prop}\label{prop:sign_bar_f}
Let $\bar{f}(t)$ be solution of the ODE in equation \eqref{eq:bar_f} on $[0,T]$. Then, $\bar{f}(t)$ is strictly positive and decreasing on $[0,T]$ if $\delta\in\mathcal{P}\cap(0,1)$ and is strictly negative and increasing if $\delta\in\mathcal{P}\cap(1,+\infty)$.
\end{prop}
\begin{proof}
The proof is provided in Appendix \ref{app:B_2}. 
\end{proof}
Next, we will use this result to show that condition (ii) of Theorem \ref{thm:verification_thm_PARTIAL_INFO} is satisfied.
\begin{prop}\label{prop:sufficient_condition_for_ver_thm_partial_info}
Assume that one of the two following conditions holds
\begin{itemize}
\item[(i)] $\delta\in\mathcal{P}\cap(1,+\infty)$,
\item[(ii)] $\delta\in\mathcal{P}\cap(0,1)$ and
\begin{equation}\label{eq:cond_3}
1-q(1+\alpha)\dfrac{\hat{f}(0)}{1-P(0)\hat{f}(0)}\max\left\lbrace P_0,\mbox{Var}[Y_T]\right\rbrace>0
\end{equation}
for some $q>1$.
\end{itemize}
Then, for any admissible strategy $\bmtheta\in\mathcal{A}^{\mathbb{F}}$, $\{\hat V(\tau, \hat{C}_{\tau}, Y_{\tau}),\mbox{ for all }\mathbb{F}\mbox{--stopping times }\tau\le T\}$ forms a uniformly integrable family.
\end{prop}
\begin{proof}
The proof is provided in Appendix \ref{app:B_3}.
\end{proof}
To close the loop, we provide sufficient conditions for admissibility of the optimal strategy.
\begin{prop}
Assume that one of the two following conditions holds
\begin{itemize} 
\item[(i)] $\delta\in\mathcal{P}\cap(0,1)$ and 
\begin{equation}
\label{eq:cond_1_ammissibility_partial_delta_in_0_1}
1-8d(1-\delta)(1+\alpha)nT\left[\left(1\vee d(1-\delta)(1+\alpha)w\right)\tilde c_1^2+a_M^2\right]\max\left\lbrace P_0,\text{Var}[Y_T]\right\rbrace>0,
\end{equation}
\item[(ii)] $\delta\in\mathcal{P}\cap(1,+\infty)$ and
\begin{equation}\label{eq:cond_1_ammissibility_delta_partial_in_1_infty}
1-8d(1-\delta)(1+\alpha)nT\left[\left(-(1+w)\wedge d(1-\delta)(1+\alpha)\tilde w\right)\tilde{c}_1^2-a_M^2\right]\max\left\lbrace P_0,\text{Var}[Y_T]\right\rbrace>0,
\end{equation}
where $w$ and $\tilde w$ are given by equations \eqref{eq:exp_for_w} and \eqref{eq:exp_for_tilde_w} respectively, and $\tilde c_1$ is given by
\begin{equation}\label{eq:exp_for_tilde_c_1}
\tilde c_1=\max_{i=1,\dots,n}\bigg|\left(\mathbf{\hat\Theta}^{-1}\left[\mathbf{a}+\mathbf{\tilde\Sigma}_{\mathbf{   S}}\left(\mathbf{\tilde\Sigma}_{\mathbf{S}}^{-1}\right)^\top\left(\mathbf{a}\sup_{u\in[0,T]}P(u)\bar{f}(u)+\bm{\tilde\Sigma}_{\mathbf{S}}\bm{\tilde\Sigma}_Y^\top\sup_{u\in[0,T]}\bar{f}(u)\right)\right]\right)_i\bigg|
\end{equation}
\end{itemize}
Then the process $\bar\bmtheta^\star$ given by equation \eqref{eq:candidate_control_CRRA_PARTIAL} is an admissible strategy.
\end{prop}
\begin{proof}
The proof replicates the line of that of Theorem \ref{prop:sufficient_cond_admissibility_FULL_INFO}.
\end{proof}
Under the assumption of Proposition \ref{prop:sufficient_condition_for_ver_thm_partial_info}, the candidate optimal strategy is admissible and $\hat{V}$ in equation \eqref{eq:VAL_F_P_INFO} is the unique solution of the optimisation problem. As for the full information case, we can derive the original controls $\bar{m}^\star$ and $\bar{\bmpi}^\star$ by applying proposition \ref{prop:originalcontrols}. Adapting Example \ref{Ex:example_CRRA_full_info} to the case of a PI insurer with partial information, the optimal multiplier becomes 
\begin{align}\label{eq:OPT_mult_partial_info}
\bar{m}^\star(t,\gamma;\varepsilon,\delta)=\bar{\theta}_1^\star(t,\gamma;\delta)+\bar{\theta}_2^\star(t,\gamma;\delta,\varepsilon),
\end{align}
where 
\begin{align}
\bar{\theta}_1^\star(t,\gamma;\delta)&=\xi_1^M(t,\gamma;\delta)+\tilde\xi_1^I(t,\gamma;\delta)+\xi_1^P(t,\gamma;\delta),\\
\bar{\theta}_2^\star(t,\gamma;\delta,\varepsilon)&=\xi_2^M(t,\gamma;\delta,\varepsilon)+\tilde\xi_2^I(t,\gamma;\delta,\varepsilon)+\xi_2^P(t,\gamma;\delta,\varepsilon),
\end{align}
with
\begin{align}
\tilde{\xi}_1^I(t,\gamma;\delta)&=\dfrac{1}{\delta}\dfrac{\sigma_Y\rho_{1,Y}}{\sigma_1}\left(\bar{f}(t)\gamma+\bar{g}(t)\right),\quad \xi_1^P(t,\gamma;\delta)=\dfrac{1}{\delta}\dfrac{a_1P(t)}{\sigma_1^2}\left(\bar{f}(t)\gamma+\bar{g}(t)\right),\\
\tilde{\xi}_2^I(t,\gamma;\delta,\varepsilon)&=\dfrac{1}{\varepsilon+\delta}\dfrac{\sigma_Y\rho_{2,Y}}{\sigma_2}\left(\bar{f}(t)\gamma+\bar{g}(t)\right),\quad \xi_2^P(t,\gamma;\delta,\varepsilon)=\dfrac{1}{\varepsilon+\delta}\dfrac{a_2P(t)}{\sigma_2^2}\left(\bar{f}(t)\gamma+\bar{g}(t)\right),
\end{align}
for every $\left(t,\gamma\right)\in[0,T]\times\R$. $\xi_1^M(t,\gamma;\delta)$ and $\xi_2^M(t,\gamma;\delta,\varepsilon)$ are defined as in equations \eqref{eq:xi_1} and \eqref{eq:xi_2}. As shown in equation \eqref{eq:OPT_mult_partial_info}, the optimal multiplier retains the same structure obtained for the CRRA investor under complete information. However, in this case, two additional terms appear, namely $\xi^P_1$ and $\xi^P_2$, which act as correction factors accounting for the uncertainty due to the non-observability of the common stochastic factor $Y$. As for the previous cases, all the components related to the brown stock depend on the carbon aversion parameter $\varepsilon$.
\paragraph{Logarithmic case.} For the logarithmic case the separated problem reads as
\begin{equation}\label{eq:LOG_optimisation_problem_partial_info}
\mbox{Maximise }\mathbb{E}^{t,c,\gamma}\left[\log(\hat C_T^{\bmtheta})\right],\mbox{ over all }\bmtheta\in\mathcal{A}^{\mathbb{F}}
\end{equation}
and the corresponding value function is given by
\begin{equation}\label{eq:_LOGvalue_function_partial_info}
\tilde V(t,c,\gamma):=\sup_{\bmtheta\in\mathcal{A}^{\mathbb{F}}}\mathbb{E}^{t,c,\gamma}\left[\log(\hat C_T^{\bmtheta})\right].
\end{equation}
The next theorem characterizes the optimal strategy and the value function $\tilde V$. 
\begin{cor}\label{cor:log_case_PARTIAL_INFO} Consider a fund manager endowed with logarithmic utility function and a carbon aversion $\varepsilon\ge 0$, then the optimal controls $\bar{\bmtheta}^\star\in\mathcal{A}^{\mathbb{F}}$ is given by $\bar{\bmtheta}^\star_t=\bar{\bmtheta}^\star(t, \Gamma_t)$ where 
\begin{align}\label{eq:sol_partial_info_log}
\bar{\bmtheta}^\star(t,\gamma)=\mathbf{\Theta}^{-1}\left(\bma\gamma+\bmb-\mathbf{r}_{n}\right).
\end{align}
where $\mathbf{\Theta}$ is the same of Corollary \ref{cor:solution_full_info_log}. The value function is given by 
\begin{equation}\label{eq:PARTIAL_INFO_VALUE_FUN_LOG_CASE}
\tilde V(t,c,\gamma)=\log(c)+r(T-t)+\dfrac{f(t)}{2}\gamma^2+g(t)\gamma+\tilde{h}(t),
\end{equation}
where
\begin{equation}\label{eq:rel_log_case}
\tilde{h}(t)=h(t)+\dfrac{\mathbf{a}^\top\mathbf{\Theta}^{-1}\mathbf{a}}{2}\left(\int_t^TP(s)\de s-P(t)\dfrac{e^{2\lambda(T-t)}-1}{2}\right),
\end{equation}
for every $t\in[0,T]$, with $f$, $g$ and $h$ being the same of Corollary \ref{cor:solution_full_info_log}.
\end{cor}
\begin{proof}
The proof is provided in Appendix \ref{app:B_4}.
\end{proof}
\subsection{Loss of utility}
Since full information allows the portfolio insurer to observe the common stochastic factor directly, the fully informed portfolio insurer has an advantage over its partial-information counterpart. Therefore, as shown in \cite{lee2016pairs}, there is always an \textit{information premium}, which is non-negative. In the present paper, we quantify this premium by computing the \textit{loss of utility} $L=\{L_t\}_{t\in[0,T]}$ due to partial information, defined as
\begin{equation}
L_t=\mathbb{E}^{c}\left[V^{\mathrm{full}}(t,C,Y_t)-V^{\mathrm{partial}}(t,C,\Gamma_t)|\mathcal{F}_t\right],\quad t\in[0,T].
\end{equation}
An alternative way to assess the informational advantage is to express the information premium in monetary terms; this is the so-called \textit{efficiency} (see, e.g., \cite{rogers2001relaxed}, \cite{brendle2006portfolio} and \cite{sass2017expert}).
Specifically, in the PPI framework, the efficiency of the partially-informed strategy relative to the full-information strategy is defined as the fraction of the initial cushion $\xi$ that a fully informed investor would need to obtain the same the expected utility of the terminal cushion achieved by a partially informed investor starting with a unitary cushion. Hence, it is found by solving the following equation for $\zeta$:
\begin{equation}\label{eq:efficiency_definition}
\mathbb{E}\left[V^{\mathrm{full}}(0,\zeta,Y_0)-V^{\mathrm{partial}}(0,1,\Gamma_0)|\mathcal{F}_0\right]=0.
\end{equation}
In what follows, we analytically characterise the loss of utility and the efficiency of a portfolio insurer who does not directly observe the common stochastic factor $Y$, for both the CRRA and log-utility cases.
\begin{prop}\label{prop:LOSS_UTILITY_CRRA} 
The loss of utility of a partially informed portfolio insurer endowed with a CRRA utility function is given by
\begin{equation}\label{eq:LOSS_UTILITY_CRRA_CASE}
L_t=\frac{c^{1-\delta}}{1-\delta}\left(e^{\frac{1-\delta}{2}\int_t^T\frac{P(s)}{1-P(s)\hat{f}(s)}\left[\bm{\tilde\Sigma}_Y\bm{\tilde\Sigma}_{\mathbf{S}}^\top\hat{f}(s)+\mathbf{a}^\top\right]\mathbf{\hat\Theta}^{-1}\left[\bm{\tilde\Sigma}_Y\bm{\tilde\Sigma}_{\mathbf{S}}^\top\hat{f}(s)+\mathbf{a}^\top\right]^\top\de s}-1\right)e^{\frac{\bar{f}(t)}{2}\Gamma^2_t+\bar{g}(t)\Gamma_t+\bar{h}(t)},
\end{equation}
for every $t\in[0,T]$, and the corresponding efficiency of the carbon-penalised PPI strategy is given by
\begin{equation}\label{eq:efficiency_CRRA_CASE}
\zeta=\exp\left\lbrace-\frac{1}{2}\int_0^T\frac{P(s)}{1-P(s)\hat{f}(s)}\left[\bm{\tilde\Sigma}_Y\bm{\tilde\Sigma}_{\mathbf{S}}^\top\hat{f}(s)+\mathbf{a}^\top\right]\mathbf{\hat\Theta}^{-1}\left[\bm{\tilde\Sigma}_Y\bm{\tilde\Sigma}_{\mathbf{S}}^\top\hat{f}(s)+\mathbf{a}^\top\right]^\top\de s\right\rbrace.
\end{equation}
\end{prop}
\begin{proof}
The proof is provided in Appendix \ref{app:B_5}.
\end{proof}
\begin{prop}\label{cor:loss_utility_eff_log}
The loss of utility of a partially informed portfolio insurer endowed with a logarithmic utility function is given by
\begin{equation}\label{eq:LOSS_OF_UTILITY_LOG__CASE}
L_t=\frac{\mathbf{a}^\top\mathbf{\Theta}^{-1}\mathbf{a}}{2}\int_t^TP(s)\de s,
\end{equation}
for every $t\in[0,T]$, and the efficiency of the corresponding carbon-penalised PPI strategy is given by
\begin{equation}\label{eq:EFFICIENCY_LOG__CASE}
\zeta=\exp\left\lbrace-\frac{\mathbf{a}^\top\mathbf{\Theta}^{-1}\mathbf{a}}{2}\int_0^TP(s)\de s\right\rbrace.
\end{equation}
\end{prop}
\begin{proof}
See Appendix \ref{app:B_6}.
\end{proof}

Proposition \ref{prop:LOSS_UTILITY_CRRA}, and more evidently Proposition \ref{cor:loss_utility_eff_log}, show that the loss of utility is strictly positive. This outcome was to be expected, since partially informed strategies constitute a subset of the fully informed ones. Consequently, a portfolio insurer with full information can always replicate, or improve upon, the performance achievable under partial information. Equivalently, the relative efficiency of the carbon-penalised strategy under partial information, vis-à-vis its full-information counterpart, is given by $\zeta < 1$, confirming that partial information entail a reduction in attainable utility.

\section{Numerical experiments}\label{sect:num_experiments}

In this section, we perform an {illustrative} simulation study to examine the behavior of the optimal carbon-penalised PPI strategy and to compare the strategies of a fully informed versus a partially informed portfolio insurer. We consider $n=4$ traded stocks and start from a given carbon classification: the first two assets are labelled as low-carbon (green), while the remaining two are labelled as high-carbon (brown). {In empirical applications, such a classification can be obtained following the procedure provided by \cite{ardia2023factor}. The approach is as follows: at each evaluation date, stocks are ranked according to their most recently available carbon intensity values. Firms in the bottom (top) quartile of the resulting cross-sectional distribution are classified as green (brown), respectively. This percentile-based sorting yields two groups that capture firms with relatively low and high carbon intensity.\footnote{{In our stylised numerical experiments, we do not implement the clustering step explicitly and work with a given green/brown classification. Carbon intensity values are therefore not reported, as they determine the group membership that triggers the penalisation.}}}\\
Unless otherwise stated, model parameters are fixed as in Table \ref{tab:model_params}. Moreover, throughout the numerical experiments, we fix the risk-free rate at $r=0.01$, the PPI protection level at $\mathrm{PL}=1$, and the initial wealth at $V_0=1$.\\

\begin{table}[!htbp]
\centering
\captionsetup[subtable]{justification=centering} 
\setlength{\tabcolsep}{6pt}
\renewcommand{\arraystretch}{1.1}
\begin{subtable}[t]{0.28\textwidth}
\centering
\begin{tabular}{lccc}
\Xhline{1.5pt}
& $\mathbf{a}$ & $\mathbf{b}$ & $\mathbf{\Sigma}_{\mathbf{S}}$ \\
\midrule
$S_1$ & $0.080$ & $-0.03$ & $0.19$ \\
$S_2$ & $0.055$ & $\phantom{-}0.01$  & $0.21$ \\
$S_3$ & $0.045$ & $\phantom{-}0.01$  & $0.22$ \\
$S_4$ & $0.075$ & $-0.03$ & $0.15$ \\
\Xhline{1.5pt}
\end{tabular}
\caption{Parameters of the stock prices.}
\label{Stocks_parameters}
\end{subtable}
\hfill
\begin{subtable}[t]{0.28\textwidth}
\centering
\begin{tabular}{ccccc}
\Xhline{1.5pt}
$\lambda$ & $\beta$ & $\sigma_Y$ & $\Gamma_0$ & $P_0$\\
\midrule
$-0.5$ & $0.5$ & $0.05$ & $1$ & $0.0025$ \\
\Xhline{1.5pt}
\end{tabular}
\caption{Parameters of the common stochastic factor $Y$.}
\label{latent_factor_Y_parameters}
\end{subtable}
\hfill
\begin{subtable}[t]{0.32\textwidth}
\centering
\setlength{\arraycolsep}{3pt}
\renewcommand{\arraystretch}{1.05}
\small
\resizebox{\linewidth}{!}{$
\mathbf{R}=
\begin{pmatrix}
1.00 & \phantom{-}0.32  & \phantom{-}0.25  & 0.10  & \phantom{-}0.35 \\
0.32 & \phantom{-}1.00  & \phantom{-}0.30  & 0.12  & -0.25 \\
0.25 & \phantom{-}0.30  & \phantom{-}1.00  & 0.20  & -0.15\\
0.10 & \phantom{-}0.12  & \phantom{-}0.20  & 1.00  & \phantom{-}0.325\\
0.35 &           -0.25  &           -0.15  & 0.325 & \phantom{-}1
\end{pmatrix}
$}
\caption{Correlation matrix $\mathbf{R}$.}
\label{R_matrix}
\end{subtable}
\caption{General parameters for the numerical study.}
\label{tab:model_params}
\end{table}

To understand the relationship between the unobservable factor process $Y$ at its filtered estimate $\Gamma$, we compare a single trajectory of these processes in Figure \ref{fig:Y_vs_filter}. The filter (dashed magenta line) shows less variability than the true trajectory, yet is able to capture the upward and downward trends of the factor $Y$ (solid blue line). We recall that the goodness of the filter depends highly on the signal-to-noise ratio. In particular, if volatility of stock prices is large, the observation is noisy, the filter gets worse.
\begin{figure}[!htbp]
\centering
\includegraphics[width=0.60\linewidth]{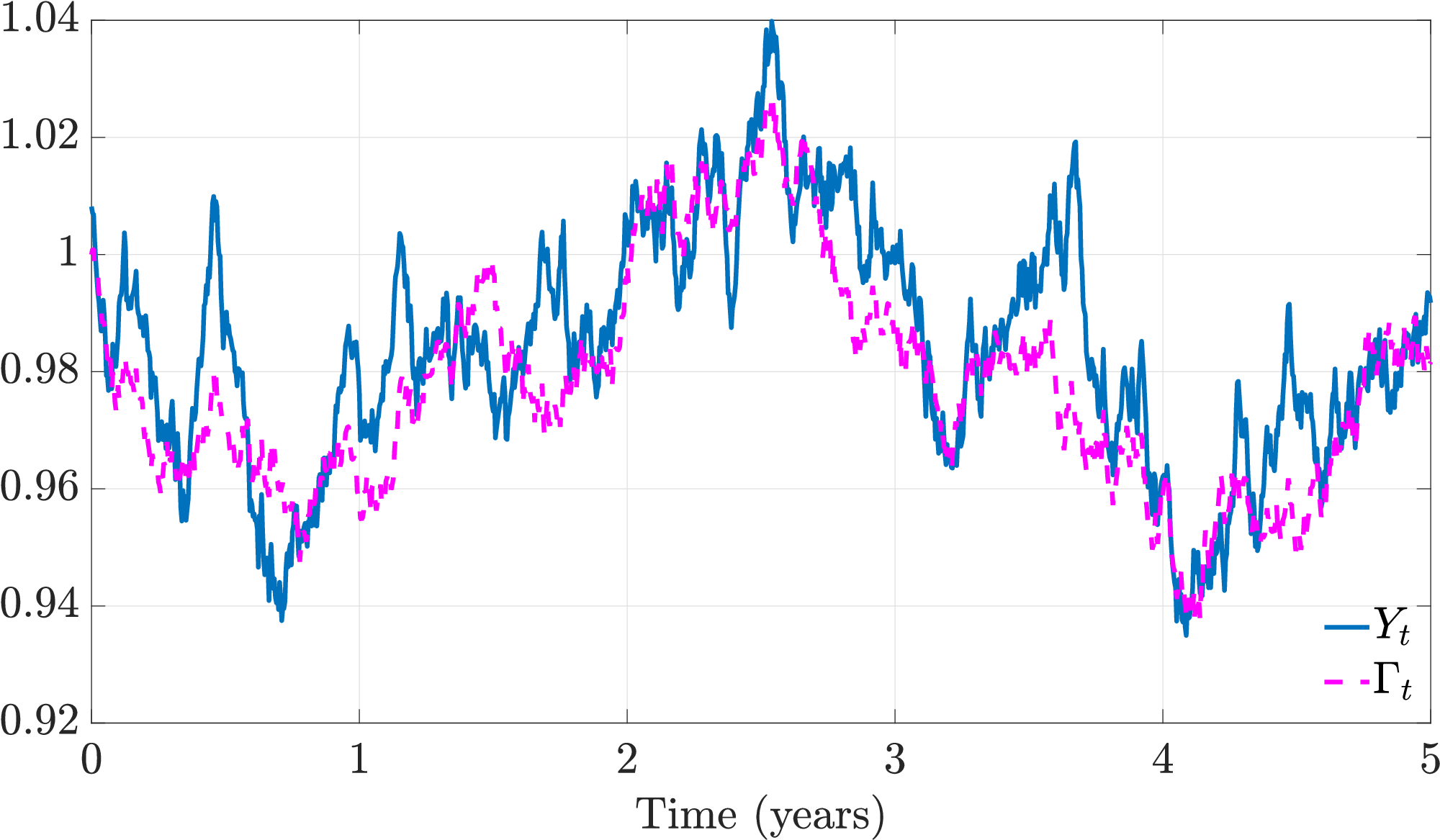}
\caption{True trajectory of the common stochastic factor $Y$ (solid blue line) and trajectory of its filtered estimate $\Gamma$ (dashed magenta line).}\label{fig:Y_vs_filter}
\end{figure}
\medskip
\subsection{Numerical Experiments for the partial information case}
We begin our analysis with a numerical study of  the optimal exposures of the carbon-penalised PPI strategy to the traded stocks. We focus on the partial information case, which is one of the key features of our model. We denote by $\mathbf{\bar E}^\star=\{\mathbf{\bar E}_{t}^\star\}_{t\in[0,T]}$ the exposure to the risky assets, where $\mathbf{\bar E}^\star_t=\left(\bar E^\star_{1,t},\,\dots,\,\bar E^\star_{n,t}\right)^{\top}$, are given by 
\begin{equation}\label{eq:exposures}
\bar{E}^{\star}_{i,t}:=\bar{m}^\star_t\bar{\pi}_{i,t}^\star\frac{(V^{\bar m^\star,\bar\bmpi^{\star}}_t-F_t)^{+}}{V^{\bar m^\star,\bar\bmpi^{\star}}_t},\quad t\in[0,T],
\end{equation}
for every $i=1,\,\dots,\,n$, and for the optimal strategy under partial information $(\bar m^\star,\bar\bmpi^{\star})$ (we recall here that $V^{\bar m^\star,\bar\bmpi^{\star}}$ is the value of the strategy under partial information). We conduct a static analysis at $t=0$ and a dynamic one thereafter. The bar charts in Figure \ref{fig:comp_risky_reference_port_t_0} show the optimal exposures $\bar E^{\star}_{i,0}$ to each traded stock at $t=0$, for every $i=1,\dots,n$, across different levels of the portfolio insurer’s risk aversion $\delta$ and carbon aversion $\varepsilon$. Each panel corresponds to a specific combination of $\delta\in\{0.7,\,1,\,3\}$ and $\varepsilon\in\{0,\,1\}$, for a direct comparison of the effects of carbon aversion. The results show that, as carbon aversion $\varepsilon$ increases, the optimal exposures to brown (carbon-intensive) stocks decrease and those to green stocks increase. 
A reduction in exposure to carbon-intensive stocks appears in every configuration, but the magnitude of this reduction depends on risk aversion. In particular, when $\delta = 3$, which corresponds to a high level of risk aversion, the percentage reduction is smaller. This is because the risky reference portfolio is conservative, hence the exposure is already low in that case. Similar results apply to the optimal PPI strategy under full information.\\
\begin{figure}[!htbp]
\centering
\includegraphics[width=0.48\linewidth]{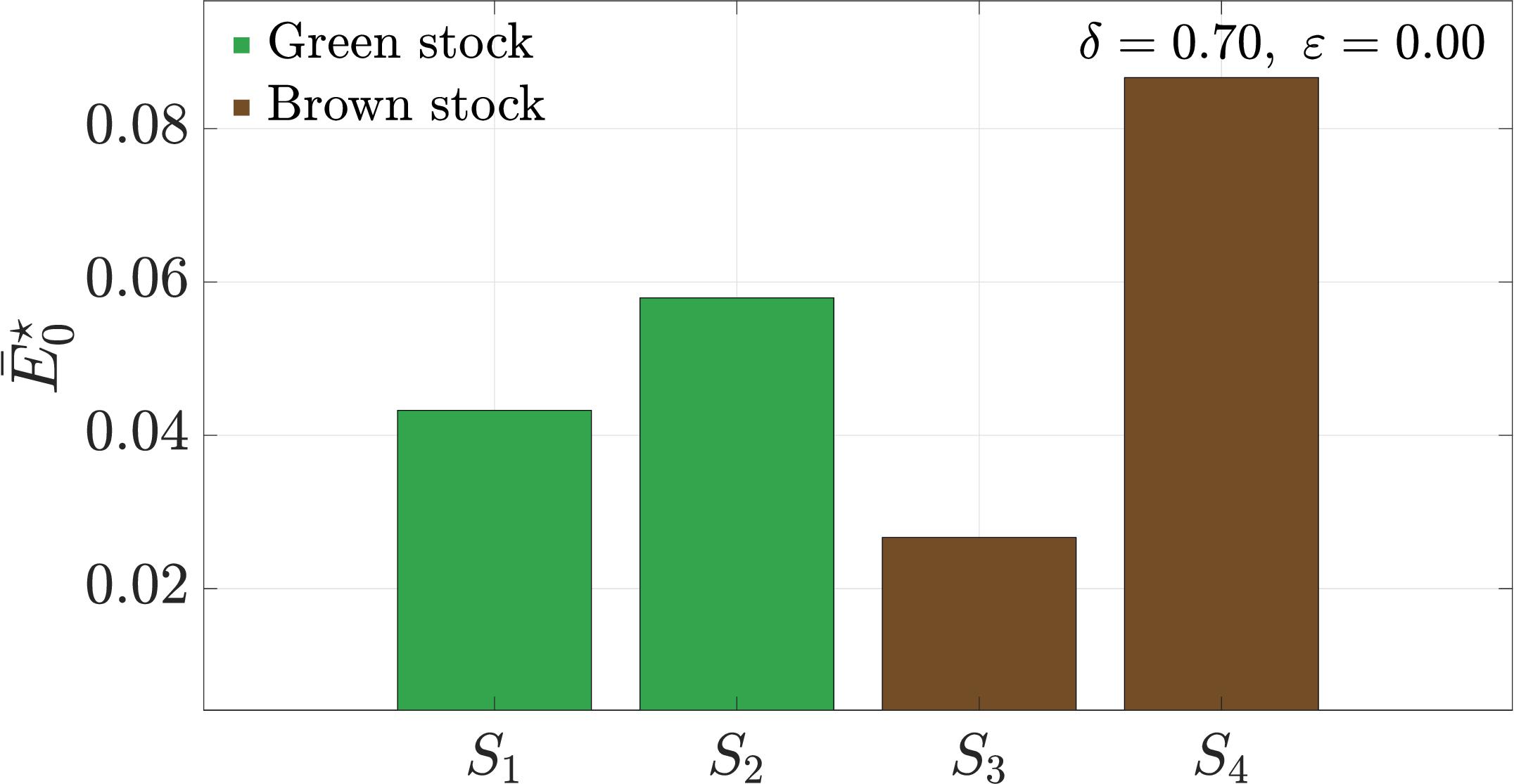} \vspace{.5cm}
\hfill 
\includegraphics[width=0.48\linewidth]{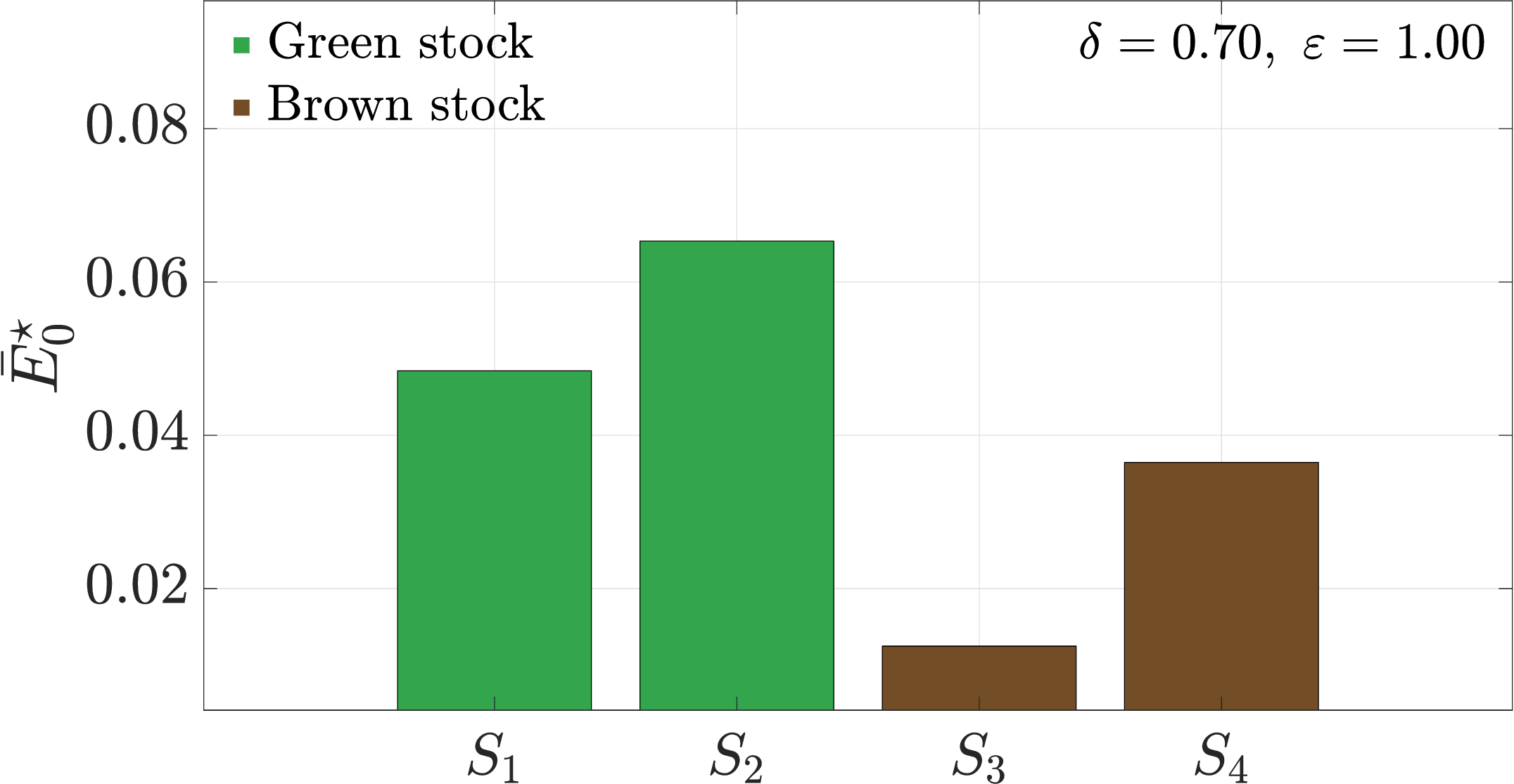}
\includegraphics[width=0.48\linewidth]{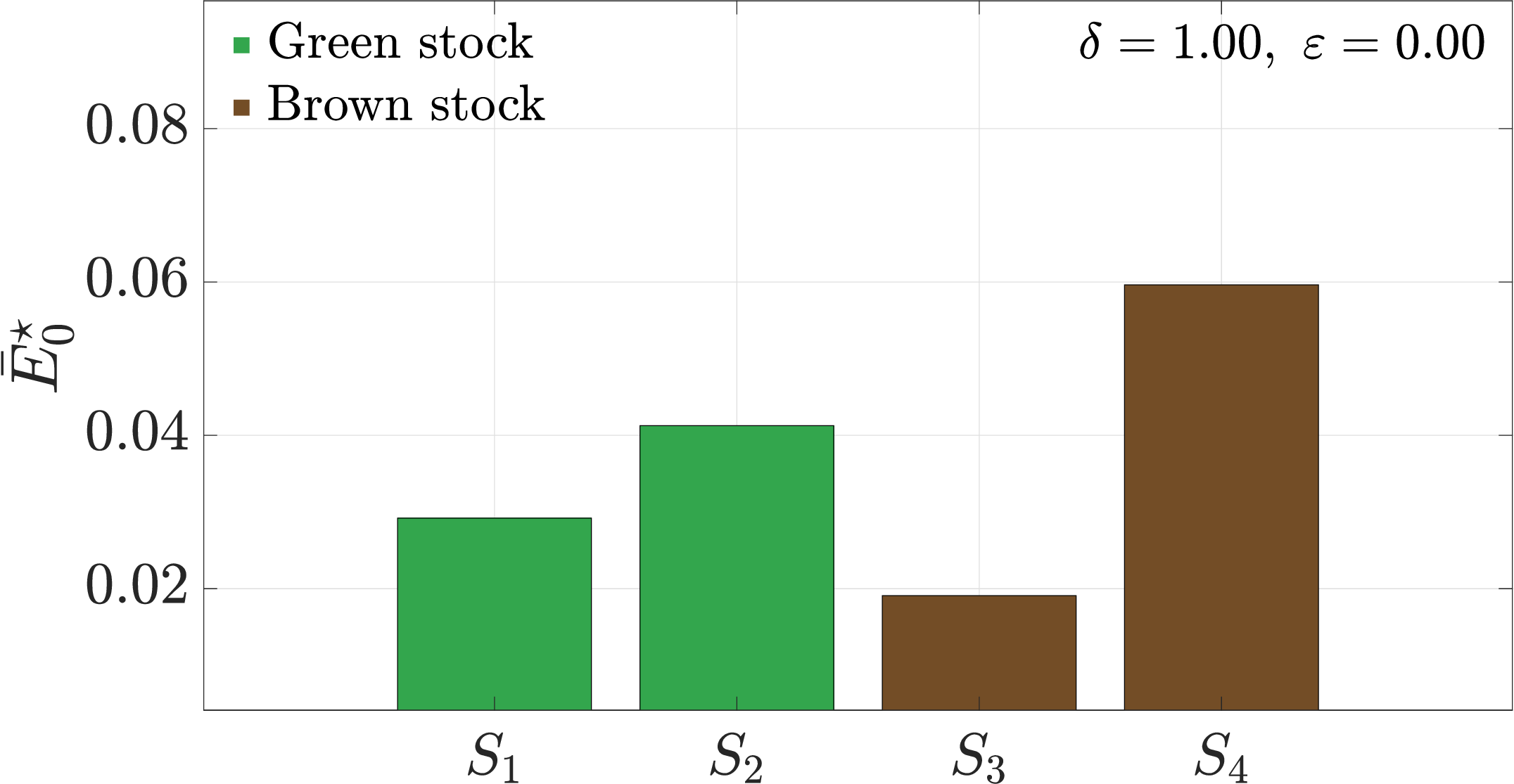} \vspace{.5cm}
\hfill 
\includegraphics[width=0.48\linewidth]{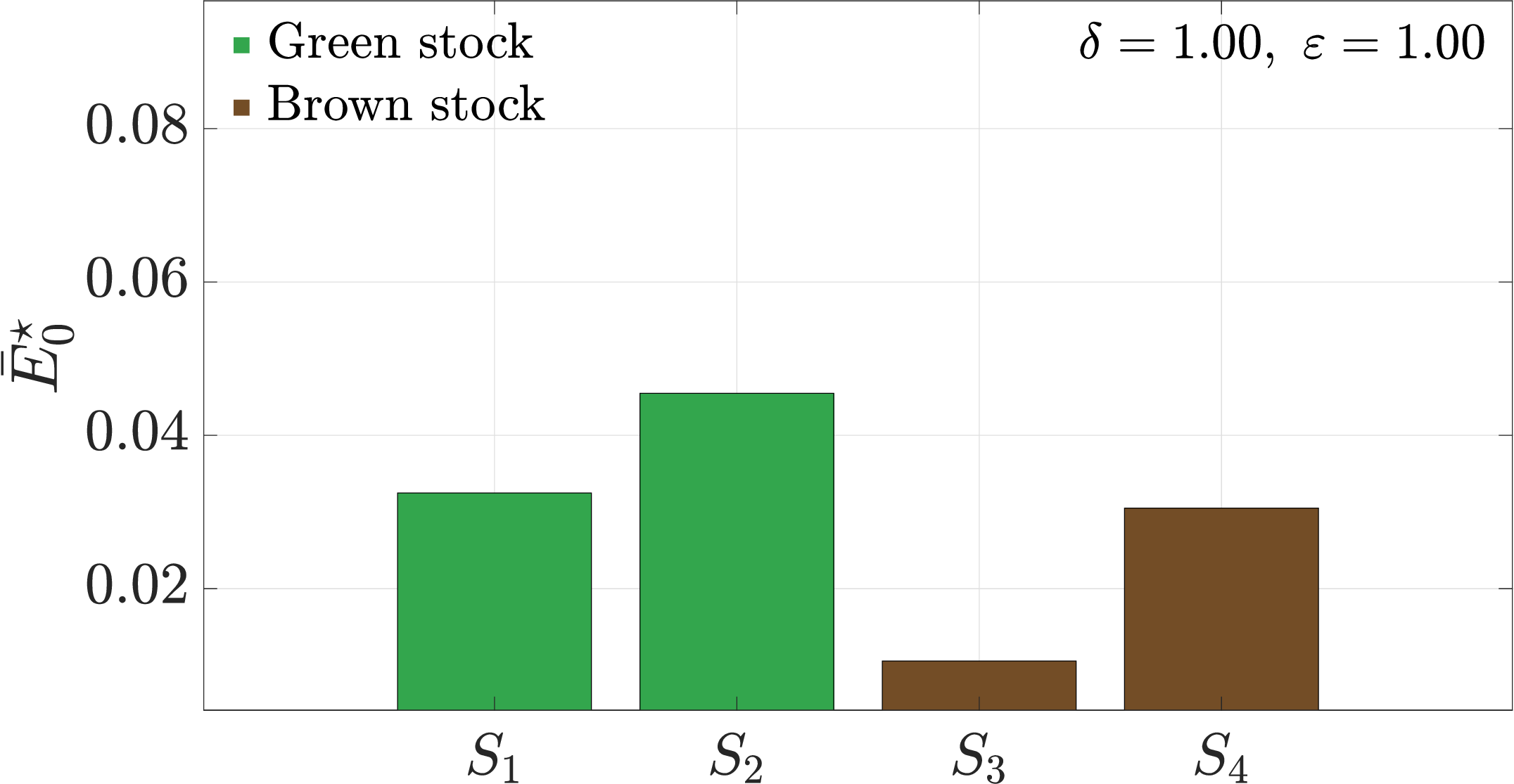}
\includegraphics[width=0.48\linewidth]{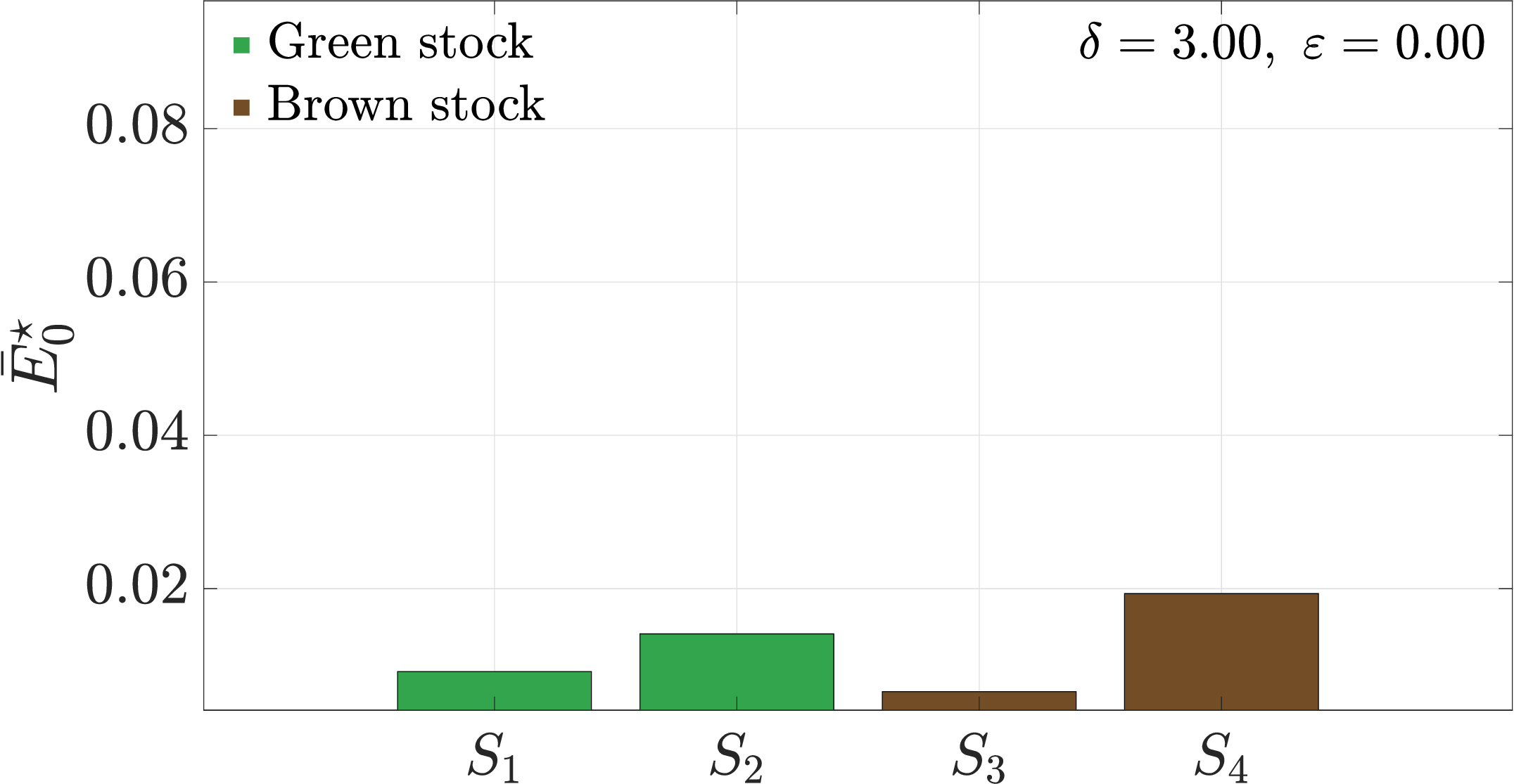}
\hfill
\includegraphics[width=0.48\linewidth]{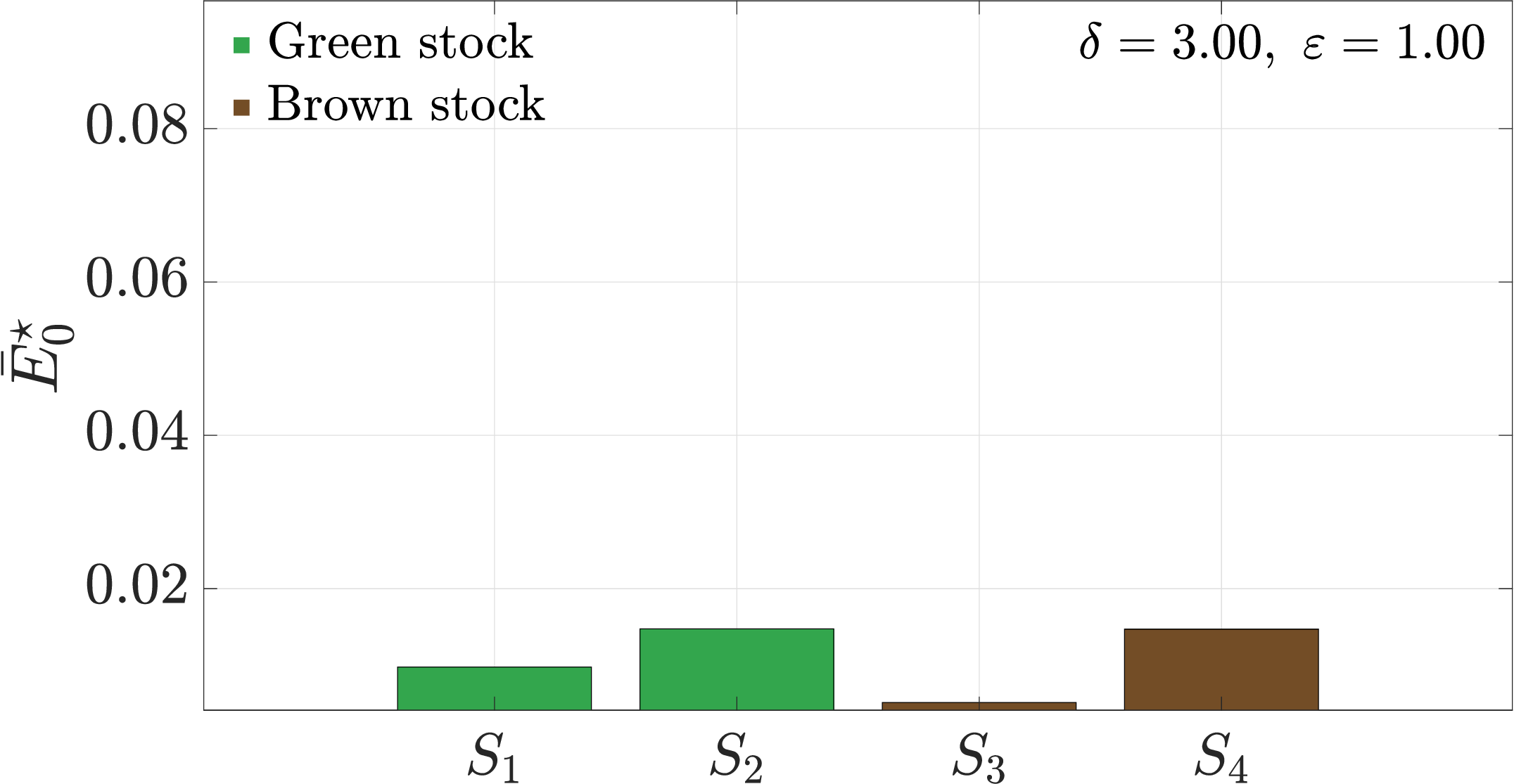}
\caption{{Bar charts displaying the optimal exposure to the $i$-th stock of the carbon-penalised PPI strategy 
at $t=0$ for different levels of $\delta$ and $\varepsilon$}.}\label{fig:comp_risky_reference_port_t_0}
\end{figure}

Figure \ref{fig:expsoure_risk_free}, illustrates the optimal multiplier $\bar m^\star$ (left panel) and the optimal exposure to the risk-free asset $S^0$ (right panel) at $t=0$ as functions of carbon aversion $\varepsilon$, and offers a description of the same effect from a different angle.
When $\varepsilon=0$, the PPI strategy’s exposure to the risky assets is entirely determined by risk aversion $\delta$. In particular, relatively low levels of the risk-aversion parameter (e.g., $\delta = 0.7$ and $\delta = 1$) lead to high values of the multiplier and large exposures to $\mathbf{S}$ (hence lower exposure to the risk free asset). Conversely, a higher $\delta$ implies a lower optimal multiplier $\bar m^\star$ and thus a smaller exposure to $\mathbf{S}$, which -- under the PPI mechanism -- results in a larger allocation to the risk-free asset $S^0$. Similarly, as $\varepsilon$ increases, $\bar m^\star$ decreases, implying a lower exposure to carbon intensive stocks. This translates in a higher allocation to $S^0$, in particular in cases where the risk aversion $\delta$ is low.
\begin{figure}[!htbp]
\centering
\includegraphics[width=0.48\linewidth]{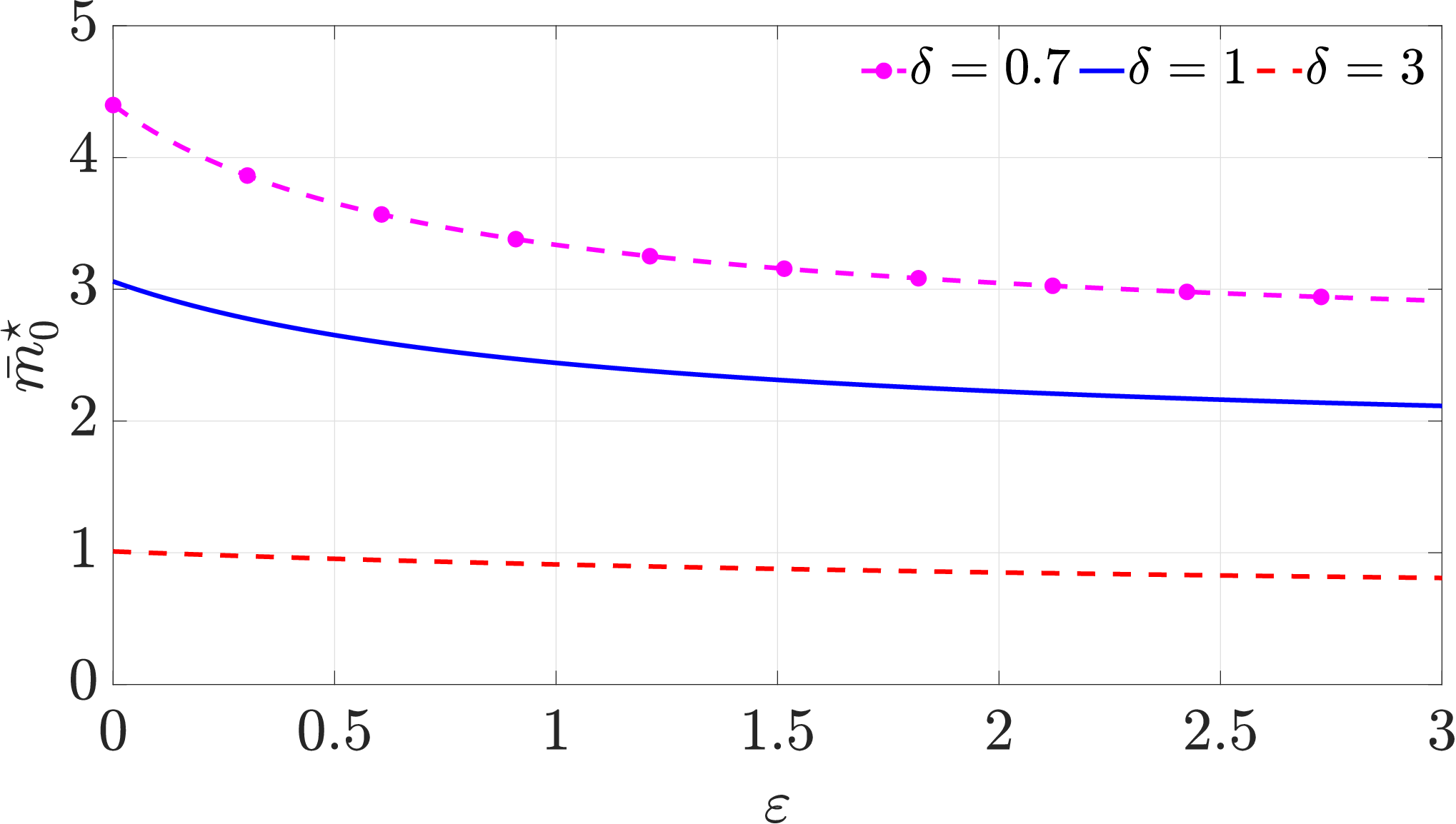}
\hfill 
\includegraphics[width=0.48\linewidth]{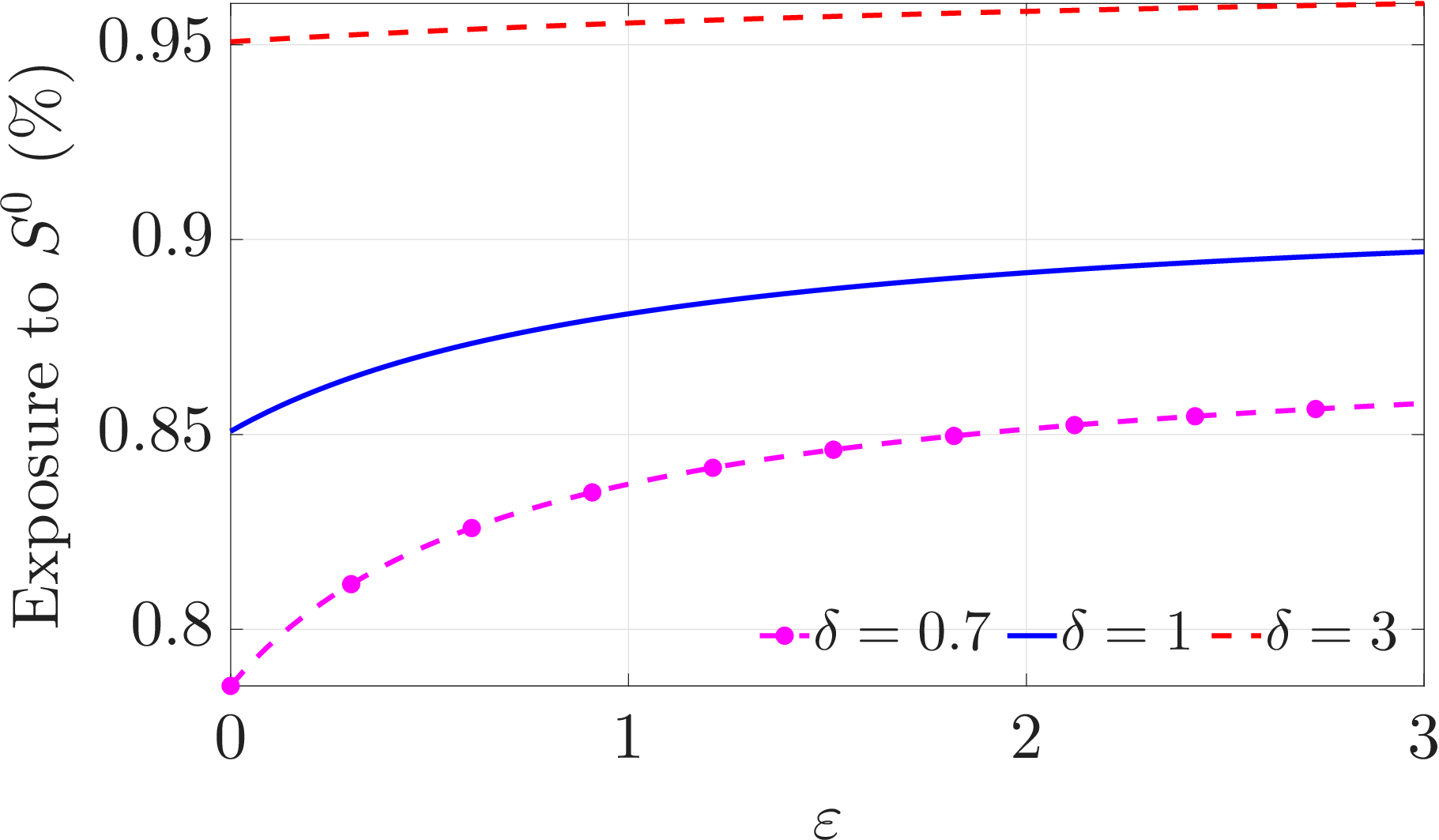}
\caption{Optimal multiplier $\bar m^\star_0$ (left panel) and optimal exposure to the risk-free asset $S^0$ (right panel) as a function of carbon aversion $\varepsilon$. The optimal PPI strategy's exposure to $S^0$ is given by $1-\mathbf{1}^\top\mathbf{\bar E}^\star_t$ for every $t\in[0,T]$.}\label{fig:expsoure_risk_free}
\end{figure}

\medskip

We now turn to the dynamic analysis. To illustrate how the proposed carbon-penalised PPI strategy shapes the allocation mechanism, we simulate the optimal exposures $\mathbf{\bar E}^\star$ over the entire investment horizon. The results, reported in Figure \ref{fig:dynamic_analysis_partial_info}, indicate that our strategy successfully {incorporates joint consideration of financial performance and carbon exposure}. In particular, brown stock number $3$ is assigned the lowest average exposure, reflecting the strategy’s sensitivity to sustainability criteria. However, the methodology is not limited to a naïve exclusion of carbon-intensive assets. Indeed, although stock number $4$ is also brown, it has a similar exposure as that of the green stock number $1$. This is because stock number $4$ exhibits the highest Sharpe ratio ($\mathrm{SR}$). This demonstrates that the penalisation mechanism does not merely exclude high-carbon assets; rather, it adjusts allocations based on a balanced evaluation of both environmental and financial features.
\begin{figure}[!htbp]
\centering
\includegraphics[width=0.75\linewidth]{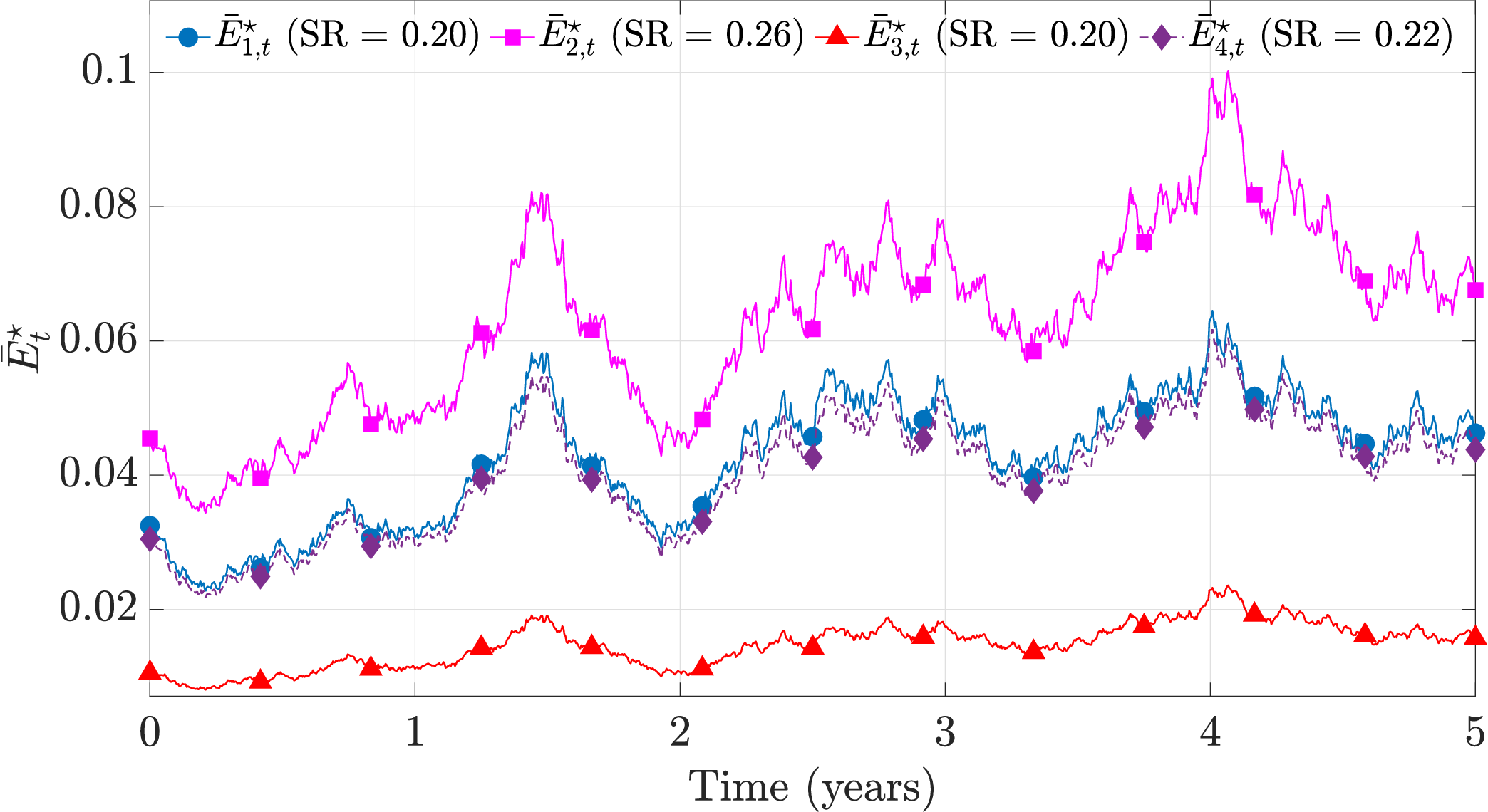}
\caption{Simulated paths of the carbon-penalised PPI strategy’s optimal exposures to $\mathbf{S}$. Parameters of $\mathbf{S}$ and $Y$ are reported in Table \ref{tab:model_params}. PPI strategy parameters: $\delta=1$, $\varepsilon=1$, $V_0=1$, $\mathrm{PL}=1$ and $T=5$ years.}\label{fig:dynamic_analysis_partial_info}
\end{figure}

Table \ref{tab:three_scenarios_perf_analysis_partial_info} shows how carbon aversion $\varepsilon$ and risk aversion $\delta$ shape the distribution of the terminal wealth of the optimal PPI strategy, under three scenarios: Scenario $1$, where green stocks outperform brown stocks; Scenario $2$, where green and brown stocks perform similarly; and Scenario $3$, where green stocks underperform brown stocks.
To generate the three scenarios, we specify three different drift vectors $\mathbf{a}$ for the stock price process $\mathbf{S}$ (reported in Table \ref{tab:scenarios}), while keeping all other parameters fixed as in Table \ref{tab:model_params}.
\begin{table}[!htbp]
\centering
\begin{tabular}{ccccc} 
\Xhline{1.5pt}
           & $a_1$   & $a_2$   & $a_3$   & $a_4$    \\ 
\hline
Scenario 1 & $0.090$ & $0.080$ & $0.045$ & $0.045$  \\
Scenario 2 & $0.080$ & $0.055$ & $0.045$ & $0.075$  \\
Scenario 3 & $0.045$ & $0.045$ & $0.080$ & $0.090$  \\
\Xhline{1.5pt}
\end{tabular}
\caption{Drift vector $\mathbf{a}$ for the three different scenarios.}
\label{tab:scenarios}
\end{table}
The results in Table \ref{tab:three_scenarios_perf_analysis_partial_info} show that, comparing the cases $\varepsilon=0$ and $\varepsilon =1$, the expected value of the optimal PPI strategy remains essentially unchanged, while the variance markedly reduces, in all scenarios and  level of risk aversion $\delta$. Furthermore, looking at the $5$th and $90$th quantiles, an increase in $\varepsilon$ raises the left tail and lowers the right tail, improving downside protection while reducing upside capture. Such a shrinkage effect is weaker in Scenario $1$ where green stocks outperform brown ones, moderate in Scenario $2$ where green and brown securities have similar performance, and stronger where brown stocks are more attractive than green ones.
\begin{table}[!htbp]
\centering
\begin{tabular}{ccccccccc} 
\Xhline{1.5pt}
\multicolumn{9}{c}{$\delta=0.7$}                                                                                       \\ 
\Xhline{1.5pt}
          & \multicolumn{2}{c}{Scenario 1} &  & \multicolumn{2}{c}{Scenario 2} &  & \multicolumn{2}{c}{Scenario 3}  \\ 
\cline{2-3}\cline{5-6}\cline{8-9}
          & $\varepsilon=0 $ & $\varepsilon=1$ &  & $\varepsilon=0$  & $\varepsilon=1$ &  & $\varepsilon=0$  & $\varepsilon=1$ \\ 
\cline{2-3}\cline{5-6}\cline{8-9}
$\mathbb{E}[V^{\bar m^\star,\bar\bmpi^{\star}}_T]$ & $1.1575$ & $1.1534$  & & $1.1208$ &   $1.1025$ & & $1.2445$	& $1.1213 $\\
$\mathrm{Var}[V^{\bar m^\star,\bar\bmpi^{\star}}_T]$ & $0.0821$ & $0.0790$  & & $0.0347$ &	$0.0142$ & & $0.3522$	& $0.0177$ \\
$q_{0.05}(V^{\bar m^\star,\bar\bmpi^{\star}}_T)$ & $1.0084$ & $1.0085$  & & $1.0117$ &	$1.0139$ & & $1.0076$	& $1.0186$ \\
$q_{0.50}(V^{\bar m^\star,\bar\bmpi^{\star}}_T)$ & $1.0771$ & $1.0773$  & & $1.0692$ &	$1.0668$ & & $1.0817$	& $1.0845$ \\
$q_{0.90}(V^{\bar m^\star,\bar\bmpi^{\star}}_T)$ & $1.3538$ & $1.3353$  & & $1.2556$ &	$1.2208$ & & $1.4960$	& $1.2473$ \\
\Xhline{1.5pt}
\multicolumn{9}{c}{$\delta=1$}                                                                                       \\ 
\Xhline{1.5pt}
          & \multicolumn{2}{c}{Scenario 1} &  & \multicolumn{2}{c}{Scenario 2} &  & \multicolumn{2}{c}{Scenario 3}  \\ 
\cline{2-3}\cline{5-6}\cline{8-9}
          & $\varepsilon=0 $ & $\varepsilon=1$ &  & $\varepsilon=0$  & $\varepsilon=1$ &  & $\varepsilon=0$  & $\varepsilon=1$ \\ 
\cline{2-3}\cline{5-6}\cline{8-9}
$\mathbb{E}[V^{\bar m^\star,\bar\bmpi^{\star}}_T]$ & $1.1154$ & $1.1135$ & & $1.0949$ & $1.0860$ & & $1.1540$ & $1.1016$ \\
$\mathrm{Var}[V^{\bar m^\star,\bar\bmpi^{\star}}_T]$ & $0.0145$ & $0.0136$ & & $0.0072$ & $0.0040$ & & $0.0447$ & $0.0060$ \\
$q_{0.05}(V^{\bar m^\star,\bar\bmpi^{\star}}_T)$ & $1.0174$ & $1.0174$	& & $1.0210$ & $1.0227$ & & $1.0172$ & $1.0263$ \\
$q_{0.50}(V^{\bar m^\star,\bar\bmpi^{\star}}_T)$ & $1.0820$ & $1.0814$ & & $1.0728$ & $1.0698$	& & $1.0912$ & $1.0816$ \\
$q_{0.90}(V^{\bar m^\star,\bar\bmpi^{\star}}_T)$ & $1.2373$ & $1.2285$	& & $1.1823$ & $1.1620$	& & $1.3175$ & $1.1873$ \\
\Xhline{1.5pt}
\multicolumn{9}{c}{$\delta=3$}                                                                                       \\ 
\Xhline{1.5pt}
          & \multicolumn{2}{c}{Scenario 1} &  & \multicolumn{2}{c}{Scenario 2} &  & \multicolumn{2}{c}{Scenario 3}  \\ 
\cline{2-3}\cline{5-6}\cline{8-9}
          & $\varepsilon=0 $ & $\varepsilon=1$ &  & $\varepsilon=0$  & $\varepsilon=1$ &  & $\varepsilon=0$  & $\varepsilon=1$ \\ 
\cline{2-3}\cline{5-6}\cline{8-9}
$\mathbb{E}[V^{\bar m^\star,\bar\bmpi^{\star}}_T]$   & $1.0681$ & $1.0679$ & & $1.0635$ & $1.0625$ & & $1.0743$ & $1.0694$ \\
$\mathrm{Var}[V^{\bar m^\star,\bar\bmpi^{\star}}_T]$ & $0.0004$ & $0.0004$ & & $0.0003$ & $0.0002$ & & $0.0007$ & $0.0004$ \\
$q_{0.05}(V^{\bar m^\star,\bar\bmpi^{\star}}_T)$     & $1.0394$ & $1.0397$ & & $1.0409$ & $1.0415$ & & $1.0403$ & $1.0422$ \\
$q_{0.50}(V^{\bar m^\star,\bar\bmpi^{\star}}_T)$     & $1.0659$ & $1.0656$ & & $1.0619$ & $1.0612$ & & $1.0706$ & $1.0666$ \\
$q_{0.90}(V^{\bar m^\star,\bar\bmpi^{\star}}_T)$     & $1.0939$ & $1.0931$ & & $1.0843$ & $1.0816$ & & $1.1068$ & $1.0935$ \\
\Xhline{1.5pt}
\end{tabular}
\caption{Mean, variance, and $5$th/$50$th/$90$th quantiles of the distribution of the optimal carbon-penalised PPI strategy at $T=5$ in the partial information case, for risk-aversion levels $\delta=0.7$ (top panel), $\delta=1$ (middle panel), and $\delta=3$ (bottom panel), comparing $\varepsilon=0$ and $\varepsilon=1$, under the three scenarios.}
\label{tab:three_scenarios_perf_analysis_partial_info}
\end{table}
As an example, at $\delta=0.7$, the variance decreases by $5.5$\% in Scenario $1$, $58.1$\% in Scenario $2$ and $94.7$\% in Scenario $3$, while the interquartile range ($[q_{0.05}, q_{0.90}]$) is reduced by $6.3$\%, $18.6$\% and $53.8$\%, respectively. Similar considerations apply to $\delta=1$ and $\delta=3$, albeit with smaller numbers. 

Figure \ref{fig:expsoure_risk_free} reports the optimal multiplier (left panel) and the corresponding exposure to the risk free asset (right panel) as functions of $\varepsilon$ for different values of $\delta$. These plots consent us to draw the following conclusions. The multiplier is decreasing in the carbon aversion $\varepsilon$ and in risk aversion $\delta$. On the contrary the exposure to the risk free asset is increasing. The effect of an increase in carbon aversion is more contained when portfolio insurer is more risk averse. In summary, $\delta$ produces a generalized reduction in the riskiness of the strategy as it indiscriminately decreases the investments in green and brown stocks. .{In contrast, carbon aversion acts  specifically on carbon-intensive stocks, reducing the exposure to assets that present undesirable brown features}. Importantly, these conclusions are not restricted to PPI strategies under partial information; they apply in a similar way to the full-information setting.

\paragraph{{Robustness check.}} {For the numerical experiments we adopt a benchmark set of parameters (see Table \ref{tab:model_params}), chosen to represent a plausible market scenario. Since results may strongly depend on this parametric specification, we perform a sensitivity analysis to assess the robustness of our findings to parameter variations. Given the large number of market parameters, we first carry out a preliminary one at a time exercise by perturbing each parameter individually by $\pm20$\% around the baseline values reported in Table \ref{tab:model_params}. Figure \ref{fig:TORNADO_PLOTS} summarizes the results through a tornado plot, showing the percentage changes in the expected value and the variance of the distribution of the optimal carbon-penalised PPI strategy at maturity relative to the baseline configuration. The plot indicates that the main drivers of the output are the factor-related drift parameters of $S_2$ and $S_4$, the volatility of $S_2$, and the latent factor dynamics, namely the long-run mean $\beta$ and the speed of mean reversion $\lambda$. By contrast, the coefficients $b_i$ and the correlation parameters (both across stocks and between stocks and the latent factor) have a negligible impact, producing less than proportional changes in $\mathbb{E}[V_T^{\bar{m}^\star,\,\bar{\pi}^\star}]$ and $\mathrm{Var}[V_T^{\bar{m}^\star,\,\bar{\pi}^\star}]$. Based on this ranking, we performe a sensitivity analysis restricted to $a_2$, $a_4$, $\sigma_2$, $\lambda$, and $\beta$. Results are reported in Appendix \ref{sect:sensitivity_analysis}.}

\begin{figure}[!htbp]
\centering
\includegraphics[width=0.48\linewidth]{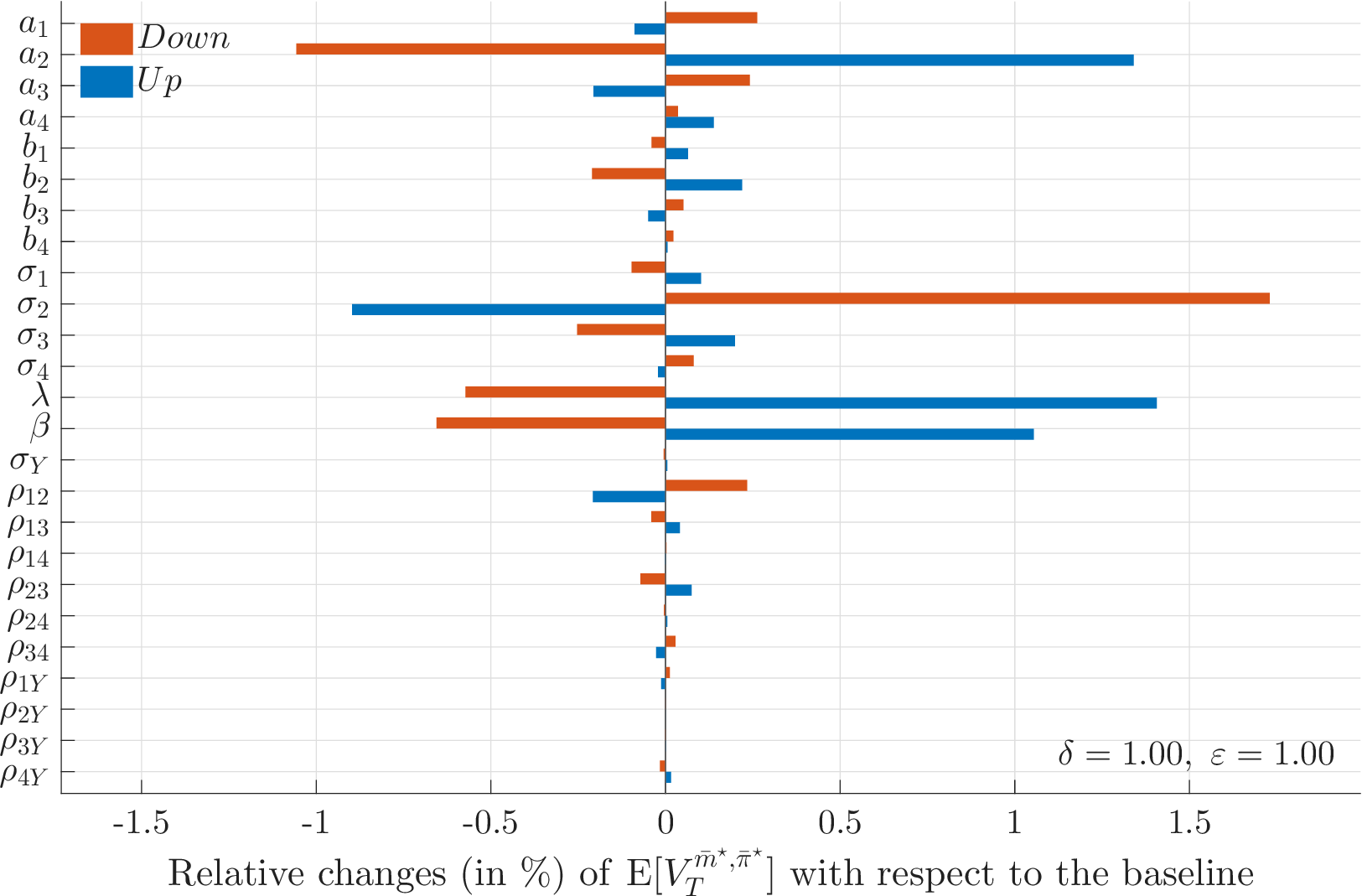} 
\vspace{.3cm}
\hfill 
\includegraphics[width=0.48\linewidth]{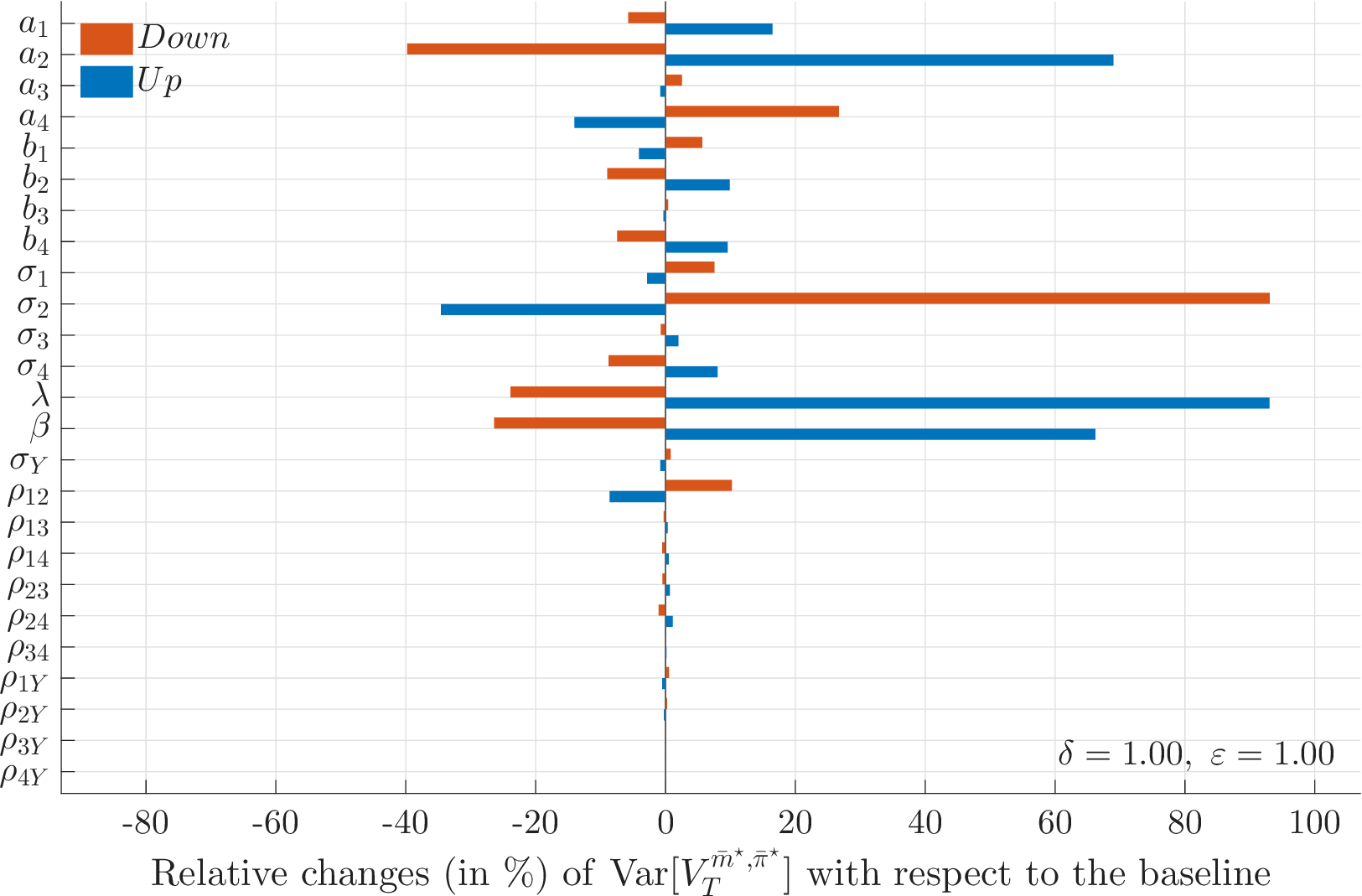}
\caption{{Tornado plots displaying the percentage change in the expected terminal value (left panels) and in the variance (right panels) of the optimal strategy at maturity, following a $20$\% one-at-a-time perturbation of the market parameters relative to the baseline values reported in Table \ref{tab:model_params}. Orange and blue bars correspond to downward and upward perturbations, respectively. Results are shown under partial information for risk aversion $\delta = 1$ and carbon aversion $\varepsilon = 1$.}
}
\label{fig:TORNADO_PLOTS}
\end{figure}

\subsection{Comparison results between the full and the partial information case}

In this final section, we compare the performance of the optimal strategies under full and partial information. Figure \ref{fig:Opt_multiplier_over_time_FULL_VS_PARTIAL} displays the optimal multiplier in full (solid blue line) and partial information (dashed magenta line).
In the left panel, we plot the standard, non-penalised case $\varepsilon =0$, and in the right panel, the penalised case $\varepsilon=1$. Both panels show that the multiplier under partial information shows slightly less variability, yet displaying very similar behaviour. The performance of the strategy under full and partial information, in terms of portfolio values are also very close as indicated in Table \ref{tab:Comparison_full_partial_Scenario_3}. This is a signal that, if markets are affected by random factors that are not easily measured, it is worth performing the portfolio analysis under partial information, rather than assuming a naive point of view and taking parameters constant.\footnote{We are ignoring here model misspecifications, which represent an additional source of error.}     

\begin{figure}[!htbp]
\centering
\includegraphics[width=0.48\linewidth]{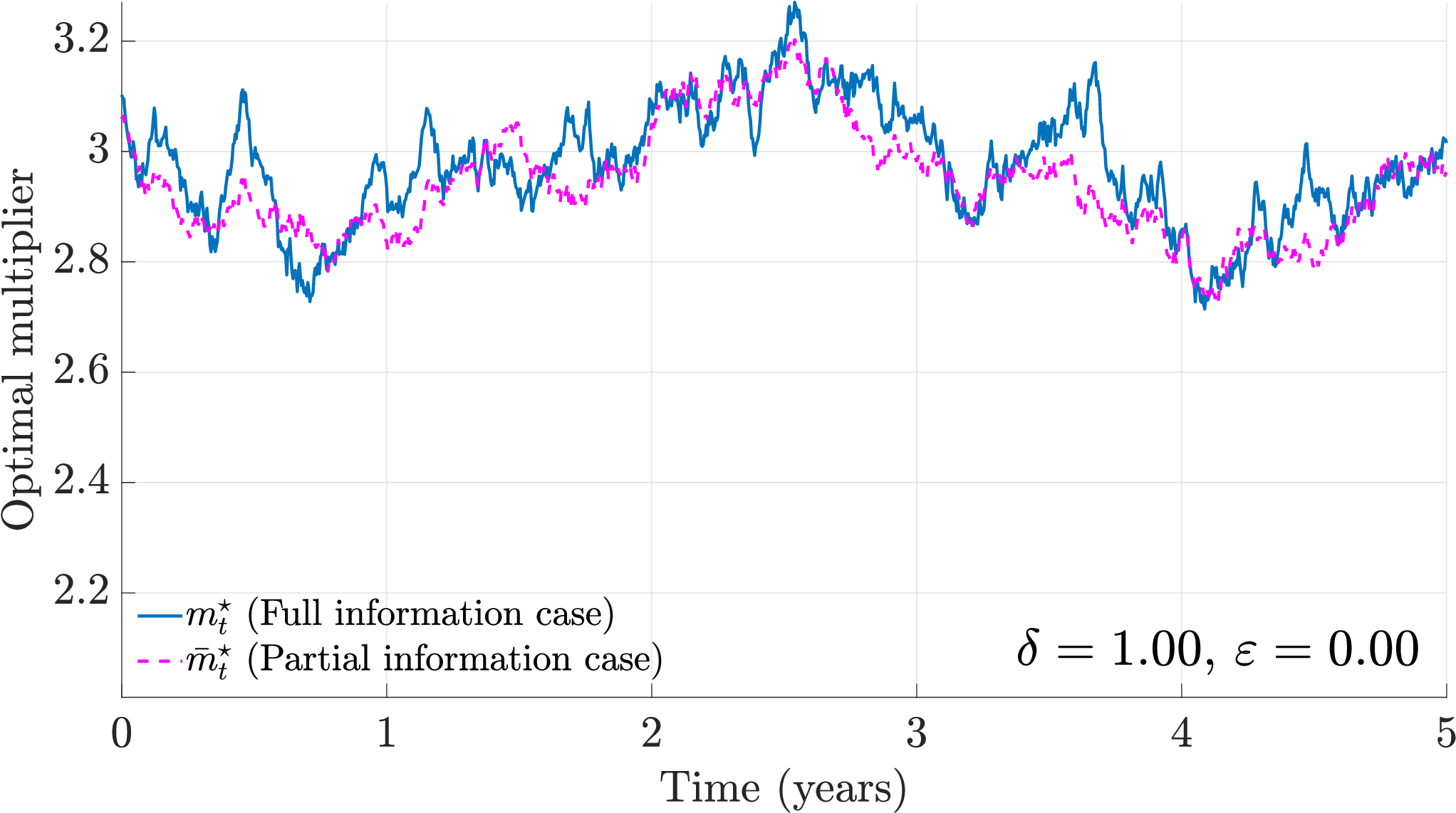}
\hfill 
\includegraphics[width=0.48\linewidth]{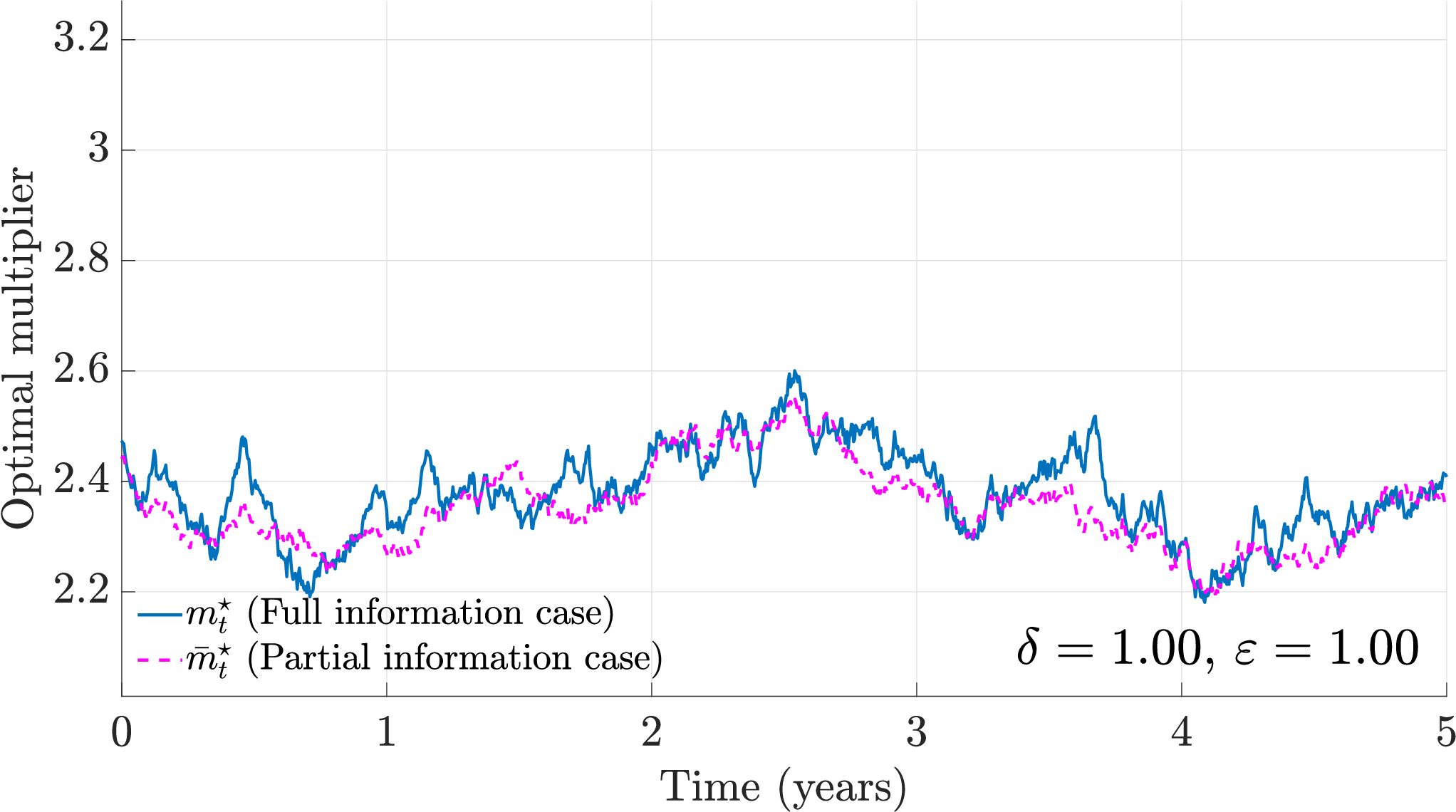}
\caption{Trajectories of the optimal multiplier under full and partial information for risk-aversion level $\delta=1$ and carbon penalisation levels $\varepsilon=0$ (left panel) and $\varepsilon=1$ (right panel). The solid blue line corresponds to the partially informed case, while the dashed magenta line corresponds to the full-information case.}\label{fig:Opt_multiplier_over_time_FULL_VS_PARTIAL}
\end{figure}

\begin{table}[!htbp]
\centering
\begin{tabular}{cccc} 
\Xhline{1.5pt}
& Full information &  & Partial information  \\ 
\cline{2-2}\cline{4-4}
$\mathbb{E}[V^{\bar m^\star,\bar\bmpi^{\star}}_T]$   & $1.1207$         &  & $1.1213$             \\
$\mathrm{Var}[V^{\bar m^\star,\bar\bmpi^{\star}}_T]$ & $0.0166$         &  & $0.0177$             \\
$q_{0.05}(V^{\bar m^\star,\bar\bmpi^{\star}}_T)$           & $1.0182$         &  & $1.0186$             \\
$q_{0.50}(V^{\bar m^\star,\bar\bmpi^{\star}}_T)$           & $1.0850$         &  & $1.0845$             \\
$q_{0.90}(V^{\bar m^\star,\bar\bmpi^{\star}}_T)$           & $1.2508$         &  & $1.2473$             \\
\Xhline{1.5pt}
\end{tabular}
\caption{Mean, variance, and $5$th/$50$th/$90$th quantiles of the distribution of the optimal carbon-penalised PPI strategy at $T=5$ in the full and partial information case, for risk-aversion level $\delta=0.7$ and $\varepsilon=1$, under Scenario $3$.}\label{tab:Comparison_full_partial_Scenario_3}
\end{table}

We conclude with an analysis of the loss of utility and efficiency. Figure \ref{fig:loss_utility_eff} reports the loss of utility at time $t=0$ (left panel) and the efficiency (right panel) as functions of the carbon-aversion parameter $\varepsilon$, for different levels of risk aversion $\delta$. The results indicate that the loss of utility decreases with both risk aversion and carbon aversion. Interestingly, the effect is more pronounced for small values of $\varepsilon$, and becomes essentially constant for larger $\varepsilon$. The opposite monotonic behavior is observed for efficiency, although the sensitivity remains greater at lower levels of $\varepsilon$. Overall, these plots suggest that carbon penalisation can improve the relative performance of the partially informed investor compared with the fully informed one, narrowing the gap between their utilities. Under high levels of risk aversion, the loss of utility becomes practically negligible, indicating that the informational advantage of the fully informed investor is largely offset by investor preferences, as both types of investors behave in an extremely prudent manner.
Finally, the presence of even a modest carbon penalisation increases the relative efficiency of the partially informed strategy vis-à-vis its fully informed counterpart.
\begin{figure}[!htbp]
\centering
\includegraphics[width=0.48\linewidth]{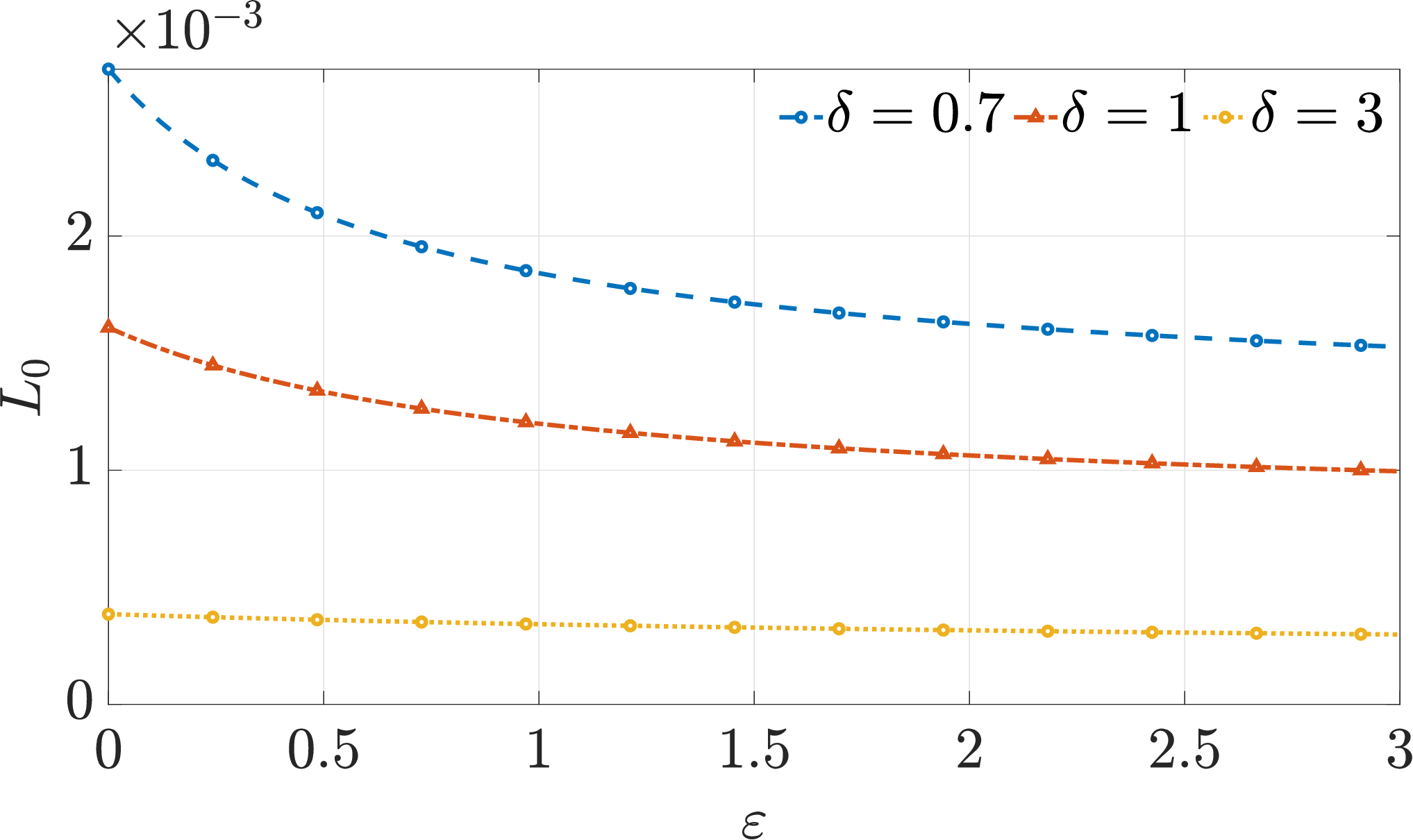}
\hfill 
\includegraphics[width=0.48\linewidth]{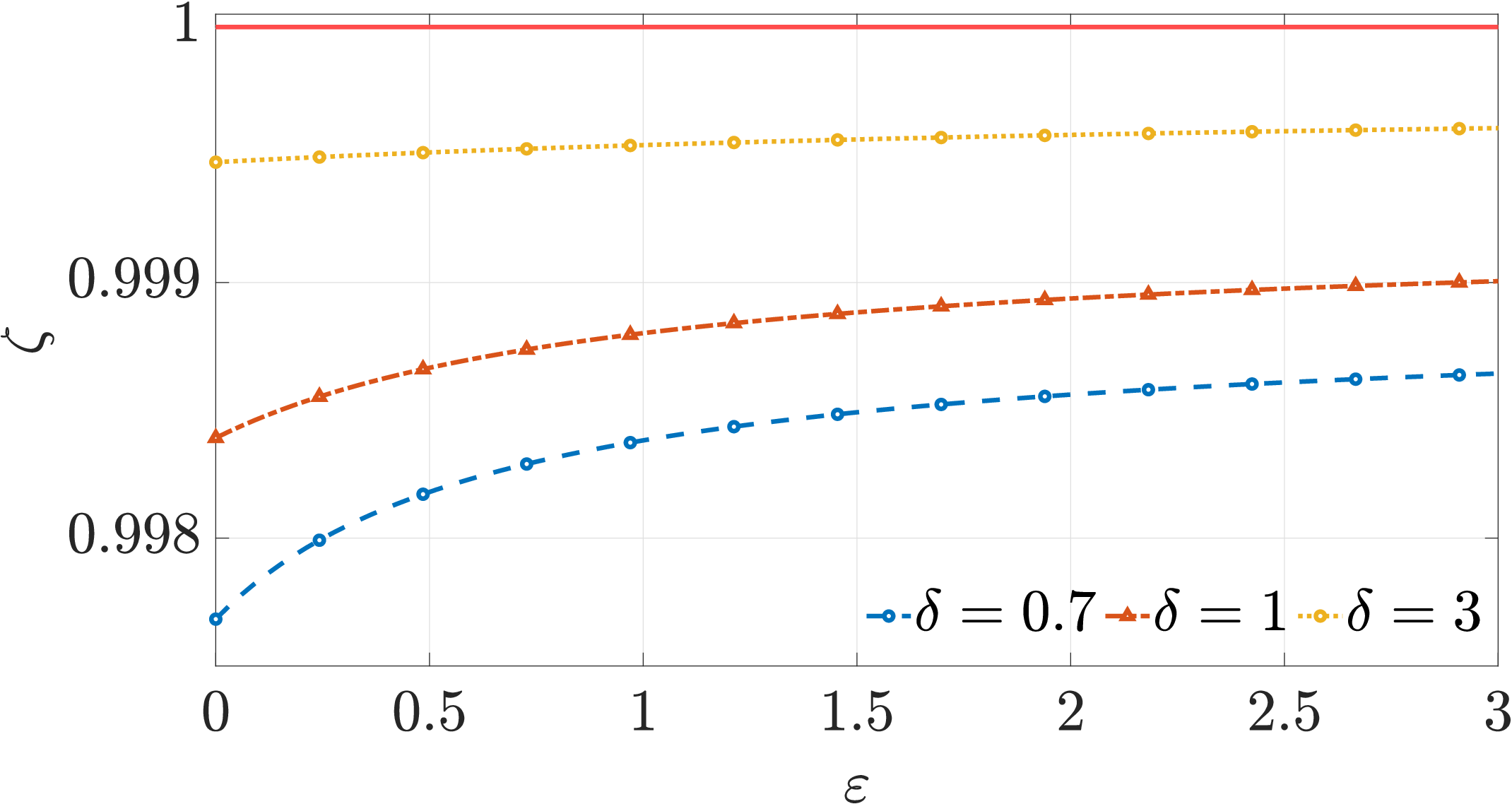}
\caption{Loss of utility (left panel) and efficiency (right panel) of a partially informed, carbon-penalised PPI insurer relative to a full-information one, with initial cushion equal to $1$, as a function of $\varepsilon$ at $t=0$ for different values of $\delta$.}\label{fig:loss_utility_eff}
\end{figure}

\section{Concluding remarks}\label{sect:conclusions}
This paper has proposed an optimal design of carbon-penalised proportional portfolio insurance (PPI) strategies in a market driven by an unobservable factor. By embedding carbon aversion into the investor’s utility function, we have shown that sustainability considerations can be consistently integrated into dynamic portfolio insurance without compromising its risk-mitigation role. The introduction of a carbon penalisation term naturally reduces exposure to carbon-intensive assets. Importantly, this reduction does not stem from an ex-ante exclusion of “brown” stocks, but from an endogenous adjustment of the optimal allocation that balances environmental impact and financial performance.

From an economic perspective, the carbon penalty operates as an implicit cost of holding high-emission assets, inducing portfolio insurers to internalise the externalities associated with carbon risk. Our numerical results indicate that even moderate levels of carbon aversion can achieve substantial emission reductions with only marginal losses in expected utility. 
Nevertheless, assets with high carbon intensity are not completely excluded; instead, a trade-off emerges between performance characteristics, e.g, a high Sharpe ratio, and carbon intensity. Consequently, a portfolio insurer considers both aspects simultaneously when designing the PPI strategy, balancing return potential against environmental impact.
Interestingly, we get that carbon penalisation improves the relative efficiency of the partially informed investor, narrowing the performance gap vis-à-vis the fully informed benchmark. When risk aversion is high, the informational premium virtually vanishes, suggesting that prudence can offset informational disadvantages.

Overall, these findings highlight that environmental preferences and informational constraints interact in shaping sustainable investment behavior. Carbon penalisation acts as a powerful mechanism to align portfolio insurance objectives with broader climate-finance goals, while partial information amplifies the conservative nature of the PPI framework.

Future research could extend this analysis in several directions. First, one may consider non-Gaussian or regime-switching latent factors to capture abrupt transitions in macro-financial or climate conditions. Second, incorporating transaction costs or market frictions would enhance the practical relevance of the model, especially for long-horizon institutional investors. Further developments might also explore multi-factor carbon risks or stochastic floors to assess how policy uncertainty and adaptive guarantees affect sustainable portfolio insurance design.

\section*{Acknowledgements and fundings}
Katia Colaneri is member of Gruppo Nazionale per l’Analisi Matematica, la Probabilità e le loro Applicazioni (GNAMPA) of Istituto
Nazionale di Alta Matematica (INdAM). The partial support through the INdAM - GNAMPA Project (CUP E5324001950001) is acknowledged. This work has been completed
while Daniele Mancinelli was affiliated with University of Rome Tor Vergata and has been funded by European Union - Next Generation EU, Mission 4, Component 2 as part of the GRINS project - Growing Resilient, INclusive and Sustainable (PE0000018, CUP: E83C22004690001) - National Recovery and Resilience Plan (PNRR). The views and opinions expressed are solely those of the authors and do not necessarily reflect those of the European Union, nor can the European Union be held responsible for them. Moreover, any views expressed are solely those of the author(s) and so cannot be taken to represent those of the Bank of England or any of its committees, or to state Bank of England policies. 

\section*{Declaration of generative AI in scientific writing}
During the preparation of this work the authors used \textit{Writefull AcademicGPT 2025} in the writing process in order to improve the readability and language of the manuscript. After using this tool, the authors reviewed and edited the content as needed and take full responsibility for the content of the published article.
\section*{Conflict of interest}
The authors declare no competing interests.
\section*{Appendix}
\appendix
\numberwithin{equation}{section}
\renewcommand{\theequation}{\thesection\arabic{equation}}
\section{Proofs of some technical results of Section \ref{sect:opt_problem_full_info}}
\subsection{Proof of Theorem \ref{thm:verification_thm_full_info_CRRA}}\label{app:A_1}
From It\^{o}'s formula applied to $f(t,\hat{C}^{\bmtheta}_t,Y_t)$ we get that, for any $0\le t\le T$ and $\bmtheta\in\mathcal{A}$, it holds
\begin{align}
f(T,\hat{C}^{\bmtheta}_T,Y_T)=&f(t,c,y)+\int_t^T\left(f_s(s,\hat{C}^{\bmtheta}_s,Y_s)+\mathcal{L}^{\bmtheta}f(s,\hat{C}^{\bmtheta}_s,Y_s)\right)\de s+\int_t^Tf_{y}(s,\hat{C}^{\bmtheta}_s,Y_{s})\tilde\sigma_YZ_s^Y\\
&+\int_t^T\left(f_{\hat{c}}(s,\hat{C}^{\bmtheta}_s,Y_s)\hat{C}^{\bmtheta}_s\bmtheta_s^\top\bmtSigma_\bmS+f_{y}(s,\hat{C}^{\bmtheta}_s,Y_{s})\bmtSigma_Y\right)\de\bmZ_s^{\bmS}. 
\end{align}
Let $M=\left\lbrace M_t\right\rbrace_{t\in[0,T]}$ be the stochastic process given by
\begin{equation}
M_t=\int_0^tf_{y}(s,\hat{C}^{\bmtheta}_s,Y_{s})\tilde\sigma_YZ_s^Y+\int_0^t\left(f_{\hat{c}}(s,\hat{C}^{\bmtheta}_s,Y_s)\hat{C}^{\bmtheta}_s\bmtheta_s^\top\bmtSigma_\bmS+f_{y}(s,\hat{C}^{\bmtheta}_s,Y_{s})\bmtSigma_Y\right)\de\bmZ_s^{\bmS},\quad t\in[0,T],
\end{equation}
and define $\tau_n= \inf\{t \ge 0: \hat C^\theta_t \ge n \text{ and } |Y_t| \le n\}$. This is an increasing sequence of stopping times such that $\tau_n \wedge T\uparrow T$ for $n\to \infty$. 
Moreover, by assumption, $f$ is a classical solution of the HJB equation \eqref{eq:HJB_equation_FULL_INFO}, hence its derivatives are continuous and bounded on compact sets. This implies that the stopped process $\{M_{t \wedge \tau_n}\}_{t \in [0,T]}$ is a martingale. Indeed, it holds that    
\begin{align}\label{eq:int_cond_ver}
&\mathbb{E}\left[\int_0^{T \wedge \tau_n}f^2_y(s,\hat{C}^{\bmtheta}_s,Y_s)\de s+\int_0^{T \wedge \tau_n}f^2_{\hat{c}}(s,\hat{C}^{\bmtheta}_s,Y_s)(\hat{C}_s^{\bmtheta})^2\bmtheta_s^\top\bmtSigma_\bmS\bmtSigma_\bmS^\top\bmtheta_s\de s\right]\\
&\le \sup_{t \le T, (c,y) \in [-n,n]^2} |f^2_y(t,\hat{C}^{\bmtheta}_t,Y_t)| T + k |f^2_c(t,\hat{C}^{\bmtheta}_t,Y_t)| n^2 \mathbb{E}\left[\int_0^T\|\bmtheta_s\|^2_2\de s\right] <\infty.
\end{align}
Now, since $f$ solves equation \eqref{eq:value_function_full_info_CRRA}, we get that for every $n \in \mathbb{N}$
\begin{align}
f(T\wedge \tau_n,\hat{C}_{T\wedge \tau_n},Y_{T\wedge \tau_n})\le&f(t\wedge \tau_n,\hat{C}_{t \wedge \tau_n},Y_{t \wedge \tau_n})+\int_{t\wedge \tau_n}^{T\wedge \tau_n}f_{y}(s,\hat{C}^{\bmtheta}_s,Y_{s})\tilde\sigma_YZ_s^Y \\
\label{eq:ineq}&+\int_{t\wedge \tau_n}^{T\wedge \tau_n}\left(f_{\hat{c}}(s,\hat{C}^{\bmtheta}_s,Y_s)\hat{C}_s\bmtheta_s^\top\bmtSigma_\bmS+f_{y}(s,\hat{C}^{\bmtheta}_s,Y_{s})\bmtSigma_Y\right)\de\bmZ_s^{\bmS},
\end{align}
for every $\bmtheta\in\mathcal{A}$. Thus, taking the conditional expectation on both sides of inequality \eqref{eq:ineq} between $t \wedge \tau_n$ and $T \wedge \tau_n$, leads to $\mathbb{E}[f(T\wedge \tau_n,\hat{C}^{\bmtheta}_{T \wedge \tau_n},Y_{T \wedge \tau_n})]\le \mathbb{E}\left[f(t\wedge \tau_n,\hat{C}^{\bm{\theta}}_{t \wedge \tau_n},Y_{t \wedge \tau_n})\right]$. 
Next we take the limit for $n \to \infty$, and thanks to condition $(i)$ of the theorem \eqref{eq:value_function_full_info_CRRA}, we obtain
\begin{equation}\label{eq:ineq_5}
\mathbb{E}^{t,c,y}\left[\dfrac{1}{1-\delta}\left(\hat{C}^{\bmtheta}_T\right)^{1-\delta}\right]\le f(t,c,y),
\end{equation}
hence $\hat{v}(t,c,y)\le f(t,c,y)$.
Similar computations prove that equality holds in \eqref{eq:ineq_5} when taking the control $\{\bmtheta^\star(t,Y_t)\}_{t\in[0,T]}\in\mathcal{A}$. Consequently, $\hat{v}(t,c,y)=\mathbb{E}^{t,c,y}\left[\frac{1}{1-\delta}(\hat{C}^{\bmtheta^\star}_T)^{1-\delta}\right]=f(t,c,y)$. This concludes the proof.
\subsection{Proof of Theorem \ref{thm:existence_CRRA}}\label{proof_existence_CRRA}
Assume that a classical solution $f$ of the Hamilton Jacobi Bellman equation \eqref{eq:HJB_equation_FULL_INFO} can be rewritten as
\begin{equation}\label{eq:GUESS}
f(t,c,y)=\dfrac{c^{1-\delta}}{1-\delta}\hat{\varphi}(t,y),
\end{equation}
where $\hat{\varphi}(t,y)$ does not depend on $c$ and is a positive function. Then, equation \eqref{eq:HJB_equation_FULL_INFO} can be rewritten as
\begin{equation}
\begin{cases}
\dfrac{\hat\varphi_t(t,y)}{1-\delta}+r\hat\varphi(t,y)+\dfrac{\left(\lambda y+\beta\right)}{1-\delta}\hat\varphi_y(t,y)+\dfrac{1}{2}\dfrac{\sigma_Y^2}{1-\delta}\hat\varphi_{y,y}(t,y)+\displaystyle\max_{\bmtheta\in\R^n}\Psi^{\bmtheta}(t,y)=0,&(t,y)\in[0,T)\times\R,\\
\hat\varphi(T,y)=1,&y\in\R,
\end{cases}
\end{equation}
where
\begin{equation}
\Psi^{\bmtheta}(t,y):=\bmtheta^\top\left(\bma y+\bmb-\bm{r}_{n}\right)\hat\varphi(t,y)-\dfrac{1}{2}\bmtheta^\top\bm{\hat{\Theta}}\bmtheta\hat\varphi(t,y)+\bmtheta^\top\bmtSigma_{\bmS}\bmtSigma_Y^\top\hat\varphi_y(t,y),\quad(t,y)\in[0,T]\times\R,
\end{equation}
with $\bmhatTheta=\left(\bmSigma_{\bmS}\bmSigma_{\bmS}^\top\odot\bme\right)+\delta\bmtSigma_{\bmS}\bmtSigma_{\bmS}^\top$. We let $\bmtheta^\star=\argmax\Psi^{\bmtheta}(t,y)$. Taking the gradient and the Hessian of $\Psi^{\bmtheta}$ with respect to $\bmtheta$, we get that
\begin{align}
\nabla_{\theta}\Psi^{\bmtheta}(t,y)&=\left(\bma y+\bmb-\bm{r}_{n}\right)\hat\varphi(t,y)-\bm{\hat{\Theta}}\bmtheta\hat\varphi(t,y)+\bmtSigma_{\bmS}\bmtSigma_Y^\top\hat\varphi_y(t,y),\\
\text{Hess}_{\bmtheta}\Psi^{\bmtheta}(t,y)&=-\bm{\hat{\Theta}}\hat\varphi(t,y).
\end{align}
Then, setting $\nabla_{\bmtheta}\Psi^{\bmtheta}(t,y)=\mathbf{0}$, provides the candidate optimal strategy $\bmtheta^{\star}(t, y)$ given by
\begin{equation}\label{eq:optimal_feedback_control}
\bmtheta^{\star}(t,y)=\bmhatTheta^{-1}\left(\bma y+\bmb-\bm{r}_{n}\right)+\bmhatTheta^{-1}\bmtSigma_{\bmS}\bmtSigma_Y^\top\dfrac{\hat\varphi_y(t,y)}{\hat\varphi(t,y)}.
\end{equation}
Moreover, since $\text{Hess}_{\bmtheta}\Psi^{\bmtheta}(t,y)$ is negative definite for every $\bmtheta\in\R^n$, this ensure that $\bmtheta^\star(t,y)$ is the well defined global maximiser. Next, we insert the optimal strategy in the HJB equation, yielding to the following PDE
\begin{align}
0=&\hat\varphi_t(t,y)+\left(1-\delta\right)r\hat\varphi(t,y)+\dfrac{1-\delta}{2}\hat\varphi(t,y)\left(\bma y+\bmb-\bm{r}_{n}\right)^\top\bmhatTheta^{-1}\left(\bma y+\bmb-\bm{r}_{n}\right)\\
&+\left(1-\delta\right)\hat\varphi_y(t,y)\bmtSigma_Y\bmtSigma_{\bmS}^\top\bmhatTheta^{-1}\left(\bma y+\bmb-\bm{r}_{n}\right)+\dfrac{1-\delta}{2}\dfrac{\left(\hat\varphi_y(t,y)\right)^2}{\hat\varphi(t,y)}\bmtSigma_Y\bmtSigma_{\bmS}^\top\bmhatTheta^{-1}\bmtSigma_{\bmS}\bmtSigma_Y^\top\\
\label{eq:PDE_CRRA_full}&+\left(\lambda y+\beta\right)\hat\varphi_y(t,y)+\dfrac{1}{2}\sigma_Y^2\hat\varphi_{y,y}(t,y),\quad (t,y)\in[0,T)\times\R,
\end{align}
with terminal condition $\varphi(T,y)=1$, for every $y\in\R$. We conjecture that $\hat{\varphi}(t,y)$ has an exponential affine form, namely 
\begin{equation}\label{eq:GUESS_CRRA_FULL}
\hat{\varphi}(t,y)=\exp\left\lbrace\frac{\hat{f}(t)}{2}y^2+\hat{g}(t)y+\hat{h}(t)\right\rbrace,
\end{equation}
with $\hat{f}(T)=\hat{g}(T)=\hat{h}(T)=0$. Clearly, the terminal value of the function in \eqref{eq:GUESS_CRRA_FULL} satisfies the terminal condition in \eqref{eq:PDE_CRRA_full} and $\hat\varphi(t,y)>0$, for every $(t,y)\in[0,T]\times\R$. Substituting this ansatz in
equation \eqref{eq:PDE_CRRA_full} results in a quadratic equation for $y$. Setting the coefficients of the terms $y^2$, $y$ and the independent term to zero yields that the functions $\hat{f}$, $\hat{g}$ and $\hat{h}$ solve the system of ODEs in equations \eqref{eq:f_hat}, \eqref{eq:g_hat} and \eqref{eq:h_hat}. If $\hat{f}$, $\hat{g}$ and $\hat{h}$  belong to the class $\mathcal{C}^{1}_b([0,T])$, then $f$ in equation \eqref{eq:GUESS} is also regular and solves the HJB equation \eqref{eq:HJB_equation_FULL_INFO}. Finally, by substituting equation \eqref{eq:GUESS_CRRA_FULL} in \eqref{eq:optimal_feedback_control}, we obtain the candidate for the optimal control in equation \eqref{eq:opt_controls_CRRA_full}. This concludes the proof.
\subsection{Proof of Proposition \ref{prop:sufficient_cond_ver_FULL_INFO}}\label{app:A_4}
We will show that $\sup_{t\in[0,T]}\mathbb{E}\left[\hat{v}^{1+\alpha}(t,\hat{C}_t,Y_t)\right]<\infty$, for some $\alpha>0$. Using the form of the function $v$ (cfr. equation \eqref{eq:value_fun_CRRA_full_info}) we get that
\begin{align}
\sup_{t\in[0,T]}\mathbb{E}\left[\hat{v}^{1+\alpha}(t,\hat{C}_t^{\bmtheta},Y_t)\right]=&\sup_{t\in[0,T]}\mathbb{E}\left[\dfrac{1}{1-\delta}(\hat{C}^{\bmtheta}_t)^{(1-\delta)(1+\alpha)}e^{\frac{(1+\alpha)\hat{f}(t)}{2}Y_t^2+(1+\alpha)\hat{g}(t)Y_t+(1+\alpha)\hat{h}(t)}\right]\\
\le&\kappa\sup_{t\in[0,T]}\mathbb{E}\left[(\hat{C}^{\bmtheta}_t)^{(1-\delta)(1+\alpha)}e^{\frac{(1+\alpha)\hat{f}(t)}{2}Y_t^2+(1+\alpha)\hat{g}(t)Y_t}\right]\\
\le&\kappa\left(\sup_{t\in[0,T]}\mathbb{E}\left[(\hat{C}^{\bmtheta}_t)^{d(1-\delta)(1+\alpha)}\right]^{\frac{1}{d}}\right)\left(\sup_{t\in[0,T]}\mathbb{E}\left[e^{\frac{q(1+\alpha)\hat{f}(t)}{2}Y_t^2+q(1+\alpha)\hat{g}(t)Y_t}\right]^{\frac{1}{q}}\right),
\end{align}
for some positive constant $\kappa$ and some $d,q>1$, where in the first inequality we have used that $\hat{h}(\cdot)\in\mathcal{C}^1_b([0,T])$, and in the second comes from applying Hölder's inequality. The first expectation is finite because of admissibility of the strategy (see the second condition of Definition \ref{defn:G_admissible_strategies_theta}). The second expectation is finite because the process $Y_t$ is Gaussian. Hence,
$$\mathbb{E}\left[e^{\frac{q(1+\alpha)\hat{f}(t)}{2}Y_t^2+q(1+\alpha)\hat{g}(t)Y_t}\right]<\infty,$$
for every $t\in[0,T]$ if and only if $1-q(1+\alpha)\hat{f}(t)\mbox{Var}[Y_t]>0$, where $\mbox{Var}[Y_t]=P_0e^{2\lambda t}+V_\infty(1-e^{2\lambda t})$, with $V_\infty=-\sigma_Y/2\lambda$. To show that $1-q(1+\alpha)\hat{f}(t)\mbox{Var}[Y_t]>0$ for every $t\in[0,T]$, we need to distinguish between two cases. If $\delta\in\mathcal{P}\cap(1,+\infty)$, $\hat{f}(t)$ is strictly negative and increasing for every $t\in[0,T]$, guaranteeing that $1-q(1+\alpha)\hat{f}(t)\mbox{Var}[Y_t]>0$. If $\delta\in\mathcal{P}\cap(0,1)$, $\hat{f}(t)$ is positive and decreasing in $[0,T]$, implying that $\hat{f}(t)<\hat{f}(0)$ for every $t\in[0,T]$. If $P_0>V_\infty$ (respectively,  $P_0\le V_\infty$), $\mbox{Var}(Y_t)$ is decreasing (respectively,  increasing) meaning that $P_0\le\mbox{Var}[Y_t]\le \mbox{Var}[Y_T]$ (respectively,  $\mbox{Var}[Y_T]\le\mbox{Var}[Y_t]<P_0$). This means that $\hat{f}(t)\mbox{Var}[Y_t]<\hat{f}(0)\max\left\lbrace P_0,\mbox{Var}[Y_T]\right\rbrace$, or equivalently, $1-q(1+\alpha)\hat{f}(t)\mbox{Var}[Y_t]>1-q(1+\alpha)\hat{f}(0)\max\left\lbrace P_0,\mbox{Var}[Y_T]\right\rbrace$, for every $t\in[0,T]$. Then the result follows from equation \eqref{eq:cond_1} and concludes the proof.
\subsection{Proof of Proposition \ref{prop:sufficient_cond_admissibility_FULL_INFO}}\label{app:A_5}
First, we discuss the first condition of Definition \ref{defn:G_admissible_strategies_theta}. For the $\mathbb{G}$-predictable process $\bmtheta^\star$ given by \eqref{eq:opt_controls_CRRA_full}, it holds that \begin{align}
\mathbb{E}&\left[\int_0^T|Y_s|\|\bm{\theta}^\star_s\|_1+\|\bm{\theta}^\star_s\|_2^2\de s\right]\\
=&\mathbb{E}\left[\int_0^T|Y_s|\|\mathbf{\hat\Theta}^{-1}\left(\mathbf{a}Y_s+\mathbf{b}-\mathbf{r}_{n}\right)+\mathbf{\hat\Theta}^{-1}\bm{\tilde\Sigma}_{\mathbf{S}}\bm{\tilde\Sigma}_Y^\top\left(\hat{f}(s)Y_s+\hat{g}(s)\right)\|_1\de s\right]\\
&+\mathbb{E}\left[\int_0^T\|\mathbf{\hat\Theta}^{-1}\left(\mathbf{a}Y_s+\mathbf{b}-\mathbf{r}_{n}\right)+\mathbf{\hat\Theta}^{-1}\bm{\tilde\Sigma}_{\mathbf{S}}\bm{\tilde\Sigma}_Y^\top\left(\hat{f}(s)Y_s+\hat{g}(s)\right)\|_2^2\de s\right]\\
\le&\mathbb{E}\left[\int_0^T|Y_s|\|\mathbf{\hat\Theta}^{-1}\left(\mathbf{a}Y_s+\mathbf{b}-\mathbf{r}_{n}\right)\|_1+|Y_s|\|\mathbf{\hat\Theta}^{-1}\bm{\tilde\Sigma}_{\mathbf{S}}\bm{\tilde\Sigma}_Y^\top\left(\hat{f}(s)Y_s+\hat{g}(s)\right)\|_1\de s\right]\\
&+\mathbb{E}\left[\int_0^T\left(\|\mathbf{\hat\Theta}^{-1}\left(\mathbf{a}Y_s+\mathbf{b}-\mathbf{r}_{n}\right)\|_2+\|\mathbf{\hat\Theta}^{-1}\bm{\tilde\Sigma}_{\mathbf{S}}\bm{\tilde\Sigma}_Y^\top\left(\hat{f}(s)Y_s+\hat{g}(s)\right)\|_2\right)^2\de s\right]\\
\le&\mathbb{E}\bigg[\int_0^TY_s^2\|\mathbf{\hat\Theta}^{-1}\mathbf{a}\|_1+|Y_s|\|\mathbf{\hat\Theta}^{-1}\left(\mathbf{b}-\mathbf{r}_{n}\right)\|_1+Y_s^2|\hat{f}(s)|\|\mathbf{\hat\Theta}^{-1}\bm{\tilde\Sigma}_{\mathbf{S}}\bm{\tilde\Sigma}_Y^\top\|_1\\
&+|Y_s||\hat{g}(s)|\|\mathbf{\hat\Theta}^{-1}\bm{\tilde\Sigma}_{\mathbf{S}}\bm{\tilde\Sigma}_Y^\top\|_1\de s\bigg]+4\mathbb{E}\bigg[\int_0^TY^2_s\|\mathbf{\hat\Theta}^{-1}\mathbf{a}\|^2_2+\|\mathbf{\hat\Theta}^{-1}\left(\mathbf{b}-\mathbf{r}_{n}\right)\|_2^2\\
&+Y_s^2\hat{f}^2(s)\|\mathbf{\hat\Theta}^{-1}\bm{\tilde\Sigma}_{\mathbf{S}}\bm{\tilde\Sigma}_Y^\top\|_2^2+\hat{g}^2(s)\|\mathbf{\hat\Theta}^{-1}\bm{\tilde\Sigma}_{\mathbf{S}}\bm{\tilde\Sigma}_Y^\top\|_2^2\de s\bigg]\\
\le&\eta_1+\eta_2\mathbb{E}\left[\int_0^T|Y_s|+Y_s^2\de s\right]<\infty
\end{align}
for some positive constant $\eta_1$ and $\eta_2$. The first inequality follows by applying the triangle inequality to the $l_1$ and $l_2$ norms, then using the Cauchy–Schwarz inequality on the second term to bound the square of the sum by the sum of squares, and finally using the positive homogeneity of norms to factor out scalar terms.
The second inequality follows by applying the same arguments as the first. The third inequality holds because $\hat f(t),\,\hat g(t)\in\mathcal{C}^1_b([0,T])$ and the last inequality comes from the fact that $Y$ is a Gaussian random variable, which implies that it has finite moments of all orders. We now discuss the second condition of Definition \ref{defn:G_admissible_strategies_theta}. We would like to show that $$\sup_{t\in[0,T]}\mathbb{E}\left[(\hat{C}^{\bmtheta^\star}_t)^{d(1-\delta)(1+\alpha)}\right]<\infty,$$ for some $\alpha>0$ and $d>1$. Using the explicit solution of equation \eqref{eq:carbon_penalised_cushion_FULL_INFO}, i.e
\begin{equation}
\hat{C}^{\bm{\theta}}_t=\hat{C}_0^{\bm{\theta}}\exp\left\lbrace\int_0^t\left[r+\bm{\theta}_u^\top\left(\mathbf{a}Y_u+\mathbf{b}-\mathbf{r}_n\right)-\frac{1}{2}\bm{\theta}_u^\top\mathbf{\hat{\Theta}}\bm{\theta}_u\right]\de u+\int_0^t\bm{\theta}_u^\top\mathbf{\tilde\Sigma}_{\mathbf{S}}\de\mathbf{Z}^{\mathbf{S}}_u\right\rbrace,
\end{equation}
we get that 
\begin{align}
\sup_{t\in[0,T]}&\mathbb{E}\left[(\hat{C}^{\bm{\theta}^\star}_t)^{d(1-\delta)(1+\alpha)}\right]\\
=&\sup_{t\in[0,T]}(\hat{C}_0^{\bm{\theta}^\star}e^{rt})^{d(1-\delta)(1+\alpha)}\mathbb{E}\left[e^{d(1-\delta)(1+\alpha)\int_0^t[(\bm{\theta}^\star_u)^\top(\mathbf{a}Y_u+\mathbf{b}-\mathbf{r}_n)-\frac{1}{2}\bm{\theta}^{\star,\top}_u\mathbf{\hat{\Theta}}\bm{\theta}^\star_u]\de u+d(1-\delta)(1+\alpha)\int_0^t\bm{\theta}^{\star,\top}_u\mathbf{\tilde\Sigma}_{\mathbf{S}}\de\mathbf{Z}^{\mathbf{S}}_u}\right]\\
\le&\dfrac{(\hat{C}^{\bm{\theta}^\star}_0)^{d(1-\delta)(1+\alpha)}}{2}\left(\sup_{t\in[0,T]}e^{rd(1-\delta)(1+\alpha)t}\mathbb{E}\left[e^{2d(1-\delta)(1+\alpha)\int_0^t\bm{\theta}^{\star,\top}_u\left(\mathbf{a}Y_u+\mathbf{b}-\mathbf{r}_n\right)\de u}\right.\right.\\
&\left.\left.e^{-d(1-\delta)(1+\alpha)\int_0^t\bm{\theta}^{\star,\top}_u\mathbf{\hat{\Theta}}\bm{\theta}^\star_u\de u}\right]+\sup_{t\in[0,T]}e^{rd(1-\delta)(1+\alpha)t}\mathbb{E}\left[e^{2d(1-\delta)(1+\alpha)\int_0^t\bm{\theta}_u^{\star,\top}\mathbf{\tilde\Sigma}_{\mathbf{S}}\de\mathbf{Z}^{\mathbf{S}}_u}\right]\right)\\
\le&\dfrac{\kappa}{4}\left(\sup_{t\in[0,T]}\mathbb{E}\left[e^{4d(1-\delta)(1+\alpha)\int_0^t\bm{\theta}^{\star,\top}_u\left(\mathbf{a}Y_u+\mathbf{b}-\mathbf{r}_n\right)\de u}\right]+\sup_{t\in[0,T]}\mathbb{E}\left[e^{-2d(1-\delta)(1+\alpha)\int_0^t\bm{\theta}^{\star,\top}_u\mathbf{\hat{\Theta}}\bm{\theta}^\star_u\de u}\right]\right.\\
&\left.+2\sup_{t\in[0,T]}\mathbb{E}\left[e^{2d(1-\delta)(1+\alpha)\int_0^t\bm{\theta}_u^{\star,\top}\mathbf{\tilde\Sigma}_{\mathbf{S}}\de\mathbf{Z}^{\mathbf{S}}_u}\right]\right)\\
=&\dfrac{\kappa}{4}\left(\sup_{t\in[0,T]}\mathbb{E}\left[e^{4d(1-\delta)(1+\alpha)\int_0^t\bm{\theta}^{\star,\top}_u\left(\mathbf{a}Y_u+\mathbf{b}-\mathbf{r}_n\right)\de u}\right]+\sup_{t\in[0,T]}\mathbb{E}\left[e^{-2d(1-\delta)(1+\alpha)\int_0^t\bm{\theta}^{\star,\top}_u\mathbf{\hat{\Theta}}\bm{\theta}^\star_u\de u}\right]\right.\\
&\label{eq:relation_1}\left.+2\sup_{t\in[0,T]}\mathbb{E}\left[e^{2d^2(1-\delta)^2(1+\alpha)^2\int_0^t\|\bmtheta_u^{\star,\top}\bmtSigma_{\bmS}\|_2^2\de u}\right]\right),
\end{align}
where $\kappa=(\hat{C}^{\bmtheta^\star}_0e^{rT})^{d(1-\delta)(1+\alpha)}$. In the first and second inequality we have used $ab\le\frac{1}{2}(a^2+b^2)$ for any $a,b\in\R$, and the last equality comes from the fact that $\mathbb{E}\left[e^{2d(1-\delta)(1+\alpha)\int_0^t\bmtheta_u^{\star,\top}\bmtSigma_{\bmS}\de\bmZ^{\bmS}_u}\right]=\mathbb{E}\left[e^{2d^2(1-\delta)^2(1+\alpha)^2\int_0^t\|\bmtheta_u^{\star,\top}\bmtSigma_{\bmS}\|_2^2\de u}\right]$. Now, we need to distinguish between two cases: $\delta\in\mathcal{P}\cap(0,1)$ and $\delta\in\mathcal{P}\cap(1+\infty)$. Assuming that $\delta\in\mathcal{P}\cap(0,1)$, equation \eqref{eq:relation_1} becomes 
\begin{align}
&\sup_{t\in[0,T]}\mathbb{E}\left[(\hat{C}_t^{\bm{\theta}^\star})^{d(1-\delta)(1+\alpha)}\right]\\
&\le\dfrac{\kappa}{4}\left(1+\sup_{t\in[0,T]}\mathbb{E}\left[e^{4d(1-\delta)(1+\alpha)\int_0^t\bm{\theta}^{\star,\top}_u\left(\mathbf{a}Y_u+\mathbf{b}-\mathbf{r}_n\right)\de u}\right]+2\sup_{t\in[0,T]}\mathbb{E}\left[e^{2d^2(1-\delta)^2(1+\alpha)^2\int_0^t\|\bmtheta_u^{\star,\top}\bmtSigma_{\bmS}\|_2^2\de u}\right]\right)\\
&\le\dfrac{\kappa}{4}\left(1+\sup_{t\in[0,T]}\mathbb{E}\left[e^{2d(1-\delta)(1+\alpha)\int_0^t\left(\|\bm{\theta}^\star_u\|_2^2+ \|\mathbf{a}Y_u+\mathbf{b}-\mathbf{r}_n\|^2\right)\de u}\right]+2\sup_{t\in[0,T]}\mathbb{E}\left[e^{2d^2(1-\delta)^2(1+\alpha)^2\int_0^tw\|\bm{\theta}_u^\star\|_2^2\de u}\right]\right)\\
\label{eq:ineq_1}&\le\dfrac{\kappa}{4}\left(1+\mathbb{E}\left[e^{2d(1-\delta)(1+\alpha)\int_0^T\left(\|\bm{\theta}^\star_u\|_2^2+\|\mathbf{a}Y_u+\mathbf{b}-\mathbf{r}_n\|_2^2\right)\de u}\right]+2\mathbb{E}\left[e^{2d^2(1-\delta)^2(1+\alpha)^2w\int_0^T\|\bm{\theta}^\star_u\|_2^2\de u}\right]\right),
\end{align}
where in the second inequality we have used $\bmtheta^{\star,\top}_u\left(\bma Y_u+\bmb-\mathbf{r}_n\right)\le \frac{1}{2}\left(\|\bmtheta^{\star}_u\|_2^2+\|\bma Y_u+\bmb-\mathbf{r}_n\|_2^2\right)$, 
and $\|\bmtheta_u^{\star,\top}\bmtSigma_{\bmS}\|_2^2\leq w\|\bmtheta_u^{\star}\|_2^2$, for every $u\in[0,T]$, with $w$ given by equation \eqref{eq:exp_for_w}. The third inequality follows from the monotonicity of the integrals in $t$, which implies that the supremum over $t\in[0,T]$ is attained at $t=T$. By Jensen's inequality, we get that 
\begin{align}
e^{2d(1-\delta)(1+\alpha)\int_0^T\left(\|\bm{\theta}^\star_u\|_2^2+\|\mathbf{a}Y_u+\mathbf{b}-\mathbf{r}_n\|_2^2\right)\de u}&\le\frac{1}{T}\int_0^Te^{2d(1-\delta)(1+\alpha)T\left(\|\bm{\theta}^\star_u\|_2^2+\|\mathbf{a}Y_u+\mathbf{b}-\mathbf{r}_n\|_2^2\right)}\de u,\\
e^{2d^2(1-\delta)^2(1+\alpha)^2w\int_0^T\|\bm{\theta}^\star_u\|_2^2\de u}&\le\frac{1}{T}\int_0^Te^{2d^2(1-\delta)^2(1+\alpha)^2wT\|\bm{\theta}^\star_u\|_2^2}\de u, 
\end{align}
therefore
\begin{align}
\mathbb{E}\left[e^{2d(1-\delta)(1+\alpha)\int_0^T\left(\|\bm{\theta}^\star_u\|_2^2+\|\mathbf{a}Y_u+\mathbf{b}-\mathbf{r}_n\|_2^2\right)\de u}\right]&\le\frac{1}{T}\int_0^T\mathbb{E}\left[e^{2d(1-\delta)(1+\alpha)T\left(\|\bm{\theta}^\star_u\|_2^2+\|\mathbf{a}Y_u+\mathbf{b}-\mathbf{r}_n\|_2^2\right)}\right]\de u,\\
\mathbb{E}\left[e^{2d^2(1-\delta)^2(1+\alpha)^2w\int_0^T\|\bm{\theta}^\star_u\|_2^2\de u}\right]&\le\dfrac{1}{T}\int_0^T\mathbb{E}\left[e^{2d^2(1-\delta)^2(1+\alpha)^2wT\|\bm{\theta}^\star_u\|_2^2}\right]\de u.
\end{align}
Hence, equation \eqref{eq:ineq_1} becomes
\begin{align}
&\sup_{t\in[0,T]}\mathbb{E}\left[(\hat{C}_t^{\bm{\theta}^\star})^{d(1-\delta)(1+\alpha)}\right]\\
&\le\dfrac{\kappa}{4}\left(1+\dfrac{1}{T}\int_0^T\mathbb{E}\left[e^{2d(1-\delta)(1+\alpha)T\left[\|\bm{\theta}^\star_u\|_2^2+ \|\mathbf{a}Y_u+\mathbf{b}-\mathbf{r}_n\|^2_2\right]}\right]\de u+\dfrac{2}{T}\int_0^T\mathbb{E}\left[e^{2d^2(1-\delta)^2(1+\alpha)^2wT\|\bm{\theta}^\star_u\|_2^2}\right]\de u\right)\\
&\le\dfrac{\kappa}{4}\left(1+\dfrac{3}{T}\int_0^T\mathbb{E}\left[e^{2d(1-\delta)(1+\alpha)\left[\left(1\vee d(1-\delta)(1+\alpha)w\right)T\|\bm{\theta}^\star_u\|_2^2 + T \|\mathbf{a}Y_u+\mathbf{b}-\mathbf{r}_n\|^2_2\right]} \right]\de u\right)\\
&\le\dfrac{\kappa}{4}\left(1+\dfrac{3}{T}\int_0^T\mathbb{E}\left[e^{2d(1-\delta)(1+\alpha)\left[\left(1\vee d(1-\delta)(1+\alpha)w\right)2nT(c_1^2Y_u^2+c_2^2)+2nT(a_M^2Y_u^2+b_M^2)\right]} \right]\de u\right)\\
&\le\dfrac{\kappa}{4}\left(1+\dfrac{3\kappa_1}{T}\int_0^T\mathbb{E}\left[e^{4d(1-\delta)(1+\alpha)nT\left[\left(1\vee d(1-\delta)(1+\alpha)w\right)c_1^2+a_M^2\right]Y_u^2} \right]\de u\right),
\end{align}
for some positive constant $\kappa_1$. In the third inequality we have used
$\max_{i=1,\dots,n}|\theta^{\star}_{i,u}|\le c_1|Y_u|+c_2$ and
$\max_{i=1,\dots,n}|\left(\mathbf{a}Y_u+\mathbf{b}-\mathbf{r}_n\right)_i|\le a_M|Y_u|+b_M$,
for every $u\in[0,T]$, where $c_1$ and $a_m$ are given by equations \eqref{eq:a} and \eqref{eq:exp_for_c_1} respectively, and 
\begin{align}
c_2&=\max_{i=1,\dots,n}\bigg|\left(\mathbf{\hat\Theta}^{-1}\left(\mathbf{b}-\bm{r}_{n}+\bm{\tilde\Sigma}_{\mathbf{S}}\bm{\tilde\Sigma}_Y^\top\sup_{t\in[0,T]}\hat{g}(t)\right)\right)_i\bigg|,\\
b_M&=\max_{i=1,\dots,n}|\left(\mathbf{b}-\mathbf{r}_n\right)_i|.
\end{align}
Consequently,
\begin{align}
\|\bm{\theta}_u\|_2^2&\le 2n\left(c_1^2Y_u^2+c_2^2\right),\\
\|\mathbf{a}Y_u+\mathbf{b}-\mathbf{r}_n\|_2^2&\le n\left(a_M^2|Y_u|+b_M\right)^2\le 2n\left(a_M^2Y_u+b_M^2\right),
\end{align}
for every $u\in[0,T]$. Finally, since $Y_t$ is Gaussian, 
\begin{equation}
\mathbb{E}\left[e^{4d(1-\delta)(1+\alpha)nT\left[\left(1\vee d(1-\delta)(1+\alpha)w\right)c_1^2+a_M^2\right]Y_u^2} \right]<\infty
\end{equation}
for every $u\in[0,T]$ if and only if 
\begin{equation}\label{eq:relation_delta_in_0_1}
1-8d(1-\delta)(1+\alpha)nT\left[\left(1\vee d(1-\delta)(1+\alpha)w\right)c_1^2+a_M^2\right]\mbox{Var}[Y_u]>0.
\end{equation}
Recalling that $\mbox{Var}[Y_u]<\max\left\lbrace P_0,\mbox{Var}[Y_T]\right\rbrace$, we get
\begin{multline}
1-8d(1-\delta)(1+\alpha)nT\left[\left(1\vee d(1-\delta)(1+\alpha)w\right)c_1^2+a_M^2\right]\mbox{Var}[Y_u]>\\1-8d(1-\delta)(1+\alpha)nT\left[\left(1\vee d(1-\delta)(1+\alpha)w\right)c_1^2+a_M^2\right]\max\left\lbrace P_0,\mbox{Var}[Y_T]\right\rbrace,
\end{multline}
for every $u\in[0,T]$. Then, the result then follows from \eqref{eq:cond_1_ammissibility_delta_in_0_1}. Now we discuss the second case where $\delta\in(1,\infty)$. Applying the same steps as in the previous case, equation \eqref{eq:relation_1} becomes
\begin{align}
\sup_{t\in[0,T]}\mathbb{E}\left[(\hat{C}_t^{\bm{\theta}^\star})^{d(1-\delta)(1+\alpha)}\right]\le\dfrac{\kappa_2}{T}\int_0^T\mathbb{E}\left[e^{4d(1-\delta)(1+\alpha)nT\left[\left(-(1+w)\wedge d(1-\delta)(1+\alpha)\tilde{w}\right)c_1^2-a_M^2\right]Y_u^2}\right]\de u.
\end{align}
where $\tilde{w}$ is given by equation \eqref{eq:exp_for_tilde_w}. As in the previous case,
\begin{equation}
\mathbb{E}\left[e^{4d(1-\delta)(1+\alpha)nT\left[\left(-(1+w)\wedge d(1-\delta)(1+\alpha)\tilde{w}\right)c_1^2-a_M^2\right]Y_u^2}\right]<\infty,
\end{equation}
for every $u\in[0,T]$ if and only if 
\begin{equation}
1-8d(1-\delta)(1+\alpha)nT\left[\left(-(1+w)\wedge d(1-\delta)(1+\alpha)\tilde{w}\right)c_1^2-a_M^2\right]\max\left\lbrace P_0,\mbox{Var}[Y_T]\right\rbrace>0.
\end{equation}
The result then follows from \eqref{eq:cond_1_ammissibility_delta_in_1_infty}.
\subsection{Proof of Proposition \ref{prop:originalcontrols}}\label{app:A_6}
The optimal controls $\bmtheta^\star$ are linked to $m^\star$ and $\bmpi^\star$ through the following system
\begin{equation}
\begin{cases}
\begin{aligned}
m_t\bm{\pi}_t&=\bm{\theta}^\star_t,\\
\bm{\pi}_t^\top\mathbf{1}_n&=1,
\end{aligned}
\end{cases}
\end{equation}
whose solutions are given by $(m^\star_t,\,\bm{\pi}_t^\star)=\left(\bm{\theta}_t^{\star,\top}\mathbf{1}_n,\,\frac{\theta^\star_{1,t}}{\bm{\theta}_t^{\star,\top}\mathbf{1}_n},\,\dots, \frac{\theta^\star_{n,t}}{\bm{\theta}_t^{\star,\top}\mathbf{1}_n}\right)$
for every $t\in[0,T]$. This concludes the proof.
\subsection{Proof of Corollary \ref{cor:solution_full_info_log}}\label{app:A_7}
We apply pointwise optimisation to obtain the optimal controls. Computing the expectation in \eqref{eq:LOG_optimisation_problem_full_info}, we get 
\begin{equation}\label{eq:obj_to_max}
\log\left(c\right)+r\left(T-t\right)+\mathbb{E}^{t,y}\left[\int_t^T\bmtheta_s^\top\left(\bma Y_s+\bmb-\mathbf{r}_n\right)\de s\right]-\dfrac{1}{2}\mathbb{E}^{t,y}\left[\int_t^T\bmtheta_s^\top\mathbf{\Theta}\bmtheta_s\de s\right],
\end{equation}
where $\mathbf{\Theta}=\bmtSigma_\bmS\bmtSigma_\bmS^\top+\bmSigma_{\bmS}\bmSigma_{\bmS}^\top\odot\bme$. 
Taking the first order conditions, we obtain the following system of linear equations $\bma Y_t+\bmb-\mathbf{r}_{n}-\mathbf{\Theta}\bmtheta_t=\mathbf{0}_n$, whose solution provide a candidate for the optimal control $\bmtheta^\star(t,y)=\mathbf{\Theta}^{-1}\left(\bma y+\bmb-\mathbf{r}_{n}\right)$. The Hessian matrix $-\mathbf{\Theta}$ is negative definite for every $\bmtheta$, ensuring that $\bmtheta^\star$ is the the unique well-defined maximiser of \eqref{eq:obj_to_max} and hence the optimal control. By inserting $\bmtheta^\star$ into \eqref{eq:obj_to_max}, we obtain a stochastic representation of the value function, namely
\begin{align}
v(t,c,y)=&\log\left(c\right)+\left[r+\dfrac{1}{2}\left(\bmb-\mathbf{r}_{n}\right)^\top\mathbf{\Theta}^{-1}\left(\bmb-\mathbf{r}_n\right)\right]\left(T-t\right)+\dfrac{1}{2}\bma^\top\mathbf{\Theta}^{-1}\bma\mathbb{E}^{t,y}\left[\int_t^TY^2_s\de s\right]\\
\label{eq:value_fun}&+\bma^\top\mathbf{\Theta}^{-1}\left(\bmb-\mathbf{r}_n\right)\mathbb{E}^{t,y}\left[\int_t^TY_s\de s\right].
\end{align}
Since $Y$ is modeled as an OU process, we can explicitly compute $\mathbb{E}^{t,y}\left[\int_t^TY_s\de s\right]$ and $\mathbb{E}^{t,y}\left[\int_t^TY_s^2\de s\right]$, which are given by
\begin{align}
\label{eq:exp_1}\mathbb E^{t,y}\left[\int_t^T Y_s\de s\right]=&y\dfrac{e^{\lambda\left(T-t\right)}-1}{\lambda}+\dfrac{\beta}{\lambda}\left[\frac{e^{\lambda\left(T-t\right)}-1}{\lambda}-\left(T-t\right)\right], \\[6pt]
\label{eq:exp_2}\mathbb E^{t,y}\left[\int_t^T Y_s^2\,ds\right]
=&\left(y+\dfrac{\beta}{\lambda}\right)^{2}\dfrac{e^{2\lambda\left(T-t\right)}-1}{2\lambda}-\frac{2\beta}{\lambda}\left(y+\dfrac{\beta}{\lambda}\right)\frac{e^{\lambda\left(T-t\right)}-1}{\lambda}+\dfrac{\beta^2}{\lambda^2}\left(T-t\right)\\
&+\dfrac{\sigma_Y^2}{2\lambda}\left[\frac{e^{2\lambda\left(T-t\right)}-1}{2\lambda}-\left(T-t\right)\right].
\end{align}
for every $t\in[0,T]$, respectively. By inserting the above expressions into \eqref{eq:value_fun} and rearranging the terms, we obtain the closed-form expression of the value function in equation \eqref{eq:value_fun_full_info_log}. This concludes the proof.
\section{An example involving two uncorrelated assets, independent of the factor process}\label{sect:example}
We consider a simplified setting in which only two stocks, $S_1$ and $S_2$, are traded on the market, representing a green and a brown stock, respectively. Moreover, we assume that $S_1$ and $S_2$ are driven by independent Brownian motions, and are also independent of the factor process $Y$. In this case, it is possible to show that the function $\Delta(x)$ is positive for $x\in(\delta^*,+\infty)$, for some $\delta^*<1$ that can be explicitly computed. In particular, we find that
\begin{equation}
\Delta(x)=\lambda^2-(1-x)\left(\dfrac{a_1^2}{x\sigma_1^2}+\dfrac{a_2^2}{(x+\varepsilon)\sigma_2^2}\right)\sigma_Y^2,\quad x\in(0,+\infty).
\end{equation}
This function is monotonic increasing, concave,  
$
\lim_{\delta\to+\infty}\Delta(\delta)=\lambda^2+\left(\dfrac{a_1^2}{\sigma_1^2}+\dfrac{a_2^2}{\sigma_2^2}\right)\sigma_Y^2>0,
$
for all $a_1,a_2\in\R$ and\footnote{Recall that for $a_1=0$ and  $a_2=0$ there is no factor process $Y$ and in this case $\Delta(x)=\lambda^2>0$, hence trivially, $\mathcal{P}=(0,1)\cup (1, +\infty)$.} 
\begin{equation}
\lim_{\delta\to 0^{+}}\Delta(\delta)=
\begin{cases}
-\infty, &\quad\mbox{if }a_1\neq 0,\\
\lambda^2-\dfrac{a_2^2}{\varepsilon\sigma_2^2}\sigma_Y^2,&\quad\mbox{if }a_1=0\mbox{ and }a_2\neq 0.
\end{cases}
\end{equation}
For $a_1=0$ and $a_2\neq 0$, we distinguish between two cases: 
\begin{itemize}
\item[(i)] if $\lambda^2-\frac{\sigma_Y^2a_2^2}{\varepsilon\sigma_2^2}\geq 0$, then $\Delta(\delta)>0$ for every $\delta\in(0,+\infty)$, hence $\delta^*=0$ and $\mathcal{P}=(0,1)\cup(1,+\infty)$,
\item[(ii)] if $\lambda^2<\frac{\sigma_Y^2a_2^2}{\varepsilon\sigma_2^2}$, then there exists a unique  $\bar \delta(\varepsilon)=\frac{a_2^2\sigma_Y^2-\varepsilon\lambda^2\sigma_2^2}{\lambda^2\sigma_2^2+a_2^2\sigma_Y^2}<1$, which depends on $\varepsilon$ such that $\Delta(\bar \delta)=0$. Hence, setting $\delta^*=\bar \delta \wedge 0$, we get that $\mathcal{P}=(\delta^*,1)\cup (1, +\infty)$. Note that the larger the value of $\varepsilon$, the larger the set of admissible risk aversion parameters.
\end{itemize}
In the case $a_1\neq 0$, $\delta^*$ is the positive solution of the equation 
\begin{equation}
\left[\lambda^2\sigma_1^2\sigma_2^2+\left(a_1^2\sigma_2^2+a_2^2\sigma_1^2\right)\sigma_Y^2\right]\delta^2+\left[\varepsilon\lambda^2\sigma_1^2\sigma_2^2-\left(\left(1-\varepsilon\right)a_1^2\sigma_2^2+a_2^2\sigma_1^2\right)\sigma_Y^2\right]\delta-\varepsilon a_1^2\sigma_2^2\sigma_Y^2=0.
\end{equation}
Note that this solution is still smaller than  $1$ and depends on $\varepsilon$, but it can never become zero or negative. Hence, $\mathcal{P}=(\delta^*, 1)\cup(1, +\infty)\subset (0,1)\cup(1,+\infty)$. This example provides additional insight. Indeed, by comparing the critical value $\delta^\star$ for different values of the penalisation $\varepsilon$, we find that the penalty for brown assets generally enlarges the set of admissible risk aversion parameters, which in turn implies that a lower risk aversion may be allowed for green assets.
\section{Proofs of some technical results of Section \ref{sect:opt_partial_info}}
\subsection{Proof of Theorem \ref{thm:CRRA_case_PARTIAL_INFO}}\label{app:B_1}
To prove the first part of the theorem we replicate the same argument as in the proof of Theorem \ref{thm:existence_CRRA}, with the ansatz 
\begin{equation}\label{eq:partial_info_GUESS}
f(t,c,\gamma)=\dfrac{c^{1-\delta}}{1-\delta}\hat{\psi}(t,\gamma) \,
\text{ and } \,
\hat{\psi}(t,\gamma)=\exp\left\lbrace\dfrac{\bar{f}(t)}{2}\gamma^2+\bar{g}(t)\gamma+\bar{h}(t)\right\rbrace.
\end{equation}

In the second part of the proof we establish 
the relationship between  between the solutions to the ODE systems in the full and partial information settings. In particular, applying equations $(28)$–$(30)$ in \cite{brendle2006portfolio}, we get \eqref{eq:rel_hat_f_bar_f}, \eqref{eq:rel_hat_g_bar_g}, and \eqref{eq:rel_hat_h_bar_h}. Moreover, since $\hat f(t),\,\hat g(t),\,\hat h(t) \in \mathcal{C}_b^{1}([0,T])$ (see Section \ref{sect:opt_problem_full_info}), to show that $\bar f(t)$, $\bar g(t)$, and $\bar h(t)$ belong to the same class of regularity, it suffices to prove that $1 - P(t)\hat f(t) > 0$ for all $t\in[0,T]$. To show $1-\hat f(t)P(t)>0$ for every $t\in[0,T]$, we start by proving that the closed set $\mathcal{T}:=\{t\in[0,T]:1-P(t)\hat{f}(t)=0\}$ is empty. Let us assume by contradiction that it is not empty and let $\bar{t}$ be its maximum. From the boundary condition of $\hat{f}$ we see that $1-P(T)\hat{f}(T)=1$, hence $\bar{t}<T$. Relation in \eqref{eq:rel_hat_f_bar_f} hold in the set $\mathcal{T}^C\cap[0,T]$, where $\mathcal{T}^C$ is the complement of $\mathcal{T}$. Therefore, for any $z>0$ such that $\bar{t}+z<T$, $(1-P(\bar t+z)\hat{f}(\bar t+z))\bar{f}(\bar t+z)=\hat{f}(\bar t+z)$ and, by continuity of all the functions involved in the equality, $(1-P(\bar t)\hat{f}(\bar t))\bar{f}(\bar t)=\hat{f}(\bar t)$. Since $\hat{f}(t)$ is a monotone function (either increasing or decreasing, depending on the parameter $\delta$) and $\hat{f}(T)=0$, then $\hat{f}(\bar{t})=0$, hence $\bar{t}\not\in\mathcal{T}$, which is a contradiction and $\mathcal{T}$ is the empty set. Since $\mathcal{T}$ is empty, $1-P(t)\hat{f}(t)$ is continuous on $[0, T ]$ and $\hat{f}(T ) = 1$, it follows that $1-P(t)\hat{f}(t)>0$ is strictly positive on $[0, T ]$. This concludes the proof.

\subsection{Proof of Proposition \ref{prop:sign_bar_f}}\label{app:B_2}
Since, as shown in Proposition \ref{thm:CRRA_case_PARTIAL_INFO}, $1 - P(t),\hat f(t) > 0$, it follows that $\hat f(t)$ and $\bar f(t)$ must have the same sign (positive if $\delta\in(0,1)\cap\mathcal{P}$ and negative if $\delta\in(1,+\infty)\cap\mathcal{P}$). We now prove that, if $\delta\in\mathcal{P}\cap(0,1)$, $\bar{f}(t)$ is positive strictly decreasing on $[0,T]$. This can be proved by rewriting the ODE in equation \eqref{eq:bar_f} as $\bar{f}_t(t)=G(\bar{f}(t))$, where
\begin{align}
G(t):=&-\left[\left(1-\delta\right)\mathbf{\bar{P}}(t)\mathbf{\hat{\Theta}}^{-1}\left(\mathbf{\bar{P}}(t)\right)^\top+\mathbf{\bar{P}}(t)\left(\mathbf{\tilde\Sigma}_{\mathbf{S}}\mathbf{\tilde\Sigma}_{\mathbf{S}}^\top\right)^{-1}\left(\mathbf{\bar{P}}(t)\right)^\top\right]t^2\\
&-2\left[\left(1-\delta\right)\mathbf{\bar{P}}(t)\mathbf{\hat{\Theta}}^{-1}\mathbf{a}+\lambda\right]t-\left(1-\delta\right)\mathbf{a}^\top\mathbf{\hat{\Theta}}^{-1}\mathbf{a},\quad t\in[0,T].
\end{align}
The boundary condition
implies that $\bar{f}(T)= 0$ and that $G(0)=-\left(1-\delta\right)\bma^\top\bmhatTheta^{-1}\bma<0$. Then, the function $G(t)$ must be negative
on $[0,T]$ for the boundary condition to be satisfied, and hence $\bar{f}(t)$ is strictly decreasing. The
same argument applies to the case $\delta\in(1+\infty)\cap\mathcal{P}$, where the derivative of $\bar{f}(t)$ is positive, and hence $\bar{f}(t)$ is strictly increasing. This concludes the proof.
\subsection{Proof of Proposition \ref{prop:sufficient_condition_for_ver_thm_partial_info}}\label{app:B_4}
The proof replicates the lines of that of Proposition \ref{prop:sufficient_cond_ver_FULL_INFO}. Also in this case, we will show that $\sup_{t\in[0,T]}\mathbb{E}\left[\hat V^{1+\alpha}(t,\hat{C}_t,\Gamma_t)\right]<\infty$, for some $\alpha>0$. Using the form of the function $\hat V$ (cfr. equation \eqref{eq:VAL_F_P_INFO}) we get that
\begin{align}
\sup_{t\in[0,T]}\mathbb{E}\left[\hat{V}^{1+\alpha}(t,\hat{C}_t^{\bmtheta},\Gamma_t)\right]=&\sup_{t\in[0,T]} \mathbb{E}\left[\dfrac{1}{1-\delta}(\hat{C}^{\bmtheta}_t)^{(1-\delta)(1+\alpha)}e^{\frac{(1+\alpha)\bar{f}(t)}{2}\Gamma_t^2+(1+\alpha)\bar{g}(t)\Gamma_t+(1+\alpha)\bar{h}(t)}\right]\\
\le&\tilde\kappa\sup_{t\in[0,T]}\mathbb{E}\left[(\hat{C}^{\bmtheta}_t)^{(1-\delta)(1+\alpha)}e^{\frac{(1+\alpha)\bar{f}(t)}{2}\Gamma_t^2+(1+\alpha)\bar{g}(t)\Gamma_t}\right]\\
\le&\tilde\kappa\left(\sup_{t\in[0,T]}\mathbb{E}\left[(\hat{C}^{\bmtheta}_t)^{d(1-\delta)(1+\alpha)}\right]^{\frac{1}{d}}\right)\left(\sup_{t\in[0,T]}\mathbb{E}\left[e^{\frac{q(1+\alpha)\bar{f}(t)}{2}\Gamma_t^2+q(1+\alpha)\bar{g}(t)\Gamma_t}\right]^{\frac{1}{q}}\right),
\end{align}
for some positive constant $\kappa$ and some $d,\,q>1$, where the first inequality comes from the fact that $\bar{h}(\cdot)\in\mathcal{C}^1_b([0,T])$, and the second follows from Hölder's inequality. The first expectation is finite because of admissibility of the strategy (see the second condition of Definition \ref{defn:F_admissible_strategies_theta}). The second expectation, instead, is finite because the process $\Gamma$ is Gaussian. Hence,
\begin{equation}\label{eq:RELATION}
\mathbb{E}\left[e^{\frac{q(1+\alpha)\bar{f}(t)}{2}\Gamma_t^2+q(1+\alpha)\bar{g}(t)\Gamma_t}\right]<\infty,
\end{equation}
for every $t\in[0,T]$ if and only if $1-q(1+\alpha)\bar{f}(t)\mbox{Var}[\Gamma_t]>0$, where $\mbox{Var}[\Gamma_t]=\mbox{Var}[Y_t]-P(t)$. If $\delta\in\mathcal{P}\cap(1,+\infty)$, from Proposition \ref{prop:sign_bar_f}, $\bar{f}(t)<0$. Hence, $1-q(1+\alpha)\bar{f}(t)\mbox{Var}[\Gamma_t]>0$ and \eqref{eq:RELATION} is satisfied. If $\delta\in\mathcal{P}\cap(0,1)$, still from Proposition \ref{prop:sign_bar_f}, $\bar{f}(t)$ is strictly positive and decreasing for every $[0,T]$. Therefore, 
\begin{align}
1-q(1+\alpha)\bar{f}(t)\mbox{Var}[\Gamma_t]&>1-q(1+\alpha)\bar{f}(0)\mbox{Var}[Y_t]\\
&\ge1-q(1+\alpha)\dfrac{\hat{f}(0)}{1-P(0)\hat{f}(0)}\max\left\lbrace P_0,\mbox{Var}[Y_T]\right\rbrace,
\end{align}
where the first inequality follows from the monotonicity of $\bar f$ and from the fact that $\mbox{Var}[\Gamma_t]<\mbox{Var}[Y_t]$. The second inequality follows from $\bar{f}(t)=\frac{\hat{f}(t)}{1-P(t)\hat{f}(t)}$ for every $t\in[0,T]$, and from the fact that $\mbox{Var}[Y_t]$ is always lower than its maximum value on $[0,T]$, that is $P_0$ or $\mbox{Var}[Y_T]$ depending on $\mbox{Var}[Y_t]$ being decreasing or increasing. Then the result follows immediately from \eqref{eq:cond_3}.
\subsection{Proof of Corollary \ref{cor:log_case_PARTIAL_INFO}}\label{app:B_3}
The proof follows the same lines as that of Corollary \ref{cor:solution_full_info_log}. Computing $\mathbb{E}^{t,c,\gamma}\left[\log(\hat C^{\bm{\theta}}_T)\right]$,
we get 
\begin{equation}
\log\left(c\right)+r\left(T-t\right)+\mathbb{E}^{t,\gamma}\left[\int_t^T\bmtheta_s^\top\left(\bma\Gamma_s+\bmb-\mathbf{r}_n\right)\de s\right]-\dfrac{1}{2}\mathbb{E}^{t,\gamma}\left[\int_t^T\bmtheta_s^\top\mathbf{\Theta}\bmtheta_s\de s\right],\\
\end{equation}
Taking the first order conditions, we obtain the following system of linear equations $$\bma\Gamma_t+\bmb-\mathbf{r}_{n}-\mathbf{\Theta}\bmtheta_t=\mathbf{0}_n,$$ whose solution $\bar{\bmtheta}^\star$ is given in equation \eqref{eq:sol_partial_info_log}. The Hessian matrix is given by $-\bm{\Theta}$ and it is negative definite for every $\bmtheta\in\R^n$. This ensure that $\bar{\bmtheta}^\star$ is the unique well-defined maximiser and hence the optimal controls.  Inserting the optimal strategy into the value function, we get
\begin{align}
\tilde V(t,c,\gamma)=&\log\left(c\right)+\left[r+\dfrac{1}{2}\left(\bmb-\mathbf{r}_{n}\right)^\top\mathbf{\Theta}^{-1}\left(\bmb-\mathbf{r}_n\right)\right]\left(T-t\right)+\dfrac{1}{2}\bma^\top\mathbf{\Theta}^{-1}\bma\mathbb{E}^{t,\gamma}\left[\int_t^T\Gamma^2_s\de s\right]\\
\label{eq:tilde_V}&+\bma^\top\mathbf{\Theta}^{-1}\left(\bmb-\mathbf{r}_n\right)\mathbb{E}^{t,\gamma}\left[\int_t^T\Gamma_s\de s\right].
\end{align}
Since $\Gamma_t$ is a Gaussian process, we can easily compute which are given by
\begin{align}\label{eq:Exp_1_Partial_Info}
\mathbb{E}^{t,\gamma}\left[\int_t^T\Gamma_s\de s\right]=&\left(\gamma+\dfrac{\beta}{\lambda}\right)\dfrac{e^{\lambda(T-t)}-1}{\lambda}-\dfrac{\beta}{\lambda}(T-t),\\[6pt]
\mathbb{E}^{t,\gamma}\left[\int_t^T\Gamma_s^2\de s\right]=&\left(\gamma+\dfrac{\beta}{\lambda}\right)^{2}\dfrac{e^{2\lambda(T-t)}-1}{2\lambda}
-2\left(\gamma+\dfrac{\beta}{\lambda}\right)\left(\dfrac{\beta}{\lambda}\right)\dfrac{e^{\lambda(T-t)}-1}{\lambda}+\left(\dfrac{\beta}{\lambda}\right)^{2}(T-t)\\
\label{eq:Exp_2_Partial_Info}&+\int_t^T\mathbf{\bar{P}}(u)\left(\mathbf{\tilde\Sigma}_{\mathbf{S}}\mathbf{\tilde\Sigma}_{\mathbf{S}}^\top\right)^{-1}\mathbf{\bar{P}}(u)^\top\dfrac{e^{2\lambda(T-u)}-1}{2\lambda}\de u,
\end{align}
for every $t\in[0,T]$, respectively.  By inserting equations \eqref{eq:Exp_1_Partial_Info} and \eqref{eq:Exp_2_Partial_Info} into \eqref{eq:tilde_V} and rearranging the
terms, we obtain the value function $\tilde V$ in equation equation \eqref{eq:PARTIAL_INFO_VALUE_FUN_LOG_CASE}. This concludes the proof.
\subsection{Proof of Proposition \ref{prop:LOSS_UTILITY_CRRA}}\label{app:B_5}
Applying the definition of $L_t$ for the CRRA utility case, we get that 
\begin{equation}\label{eq:L}
L_t=\mathbb{E}^c\left[\hat{v}(t,\hat{C}_t,Y_t)-\hat{V}(t,\hat{C},\Gamma_t)|\mathcal{F}_t\right]=\dfrac{c^{1-\delta}}{1-\delta}\left(\mathbb{E}\left[e^{\frac{\hat{f}(t)}{2}Y_t+\hat{g}(t)Y_t+\hat{h}(t)}|\mathcal{F}_t\right]-e^{\frac{\bar{f}(t)}{2}Y_t+\bar{g}(t)Y_t+\bar{h}(t)}\right).
\end{equation}
Since, $Y_t|\mathcal{F}_t\sim N(\Gamma_t, P_t)$, then 
\begin{equation}\label{eq:expected_value_normal}
\mathbb{E}\left[e^{\frac{\hat{f}(t)}{2}Y_t+\hat{g}(t)Y_t+\hat{h}(t)}|\mathcal{F}_t\right]=\frac{e^{\hat{h}(t)+\frac{1}{2}\frac{\hat{g}^2(t)P(t)}{1-\hat{f}(t)P(t)}+\frac{\hat{g}(t)\Gamma_t}{1-\hat{f}(t)P(t)}+\frac{1}{2}\frac{\hat{f}(t)\Gamma^2_t}{1-\hat{f}(t)P(t)}}}{\sqrt{1-P(t)\hat{f}(t)}},\quad t\in[0,T].
\end{equation}
It is worth noting that the above expression is well defined because $1-P(t)\hat{f}(t)>0$ for every $t\in[0,T]$ (see Theorem  \ref{thm:CRRA_case_PARTIAL_INFO}). Inserting \eqref{eq:expected_value_normal} into \eqref{eq:L} and using \eqref{eq:rel_hat_f_bar_f}, \eqref{eq:rel_hat_g_bar_g}, and \eqref{eq:rel_hat_h_bar_h} yields the result in equation \eqref{eq:LOSS_UTILITY_CRRA_CASE}. Applying the definition of efficiency (see equation \eqref{eq:efficiency_definition}), $\xi$ can be found by solving the following equation:
\begin{equation}
\dfrac{\zeta^{1-\delta}}{1-\delta}\mathbb{E}\left[e^{\frac{\hat{f}(0)}{2}Y^2_0+\hat{g}(0)Y_0+\hat{h}(0)}|\mathcal{F}_0\right]=\dfrac{1}{1-\delta}e^{\frac{\bar{f}(0)}{2}\Gamma^2_0+\bar{g}(0)\Gamma_0+\bar{h}(0)}.
\end{equation}
Using \eqref{eq:expected_value_normal} together with \eqref{eq:rel_hat_f_bar_f}, \eqref{eq:rel_hat_g_bar_g}, and \eqref{eq:rel_hat_h_bar_h}, the foregoing equation can be rewritten as
\begin{equation}
\zeta^{1-\delta}e^{\frac{1-\delta}{2}\int_0^T\frac{P(s)}{1-P(s)\hat{f}(s)}\left[\bm{\tilde\Sigma}_Y\bm{\tilde\Sigma}_{\mathbf{S}}^\top\hat{f}(s)+\mathbf{a}^\top\right]\mathbf{\hat\Theta}^{-1}\left[\bm{\tilde\Sigma}_Y\bm{\tilde\Sigma}_{\mathbf{S}}^\top\hat{f}(s)+\mathbf{a}^\top\right]^\top\de s}=1.
\end{equation}
Hence, the result in \eqref{eq:efficiency_CRRA_CASE} immediately follows. This concludes the proof.
\subsection{Proof of Corollary \ref{cor:loss_utility_eff_log}}\label{app:B_6}
Applying the definition of $L_t$ for the logarithmic utility case, noticing that $\mathbb{E}[Y_t^2|\mathcal{F}_t]=\Gamma_t^2+P(t)$, and using equation \eqref{eq:rel_log_case}, we obtain
\begin{equation}\label{eq:log_loss_proof}
L_t=\frac{\mathbf{a}^\top\mathbf{\Theta}^{-1}\mathbf{a}}{4\lambda}\left[\phi(t)P(t)+\sigma_Y^2\left(\frac{\phi(t)}{2\lambda}-\left(T-t\right)\right)-
\int_{t}^{T}\mathbf{\bar{P}}(s)\left(\mathbf{\tilde\Sigma}_{\mathbf{S}}\mathbf{\tilde\Sigma}_{\mathbf{S}}^\top\right)^{-1}\mathbf{\bar{P}}(s)^\top\phi(s)\de s\right],
\end{equation}
where $\phi(t):=e^{2\lambda(T-t)}-1$, for every $t\in[0,T]$. Since,
\begin{equation}
\int_t^T\phi(s)\de P(s)=\int_t^T\phi(s)\left(2\lambda P(s)+\sigma_Y^2\right)-\int_t^T\mathbf{\bar P}(s)\left(\mathbf{\tilde\Sigma}_{\mathbf{S}}\mathbf{\tilde\Sigma}_{\mathbf{S}}^\top\right)^{-1}\mathbf{\bar P}(s)^\top\phi(s)\de s,\quad t\in[0,T],
\end{equation}
we get that
\begin{align}\label{eq:relation_log}
\int_t^T\mathbf{\bar P}(s)\left(\mathbf{\tilde\Sigma}_{\mathbf{S}}\mathbf{\tilde\Sigma}_{\mathbf{S}}^\top\right)^{-1}\mathbf{\bar P}(s)^\top\phi(s)\de s=\int_t^T\phi(s)\left(2\lambda P(s)+\sigma_Y^2\right)\de s-\int_t^T\phi(s)\de P(s),
\end{align}
for every $t\in[0,T]$. Inserting equation \eqref{eq:relation_log} into \eqref{eq:log_loss_proof}, we obtain the expression for the loss of utility stated in \eqref{eq:LOSS_OF_UTILITY_LOG__CASE}. Applying the definition of efficiency (see equation \eqref{eq:efficiency_definition}), $\xi$ can be found by solving the following equation:
\begin{equation}\label{eq:EFF_COMP}
\mathbb{E}\left[v(0,\zeta,Y_0)-\tilde{V}(0,1,\Gamma_0)|\mathcal{F}_0\right]=0.
\end{equation}
Following the same steps used to derive the loss of utility, equation \eqref{eq:EFF_COMP} simplifies to
\begin{equation}
\log(\zeta)+\frac{\mathbf{a}^\top\mathbf{\Theta}^{-1}\mathbf{a}}{2}\int_0^TP(s)\de s=0.
\end{equation}
Hence, the result in \eqref{eq:EFFICIENCY_LOG__CASE} immediately follows. This concludes the proof.

\section{{Sensitivity analysis}}\label{sect:sensitivity_analysis}
{We perform a sensitivity analysis of the main results in Section \ref{sect:num_experiments} by varying, one at a time, the parameters $a_2$, $a_4$, $\sigma_2$, $\lambda$, and $\beta$, which, based on the evidence in Figure \ref{fig:TORNADO_PLOTS}, are the key drivers of the variability in the expected value and the variance of the distribution of the carbon-penalised PPI strategy at maturity. We denote by $\bar{E}^\star_{g}=\{\bar{E}^\star_{g,t}\}_{t\in[0,T]}$ and $\bar{E}^\star_{b}=\{\bar{E}^\star_{b,t}\}_{t\in[0,T]}$ the total exposures to green
and brown risky assets, given by
\begin{equation}
\bar{E}^\star_{g,t}=\sum_{i=1}^k\bar{E}^\star_{i,t},\quad
\bar{E}^\star_{b,t}=\sum_{i=k+1}^n\bar{E}^\star_{i,t},\quad t\in[0,T],
\end{equation}
respectively, with $\bar{E}^\star_{i,t}$ defined by equation \eqref{eq:exposures}. First, we conduct a static sensitivity analysis at time $t=0$, followed by a dynamic one. Figure \ref{fig:Static_Sensitivities_exposures} reports the total exposures to green and brown stocks at $t=0$ as a function of the five aforementioned parameters, for three different levels of carbon aversion. Figure \ref{fig:opt_mult_comp_risky_reference_port_t_0} reports the associated optimal multiplier together with the strategy’s exposure to the risk-free asset $S^0$ at $t=0$, for the same set of parameter variations. The parameter ranges coincide with those used to produce the tornado plot in Figure \ref{fig:TORNADO_PLOTS}.\\
To streamline the exposition, we discuss the sensitivity results in terms of increases in the key parameters relative to the benchmark specification (marked by the red dots in Figures \ref{fig:Static_Sensitivities_exposures} and \ref{fig:opt_mult_comp_risky_reference_port_t_0}). The outcomes associated with decreases over the same parameter ranges are read directly from the same plots and admit the opposite interpretation; we therefore do not comment on them separately.\\
Increasing the coefficient $a_2$ raises the expected return of the second green asset, making it more attractive. Equivalently, $\bar\pi^\star_2$ gets larger and this reduces the optimal weights $\bar\pi^\star_3$ and $\bar\pi^\star_4$ in brown stocks. This effect is more prominent for positive values of carbon‐aversion parameter $\varepsilon$, as a combined effect of a larger rate of return of green assets and penalization of brown assets. Consequently, the optimal strategy becomes more exposed to green assets and more conservative, i.e. with a larger exposure to the risk-free asset $S^0$, as discussed in Section \ref{sect:num_experiments}.\\
A higher value of the parameter $a_4$, instead, induces the opposite behaviour (Figure \ref{fig:Static_Sensitivities_exposures}): the rate of return of the asset $S^4$ increases, hence the exposure to this asset rises, and that of green stocks $S^1$ and $S^2$ decreases.  The  the optimal multiplier at $t=0$ (see Figure \ref{fig:opt_mult_comp_risky_reference_port_t_0}) increases, meaning that the strategy also scales up the overall exposure to the risky reference portfolio, and the exposure to the risk-free asset correspondingly decreases. This tilt toward brown is dampened when $\varepsilon$ is non zero.\\
Turning to sensitivity with respect to $\sigma_2$, in the same spirit of previous comments, a rise in riskiness of the green stock $S^2$ generates a lower exposure to this asset (hence lower green exposure overall). Moreover, the optimal multiplier at $t=0$ gets smaller, and the exposure to the risk-free asset $S^0$ increases. Again, as $\varepsilon$ increases, the strategy becomes generally greener and more conservative.\\
The tornado plot \ref{fig:TORNADO_PLOTS} also shows larger variations in the expected portfolio value and variance for parameters of the latent factor, which we further investigate. An increase in $\lambda$ makes the latent factor $Y$, hence the current filtered level, more persistent. In financial terms, good times (respectively bad times) are expected to last longer. The multiplier $\bar m_0^\star$ increases and the exposure to the risk-free asset $S^0$ decreases. Total exposures to both green and brown stocks also increase with $\lambda$. However, the usual effect on brown stocks applies for $\varepsilon>0$. As for the long-run mean of the latent factor $\beta$, when it increases, so does the expected excess returns for assets with positive factor loadings, improving the risk–return trade-off of the risky opportunity set. As a result, the strategy optimally scales up exposure to the risky reference portfolio: $\bar m_0^\star$ increases and the allocation to $S^0$ falls. Hence, both green and brown total exposures rise with $\beta$.\\
We now turn to the dynamic analysis. Figures \ref{fig:DYNAMIC_MULTIPLIER_A_2} - \ref{fig:DYNAMIC_EXPOSURES_BETA} show the trajectories of the optimal multiplier and the total exposures to green and brown assets over the investment horizon under full and partial information. For each key parameter identified by the tornado plots, three scenarios are considered: a $20$\% decrease relative to the benchmark (left panels), the benchmark values reported in Table \ref{tab:model_params} (central panels), and a $20$\% increase (right panels).}\\
{The dynamic results confirm the qualitative insights of the static analysis. To avoid repetitions we only focus on parameters in the latent factor dynamics. When $\lambda$ increases or the long-run mean $\beta$ rises, the strategy optimally scales up risky exposure, raising the multiplier and total exposures over time.\\
Across all parameter configurations in Figure\ref{fig:DYNAMIC_MULTIPLIER_A_2} - \ref{fig:DYNAMIC_EXPOSURES_BETA}, the full and partial information cases remain remarkably close: the trajectories of the optimal multiplier and the corresponding total green and brown exposures largely overlap, with the partially informed strategy displaying slightly smoother dynamics, consistent with the filtering of the latent factor. Overall, this indicates that the sensitivity patterns are robust to the information set, and that the comparison between full and partial information is qualitatively stable across the considered perturbations.}
\begin{figure}[H]
\centering
\includegraphics[width=0.32\linewidth]{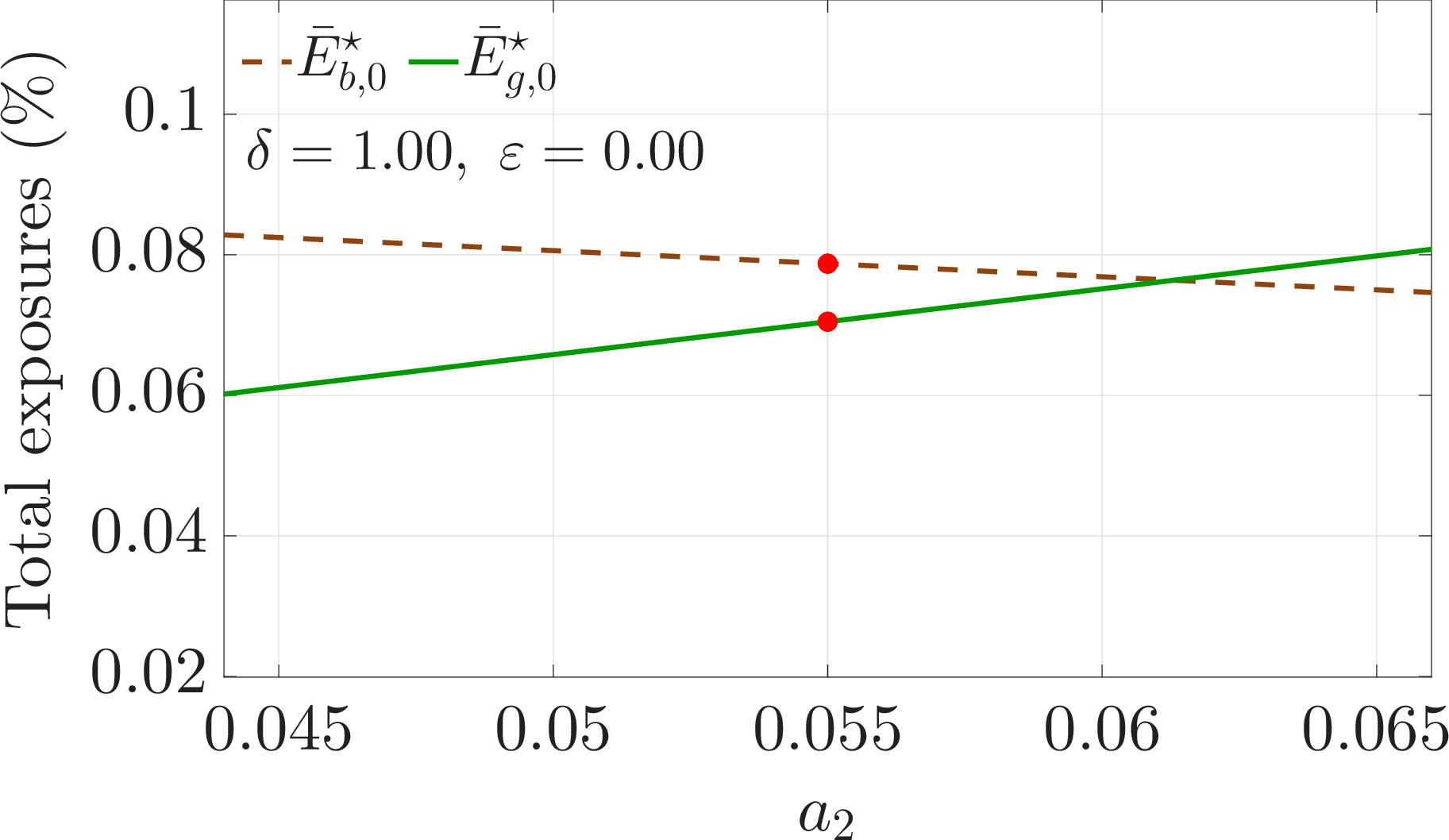} 
\vspace{.3cm}
\hfill 
\includegraphics[width=0.32\linewidth]{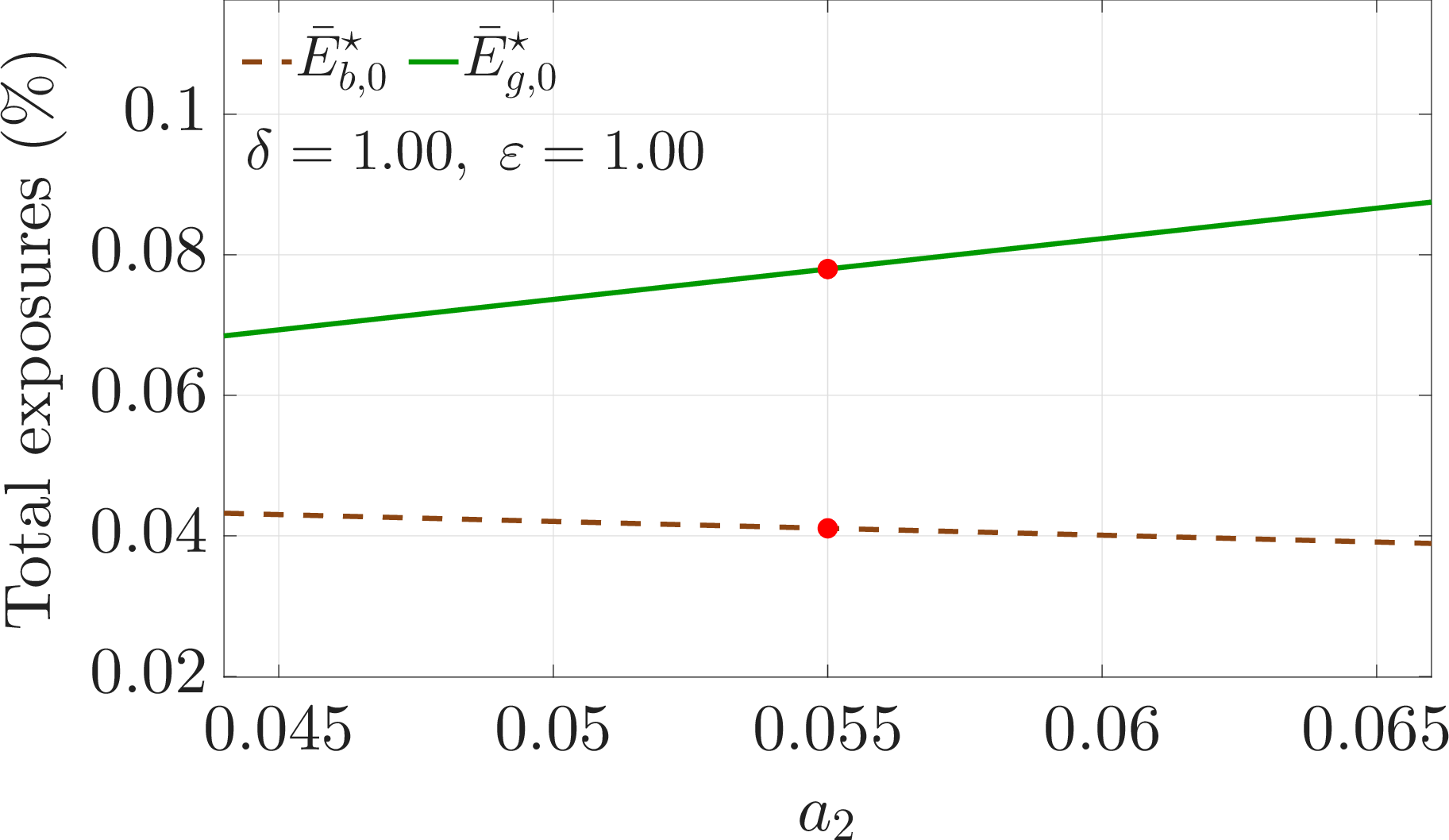}
\vspace{.3cm}
\hfill
\includegraphics[width=0.32\linewidth]{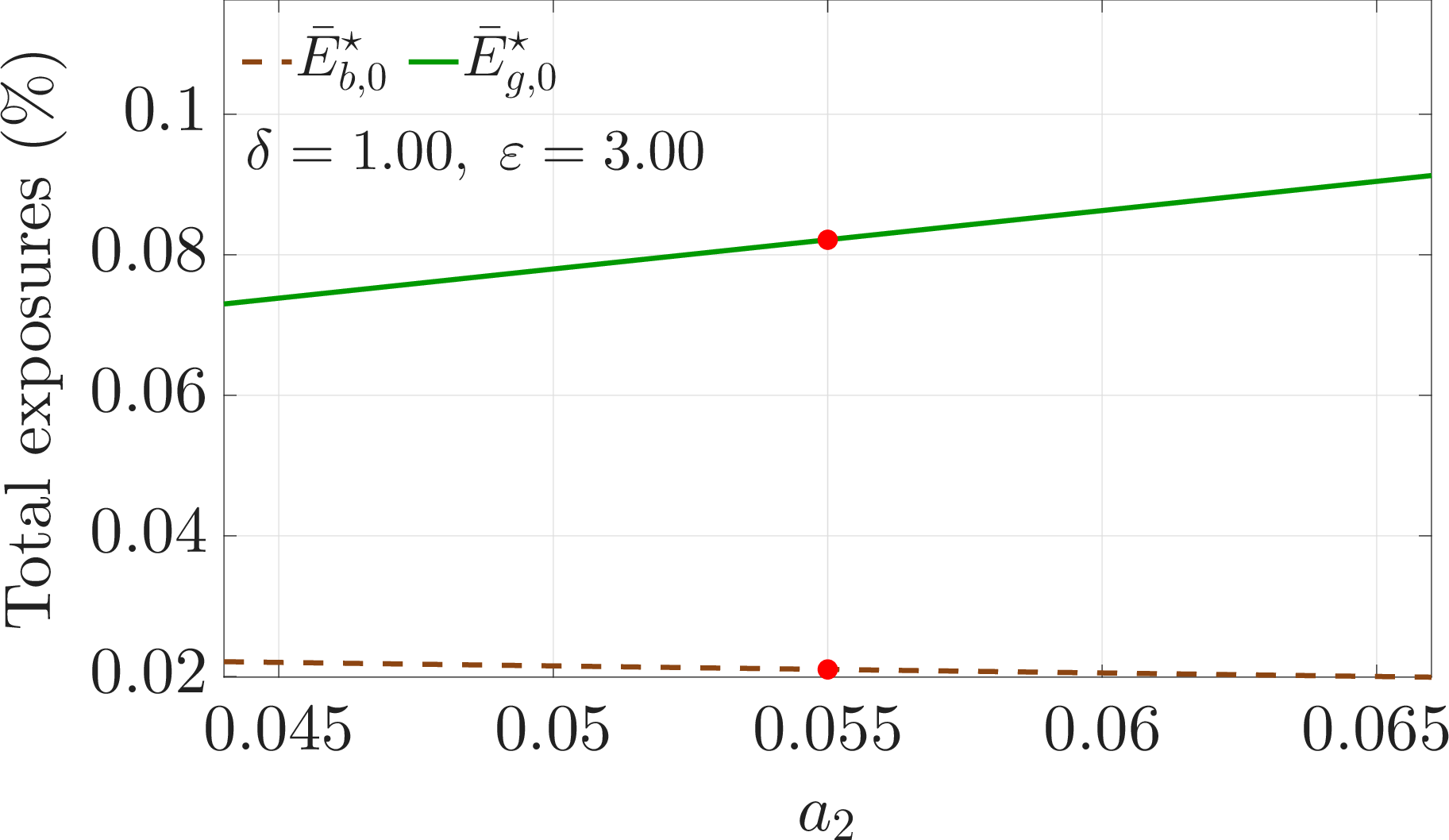}

\includegraphics[width=0.32\linewidth]{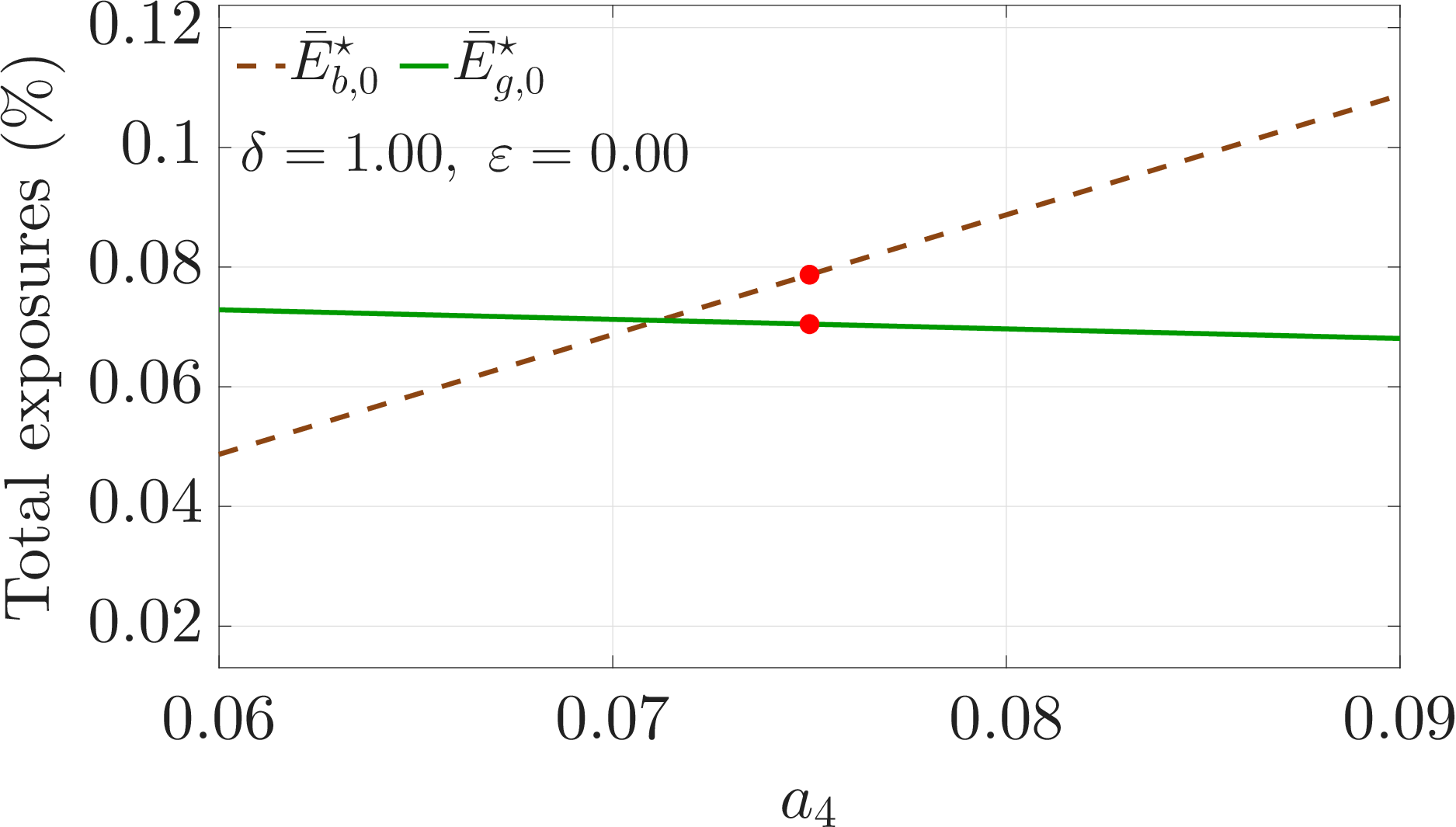} 
\vspace{.3cm}
\hfill 
\includegraphics[width=0.32\linewidth]{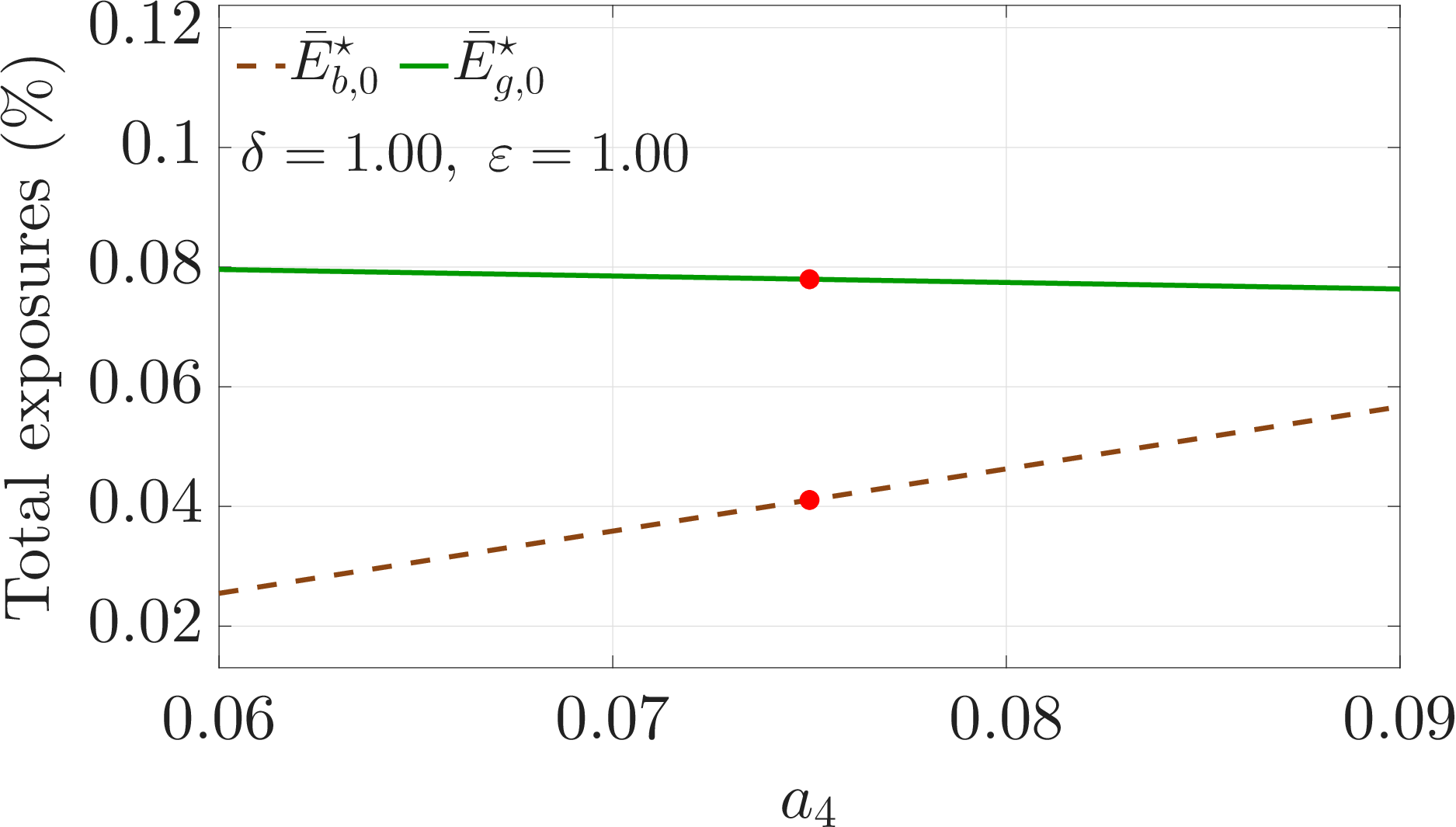}
\vspace{.3cm}
\hfill
\includegraphics[width=0.32\linewidth]{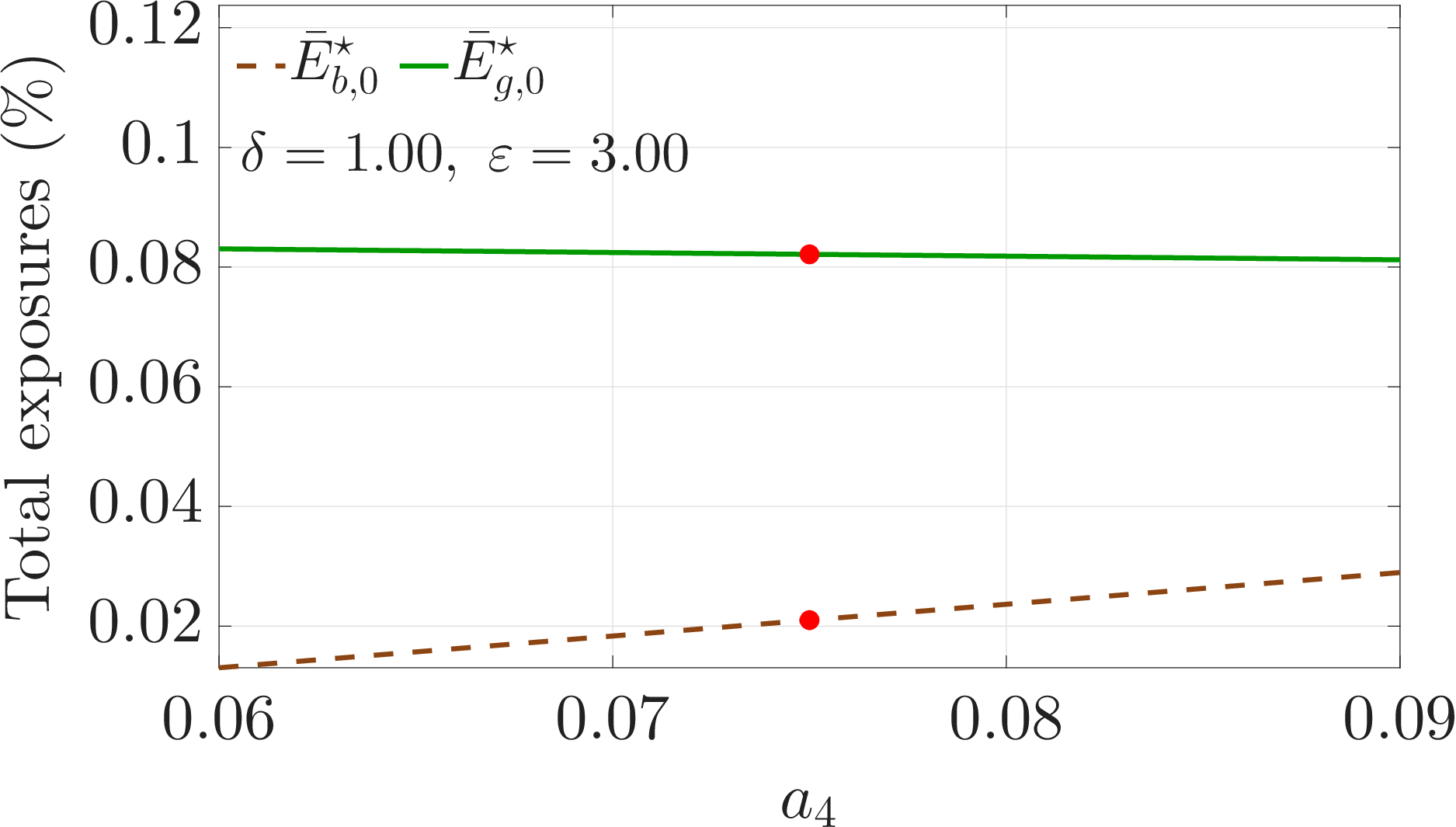}

\includegraphics[width=0.32\linewidth]{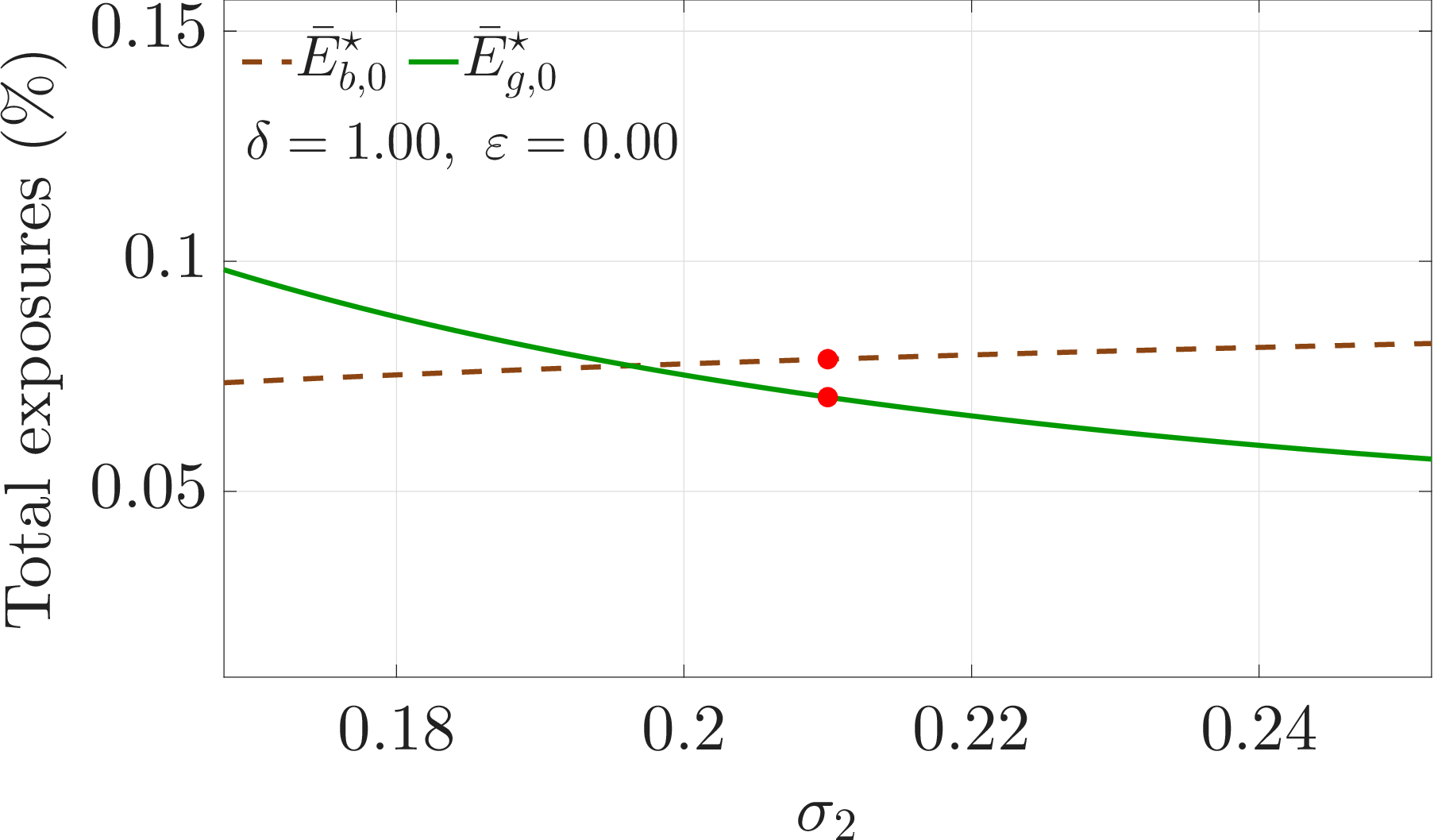} 
\vspace{.3cm}
\hfill 
\includegraphics[width=0.32\linewidth]{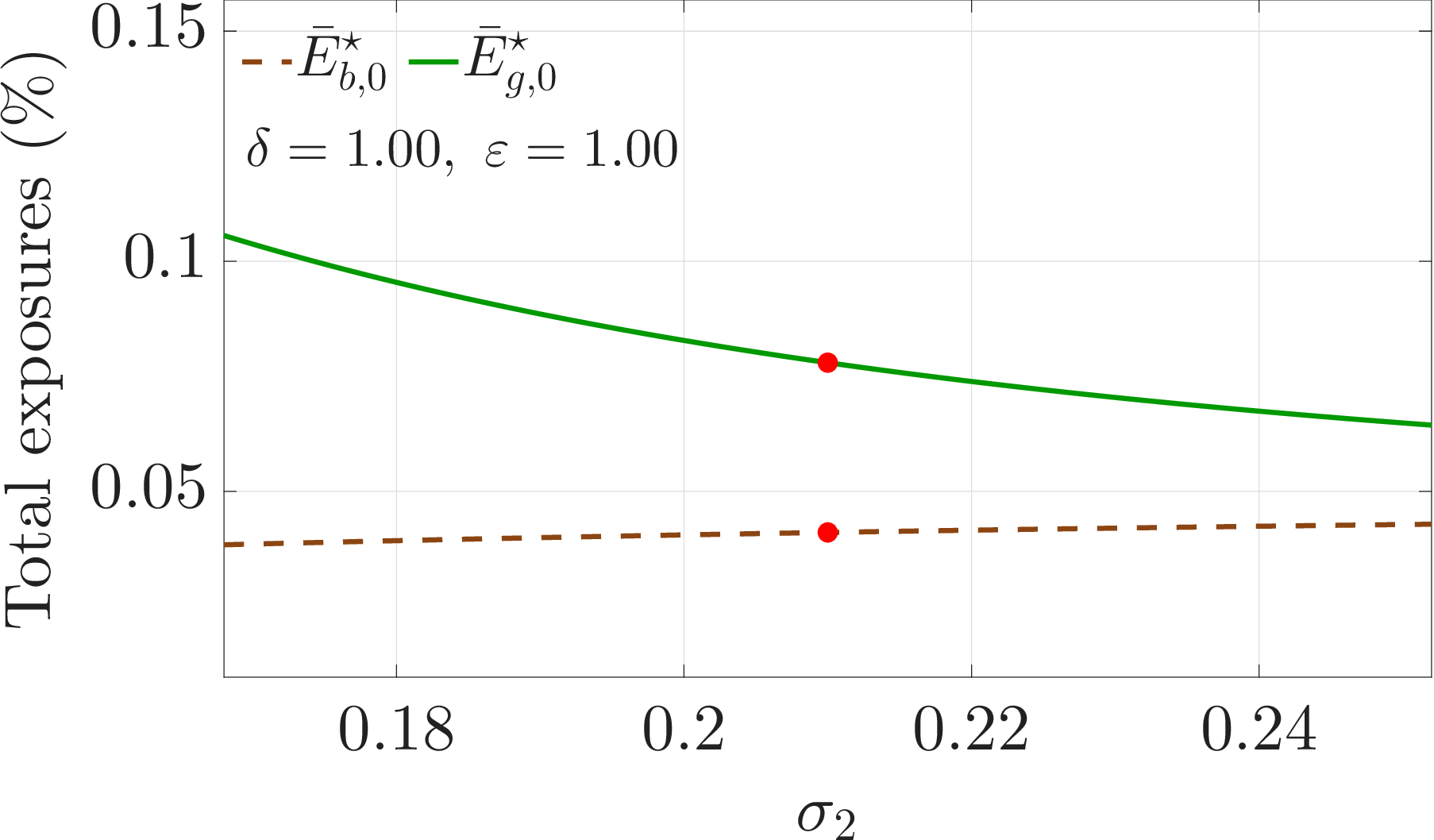}
\vspace{.3cm}
\hfill
\includegraphics[width=0.32\linewidth]{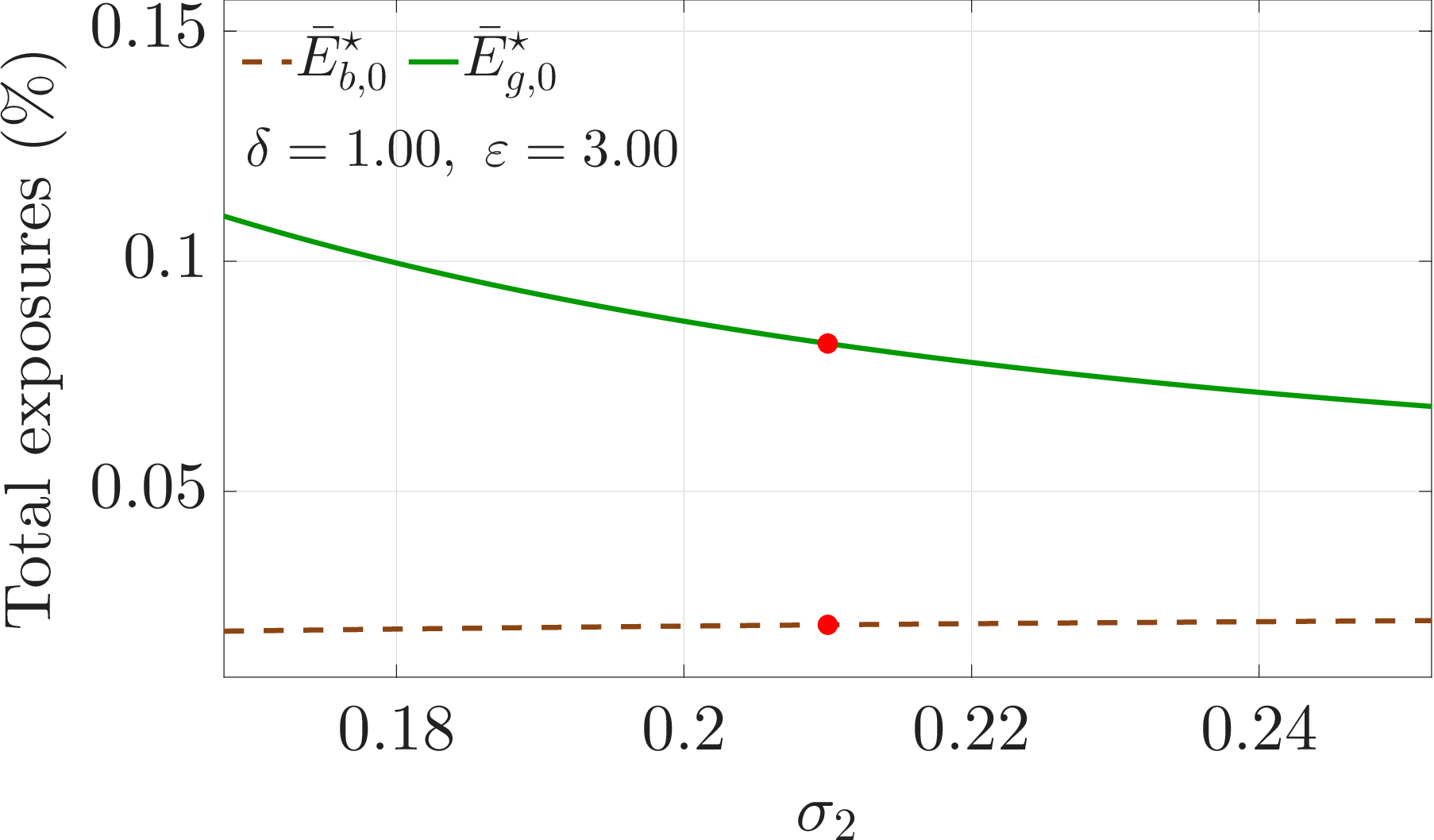}

\includegraphics[width=0.32\linewidth]{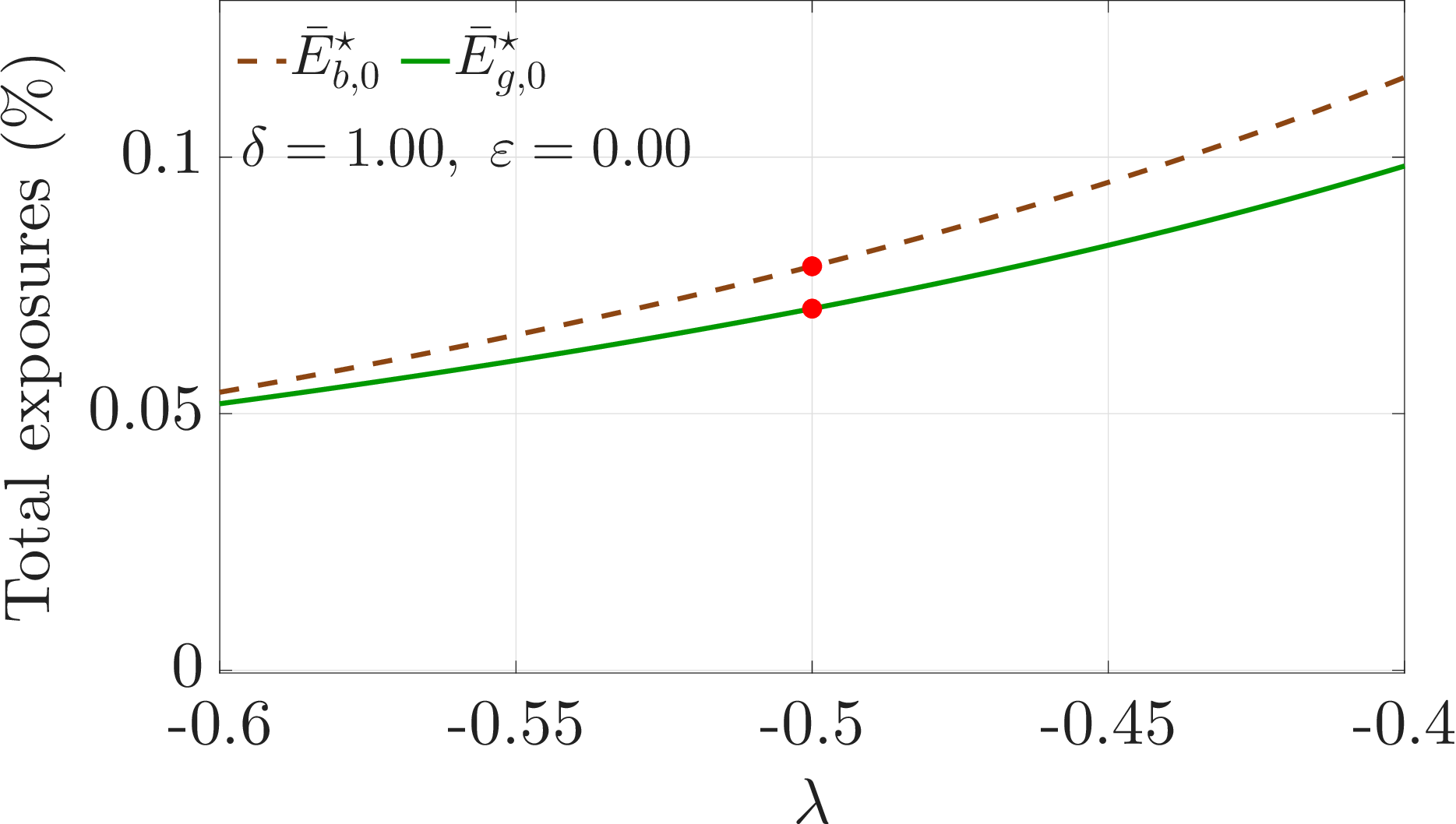} 
\vspace{.3cm}
\hfill 
\includegraphics[width=0.32\linewidth]{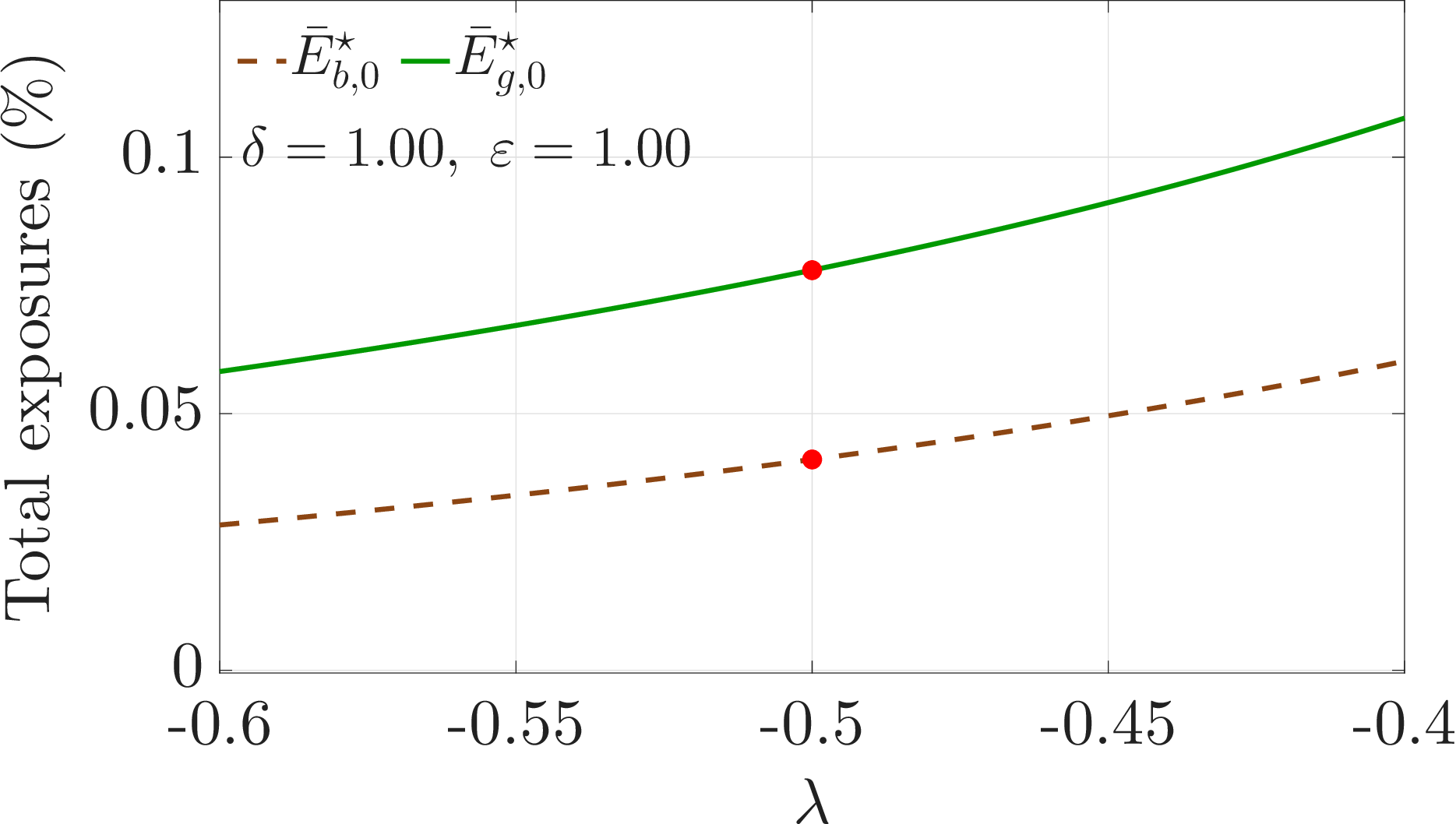}
\vspace{.3cm}
\hfill
\includegraphics[width=0.32\linewidth]{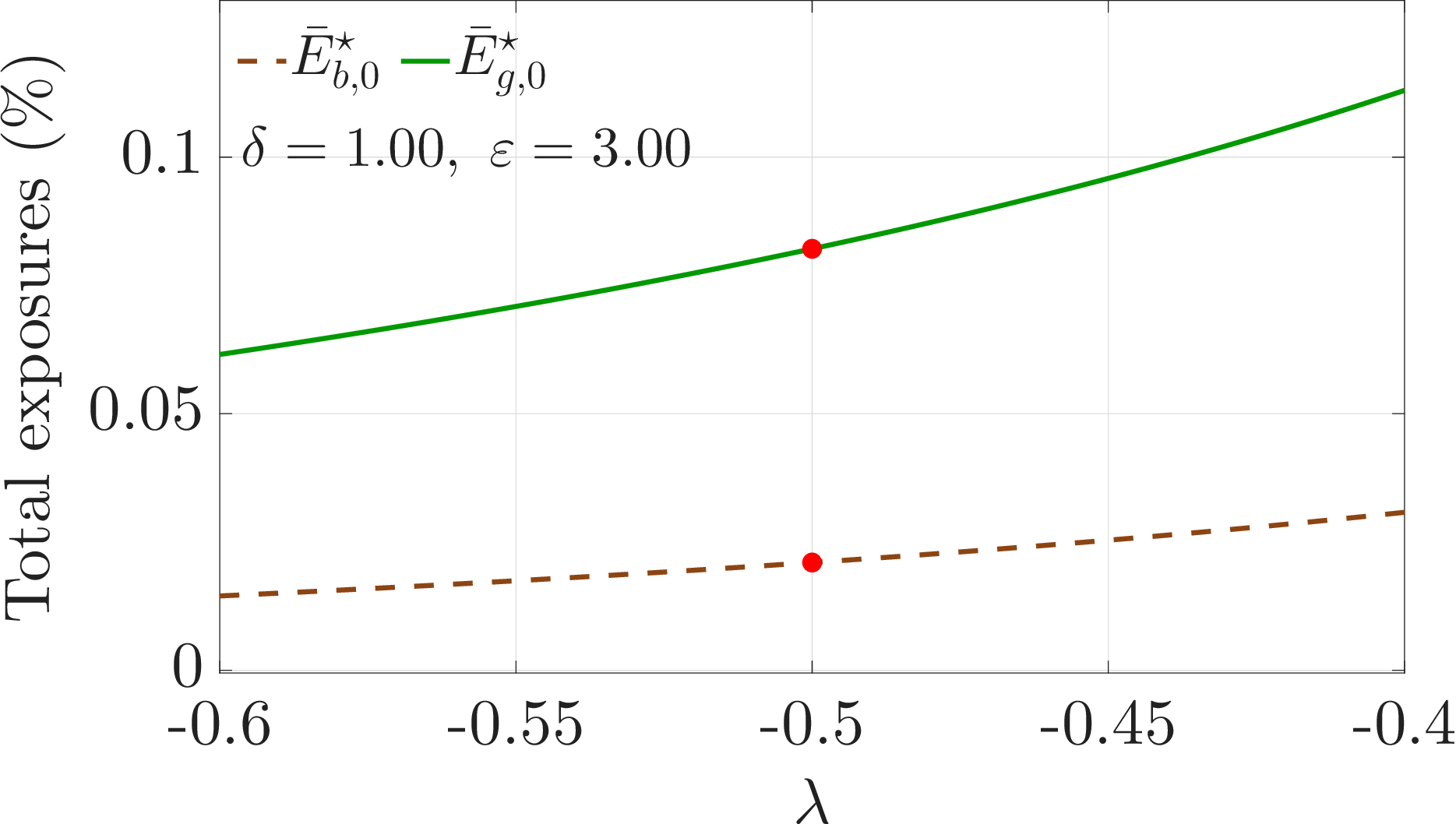}

\includegraphics[width=0.32\linewidth]{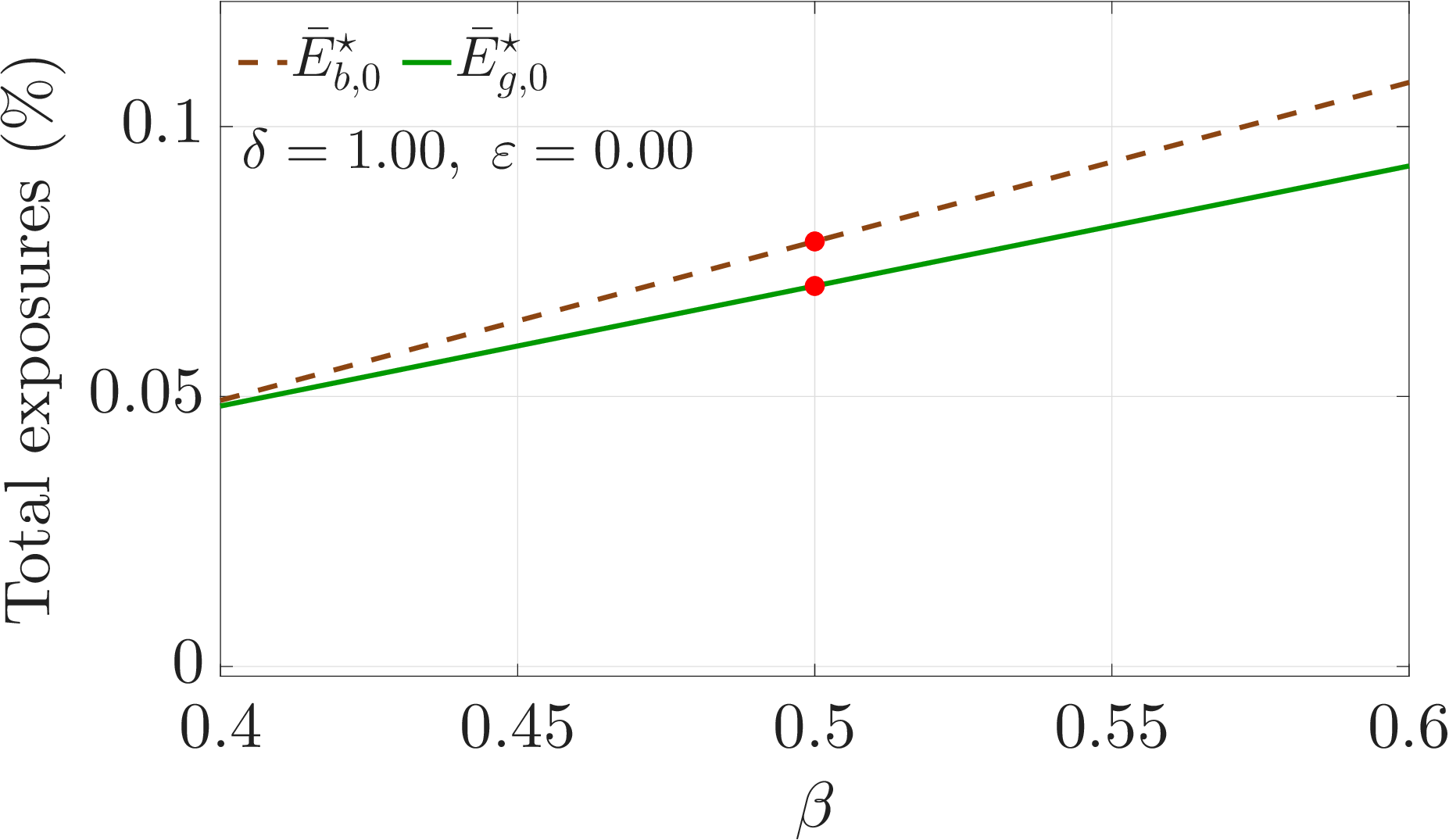} 
\vspace{.3cm}
\hfill 
\includegraphics[width=0.32\linewidth]{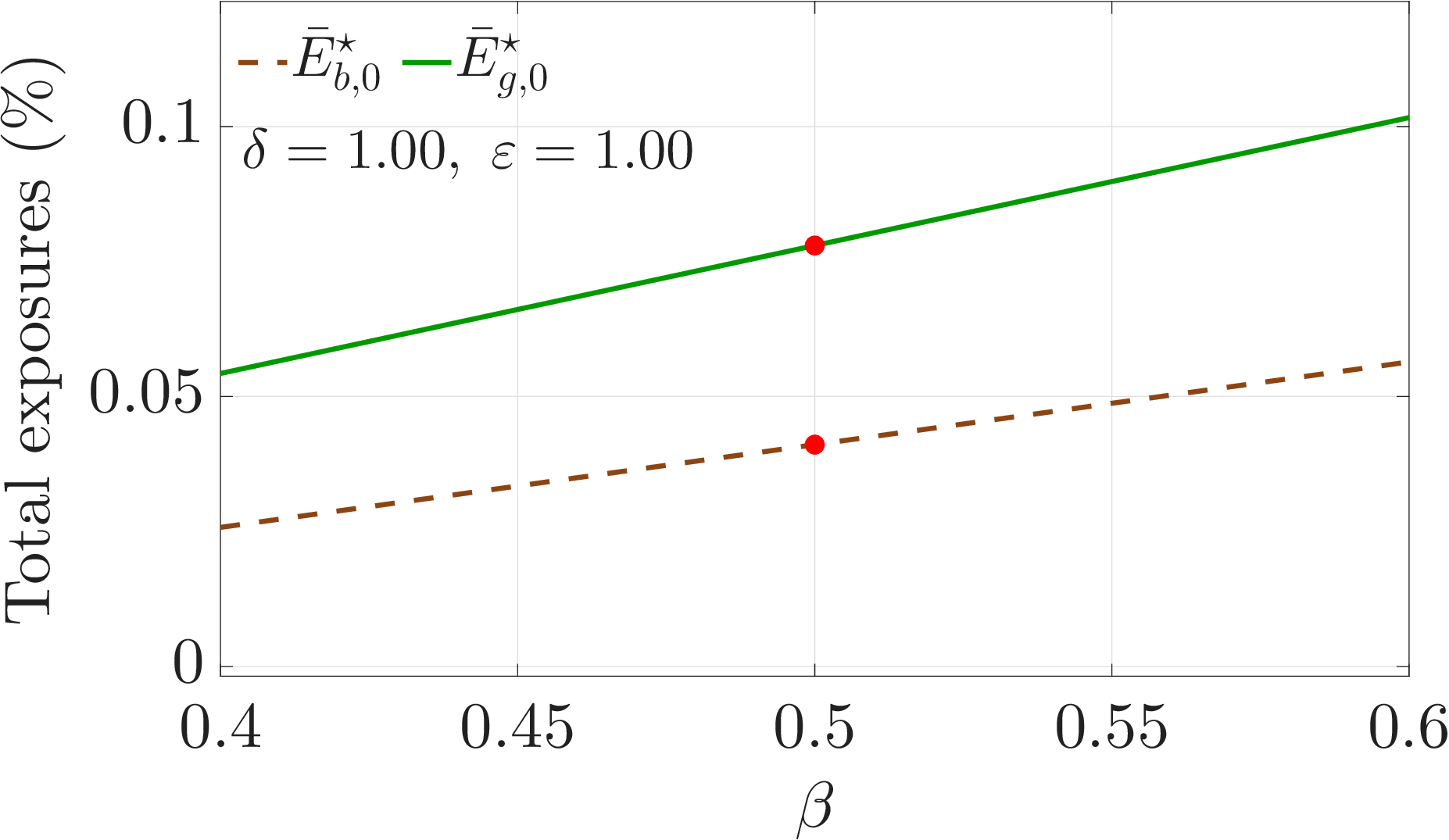}
\vspace{.3cm}
\hfill
\includegraphics[width=0.32\linewidth]{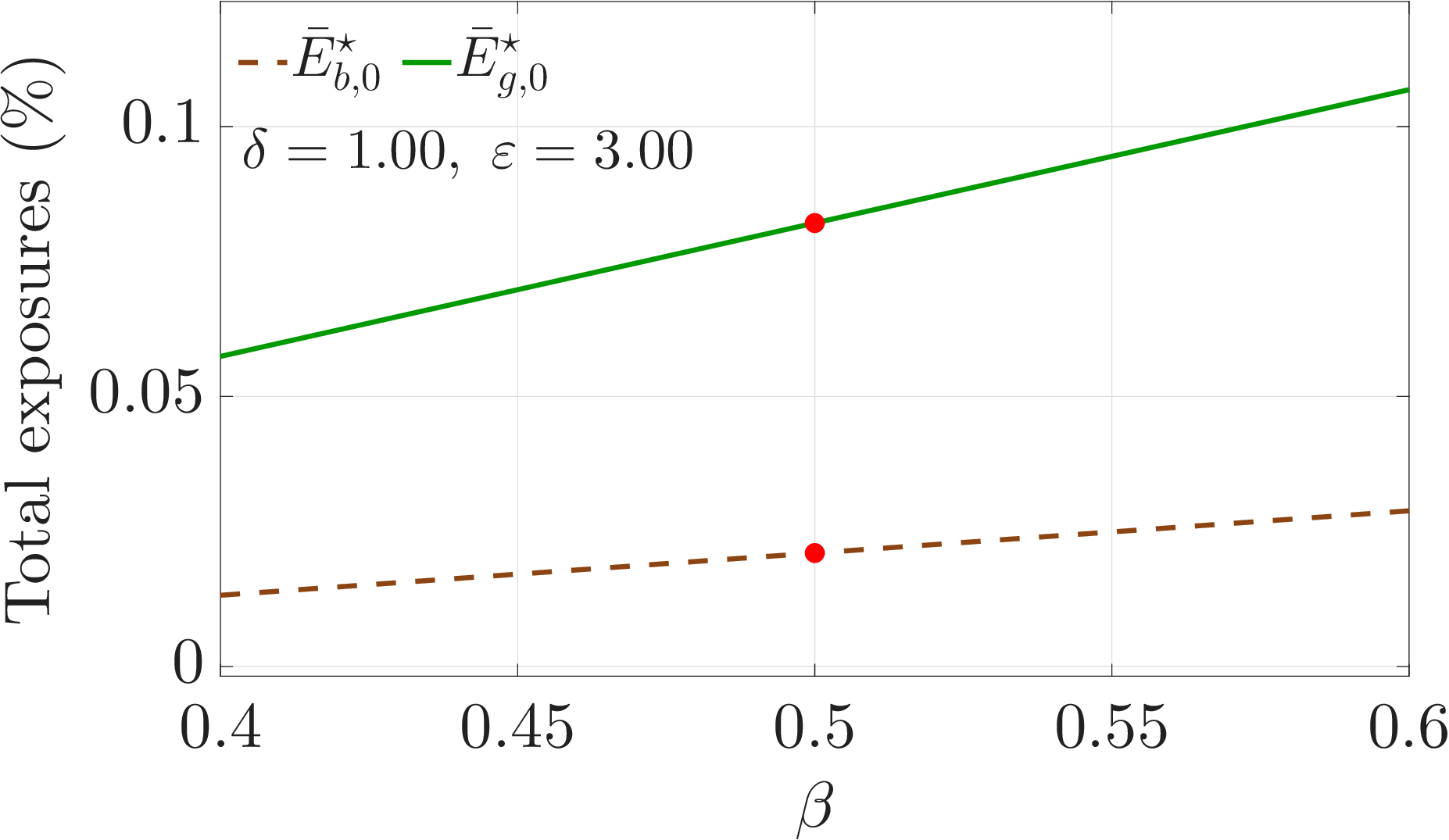}
\caption{Total optimal exposure (in \%) of the carbon penalised PPI strategy to brown and green stocks at $t=0$ as a function of the market parameters. Results are reported under partial information for $\delta=1$ and carbon aversion $\varepsilon=\{0,\,1,\,3\}$ (left-to-right columns). The red dots denote the total exposures under the baseline parameter configuration in Table \ref{tab:model_params}.}\label{fig:Static_Sensitivities_exposures}
\end{figure}
\newpage
\begin{figure}[H]
\centering
\includegraphics[width=0.38\linewidth]{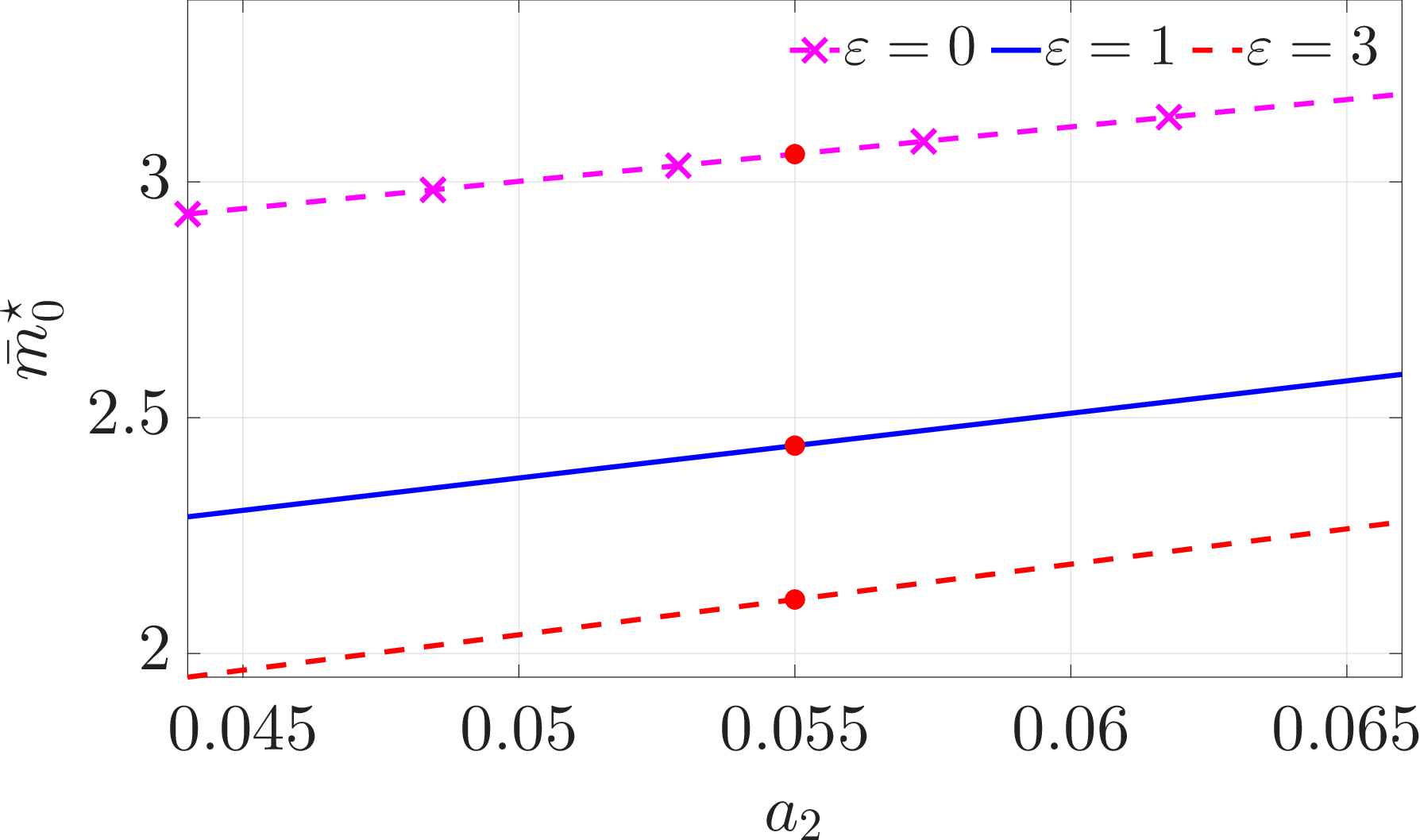} 
\vspace{.2cm}
\hspace{0.5cm}
\includegraphics[width=0.38\linewidth]{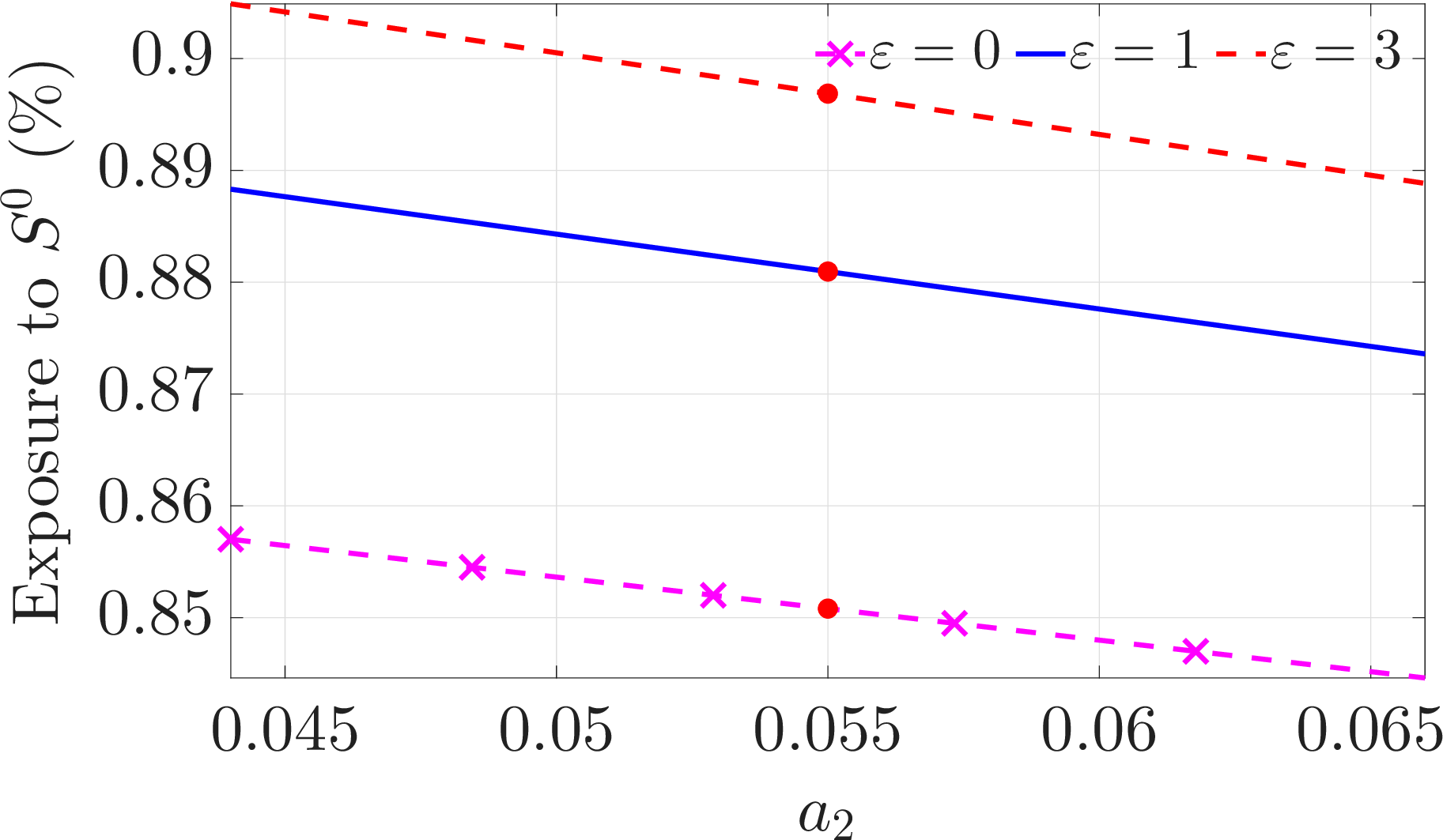}

\includegraphics[width=0.38\linewidth]{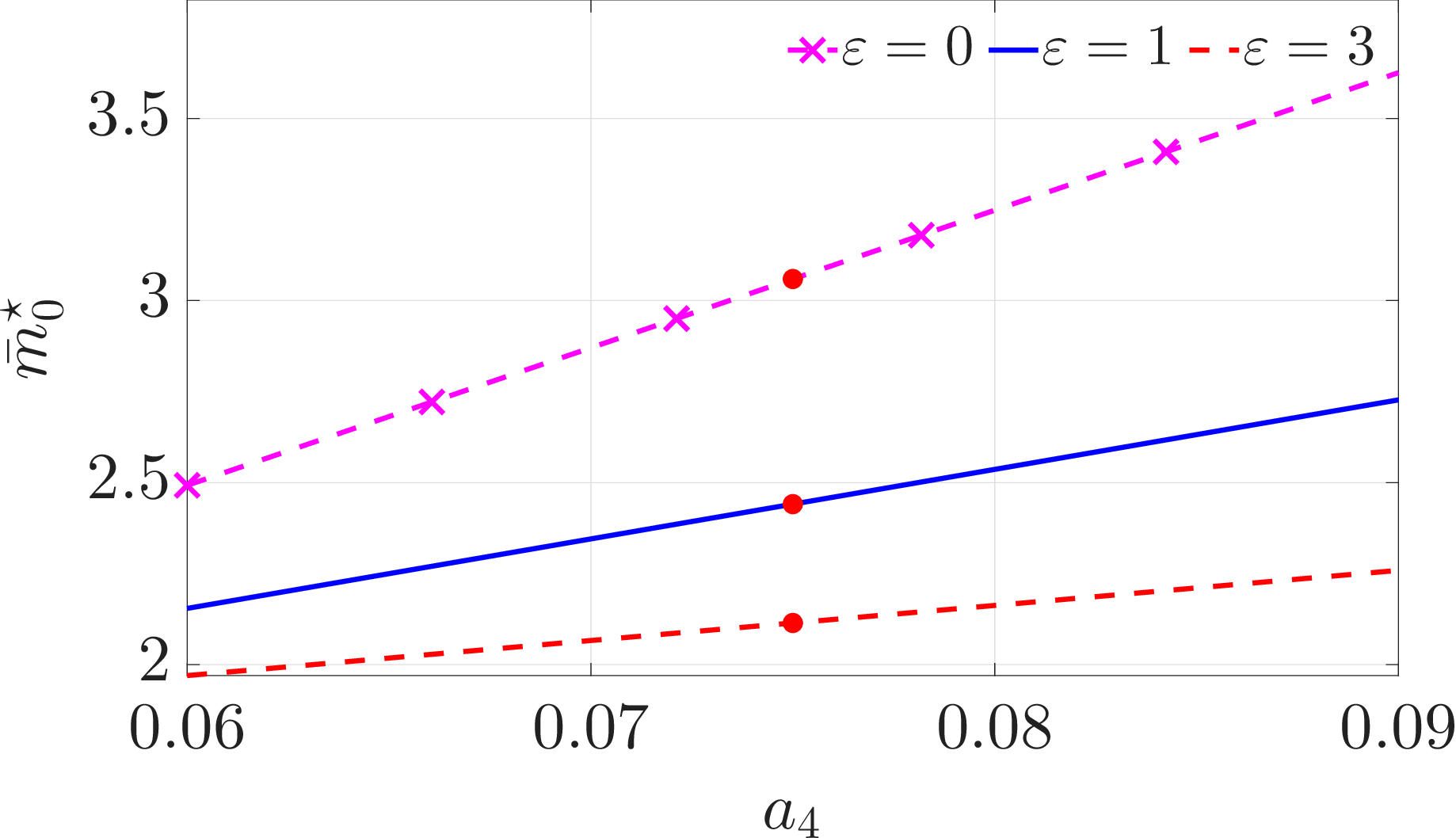} 
\vspace{.2cm}
\hspace{0.5cm}
\includegraphics[width=0.38\linewidth]{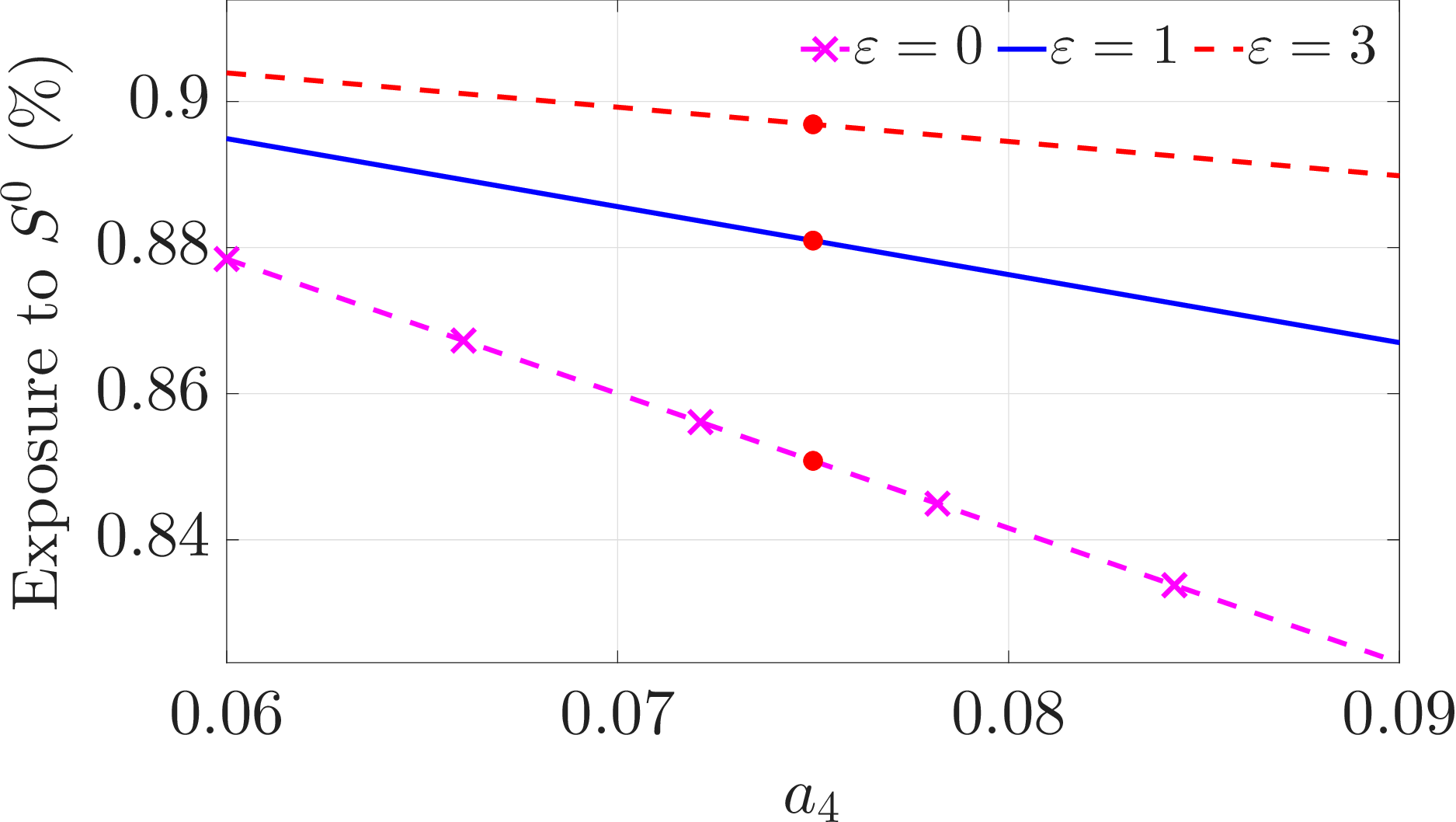}

\includegraphics[width=0.38\linewidth]{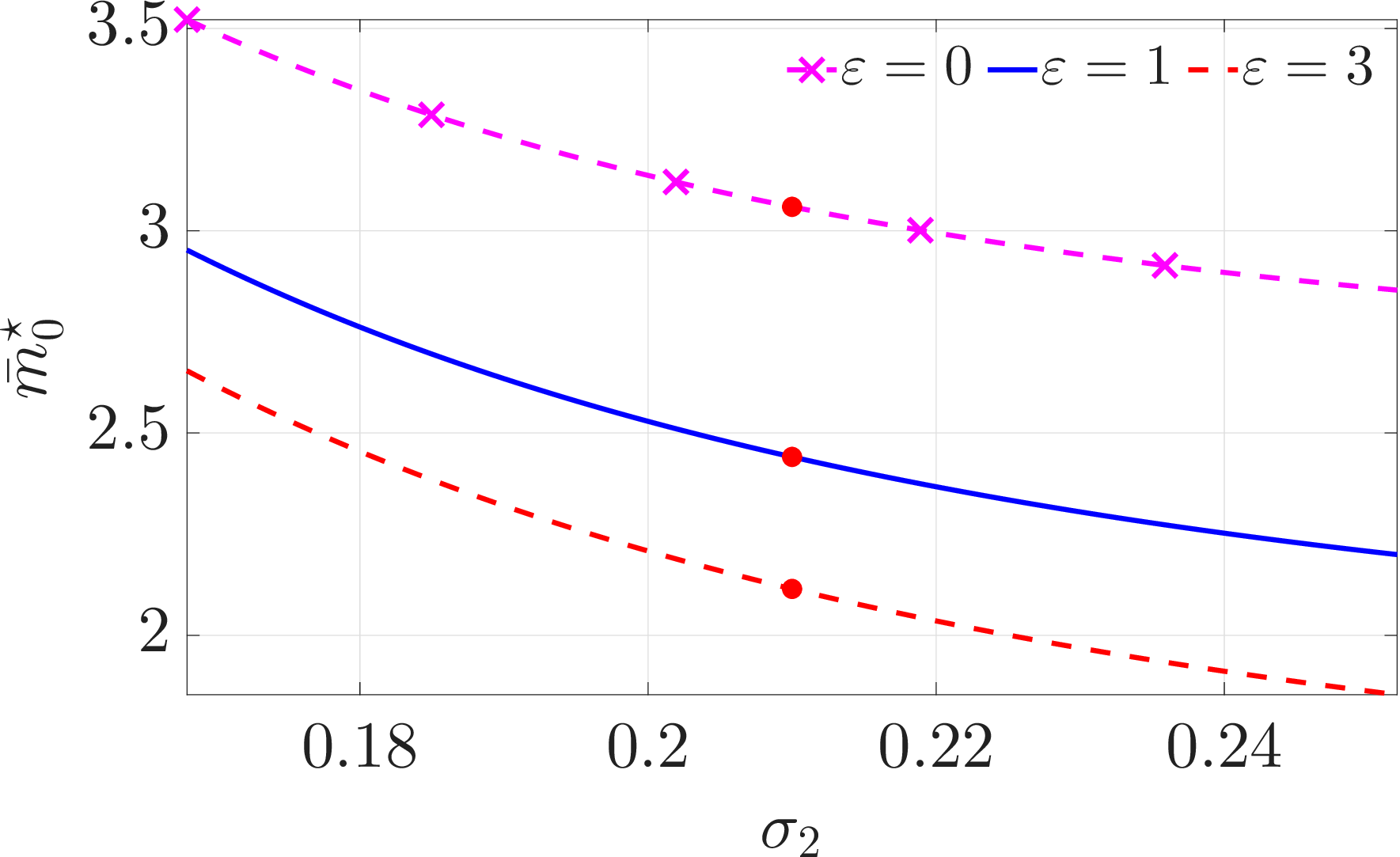} 
\vspace{.2cm}
\hspace{0.5cm}
\includegraphics[width=0.38\linewidth]{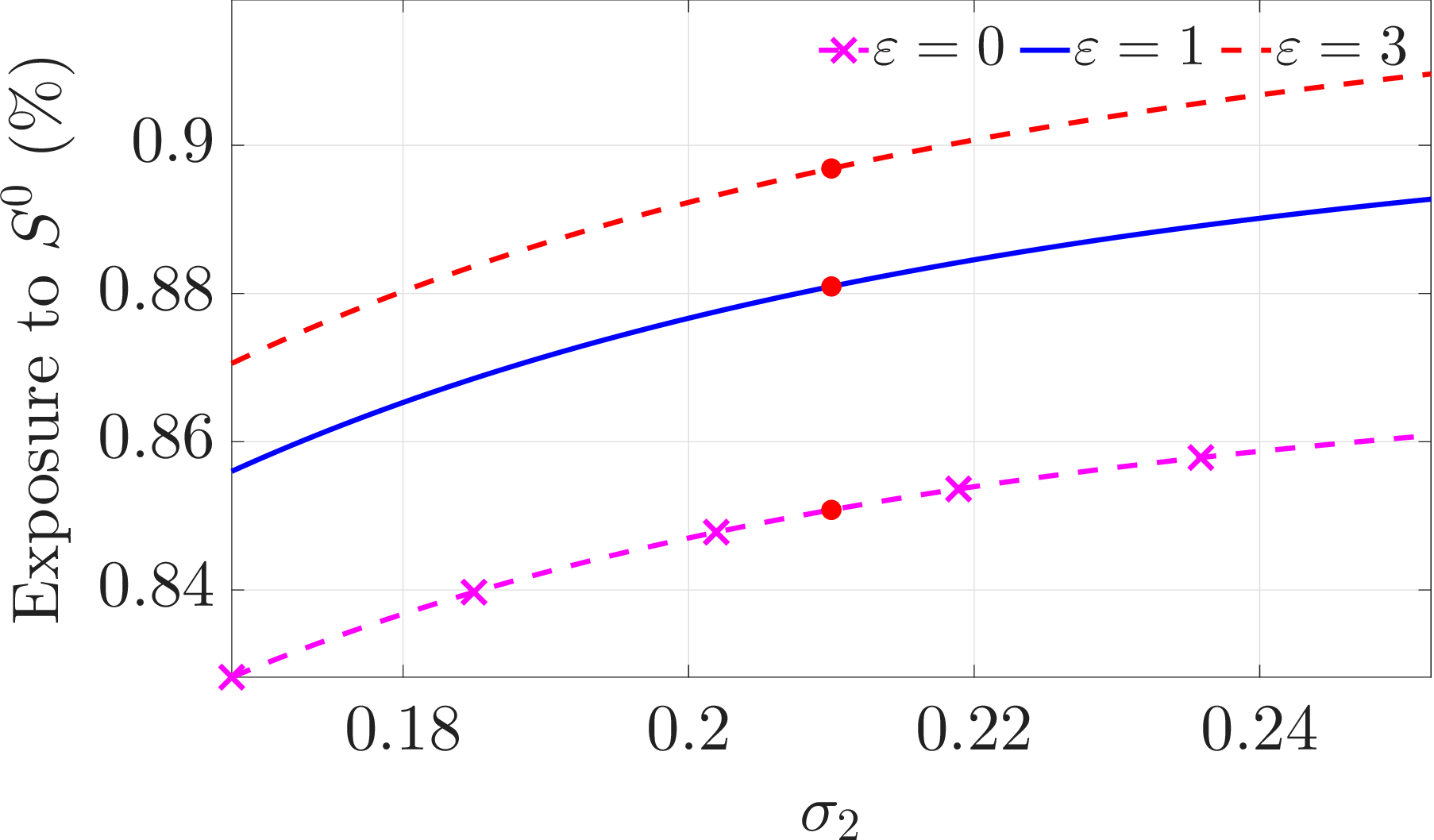}

\includegraphics[width=0.38\linewidth]{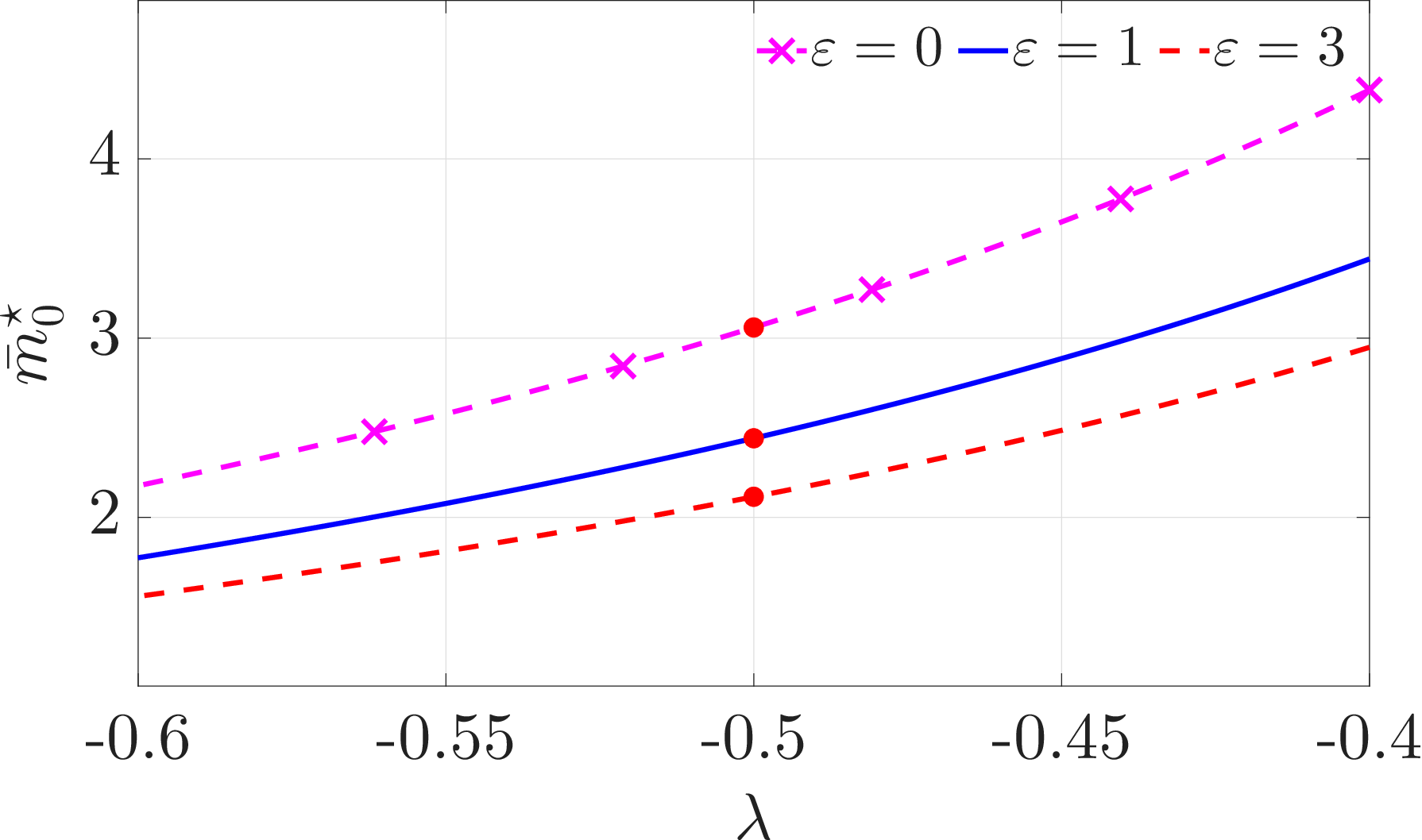} 
\vspace{.2cm}
\hspace{0.5cm}
\includegraphics[width=0.38\linewidth]{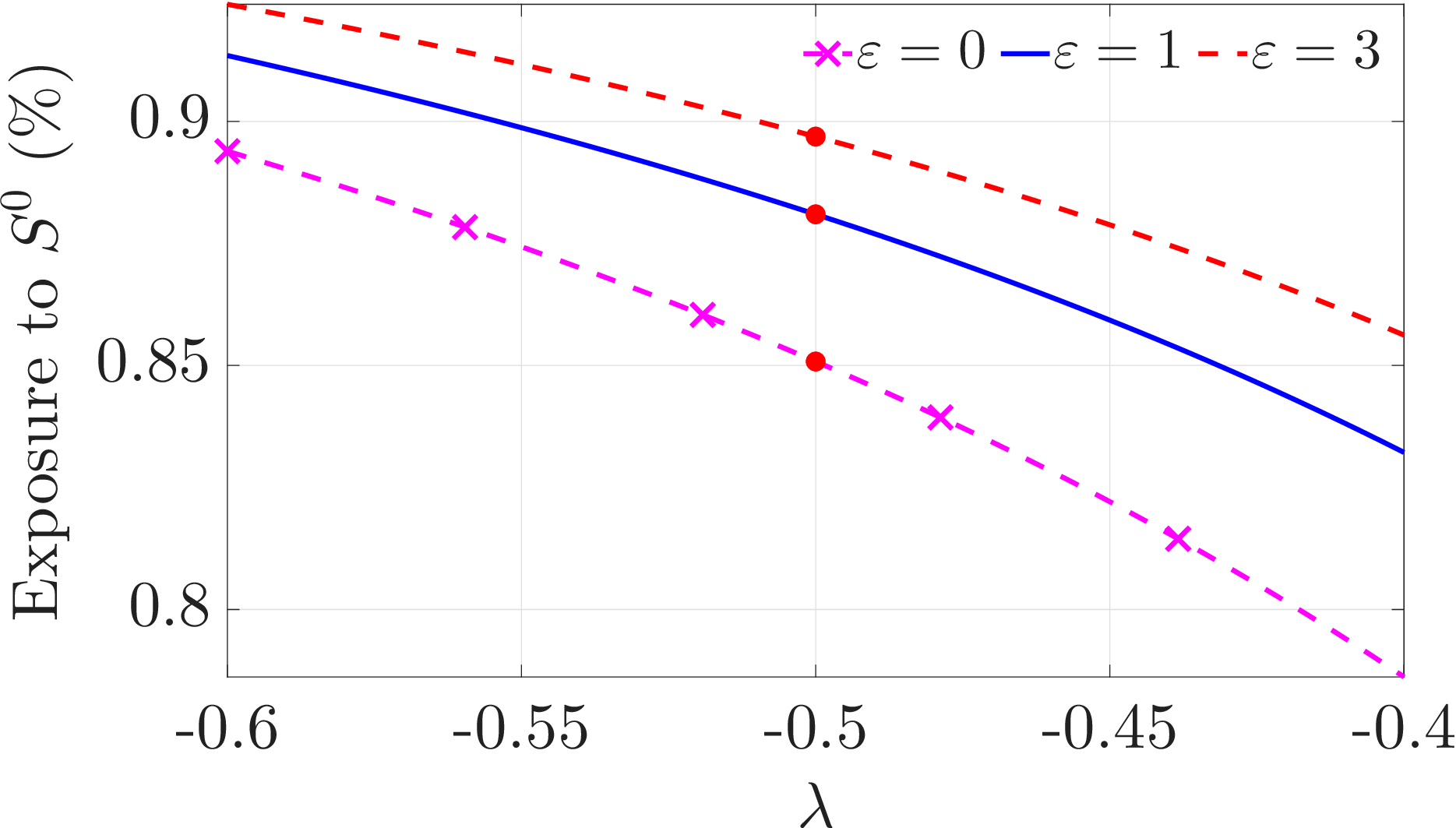}

\includegraphics[width=0.38\linewidth]{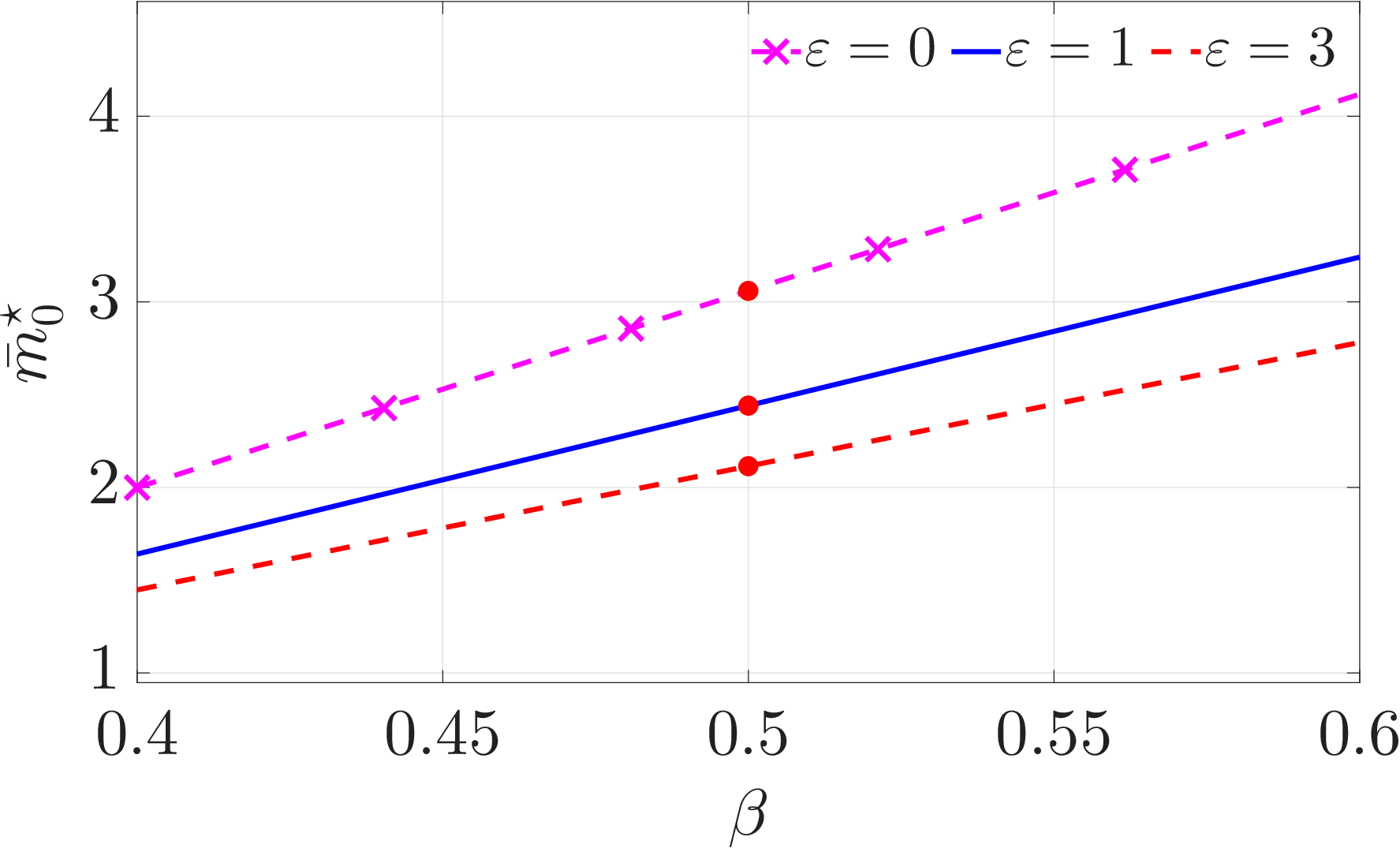}
\vspace{.2cm}
\hspace{0.5cm}
\includegraphics[width=0.38\linewidth]{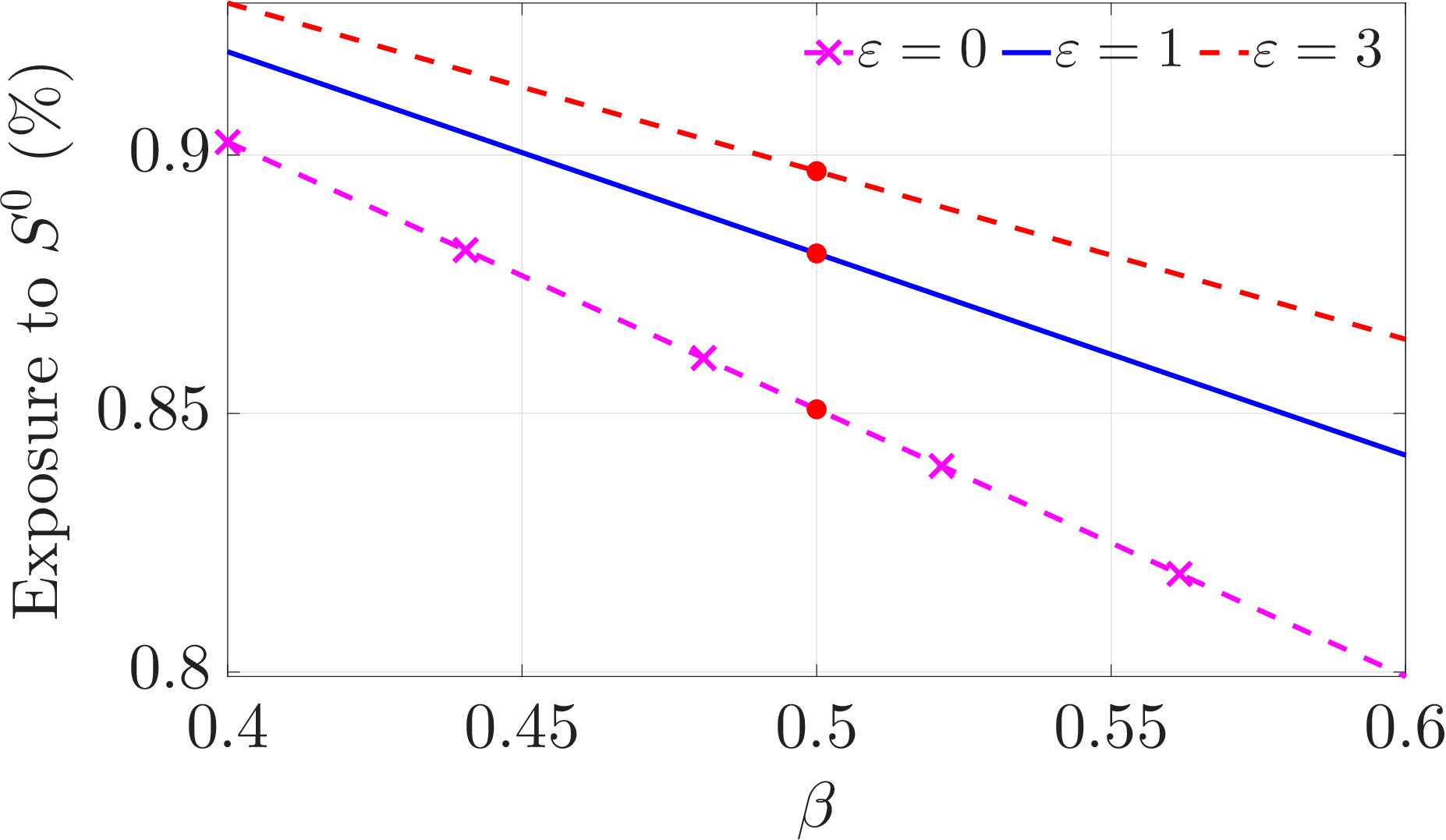}
\caption{Optimal multiplier $\bar{m}^\star_0$ (left panels) and optimal exposure to the risk-free asset $S^0$ (right panels) at $t=0$ as a function of market parameters for $\delta=1$ and different levels of carbon aversion $\varepsilon$. The red dots denote the optimal multiplier and the exposure to $S^0$ under the baseline parameter configuration reported in \ref{tab:model_params}.}\label{fig:opt_mult_comp_risky_reference_port_t_0}
\end{figure}
\begin{figure}[H]
\centering
\includegraphics[width=0.32\linewidth]{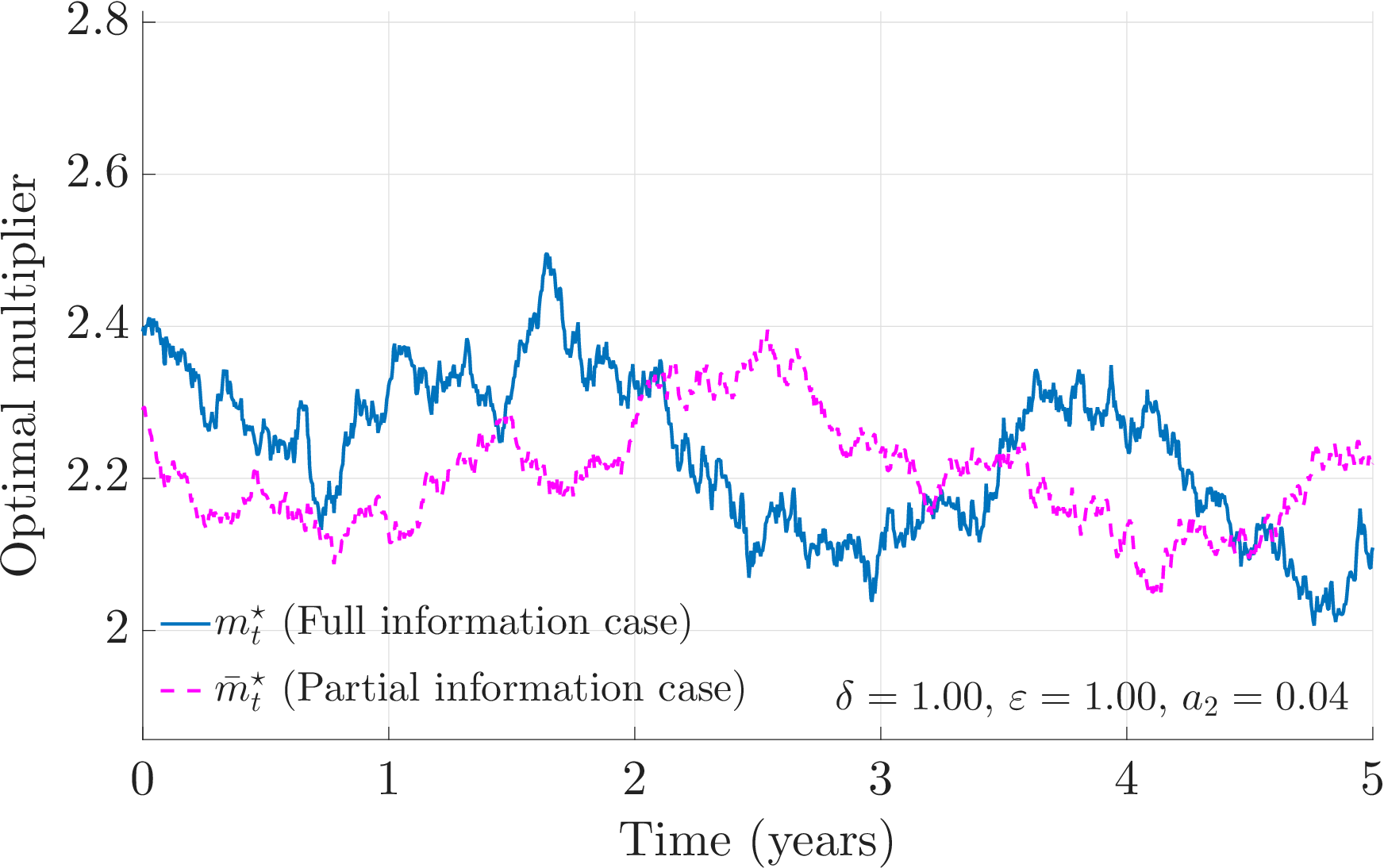} 
\vspace{.3cm}
\hfill 
\includegraphics[width=0.32\linewidth]{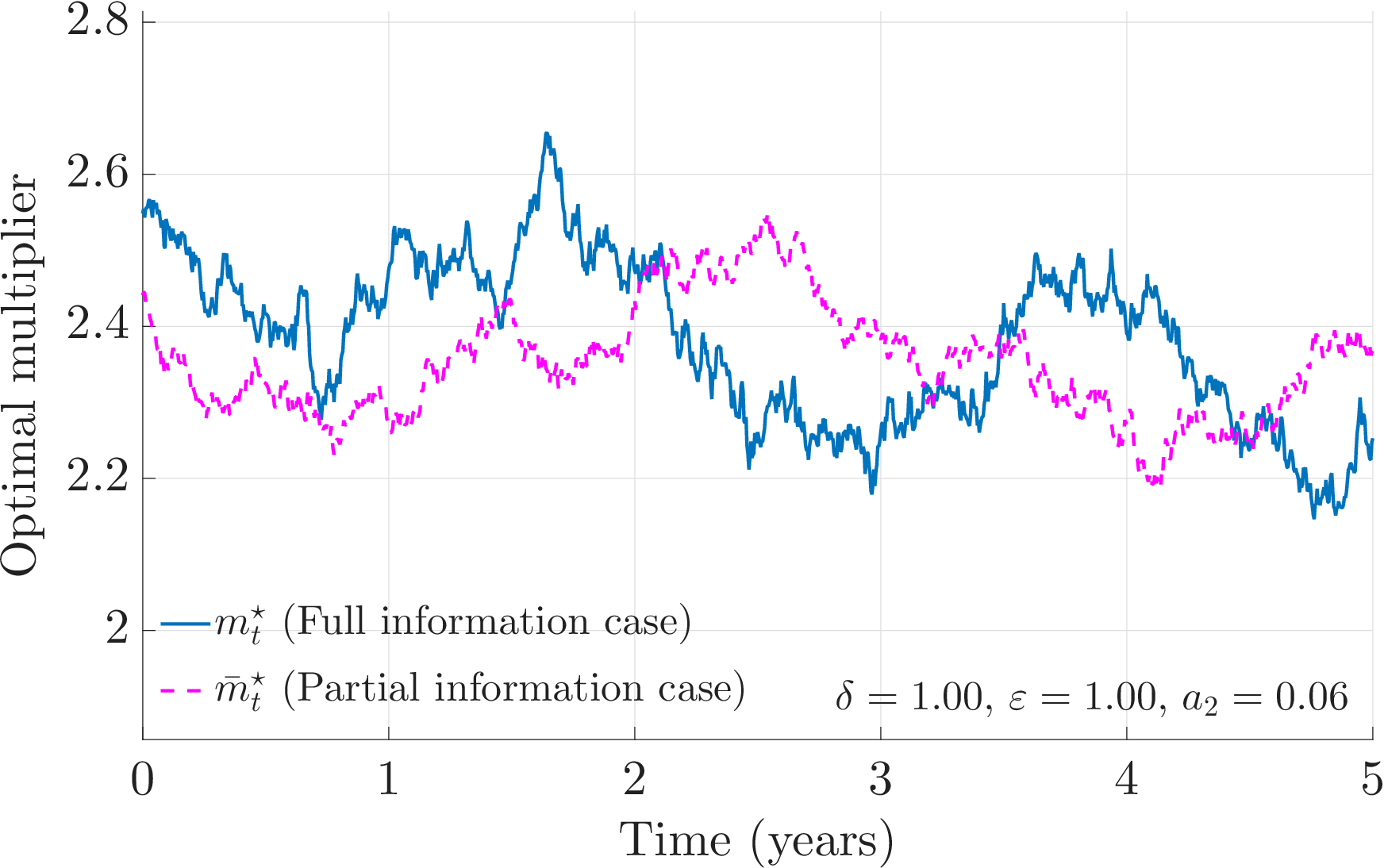}
\vspace{.3cm}
\hfill
\includegraphics[width=0.32\linewidth]{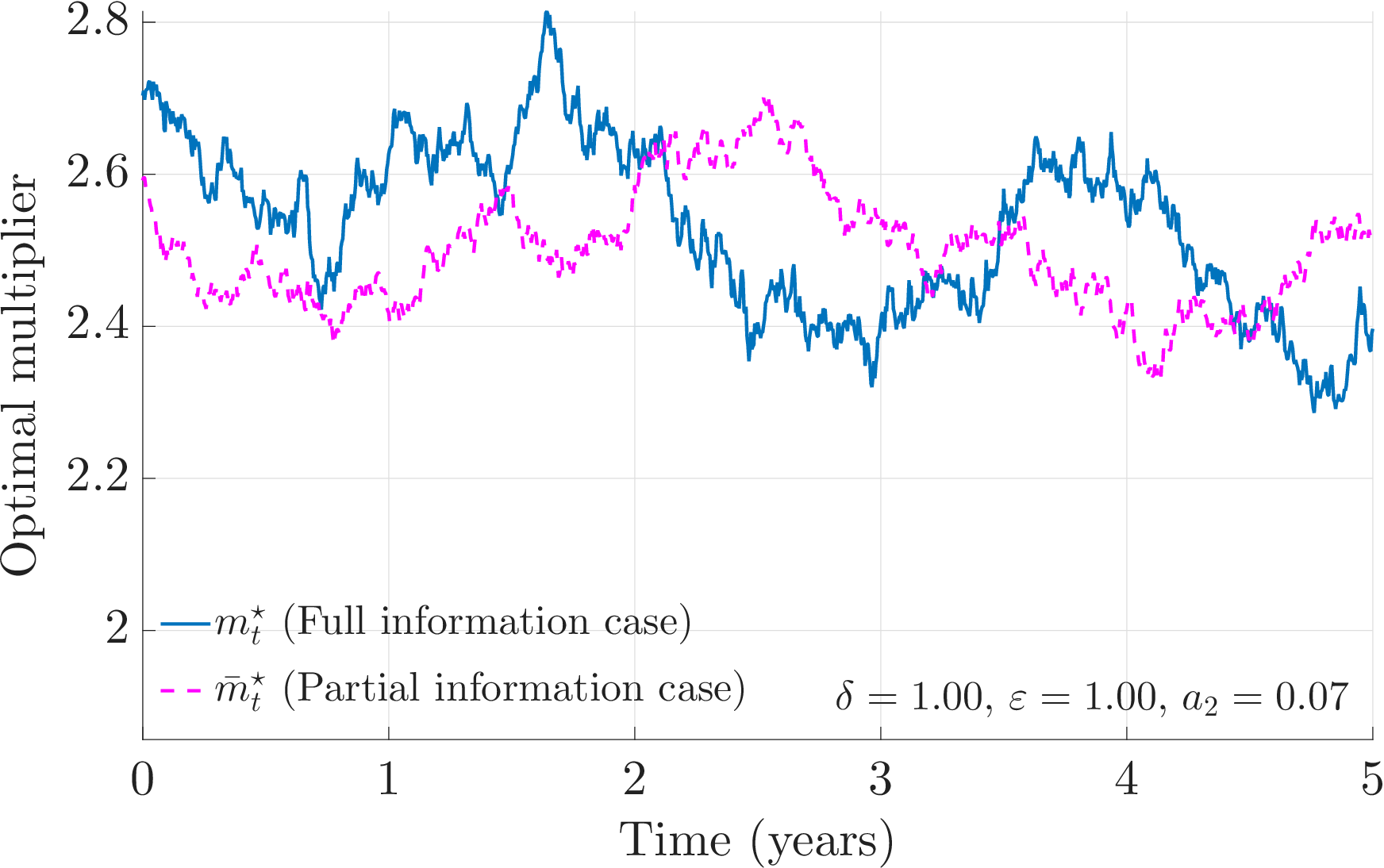}

\includegraphics[width=0.32\linewidth]{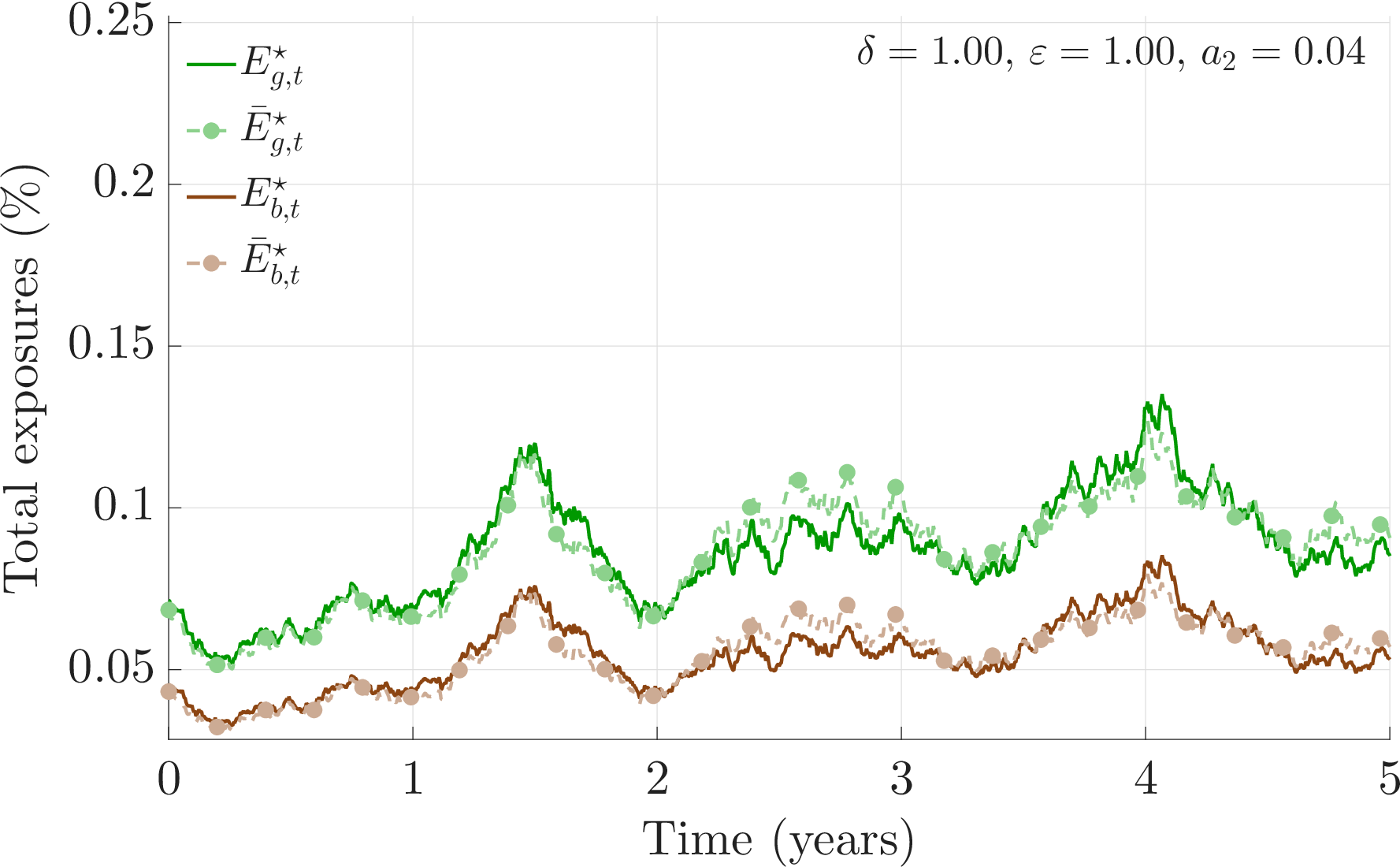} 
\vspace{.3cm}
\hfill 
\includegraphics[width=0.32\linewidth]{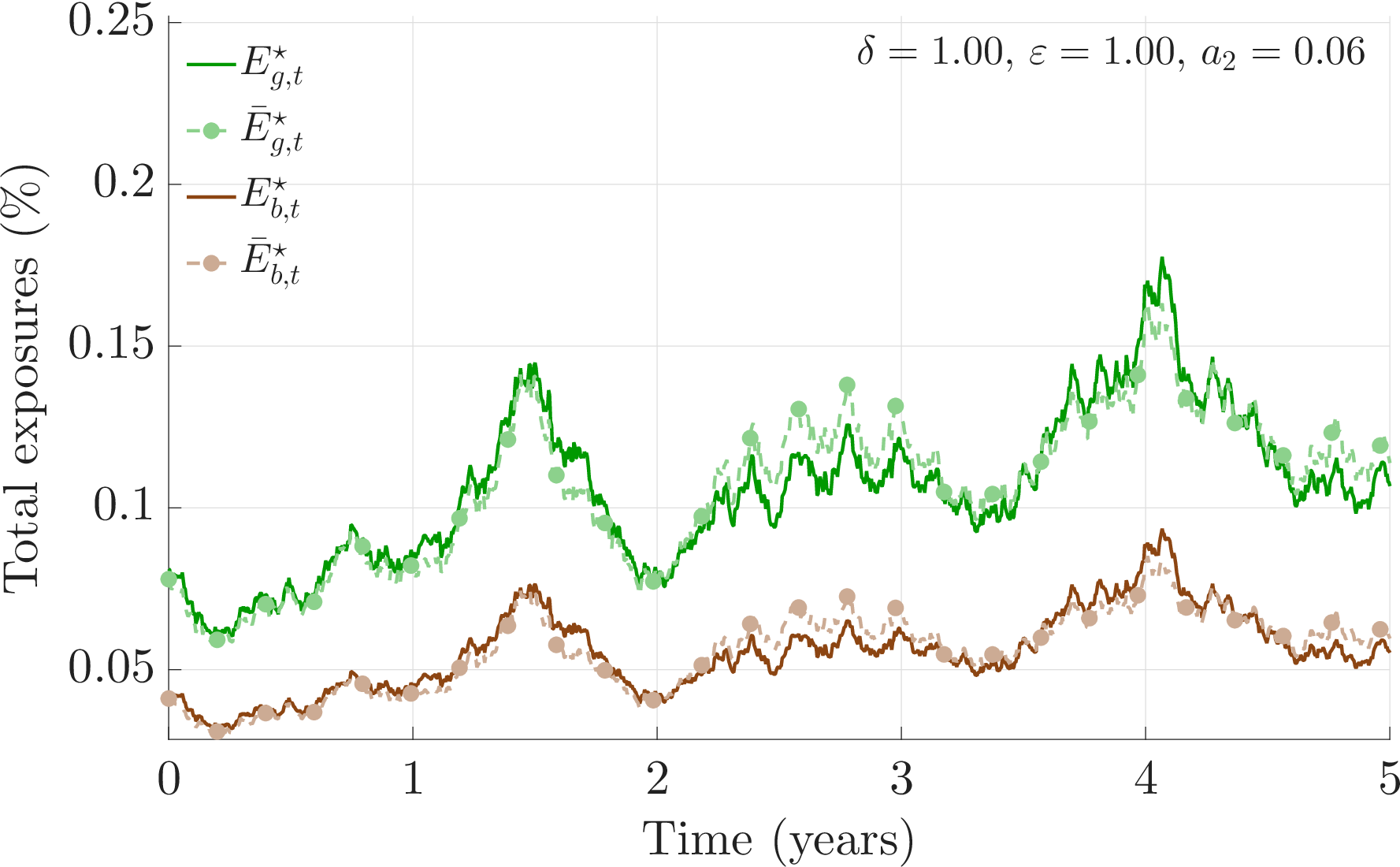}
\vspace{.3cm}
\hfill
\includegraphics[width=0.32\linewidth]{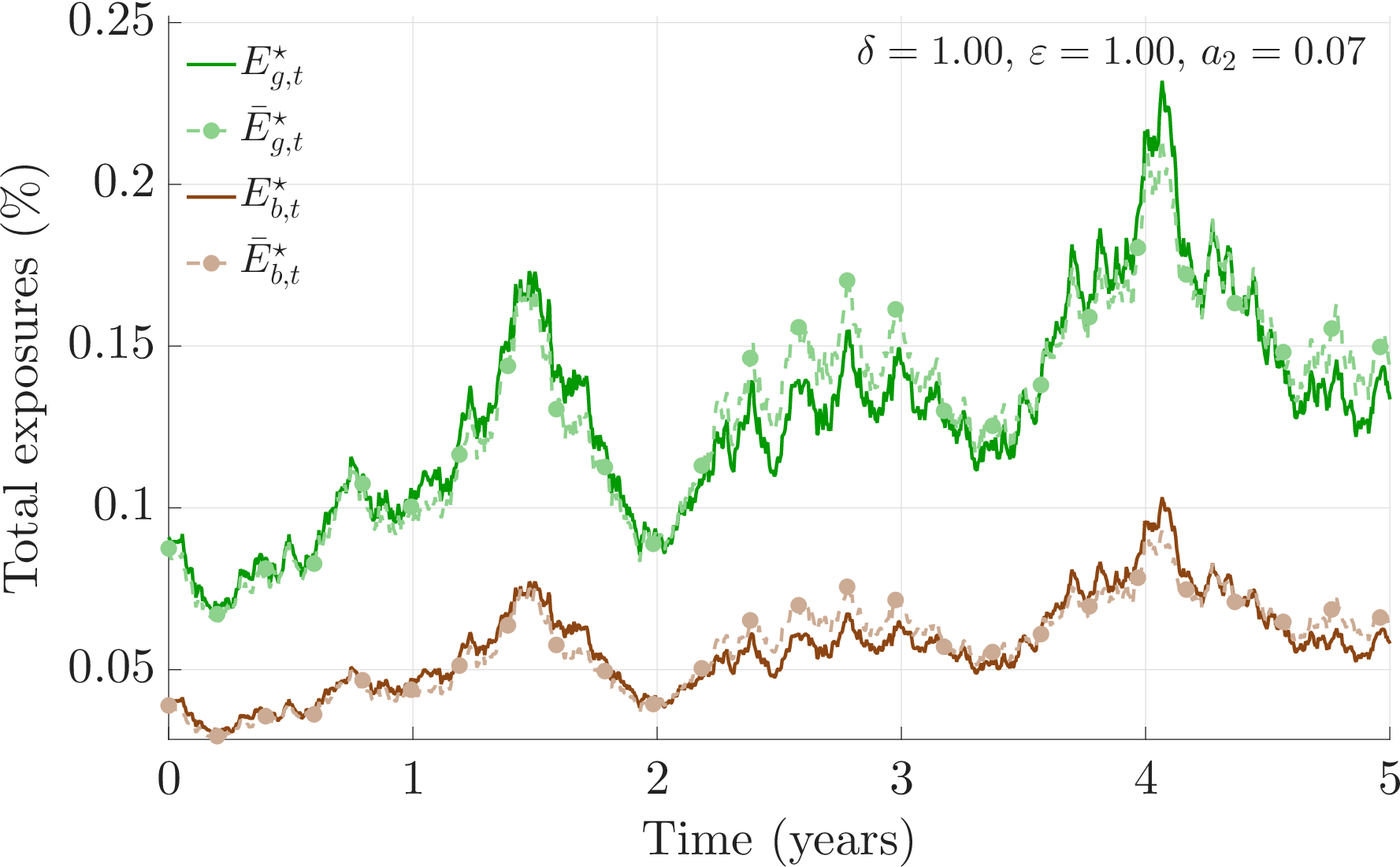}

\caption{Trajectories of the optimal multiplier under full and partial information (top panels) and of the corresponding total optimal exposure to green and brown stocks (bottom panels), for $\delta=1$ and $\varepsilon=1$. Each column corresponds to a different value of $a_2$, with the central column corresponding to the baseline configuration in Table \ref{tab:model_params}. In the bottom panels, the solid green (resp. brown) line represents the total optimal exposure to green (resp. brown) stocks under full information, while the dotted light green (resp. light brown) line represents the corresponding exposure under partial information.}
\label{fig:DYNAMIC_MULTIPLIER_A_2}
\end{figure}
\begin{figure}[H]
\centering
\includegraphics[width=0.32\linewidth]{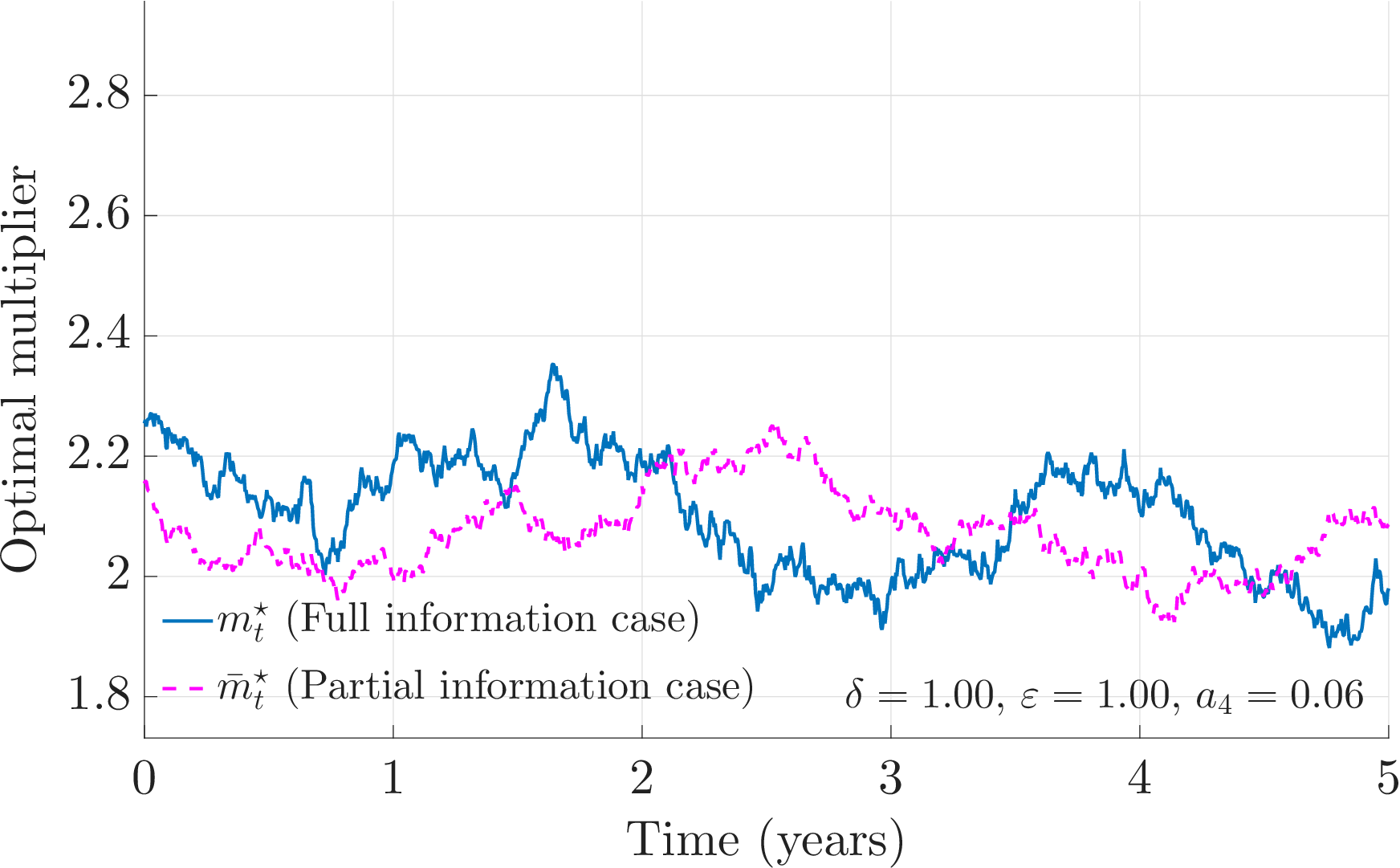} 
\vspace{.3cm}
\hfill 
\includegraphics[width=0.32\linewidth]{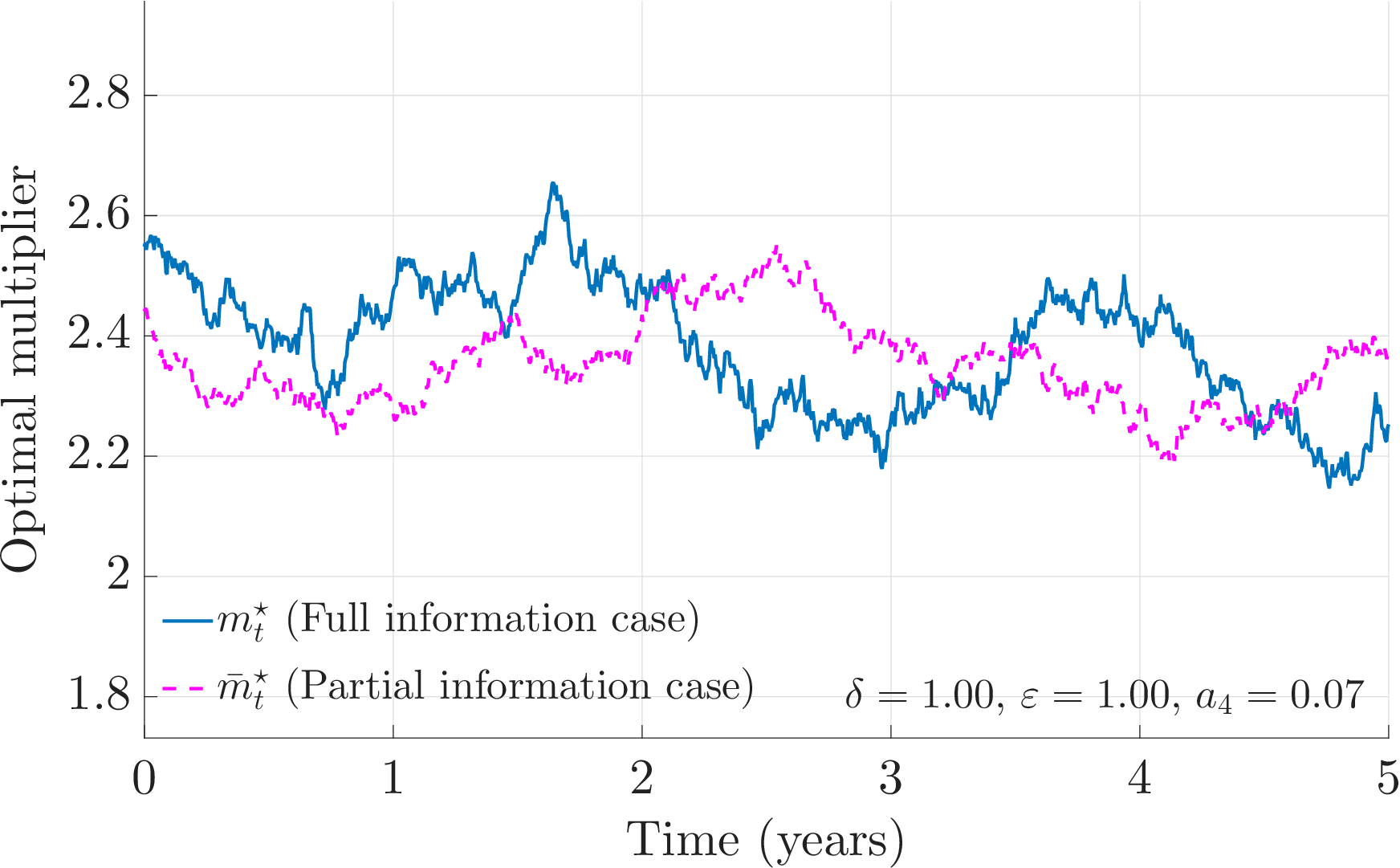}
\vspace{.3cm}
\hfill
\includegraphics[width=0.32\linewidth]{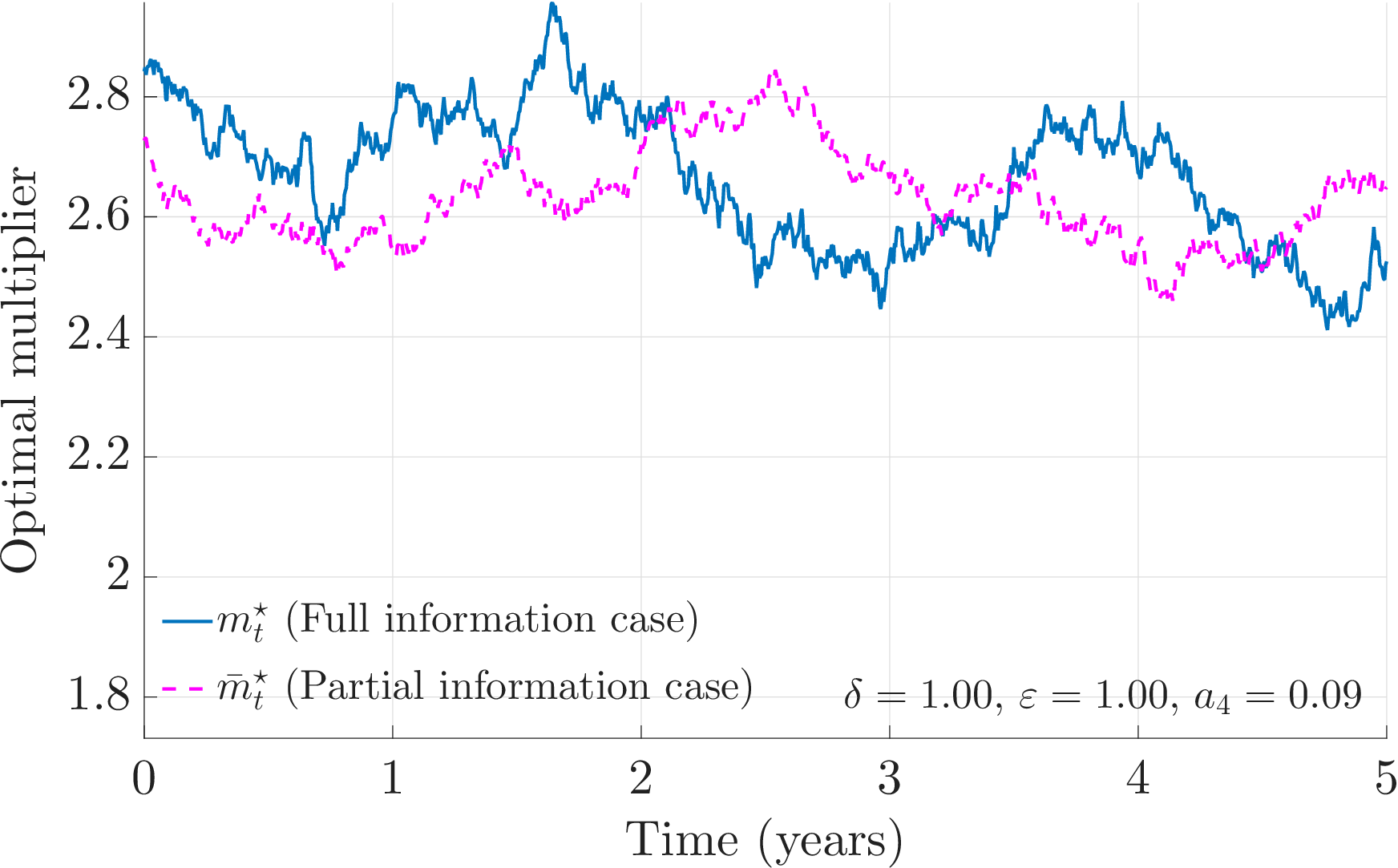}

\includegraphics[width=0.32\linewidth]{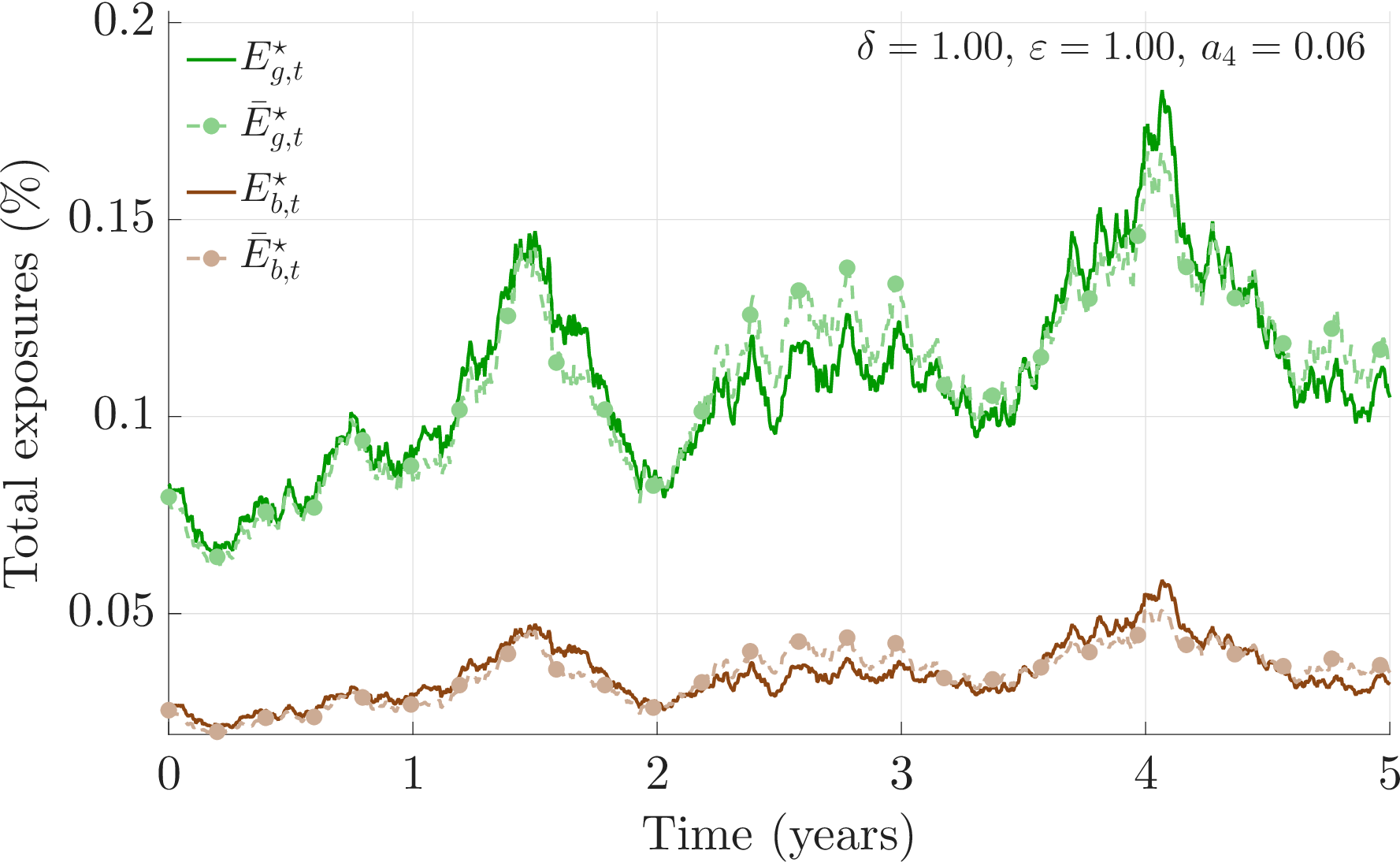} 
\vspace{.3cm}
\hfill 
\includegraphics[width=0.32\linewidth]{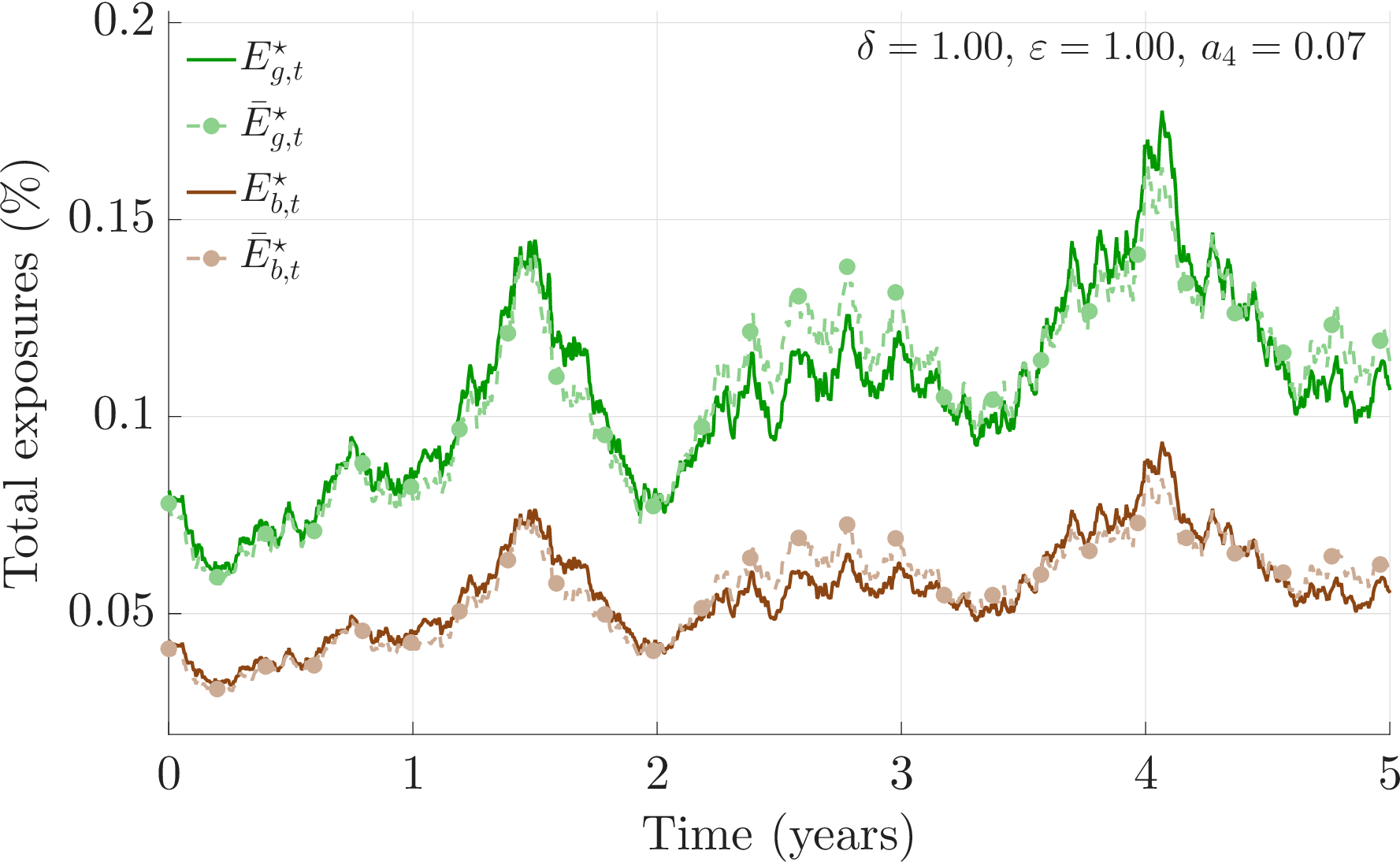}
\vspace{.3cm}
\hfill
\includegraphics[width=0.32\linewidth]{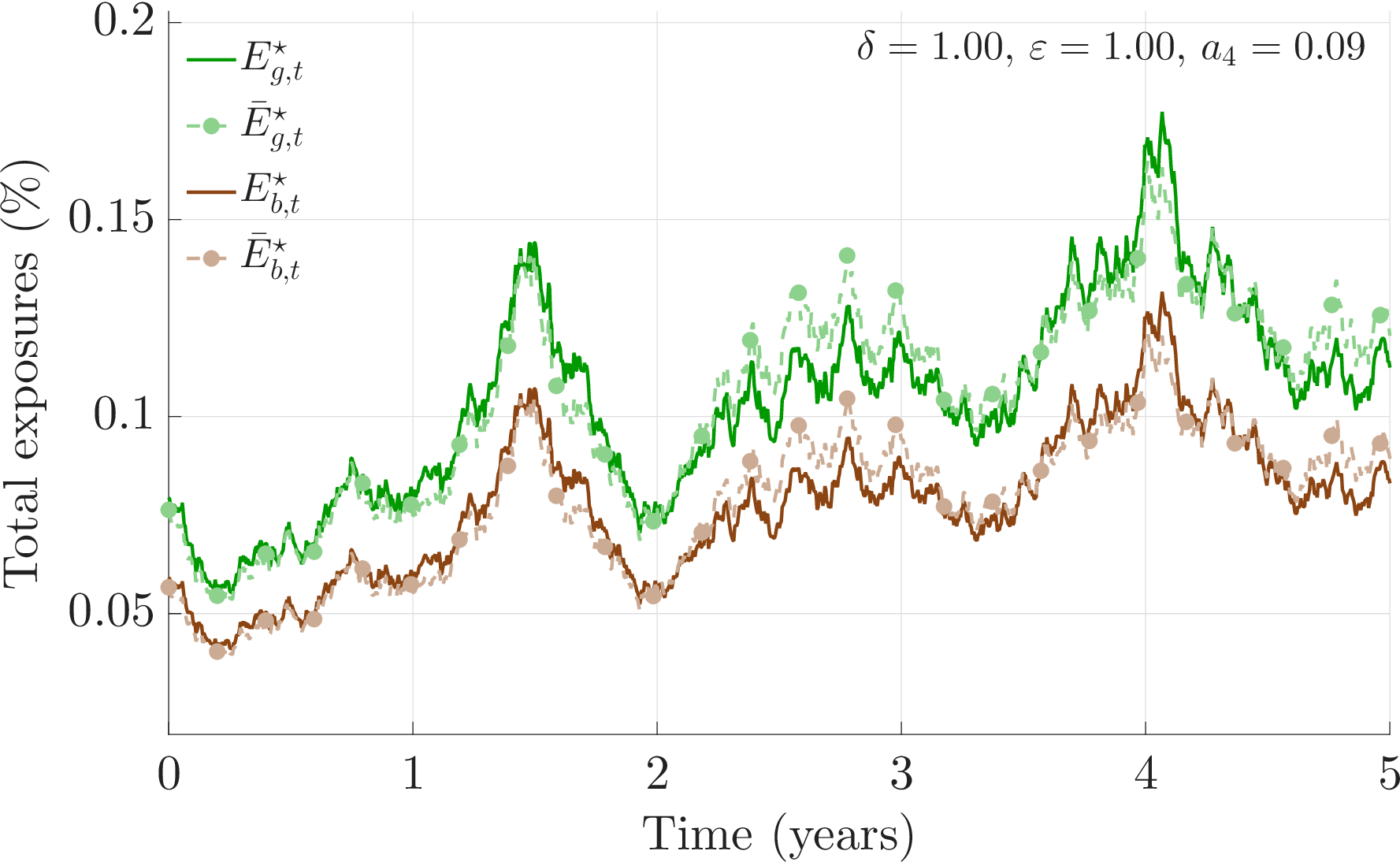}
\caption{Trajectories of the optimal multiplier under full and partial information (top panels) and of the corresponding total optimal exposure to green and brown stocks (bottom panels), for $\delta=1$ and $\varepsilon=1$. Each column corresponds to a different value of $a_4$, with the central column corresponding to the baseline configuration in Table \ref{tab:model_params}. In the bottom panels, the solid green (resp. brown) line represents the total optimal exposure to green (resp. brown) stocks under full information, while the dotted light green (resp. light brown) line represents the corresponding exposure under partial information.}\label{fig:DYNAMIC_EXPOSURES_A_4}
\end{figure}
\begin{figure}[H]
\centering
\includegraphics[width=0.32\linewidth]{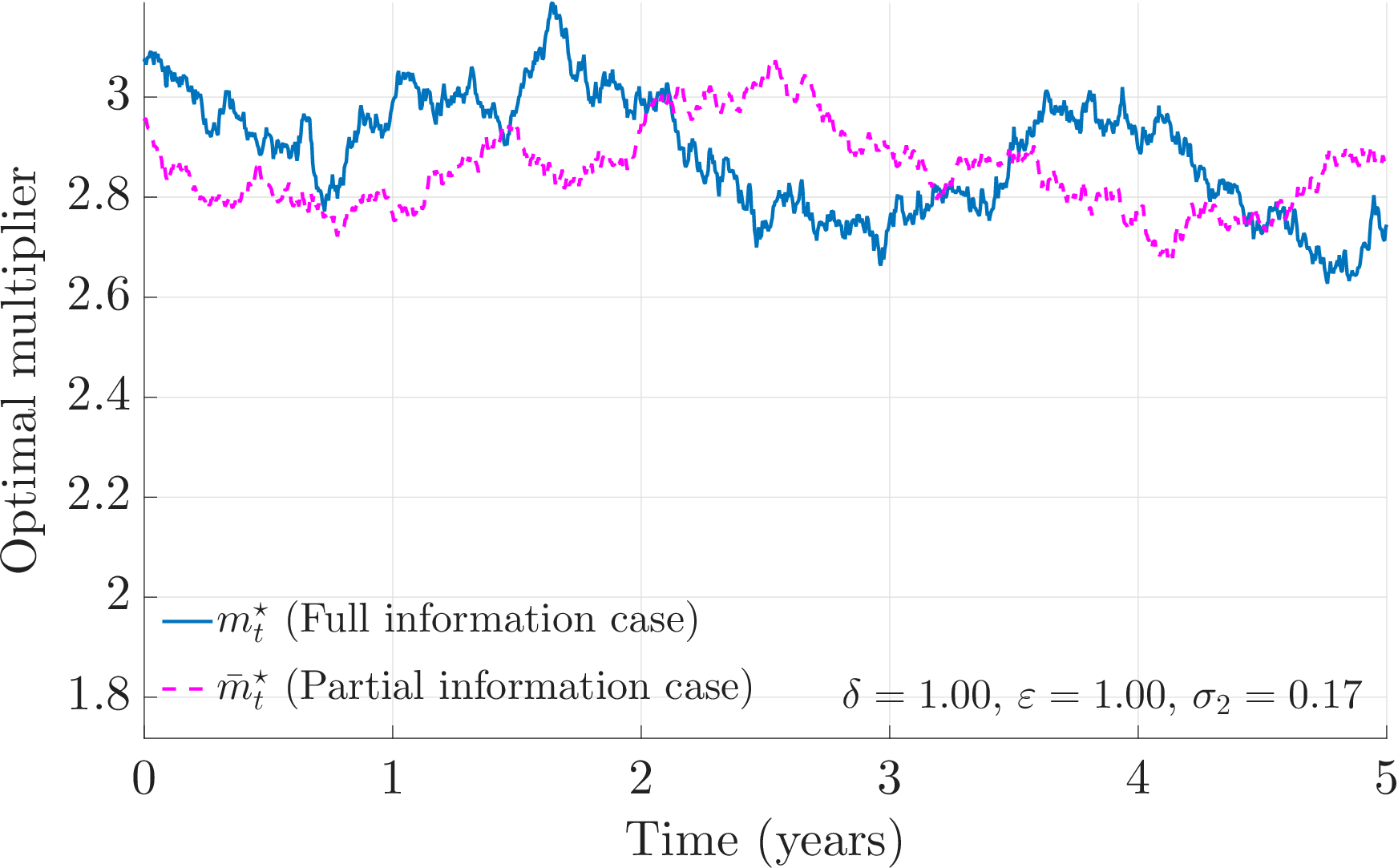} 
\vspace{.3cm}
\hfill 
\includegraphics[width=0.32\linewidth]{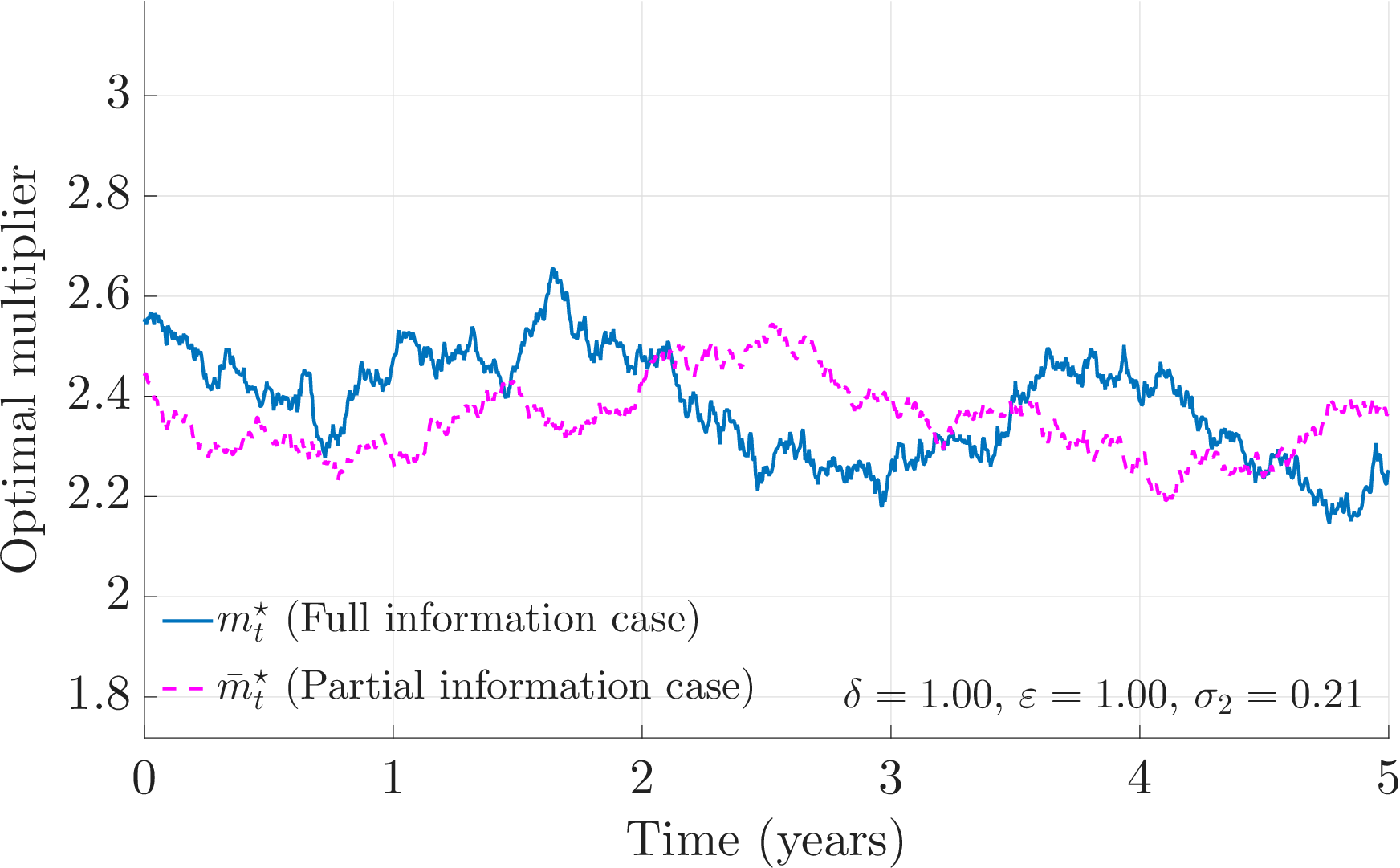}
\vspace{.3cm}
\hfill
\includegraphics[width=0.32\linewidth]{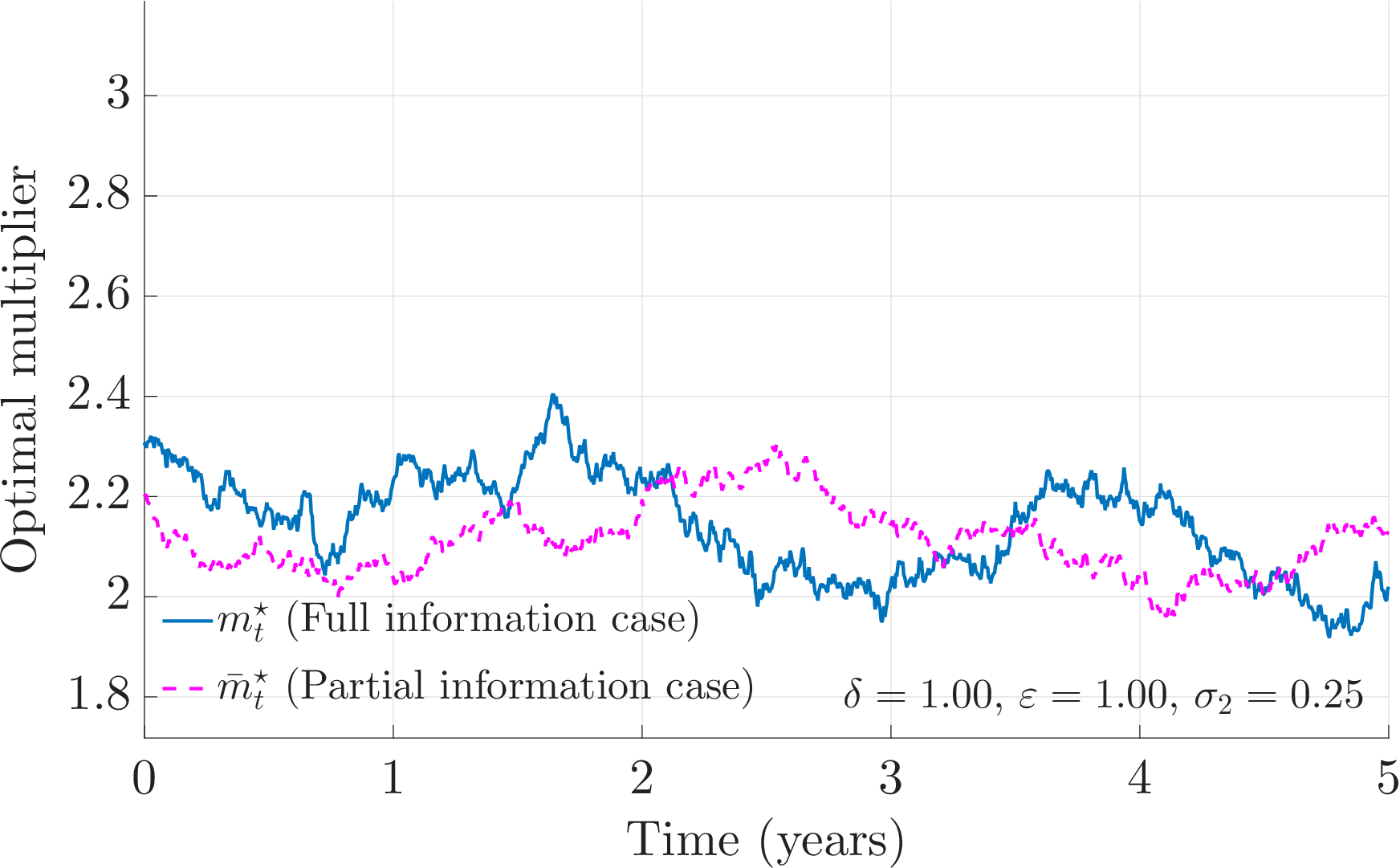}

\includegraphics[width=0.32\linewidth]{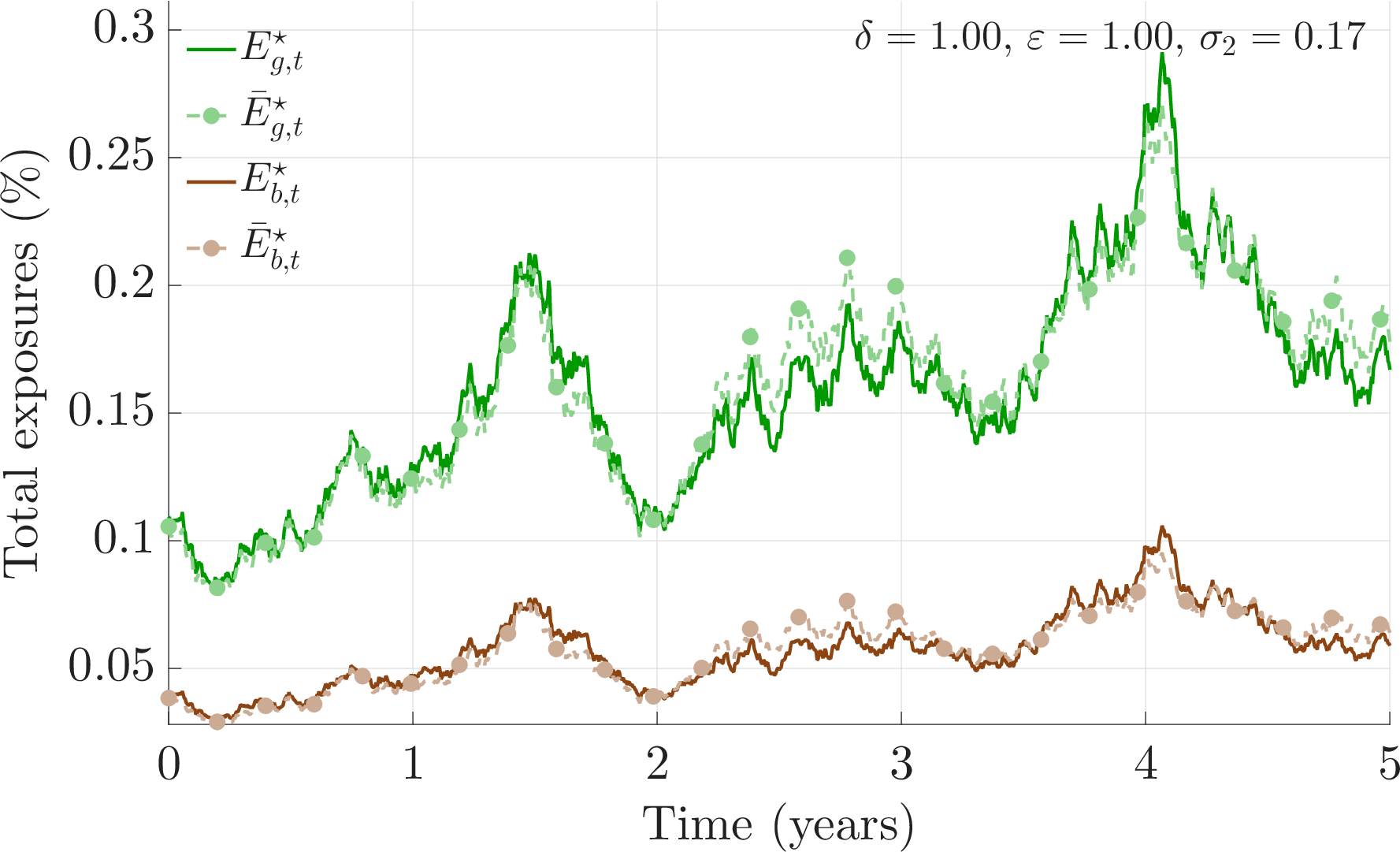} 
\vspace{.3cm}
\hfill 
\includegraphics[width=0.32\linewidth]{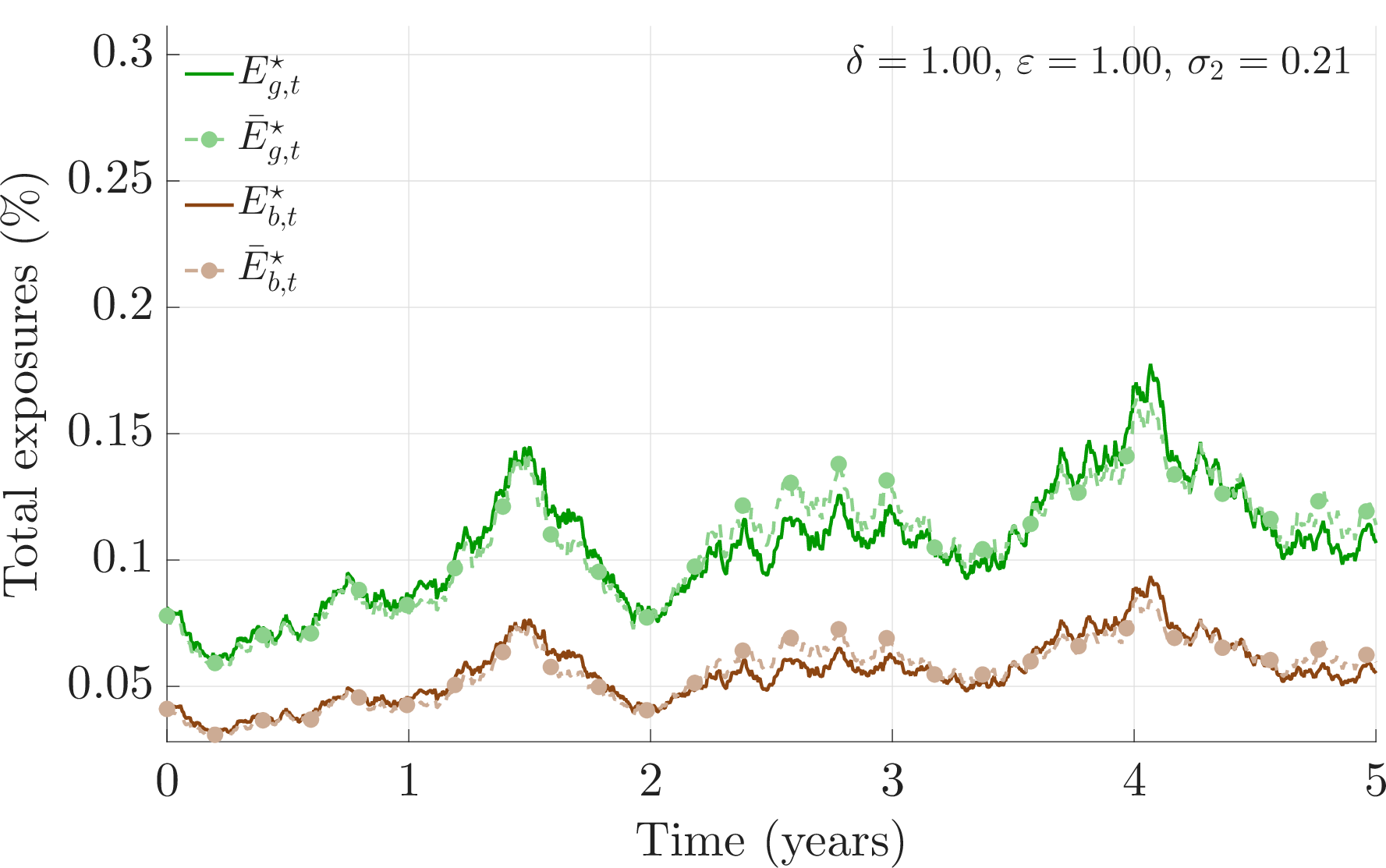}
\vspace{.3cm}
\hfill
\includegraphics[width=0.32\linewidth]{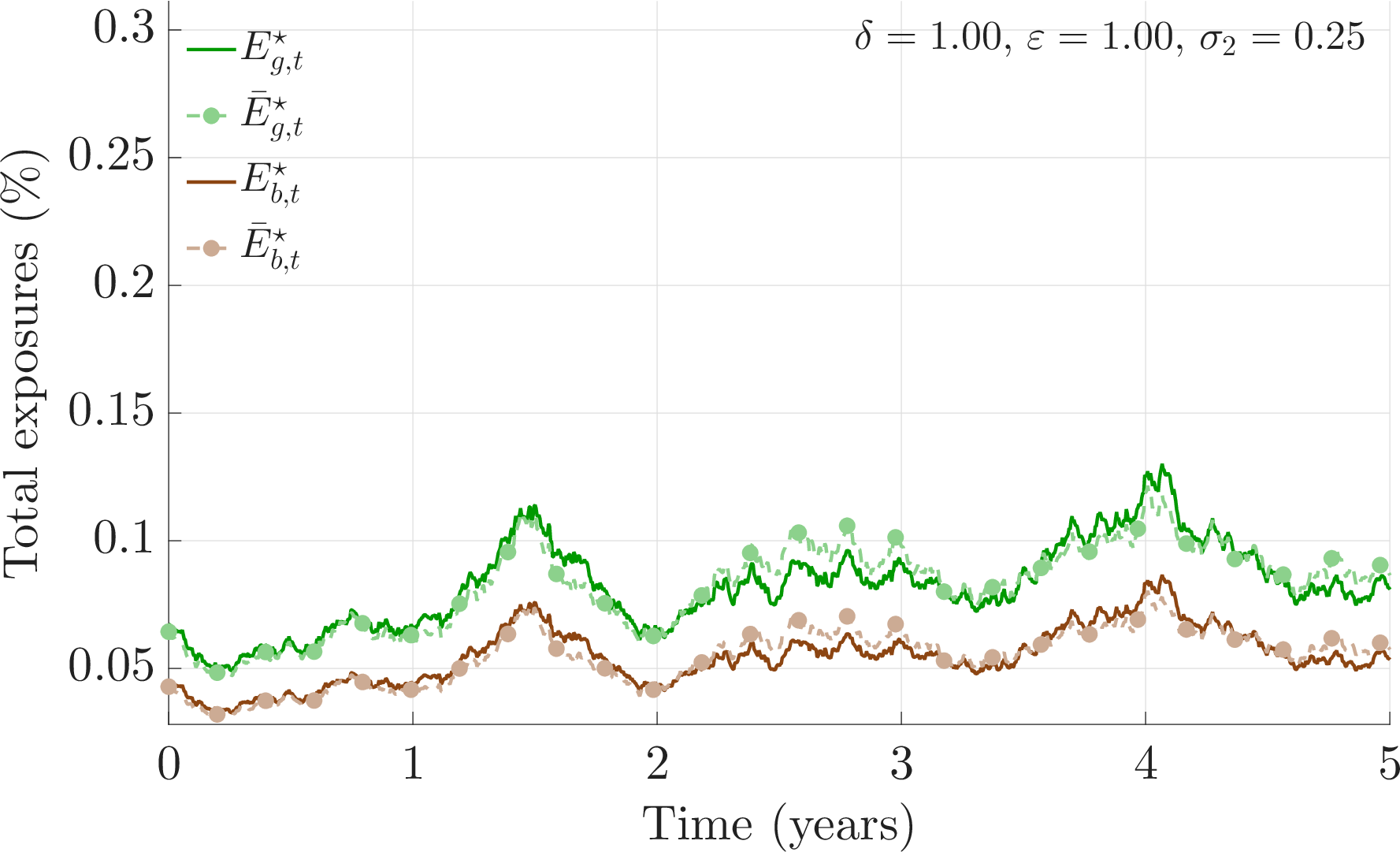}
\caption{Trajectories of the optimal multiplier under full and partial information (top panels) and of the corresponding total optimal exposure to green and brown stocks (bottom panels), for $\delta=1$ and $\varepsilon=1$. Each column corresponds to a different value of $\sigma_2$, with the central column corresponding to the baseline configuration in Table \ref{tab:model_params}. In the bottom panels, the solid green (resp. brown) line represents the total optimal exposure to green (resp. brown) stocks under full information, while the dotted light green (resp. light brown) line represents the corresponding exposure under partial information.
}\label{fig:DYNAMIC_EXPOSURES_sigma_2}
\end{figure}
\begin{figure}[H]
\centering
\includegraphics[width=0.32\linewidth]{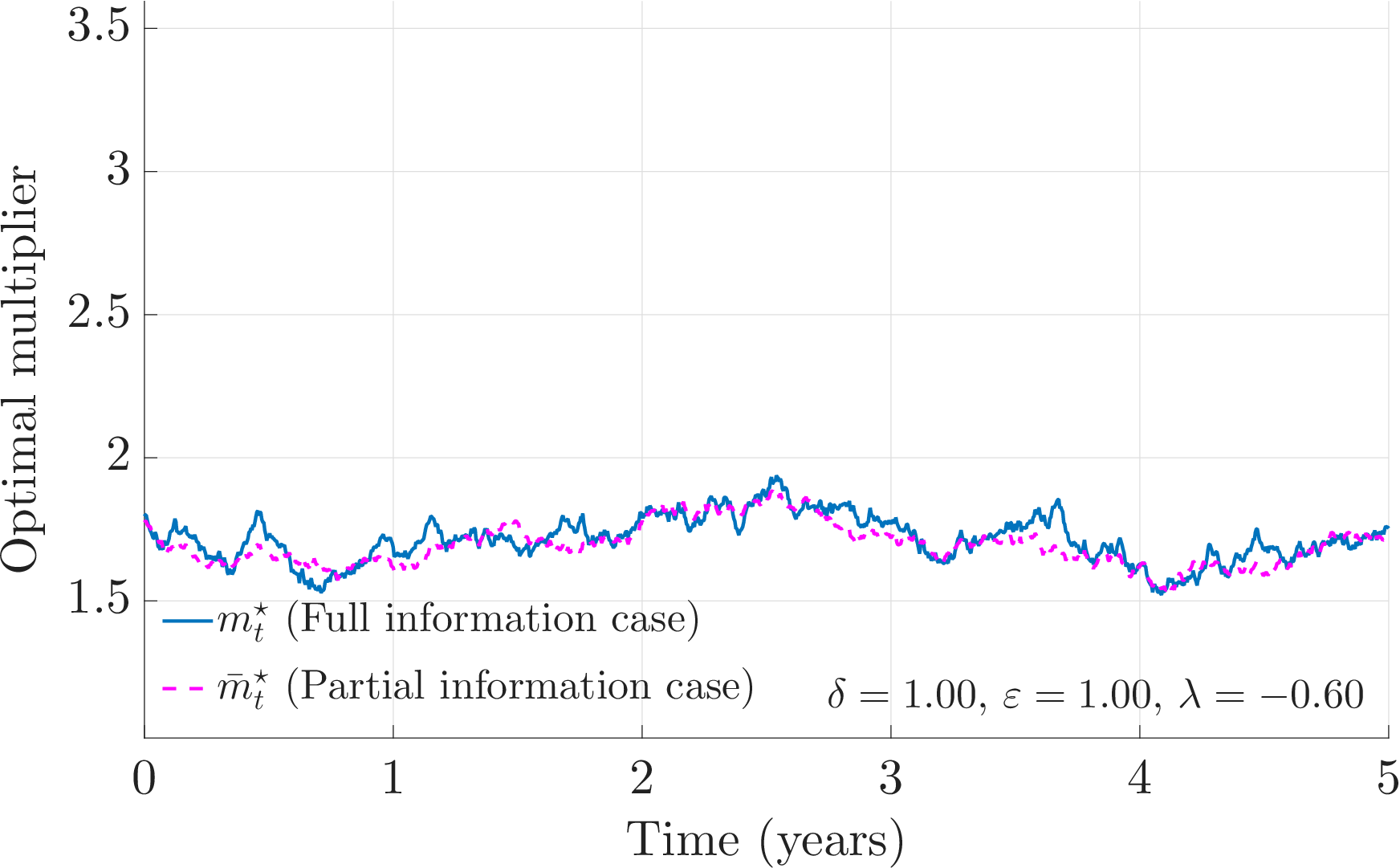} 
\vspace{.3cm}
\hfill 
\includegraphics[width=0.32\linewidth]{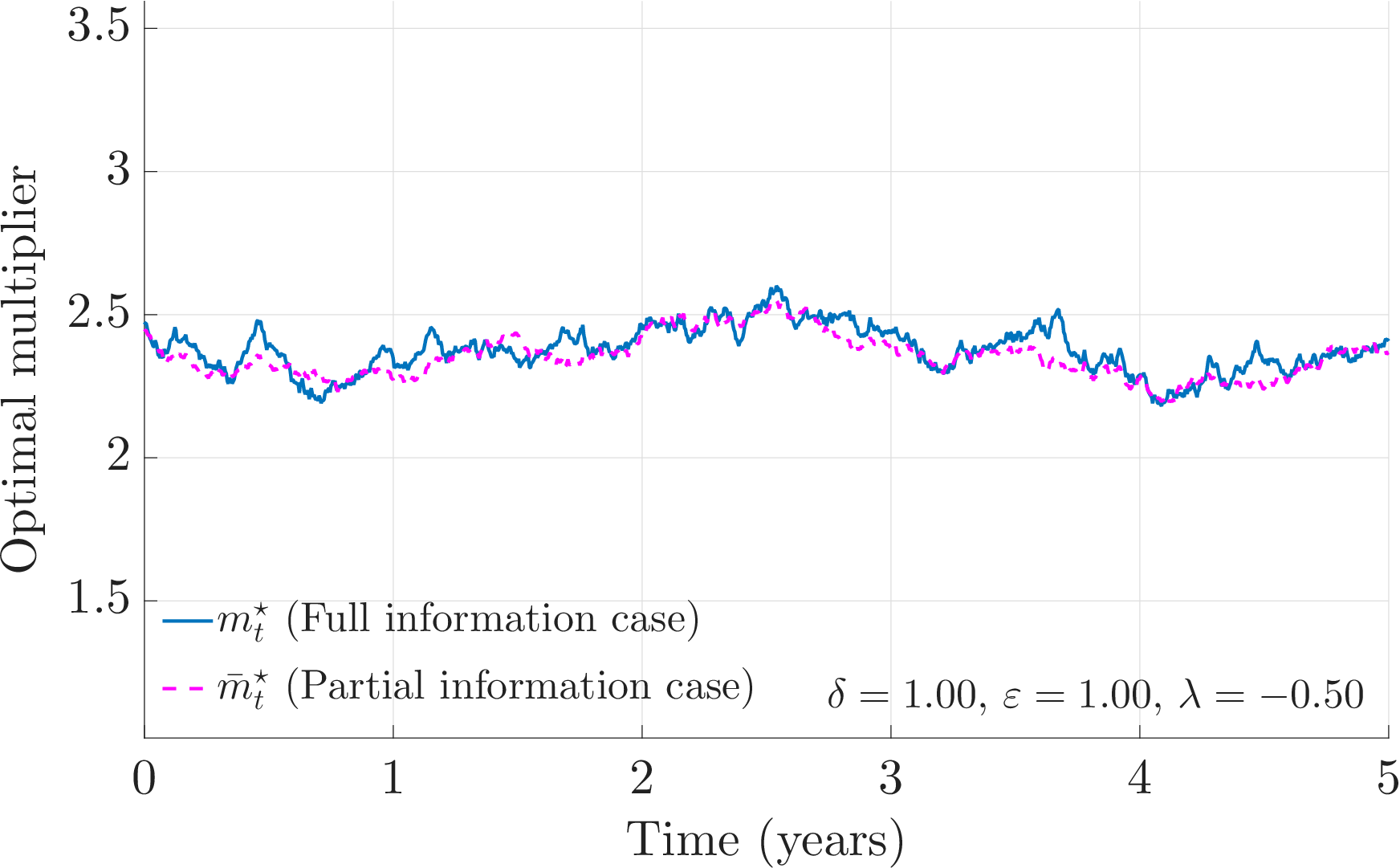}
\vspace{.3cm}
\hfill
\includegraphics[width=0.32\linewidth]{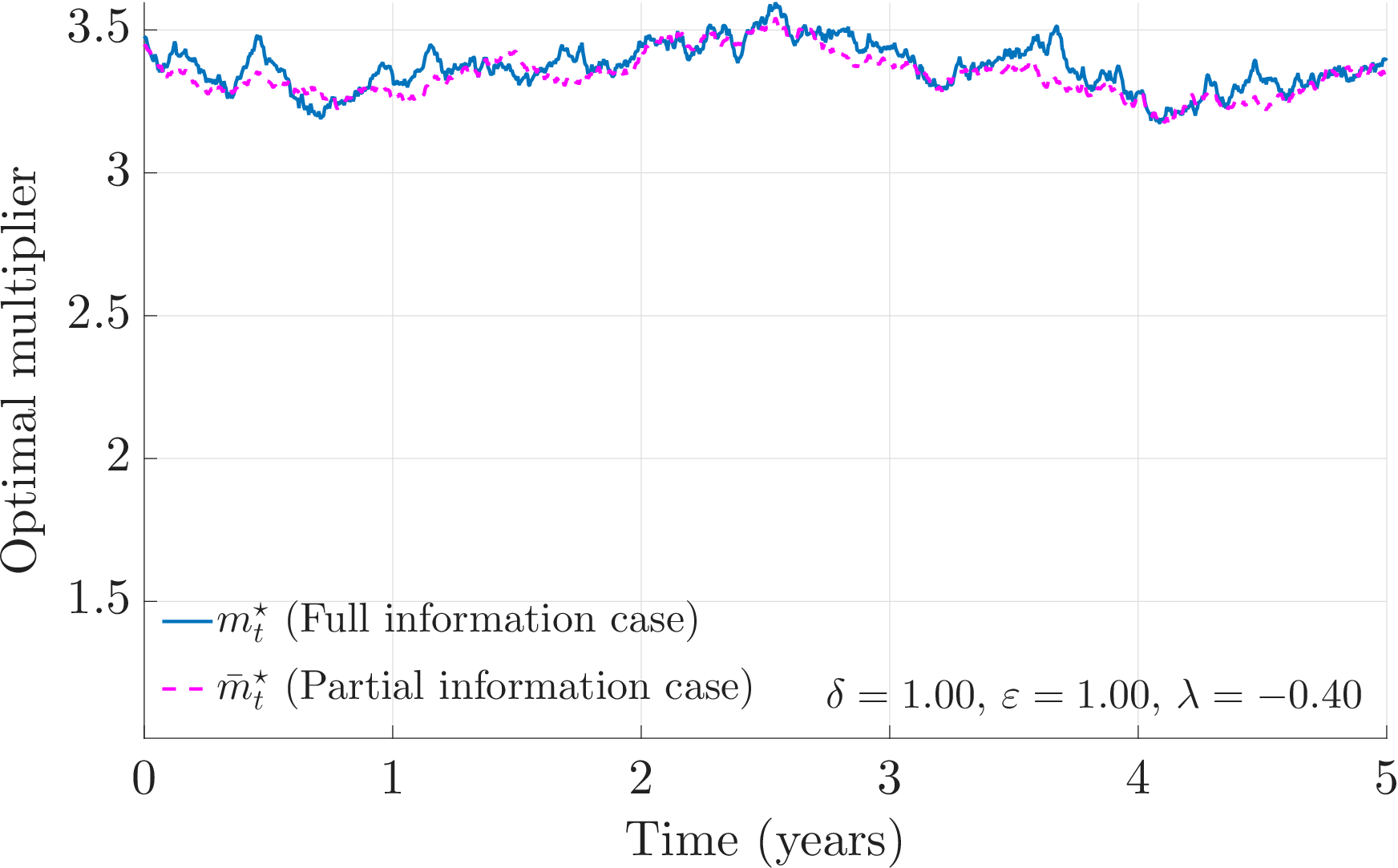}

\includegraphics[width=0.32\linewidth]{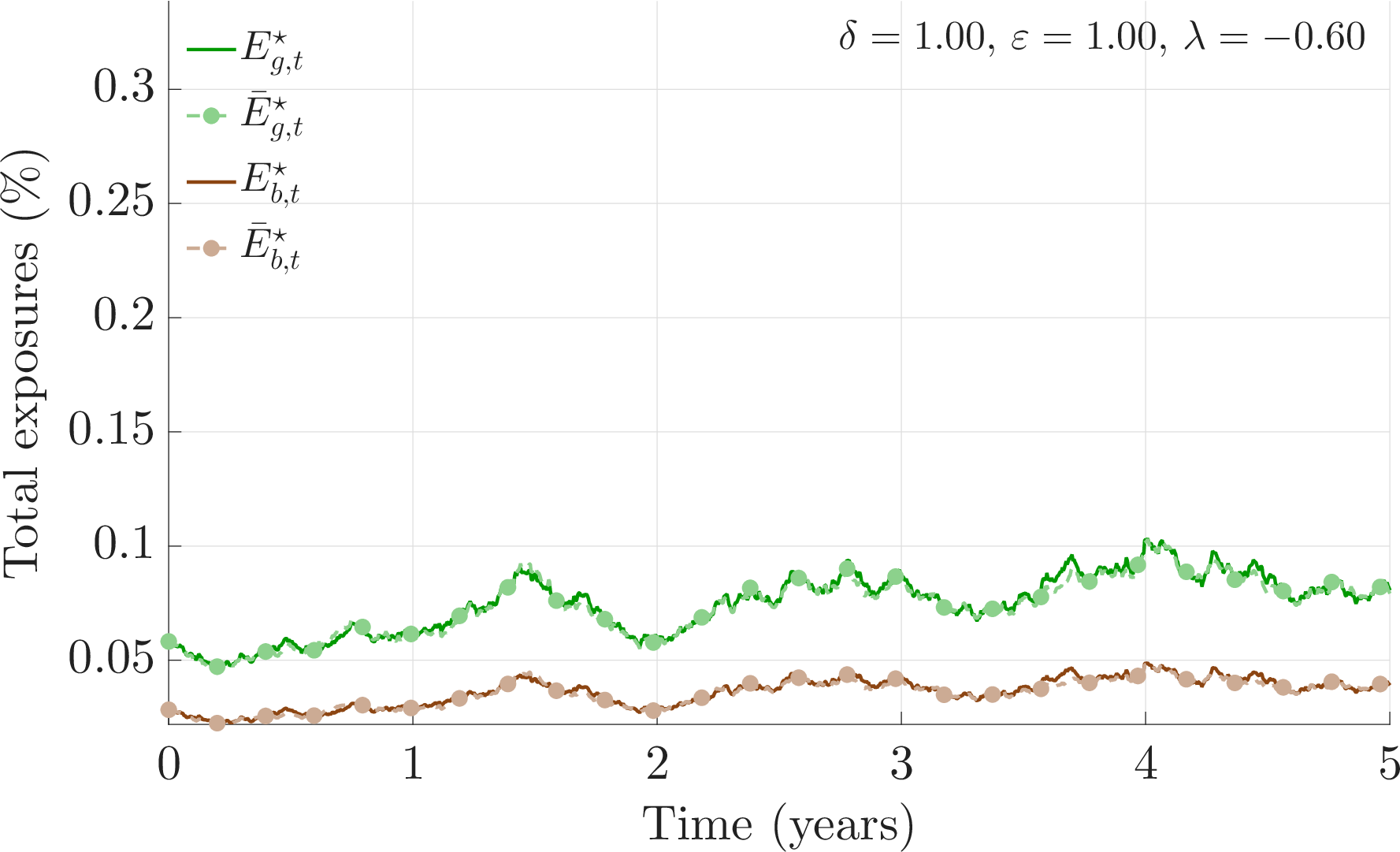} 
\vspace{.3cm}
\hfill 
\includegraphics[width=0.32\linewidth]{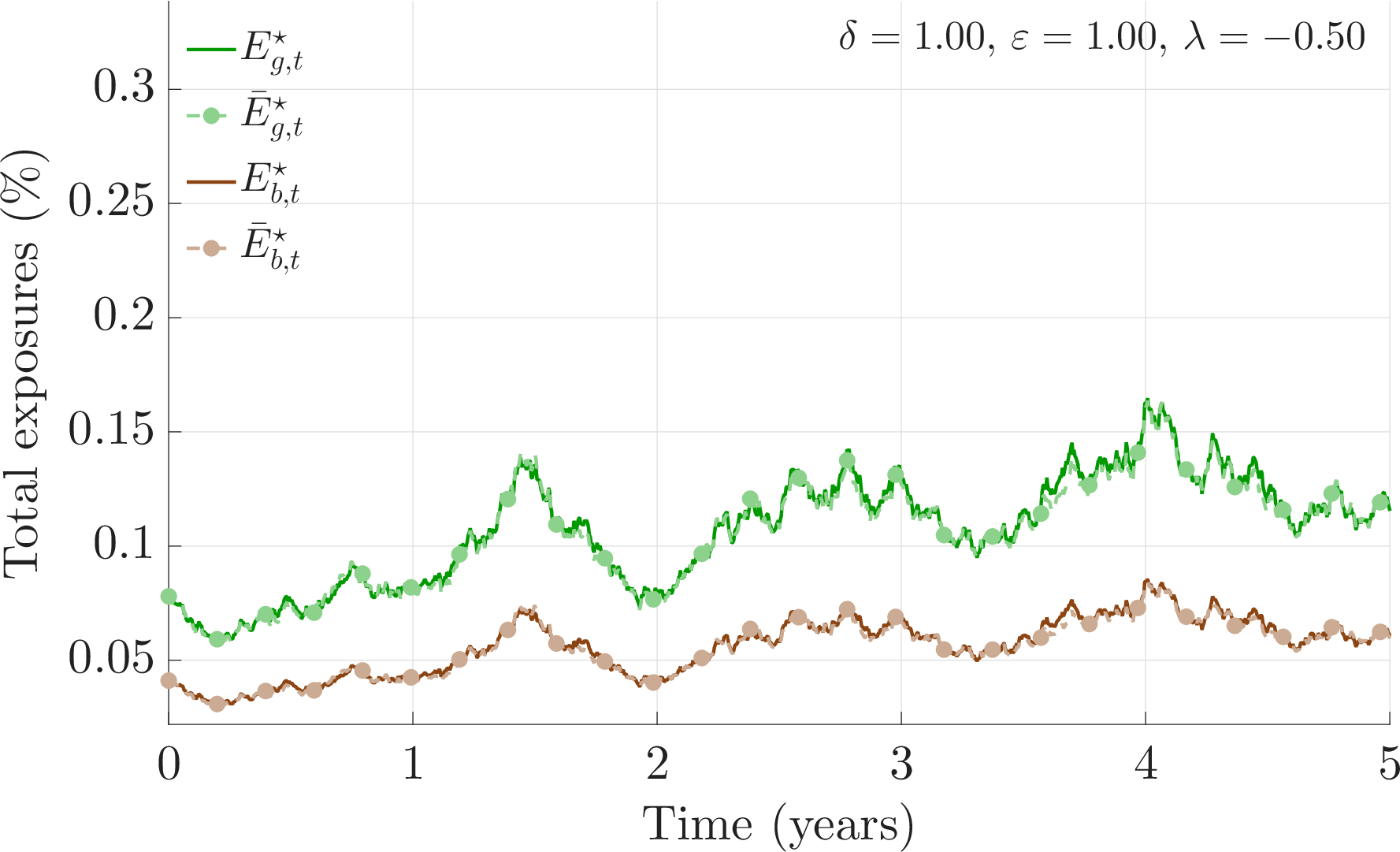}
\vspace{.3cm}
\hfill
\includegraphics[width=0.32\linewidth]{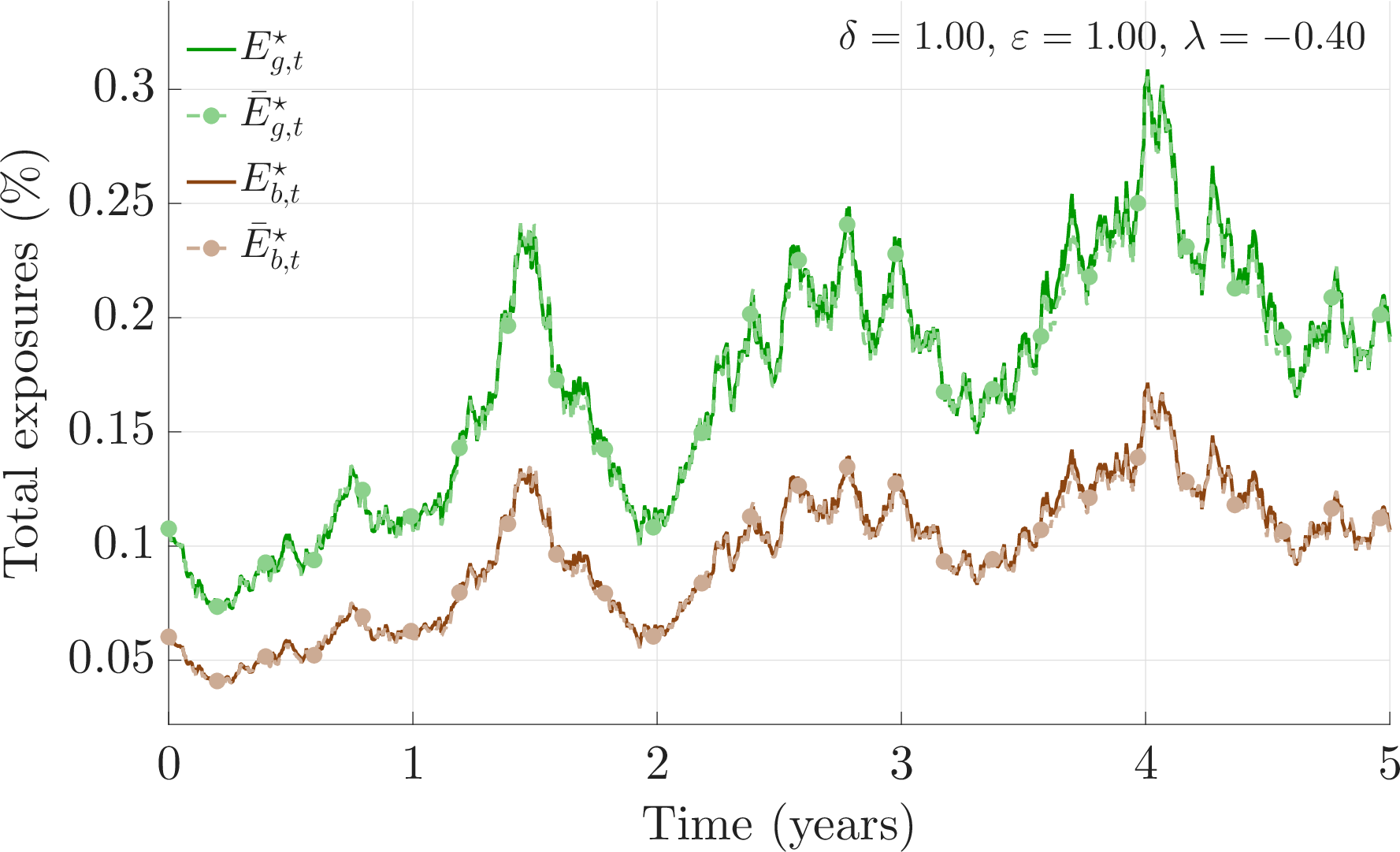}
\caption{Trajectories of the optimal multiplier under full and partial information (top panels) and of the corresponding total optimal exposure to green and brown stocks (bottom panels), for $\delta=1$ and $\varepsilon=1$. Each column corresponds to a different value of $\lambda$, with the central column corresponding to the baseline configuration in Table \ref{tab:model_params}. In the bottom panels, the solid green (resp. brown) line represents the total optimal exposure to green (resp. brown) stocks under full information, while the dotted light green (resp. light brown) line represents the corresponding exposure under partial information.}\label{fig:DYNAMIC_EXPOSURES_LAMBDA}
\end{figure}
\begin{figure}[H]
\centering
\includegraphics[width=0.32\linewidth]{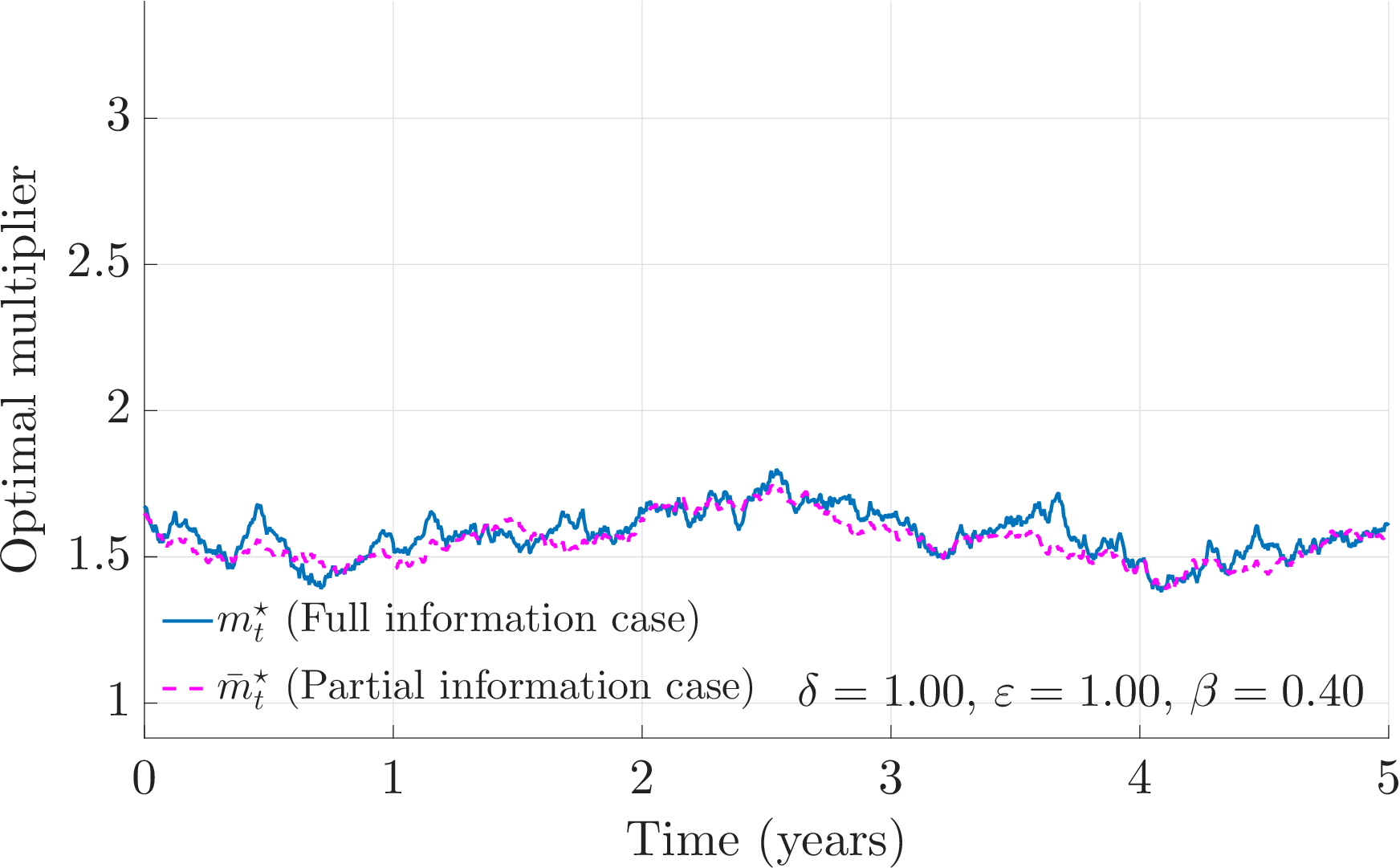} 
\vspace{.3cm}
\hfill 
\includegraphics[width=0.32\linewidth]{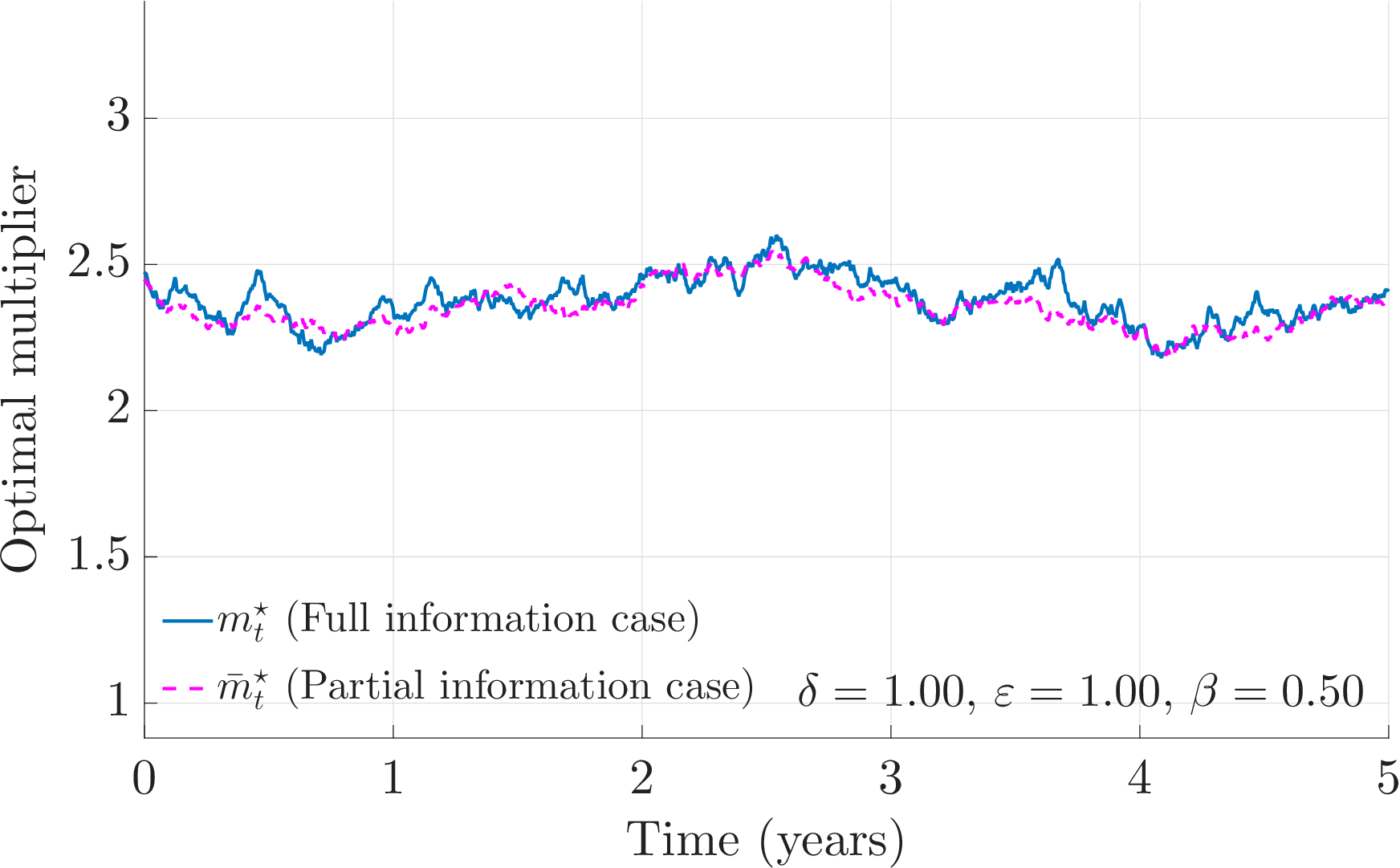}
\vspace{.3cm}
\hfill
\includegraphics[width=0.32\linewidth]{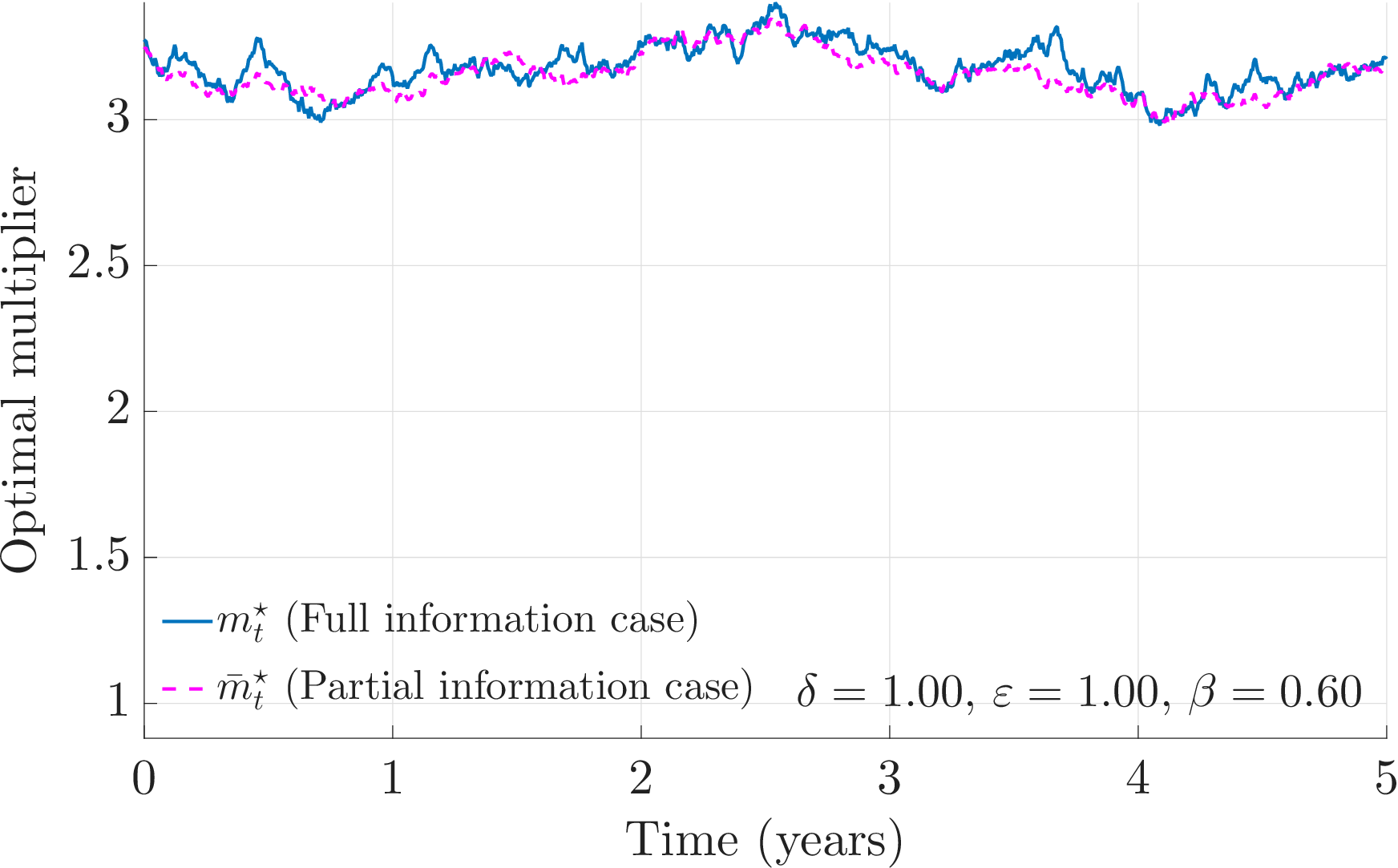}

\includegraphics[width=0.32\linewidth]{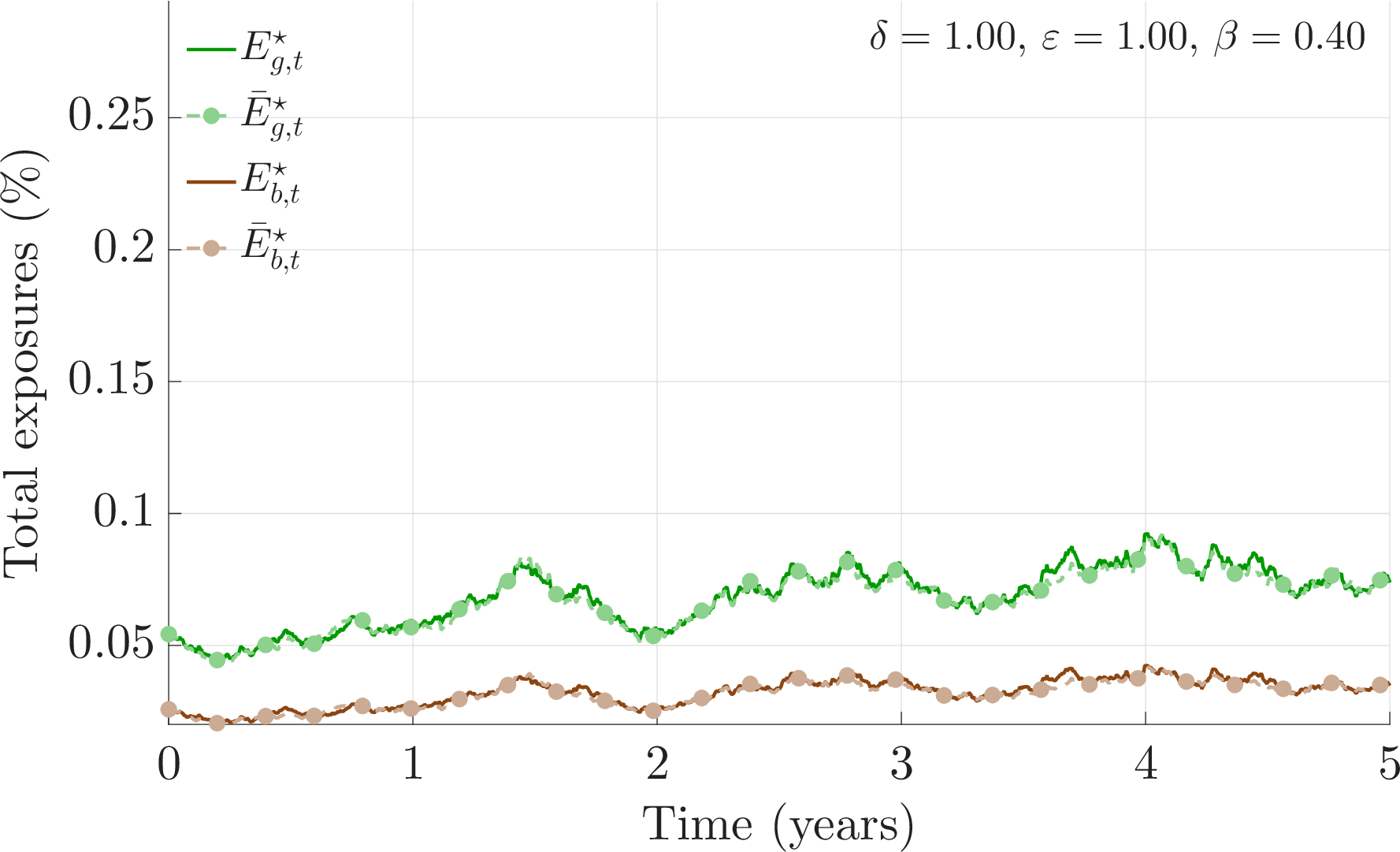} 
\vspace{.3cm}
\hfill 
\includegraphics[width=0.32\linewidth]{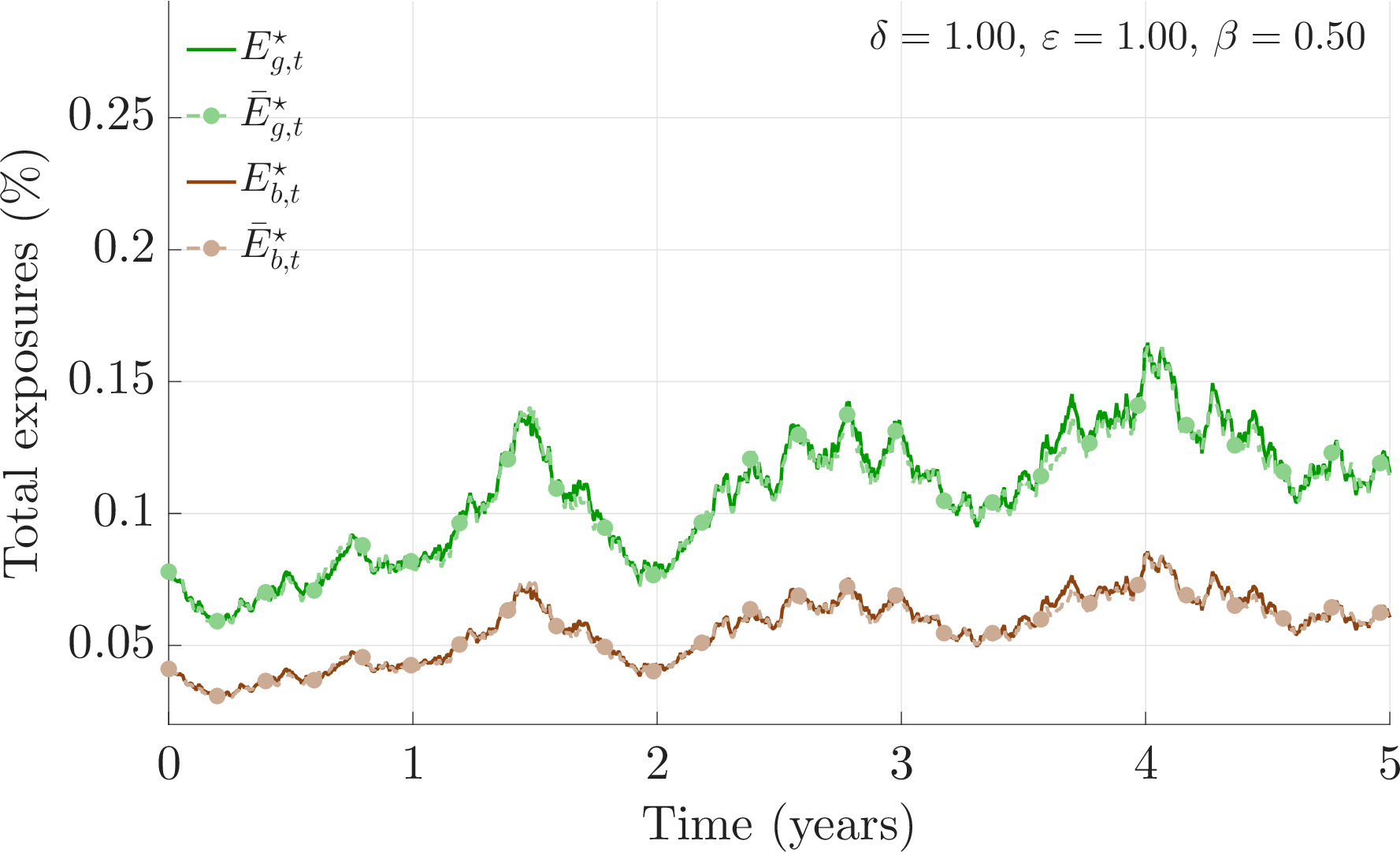}
\vspace{.3cm}
\hfill
\includegraphics[width=0.32\linewidth]{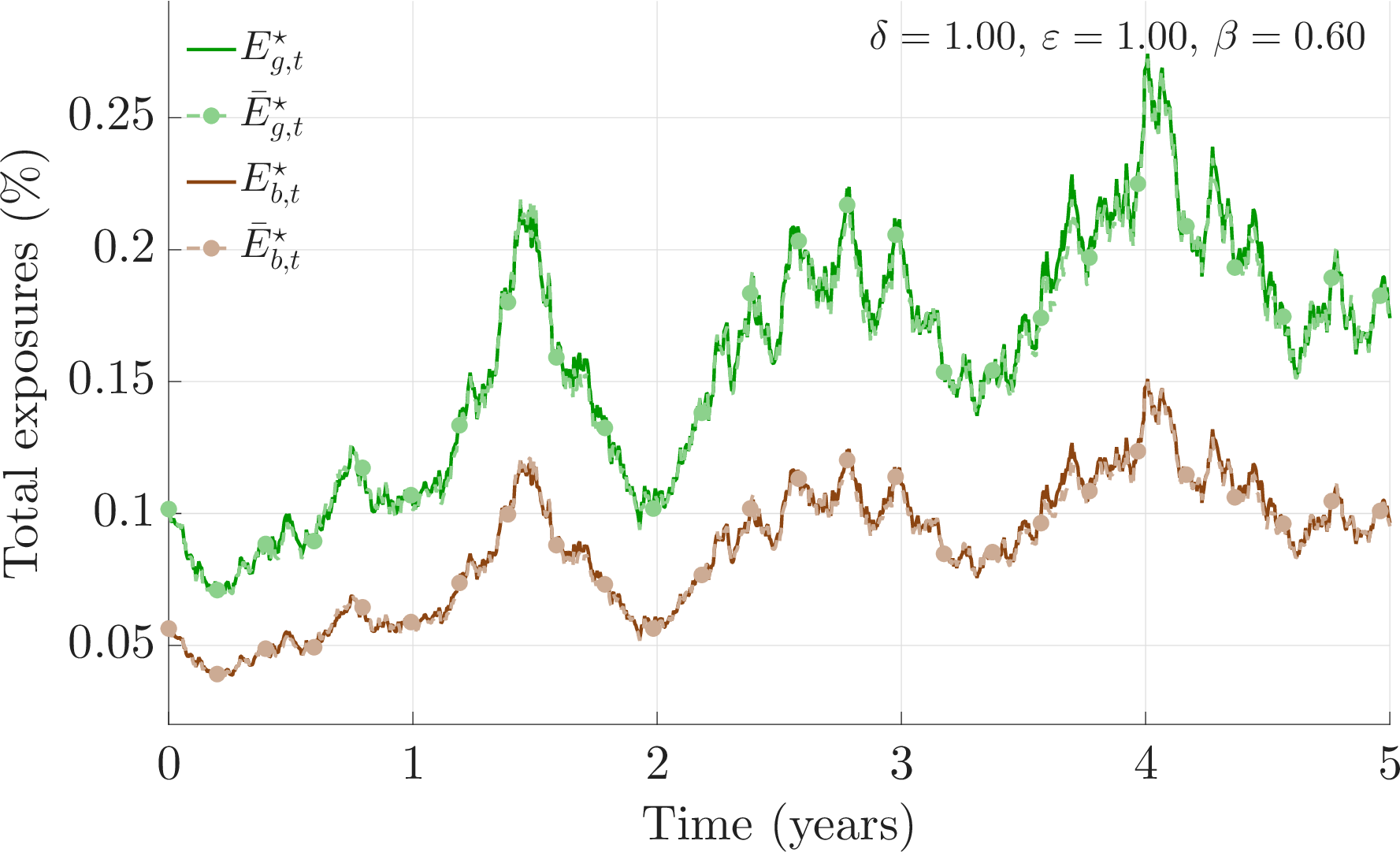}
\caption{Trajectories of the optimal multiplier under full and partial information (top panels) and of the corresponding total optimal exposure to green and brown stocks (bottom panels), for $\delta=1$ and $\varepsilon=1$. Each column corresponds to a different value of $\beta$, with the central column corresponding to the baseline configuration in Table~6.1. In the bottom panels, the solid green (resp. brown) line represents the total optimal exposure to green (resp. brown) stocks under full information, while the dotted light green (resp. light brown) line represents the corresponding exposure under partial information.}\label{fig:DYNAMIC_EXPOSURES_BETA}
\end{figure}
\bibliographystyle{plainnat}

\end{document}